\documentclass{article}
\usepackage[utf8]{inputenc}
\usepackage{amsmath}
\usepackage{amsthm}											
\usepackage{amsfonts}
\usepackage{amssymb}
\usepackage{braket}
\usepackage{caption}
\usepackage{sidecap}
\usepackage{subcaption}
\usepackage{comment}
\usepackage{esint}
\usepackage{bbm}
\usepackage[a4paper, total={6.5in, 10in}]{geometry}
\usepackage{tikz}
\usetikzlibrary{matrix}
\usepackage{graphicx}
\usepackage{mathtools}
\usepackage{pdflscape}
\usepackage{paralist}
\usepackage{xcolor}
\graphicspath{ {images/} }
\usepackage[section]{placeins}
\DeclareMathOperator{\diag}{diag}
\DeclareMathOperator{\Op}{Op}

\DeclareMathOperator{\Tr}{Tr}
\DeclareMathOperator{\tr}{tr}
\DeclareMathOperator{\sgn}{sgn}
\DeclareMathOperator{\supp}{supp}

\newcommand{\norm}[1]{\left\lVert#1\right\rVert}
\newcommand{\vertiii}[1]{\left\lvert\kern-0.25ex\left\lvert\kern-0.25ex\left\lvert #1 \right\rvert\kern-0.25ex\right\rvert\kern-0.25ex\right\rvert}
\newtheorem{proposition}{Proposition}
\newtheorem{lemma}[proposition]{Lemma}
\newtheorem{corollary}[proposition]{Corollary}
\newtheorem{theorem}[proposition]{Theorem}

\numberwithin{proposition}{section}

\newcommand{\s}{\text{span}}

\newcommand{\ssubset}{\subset\joinrel\subset}

\newcommand{\kone}{k_1}
\newcommand{\ktwo}{k_2}

\newcommand{\Hmu}{H^{(\mu)}}
\newcommand{\Hepsnot}{H^{(\eps)}}
\newcommand{\Hmueps}{H^{(\mu,\eps)}}

\newcommand{\Hone}{H^{(1)}}

\newcommand{\anot}{\alpha_0}
\newcommand{\sigmaeps}{a^{(\eps)}}

\newcommand{\sigmamuw}{a^{(\mu,w)}}

\newcommand{\sbz}{a_{\pm,z}}
\newcommand{\Thetaeps}{\Theta^{(\eps)}}
\newcommand{\Thetaepsnot}{\Theta^{(\eps)}_0}

\newcommand{\invariantplus}{I_+}
\newcommand{\invariantminus}{I_-}
\newcommand{\invariantpm}{I_\pm}

\newcommand{\teone}{E_1}
\newcommand{\tetwo}{E_2}

\newcommand{\hnot}{{\rm \bf (H0)}}
\newcommand{\hone}{{\rm \bf (H1)}}

\newcommand{\ofn}{\mathfrak{m}}

\newcommand{\Rfour}{S}

\newcommand{\sbd}{\tilde{C}}
\newcommand{\shp}{T}
\newcommand{\sm}{S^m}
\newcommand{\smeh}{ES^m}

\newcommand{\eps}{\varepsilon}
\newcommand{\Rm}{\mathbb{R}}
\newcommand{\Cm}{\mathbb{C}}
\newcommand{\Zm}{\mathbb{Z}}
\newcommand{\Nm}{\mathbb{N}}

\newcommand{\fm}{\mathfrak m}

\newcommand{\aver}[1]{\langle #1 \rangle}

\newcommand{\fs}{\mathfrak{S}}
\newcommand{\mH}{{\mathcal{H}}}
\newcommand{\mP}{{\mathcal{P}}}

\newcommand{\ess}{\text{ess}}

\newcommand{\cg}{c} 
\newcommand{\cp}{b} 
\newcommand{\kp}{\kappa} 
\newcommand{\sg}{\alpha} 
\newcommand{\sym}{a} 
\newcommand{\ul}{\Lambda} 
\newcommand{\zp}{{z'}} 

\newcommand{\gb}[1]{{{\color{blue}{GB: #1}}}}

\newcommand{\tcc}[1]{{{\color{red}{#1}}}}
\newcommand{\tcbn}[1]{{{\color{black}{#1}}}}

\title{Approximations of interface topological invariants}
\author{Solomon Quinn and Guillaume Bal}

\begin{document}

\maketitle

\begin{abstract}
This paper concerns the asymmetric transport observed along interfaces separating two-dimensional bulk topological insulators modeled by (continuous) differential Hamiltonians and how such asymmetry persists after numerical discretization. 
We first demonstrate that a relevant edge current observable is quantized and robust to perturbations for a large class of elliptic Hamiltonians. We then establish a bulk edge correspondence stating that the observable equals an integer-valued bulk difference invariant depending solely on the bulk phases.
We next show how to extend such results to periodized Hamiltonians amenable to standard numerical discretizations. A form of no-go theorem implies that the asymmetric transport of periodized Hamiltonians necessarily vanishes. We introduce a filtered version of the edge current observable and show that it is approximately stable against perturbations and converges to its quantized limit as the size of the computational domain increases.
To illustrate the theoretical results, we finally present numerical simulations that approximate the infinite domain edge current with high accuracy, and show that it is approximately quantized even in the presence of perturbations. 
\end{abstract}
\numberwithin{equation}{section}
\section{Introduction} \label{sectionIntro}

Robust asymmetric transport along interfaces separating two-dimensional insulating materials has been observed or predicted in many fields of applied science including solid state physics, photonics, and 
geophysics \cite{BH, delplace, sato, Volovik, Witten}.  This robustness of the asymmetric transport to arbitrarily large amounts of perturbation affords a topological interpretation and offers a surprising topological obstruction to the Anderson localization \cite{fouque2007wave} one expects in the standard setting of topologically trivial materials \cite{PS,1,bal2023mathbb}. The main objectives of this paper are: (i) to analyze topological invariants associated to such asymmetric transport for a large class of (continuous) Hamiltonian models; and (ii) to show that such invariants may be accurately approximated by computations on increasingly larger periodic domains.

\medskip
We consider the following framework.
A physical system is modeled by a (single-particle) Hamiltonian $H$ acting on vector- (or spinor-) valued functions of the Euclidean plane that belong to the Hilbert space $\mathcal{H} := L^2 (\mathbb{R}^2) \otimes \mathbb{C}^n$ for $n\geq1$. We label the spatial coordinates $(x,y) \in \mathbb{R}^2$.
We consider $H$ as a self-adjoint differential (or more generally pseudo-differential) operator from its domain of definition $\mathcal{D} (H) \subset \mathcal{H}$ to $\mathcal{H}$.

We consider the situation where a line, say the $x-$axis $\{y=0\}$, separates two different insulators located in the upper and lower half spaces. More precisely, we assume that $H$ is equal to an insulating Hamiltonian $H_+$ when $y\geq y_0\geq0$ and to another insulating Hamiltonian $H_-$ when $y<-y_0$. Here, by equality, we mean equality of the {\em symbols} of the Hamiltonians that will be described more precisely below. The Hamiltonian $H$ therefore describes a transition between {\em bulk} insulating phases. By insulators, we mean Hamiltonians $H_\pm$ whose spectra do not intersect a specific energy interval $[E_1,E_2]$. Physical states generated in this energy interval therefore cannot propagate into the bulks and are confined to the vicinity of the interface $\{y=0\}$. Surprisingly, transport for such states is possible along the interface and in fact necessary when the two bulks are {\em topologically} distinct. Moreover, this asymmetric transport remains robust to arbitrary perturbations of the system that do not change the invariant, such as, in our setting, arbitrary compactly supported (zero-th order) perturbations.

\medskip

In order to describe the asymmetric transport along the interface, we introduce the following {\em interface (or edge) current observable}
\begin{align}\label{eq:sigmaI}
    \sigma_I (H,P,\varphi) := \Tr i [H,P] \varphi ' (H).
\end{align}
The details of this physical observable sometimes referred to as an interface conductivity in analogy with its electronic context, are as follows. The Hamiltonian $H$ is an interface Hamiltonian as above. The function $\varphi'(h)$ is a smooth, non-negative, compactly supported function in the interval $[E_1,E_2]$ that integrates to $1$ (more precisely, $\varphi(E_1)=0$ while $\varphi(E_2)=1$). This function selects states in the energy window $[E_1,E_2]$ where the bulk Hamiltonians $H_\pm$ are gapped. We will refer to it as the {\em density of states}. 

The operator $i[H,P]$ is the {\em current} operator. Here, $P=P(x)$ is a smooth real-valued function equal to $0$ for $x<x_1$ and equal to $1$ for $x>x_2>x_1$. Intuitively, for $x_1\approx x_2$, the observable $P$ indicates the amount of signal located in the half space $x>x_2$. In the Heisenberg picture, the rate of change of such an observable, i.e., the amount of signal transporting from the half space $x<x_2$ to the half space $x>x_2$ per unit time, is precisely given by the current operator $i[H,P]$.

With this in mind, $\sigma_I$ defined in \eqref{eq:sigmaI} is thus the expected value of the current operator $i[H,P]$ for the density of states given by $\varphi'(H)$. We will refer to this quantitative estimate of asymmetric transport as the edge current observable. The main objective of this paper is to analyze this object, and in particular to show that it is quantized for a large class of continuous Hamiltonians and that it may be approximated using Hamiltonians defined on compact domains and therefore amenable to numerical simulations. 

\medskip

The explicit computation of $\sigma_I$ directly from \eqref{eq:sigmaI} remains challenging as it visibly depends on detailed spectral data associated to the operator $H$. A tremendous simplification occurs if we can relate $\sigma_I$ to properties of the bulk Hamiltonians $H_\pm$, whose analysis is often much simpler. This is the role of a general and important principle, the {\em bulk-edge correspondence}, which heuristically states that
\begin{align}\label{eq:BEC}
    2\pi \sigma_I (H,P,\varphi) = I(H_+) - I(H_-) \in \Zm
\end{align}
where $I(H_\pm)$ are invariants associated to the bulk Hamiltonians $H_\pm$. Such a result shows that $2\pi\sigma_I$ is quantized and integer-valued and that when $I(H_+)\not= I(H_-)$, then current flows along the interface irrespective of the details of the transition from $H_-$ to $H_+$ modeled by $H$. 

\medskip

The bulk-edge correspondence (BEC) may intuitively be interpreted as an imbalance between the bulk insulators that is exactly compensated by another imbalance along the interface generating anomalous asymmetric transport \cite{BH,EG,halperin1982quantized,hatsugai}. In spite of its central role in the understanding of topological phases of matter, its derivation remains difficult and often heuristic. 

The analysis of \eqref{eq:sigmaI} is the starting point of mathematical derivations of the BEC by several authors. It was introduced in \cite{SB-2000} following \cite{halperin1982quantized} in the setting of discrete Hamiltonians on the infinite two-dimensional lattice $\Zm^2$. In that setting, $H$ is an edge Hamiltonian defined on $\Zm\times \Nm$ and $\sigma_I(H)$ is then related to a Hall conductivity $\sigma_B(H_+)$ for $H_+$ a bulk Hamiltonian using $K-$theoretic techniques; see \cite{avila2013topological,BKR,PS} for details on such an approach in the discrete and continuous settings algebraically relating invariants of bulk algebras to invariants of edge (half-space) algebras.  The edge current observable is also central in the derivation of the BEC in \cite{Elbau,elgart2005equality,EG} for general discrete Hamiltonians  and based on generalizations of the index of pairs of projections introduced in \cite{ASS90,ASS94}. 

The BEC for general continuous Hamiltonians with periodic coefficients is analyzed in \cite{Drouot}. There, $H_\pm$ model insulating materials and $I(H_\pm)$ are the Chern numbers associated to appropriate dispersion surfaces of their Bloch-Floquet transformations.

\medskip

Our analysis of \eqref{eq:sigmaI} follows the method developed in \cite{3} for first-order {\em elliptic} operators, which we generalize to the much larger class of (essentially arbitrary, smooth) elliptic operators. As in \cite{3}, and unlike \cite{SB-2000,Elbau,Drouot}, the `invariants' $I(H_\pm)$, while defined, are not necessarily stable topological invariants. It is only their difference that is well defined as an integer-valued {\em bulk-difference} invariant \cite{3}. Physically, this reflects the fact that one may more generally defined {\em relative} topological phases rather than {\em absolute} topological phases. This behavior reflects the unbounded nature of the `Brillouin zone' (the whole plane $\Rm^2$ for continuous elliptic operators) associated to unbounded differential operators. 

As in \cite{3}, our approach to the BEC is based on expressing $\sigma_I$ as an integral involving the Schwartz kernel of the resolvent operator $(z-H)^{-1}$ for $z\in\Cm\backslash\Rm$. A number of continuous deformations of the data defining the edge current observable allow us to relate $\sigma_I$ to a generalized winding number formula we will refer to as a Fedosov-H\"ormander (FH) formula. This formula appears in \cite{fedosov1970direct} as an explicit Atiyah-Singer-type expression for the index of Fredholm operator in terms of integrals of its known symbol and was extended to a significantly larger class of operators in \cite[Chapter 19]{Hormander}. 

The same formula also appears in \cite{Volovik,EG} as the winding number associated to a Green's function with imaginary frequency.  Our approach shares a number of techniques used in the derivation of the BEC in \cite{EG} as well as in the aforementioned results of \cite{fedosov1970direct,Hormander}: operators written in a Weyl formalism, semiclassical expansions, and continuous deformations. These techniques are also central in the BEC derivation of \cite{Drouot}. 

\medskip

The second main contribution of this paper is an analysis of the edge current for compactly supported approximations of the Hamiltonians so that they may be amenable to numerical simulations. In doing so, we face a number of obstructions. 

Any compact approximation of $\Rm^2$, such as for instance approximating $H$ by $H_L$ on $(-\pi L,\pi L)^2$ with, say, periodic boundary conditions, results in a necessarily trivial topological phase. There are several reasons for this phenomenon: (i) any transition from $H_-$ to $H_+$ along $y\approx0$ will be augmented by a transition from $H_+$ to $H_-$ along $y\approx L$ with opposite topology; (ii) any observable $P$ modeling transport along the hyperplane $x=x_1$ will be augmented with another transition with opposite current; (iii) finally, the edge current observable $\sigma_I$ is intimately connected to the absolutely continuous spectrum of $H$ and related spectral flows (see, e.g., \cite{2,3}), which is not preserved by $H_L$, whose spectrum is purely discrete.

The best we can hope for is therefore to define a {\em modified} and {\em approximate} notion of filtered edge current $\tilde\sigma_I$ on a domain $(-L,L)$ and expect convergence of $\tilde\sigma_I$ to $\sigma_I$ as $L\to\infty$. This is precisely what we obtain for a large class of differential operators. We will show approximate stability properties for $\tilde\sigma_I$ generalizing the exact ones obtained for $\sigma_I$. This lack of exact stability requires us to use detailed spectral information on $H$ and its approximation $H_L$ with multiple appeals to the Courant-Fischer min-max theorem. This will allow us to obtain a faster-than-algebraic convergence as $L\to\infty$.

\medskip

We should mention that a number of works address the computation of bulk topological indices based on the notion of `spectral localizers' \cite{Loring,LS19, LS20,Prodan,Schulz-Baldes}. Under appropriate assumption that bulk Hamiltonians on compact domain are sufficiently similar to bulk Hamiltonians on unbounded domains, the bulk invariant is characterized by the signature of an appropriate finite dimensional matrix. Eigenvalues and edge states for discrete and continuous models are also computed numerically in \cite{TWL} using a method that avoids artificial Dirichlet boundary conditions by appropriately utilizing the resolvent of the Hamiltonian.

Our work on edge current and filtered edge current applies to large classes of Hamiltonians whose detailed spectral decompositions may not be known explicitly. When enough spectral information on an unperturbed operator is available, an interface scattering theory may be implemented and solved numerically using high accuracy integral formulations; see \cite{bal2022integral,bal2023asymmetric,bal2023mathbb,chen2023scattering} for recent results on this problem.

The edge current \eqref{eq:sigmaI} and its associated asymmetric transport are purely spectral properties of the Hamiltonian $H$. For an analysis of the propagation of wavepackets along curved interfaces generalizing the flat case $\{y=0\}$, see the recent works \cite{bal2024semiclassical,bal2023magnetic,bal2023edge,D22}.

\medskip

Finally we stress that the results of this paper apply to {\em elliptic} operators (see next section for a precise definition). This class includes Dirac-type operators (see also \cite{2} for a direct analysis), BdG Hamiltonians as they appear in the modeling of topological superconductors \cite{sato,Volovik}, but also models of Floquet topological insulators \cite{BM21} and bilayer graphene \cite{bal2023mathematical}, where the BEC we derive in this paper is applied to compute the edge current.

The method does not apply when our notion of ellipticity fails, for instance for the magnetic Schr\"odinger equations that model the integer quantum Hall effect \cite{ASS90,ASS94,bellissard1994noncommutative}, or to partial-differential models with micro-structures \cite{Drouot,FLW-ES-2015}. An interesting border-line case is a Hamiltonian modeling atmospheric mass transport along the equator \cite{delplace,souslov}, in which the BEC (barely) fails \cite{3,GJT}. We will see that ellipticity and hence the BEC may be restored after appropriate regularization of this model.

\medskip

The rest of the paper is structured as follows. Our main results on the analysis of $\sigma_I$ are presented in section \ref{sec:stab}. After some preliminaries necessary to define elliptic Hamiltonians, we show in section \ref{sec:stability} that $i[H,P]\varphi'(H)$ is indeed a trace-class operator ensuring that $\sigma_I$ in \eqref{eq:sigmaI} is well-defined and that it is stable against continuous deformations of $H$, $P$, and $\varphi(h)$. A similar stability result for $\sigma_I$ is obtained in \cite{3,bal2023topological} by relating the invariant to the index of a Fredholm operator, which is known to be invariant under continuous deformations \cite{Hormander}. In this paper, we instead prove its stability directly, as is done for instance in \cite{Drouot, Elbau, GP,PS}. The main reason for this choice is its generalization to the discrete case, where exact stability no longer holds and the identification to the index of a Fredholm operator is no longer available. The bulk-edge correspondence based on establishing a Fedosov-H\"ormander formula for $\sigma_I$ is presented in section \ref{sec:FHformula}. 
Section \ref{sec:integral} provides a number of cases where the computation of the invariant simplifies.
A number of typical examples of application of the theory are treated in section \ref{sec:applications}.

The analysis of operators $H_L$ posed on bounded domains is taken up in section \ref{sec:periodic}. After a presentation of the class of differential operators $H_L$ we consider, we define an appropriate notion $\tilde\sigma_I$ of filtered periodic edge current, show that it is a spectrally accurate approximation of $\sigma_I$ as $L\to\infty$, and finally generalize the stability results of section \ref{sec:stab} to the periodic setting. 

A number of numerical simulations illustrating the computation of $\sigma_I$ via $\tilde\sigma_I$ when $H$ is a Dirac operator with massive domain wall or a model of a $p-$wave superconductor are presented in section \ref{sec:numerics}. 

The technical derivations of the results of section \ref{sec:stab} are given in section \ref{sec:pfbec} while those of the results of section \ref{sec:periodic} are given in section \ref{sec:pfs_periodic}. Relevant material on pseudo-differential operators and semiclassical calculus as well as some proofs are postponed to the appendix. 

\section{Main results on current observable, stability, and BEC} 
\label{sec:stab}
%


This section presents our main results on the current edge observable $\sigma_I$ for continuous Hamiltonians on the Euclidean plane.
\subsection{Elliptic Hamiltonians and Weyl symbols}
We first describe the class of Hamiltonians $H$ considered in this paper and then analyze the stability properties of $\sigma_I$ defined in \eqref{eq:sigmaI} with respect to perturbations of $H$, $P$, and $\varphi$. A convenient formalism to describe functions of Hamiltonians such as $\varphi'(H)$ and obtain criteria ensuring that operators such as $i[H,P]\varphi'(H)$ are trace-class is the pseudo-differential framework and calculus. We mostly follow the presentation in \cite{DS}, to which we refer for additional information and context.

Rather than analyzing $H$ directly, we write its {\em Weyl quantization} $H=\Op(a)$ as
\begin{align}\label{eq:weylquant}
    \Op (a) \psi (x) :=
    \frac{1}{(2\pi)^d} \int_{\mathbb{R}^{2d}}
    e^{i(x-x')\cdot \xi}
    a(\frac{x+x'}{2}, \xi) \psi (x') dx' d\xi,
    \qquad
    \psi \in \mathcal{S} (\mathbb{R}^d) \otimes \mathbb{C}^n.
\end{align}
Here, $a(x,\xi)$ is the (Weyl) {\em symbol} of $H$ and $\mathcal{S} (\mathbb{R}^d)$ is the space of Schwartz functions on $\Rm^d$. While $d=2$ in this paper, we use physical variables $x$ and $x'$ and dual variable $\xi$ in $\Rm^d$ in this short summary of the calculus we need. Any linear operator admits a distribution-valued Schwartz kernel and upon taking appropriate Fourier transforms, which are allowed for a large class of kernels, may be written as above in terms of a Weyl symbol \cite{DS}. In this paper, the symbol $a(x,\xi)$ will always be Hermitian-valued in ${\mathcal M}_n(\Cm)$ and hence diagonalizable for each $(x,\xi)\in\Rm^{2d}$. The advantage of introducing $a(x,\xi)$ is that a calculus exists to make sense of operators such as $\varphi'(H)$ and $i[H,P]\varphi'(H)$. This calculus requires us to introduce a number of spaces of symbols. We need two notions of symbols to state our main results. The first one is the standard space $S^m(\Rm^{d}\times\Rm^d)$ of smooth symbols $a(x,\xi)$ such that 
\begin{align}\label{eq:Sm}
 |\partial^\alpha_x\partial^\beta_\xi a(x,\xi)| \leq C_{\alpha,\beta} |\xi|^{m-|\beta|},\quad \forall (x,\xi)\in \Rm^d\times\Rm^d,
\end{align}
for any multi-indices $\alpha\in \Nm^d$ and $\beta\in\Nm^d$ with $|\beta|=\beta_1+\ldots+\beta_d$.
Differential operators of order at most $m$ have symbols that belong to $S^m$. A number of illustrating examples are presented in section \ref{sec:applications} below.

The second notion of spaces of symbol is denoted by $a\in S(\fm)$ where $\fm$ is an {\em order} function, i.e. a function such that for all $X,Y$ in $\Rm^{2d}$, there are $C>0$ and $N>0$ such that $\fm(X)\leq C\aver{X-Y}^N \fm(Y)$. Throughout the paper, we define the Japanese bracket as $\aver{X}=\sqrt{1+|X|^2}$. We then say that $a\in S(\fm)$ when for each multi-index $\alpha\in \Nm^{2d}$, we have
\begin{align}\label{eq:Sfm}
 |\partial^\alpha a(X)| \leq C_{\alpha} \fm(X),\quad \forall X\in \Rm^{2d}.
\end{align}
For $X'\in\Rm^p$ with $0\leq p\leq 2d$ any subset of the $2d$ variables representing $X$, we observe that $\aver{X'}^s$ is an order function for any $s\in\Rm$. These are the main order functions we will be using in particular to prove that operators are trace-class. We also use the standard notation $\Op a\in \Op S(\fm)$ to mean that the operator $\Op a$ has symbol $a\in S(\fm)$.

The last important ingredient we need to introduce is that of {\em ellipticity}. We say that $a\in S^m$ is {\em elliptic}, and then use the notation $a\in ES^m$, when $\sym$ is Hermitian-valued and 
satisfies the growth condition
\begin{align}\label{eq:ellip}
  a_{\rm min}(x,\xi) \geq C \aver{\xi}^m-1
\end{align}
for some positive constant $C>0$ independent of $(x,\xi)\in\Rm^{2d}$, where $a_{\rm min}(x,\xi)$ denotes the smallest singular value of $a(x,\xi)$. Note that this imposes a constraint on $a(x,\xi)$ for all $|\xi|$ sufficiently large. 

The ellipticity condition \eqref{eq:ellip} implies in particular that $H$ is a {\em self-adjoint} operator with domain of definition the standard Sobolev space $\mathcal{H}^m$ of square-integrable derivatives up to order $m$ defined in \eqref{eq:Hm} below and that the resolvent operator $(i+H)^{-1}$ is a bounded operator from $\mH$ to $\mH^m$; see \cite{Bony,Hormander} as well as Proposition \ref{prop:resh} below. By the functional version of the spectral theorem, operators such as $\phi(H)$ for $\phi$ smooth, real-valued, and compactly supported are therefore defined as bounded operators on $\mH$.


%
\subsection{Ellipticity condition and phase transition}\label{sec:H1}
For the rest of the paper, we assume that $d=2$ with spatial coordinates $x$ replaced by $(x,y) \in \mathbb{R}^2$ and corresponding dual variables $\xi$ replaced by $(\xi, \zeta) \in \mathbb{R}^2$.
We consider operators $H=\Op a$ with $a=a(x,y,\xi,\zeta)\in ES^m=ES^m(\Rm^4)$ thus satisfying \eqref{eq:ellip} with $m>0$.

In order for $H$ to model a transition between bulk phases, we assume that the two bulk Hamiltonians $H_\pm=\Op \sym_\pm$ have constant coefficients $\sym_\pm=\sym_\pm(\xi,\zeta)$ and are gapped in the interval $[E_1,E_2]\subset\Rm$ for $E_1<E_2$. Since the symbols of $H_\pm$ are assumed independent of $(x,y)$, the interval is a spectral gap of $H_\pm$ if and only if it is a spectral gap of the Hermitian matrices $\sym_\pm(\xi,\zeta)$ for all $(\xi,\zeta)\in\Rm^2$.
\\[2mm]
{\bf (H1) }
Let $H = \Op (a)$ while $H_\pm=\Op (a_\pm)$ for $a,a_\pm \in \smeh$ and $\sym_\pm$ independent of position $(x,y)$. We assume that the intersection of the spectrum of $H_\pm$ with the interval $[E_1,E_2]$ is empty and that there exists $y_0>0$ such that $a(x,y,\xi,\zeta)=a_+(\xi,\zeta)$ when $y>y_0$ while $a(x,y,\xi,\zeta)=a_-(\xi,\zeta)$ when $y<-y_0$.
\\[2mm]
\subsection{Stability of the edge current observable}\label{sec:stability}

We are now in a position to state several properties shared by the edge current observable \eqref{eq:sigmaI}.
While all results in this section hold under {\bf (H1)} defined in section \ref{sec:H1}, some results hold under the more general assumption:
\\[2mm]
{\bf (H0) }
Let $H = \Op (a)$ with $a\in ES^m$ and $[E_1,E_2]$ be fixed. We assume that for all $\Phi \in \mathcal{C}^\infty_c (E_1, E_2)$, 
then $\Phi(H)\in \Op S(\aver{y,\xi,\zeta}^{-\infty})$.
\\[2mm]
Here $S(\fm^{-\infty})=\cap_{p\geq0} S(\fm^{-p})$ is the space of symbols that decays faster than algebraically in $\fm$. We prove in Proposition \ref{trclass} that {\bf (H1)} implies {\bf (H0)}.

\medskip

To define the remaining elements in the edge current observable $\sigma_I$ in \eqref{eq:sigmaI}, it is convenient to define the notion of {\em switch function}.  We denote by $\fs(c_1,c_2;\lambda_1,\lambda_2)$ the set of smooth real-valued functions $f$ on $\mathbb{R}$ for which there exists $\delta>0$ such that $f(x)=c_1$ for all $x\leq\lambda_1+\delta$ and $f(x)=c_2$ for all $x\geq\lambda_2-\delta$. The union over $\lambda_1<\lambda_2$ is denoted by $\fs(c_1,c_2)$.

Our main assumptions on $P$ and $\varphi$ are that $P\in \fs(0,1)$ and that $\varphi\in \fs(0,1;E_1,E_2)$ for $E_1<E_2$. The interval $[E_1,E_2]$ represents the gapped  energy range for both bulk Hamiltonians $H_\pm$. The function $\varphi'(h)$ is thus supported in the interval $[E_1,E_2]$ and integrates to $1$. If $\varphi'\geq0$, then we may interpret it as a density of states selecting modes that are not allowed to propagate into the bulk phases and hence may only transport along the interface $y\approx0$.

With this notation in place, we state our main results on the definition and stability of $\sigma_I$ in \eqref{eq:sigmaI}.
\begin{lemma}\label{lemma:tc}
    Suppose $H$ satisfies \hnot. Let $P(x) = P \in \fs (0,1)$ and $\varphi \in \fs(0,1;E_1,E_2)$. Then $[H,P] \varphi'(H)$ is trace-class.
\end{lemma}
With the above lemma in hand, we are ready to state the main result of this section, which establishes the stability of the interface current observable with respect to a large class of perturbations.
\begin{theorem}\label{thm:stabilityall}
    Suppose $H$ satisfies \hnot\ and that $P(x) = P \in \fs (0,1)$ and $\varphi \in \fs (0,1;E_1,E_2)$. Then
    \begin{enumerate}
        \item The edge observable $\sigma_I (H,P,\varphi)$ is independent of $P(x) = P \in \fs (0,1)$ and $\varphi \in \fs (0,1;E_1,E_2)$.   \label{it:P}
        \item If $W$ is symmetric with $W \in \Op (S^m \cap S (\aver{\xi, \zeta}^{m-\delta} \aver{x,y}^{-\delta}))$ for $\delta > 0$, then $\sigma_I (H+W) = \sigma_I (H)$. \label{it:compact}
        \item If $W$ is symmetric with $W \in \Op (S^m)$, then $\sigma_I (H+\mu W) = \sigma_I (H)$ for all $\mu>0$ sufficiently small.\label{it:bounded}
    \end{enumerate}
    Suppose further that $H$ satisfies \hone. Then:
    \begin{itemize}
        \item[4.] Define $H_h := \Op_h (\sym)$, then $\sigma_I (H_h)$ is independent of $h \in (0,1]$. \label{it:h}
        \item[5.] Fix $x_0 \in \mathbb{R}$ and define $a_{x_0} (x,y,\xi,\zeta) := a (x_0,y,\xi,\zeta)$ and $H_{x_0} := \Op (a_{x_0})$. Then $\sigma_I (H_{x_0}, P, \varphi) = \sigma_I (H, P, \varphi)$. \label{it:x}
    \end{itemize}
\end{theorem}
Since under hypothesis \hnot, $\sigma_I$ is independent of $P$ and $\varphi$, we will henceforth write $\sigma_I (H) = \sigma_I (H,P,\varphi)$, leaving the (in)dependence on $P$ and $\varphi$ implicit. 

The second and third results prove the {\em stability} of $\sigma_I$ under relatively compact and small perturbations, respectively. The second result is in some sense surprising. As we indicated in the introduction, this means for most models that an arbitrarily large but finite slab of random perturbations along the interfaces $\{y=0\}$ does not change the edge current. In topologically trivial materials, transmission is exponentially suppressed by the random perturbations \cite{fouque2007wave}. The above stability result shows that this is no longer the case for topologically nontrivial materials; see \cite{PS,1,bal2023mathbb} as well.

The fourth result requires additional definitions. Here $\Op_h$ means semi-classical quantization, generalizing the Weyl quantization in \eqref{eq:weylquant}. We refer to \eqref{eq:weylquanth} in the appendix for its explicit expression. 
Our proof of the fourth result requires the full hypothesis \hone. This independence with respect to semiclassical rescaling as $0<h\to 0$ is an essential ingredient in our subsequent proof of the bulk-edge correspondence. The reason for its importance is that while objects such as $\varphi'(H)$ may be defined, their computation remains difficult. It turns out that its expression is amenable to a useful series in powers of $h$ in the semiclassical limit. Such rescaling is an essential ingredient in the index formulas in \cite{fedosov1970direct,Hormander} as well as in the BEC derivations in \cite{3,Drouot,EG}. Note that the fourth result cannot possibly hold for magnetic Schr\"odinger or magnetic Dirac equations since rescaling in $h$ inevitably involves crossings of Landau levels and hence cannot lead to stable edge current.

The fifth and final result shows that the spatial dependence in $x$ of the symbol $a$ is irrelevant in the computation of the edge current. 

The proofs of Lemma \ref{lemma:tc} and Theorem \ref{thm:stabilityall} are given in section \ref{sec:pfbec}.

While these result show the stability of $\sigma_I$ with respect to a number of (continuous) deformations in $H$, $P$, and $\varphi$, they do not show that $\sigma_I$ takes quantized values or provide a means to compute it explicitly. This is the role of the BEC we now turn to.

\subsection{Fedosov-H\"ormander formula and bulk-edge correspondence} \label{sec:FHformula}
%

Our next objective is to recast the edge current observable $\sigma_I(H)$ as an integral involving only the symbol $\sym$ of $H=\Op \sym$. Such integrals are significantly easier to estimate than the defining trace formula in \eqref{eq:sigmaI}. We will also see that this integral may be modified to yield a {\em bulk-edge correspondence}, in the sense that the integral involves the symbol $\sym(x,y,\xi,\zeta)$ restricted to $y\geq y_0$ and $y<-y_0$, i.e., only the symbols $\sym_\pm$. This integral may also be recast as a Fedosov-H\"ormander formula involving a complex-frequency Green's function as introduced in \cite{Volovik,EG}.

Our main result is the following:
\begin{theorem} \label{thm:main}
Suppose $H = \Op (\tilde{\sym})$ satisfies \hone,
$P(x)=P \in \fs(0,1)$ and $\varphi \in \fs(0,1;E_1,E_2)$.
Let $\alpha \in (E_1, E_2)$, and let
$R \subset \mathbb{R}^3$ be 
bounded with 
a piecewise smooth boundary $\partial R$.
Fix $x_0 \in \mathbb{R}$, 
and define $\sym(y,\xi,\zeta) := \tilde{\sym} (x_0, y, \xi, \zeta)$.
Assume $R$ contains
all points $(y, \xi, \zeta)$ where 
$\sym(y,\xi,\zeta)$ has an eigenvalue of $\alpha$.
Then
defining $z := \alpha + i\omega$ 
and $\sym_z := z-\sym$, we have that
\begin{align}\label{eq:sigmaI1}
    \sigma_{I}(H,P,\varphi) = \frac{i}{16\pi^3} \int_{\partial R} \int_{-\infty}^{+\infty} \Theta_z d\omega d\Sigma, \qquad \Theta_z := \tr \eps_{ijk}\sym_z^{-1} \partial_i \sym_z \sym_z^{-1} \partial_j \sym_z \sym_z^{-1}\nu_k
\end{align}
where $\nu$ is the unit vector (outwardly) normal to $\partial R$, $d\Sigma= d\Sigma(y,\xi, \zeta)$ is the Euclidean measure on $\partial R$, and $\eps_{ijk}$ is the anti-symmetric tensor with $\eps_{123}=1$ and the variables identified by $(1,2,3)=(y,\xi,\zeta)$.
\end{theorem}
The domain $R$ exists thanks to hypothesis \hone\ on $\sym$. Indeed, $\alpha\in (E_1,E_2)$ cannot be an eigenvalue of $\sym(y,\xi,\zeta)$ for $|y|$ large while all singular values of $\sym$ tend to infinity as $|(\xi,\zeta)|$ tends to infinity.

The proof of the Theorem is postponed to section \ref{subsec:main}. It involves three main steps: (i) write the trace in \eqref{eq:sigmaI} as an integral involving the kernel of the resolvent operator $(z-H)^{-1}$; (ii) use the invariance of $\sigma_I$ with respect to semiclassical rescaling to write it using the resolvent operator $(z-H_h)^{-1}$ and use semiclassical expansions to identify $O(h^0)$ term; and finally (iii) use analytic properties and the Stokes formula to recast the traces as the right-hand side in \eqref{eq:sigmaI1}.

As we just mentioned, the right-hand side of \eqref{eq:sigmaI1} is stable against a number of transformations. Two additional transformations allow us to prove a bulk-edge correspondence and derive a Fedosov-H\"ormander formula. We state them in the following two corollaries.
\begin{corollary}[Bulk-Edge Correspondence]\label{cor:bic}
Suppose $H = \Op (\sym)$ satisfies \hone,
    $P(x)=P \in \fs(0,1)$ and $\varphi \in \fs(0,1;E_1,E_2)$. 
    Let $\alpha \in (E_1,E_2)$ and define $\sbz := z-\sym_{\pm}$ with $z=\alpha+i\omega$.
    Then
    \begin{align}\label{eq:BEC2}
        2\pi \sigma_I (H,P,\varphi) = \invariantplus - \invariantminus,\qquad \invariantpm = \frac{i}{8\pi^2} \int_{\mathbb{R}^3} \tr [\sbz^{-1} \partial_\xi \sbz, \sbz^{-1} \partial_\zeta \sbz] \sbz^{-1} d\omega d\xi d\zeta.
    \end{align}
\end{corollary}
We thus observe that $2\pi\sigma_I$ may be written as a difference of integrals involving only the bulk phase symbols $a_\pm$. It turns out that $I_+-I_-$ is always integral-valued as shown in \cite{3}. Indeed, the above right-hand-side may be transformed into a {\em bulk-difference} Chern number after integration in the variable $\omega$. The integrals $I_\pm$ separately are {\em not} integral-valued (and take the value $\pm\frac12$ for Dirac operators for instance). As we already noted, this shows that phase differences are more generally defined than absolute phases in the context of elliptic (pseudo-)differential Hamiltonians.

The second corollary is:
\begin{corollary}[Fedosov-H\"ormander formula]\label{cor:FH}
    Suppose $H = \Op (\tilde{\sym})$ satisfies \hone, 
    $P(x)=P \in \fs(0,1)$ and $\varphi \in \fs(0,1;E_1,E_2)$. 
Let $\alpha \in (E_1, E_2)$, and let $\Rfour \subset \mathbb{R}^4$ be bounded with piecewise smooth boundary $\partial \Rfour$.
Fix $x_0 \in \mathbb{R}$ and define $\sym(y,\xi,\zeta) := \tilde{\sym} (x_0, y, \xi, \zeta)$.
Assume $\Rfour$ contains all points $(0, y, \xi, \zeta)$ where 
$\sym(y,\xi,\zeta)$ has an eigenvalue of $\alpha$.
Then
for $z := \alpha + i\omega$ 
and $\sym_z := z-\sym$, we have that
\begin{align} \label{FH}
    2 \pi \sigma_I (H,P,\varphi) = 
    \frac{1}{24\pi^2} \int_{\partial \Rfour}
    u \cdot \nu d\Sigma_3, \qquad u_l := \tr \epsilon_{ijkl} \partial_i \sym_z \sym_z^{-1} \partial_j \sym_z \sym_z^{-1} \partial_k \sym_z \sym_z^{-1},
\end{align}
where
$\nu$ is the outward unit normal 
to $\partial \Rfour$, $d\Sigma_3= d\Sigma_3 (\omega, y,\xi, \zeta)$ is the Euclidean measure on $\partial \Rfour$, and
$\epsilon_{ijkl}$ the anti-symmetric tensor with $\epsilon_{1234} =1$ and 
the variables identified by $(1,2,3,4) = (\omega,y, \xi, \zeta)$. 
\end{corollary}
The above integral may be recast in a more geometric form as it appears in \cite{EG,Hormander,3}:
\begin{align}\label{eq:FH2}
    \int_{\partial \Rfour}
    u \cdot \nu d\Sigma_3 = \int_{\partial \Rfour}(\sym_z^{-1}d\sym_z)^{\wedge 3}.
\end{align}

Since $\sigma_I$ does not depend on the $x-$dependence of $\sym(x,y,\xi,\zeta)$, we decided to write the integral in the last three results in terms of the reduced symbol $a(x_0,y,\xi,\zeta)$. In fact, we could have defined more generally $a_z(\omega,y,\xi,\zeta)=\alpha+i\omega -a(\omega,y,\xi,\zeta)$ and \eqref{eq:sigmaI1}-\eqref{eq:BEC2}-\eqref{FH} would still hold. This reflects the fact that the right-hand side of \eqref{eq:sigmaI1} is a homotopy invariant \cite{bal2023topological,Hormander} (invariant with respect to continuous deformations of $a$ that remain elliptic). This also shows that $2\pi\sigma_I\in\Zm$, something that is not easily apparent in any of the above integrals. 

The derivations of the corollaries from the main theorem are similar to those in \cite{3}. For concreteness and since the frameworks are different, we now provide brief proofs.
\begin{proof}[Proof of Corollary \ref{cor:bic}]
Fix $y_0>0$ sufficiently large, and define $R_M := (-y_0, y_0) \times (-M,M)^2$ for all $M>0$.
Recalling the definition of $\Theta_z$ in \eqref{eq:sigmaI1},
observe that $|\Theta_z| \le C \aver{\omega}^{-3}$ and $|\Theta_z| \le C \aver{\xi,\zeta}^{-m-2}$ uniformly in $(y,\xi,\zeta) \in \partial R_M$ and $M>0$ sufficiently large, hence $|\Theta_z| \le C \aver{\omega}^{-3/2} \aver{\xi,\zeta}^{-\frac{m+2}{2}}$
by interpolation. It follows that $\int_{-\infty}^{+\infty} |\Theta_z|d\omega \le C\aver{\xi,\zeta}^{-\frac{m+2}{2}}$.
Therefore, taking $R:=R_M$ in Theorem \ref{thm:main} and
sending $M \rightarrow \infty$, we see that
    the contributions to $\sigma_{I}$ in \eqref{eq:sigmaI1} from the sides of $\partial R_M$ with normal vector in the $\xi$ and $\zeta$ directions vanish. Indeed, the area of these surfaces is proportional to $M$, with the maximum of the integrand bounded by $CM^{-1-m/2}$.
    Thus we are left with integrals over the sides corresponding to $y = \pm y_0$, over which $\sym(y,\xi,\zeta) = \sym_\pm (\xi, \zeta)$.
    As a consequence, $\sigma_{I} = \frac{1}{2\pi} (\invariantplus-\invariantminus)$, and the proof is complete.
\end{proof}

\begin{proof}[Proof of Corollary \ref{cor:FH}]
     By Theorem \ref{thm:main},
     \begin{align}\label{eq:sigmaM}
         \sigma_I = \frac{i}{16\pi^3} \lim_{M \rightarrow \infty} \int_{\partial R} \int_{-M}^{M} \Theta_z d\omega d\Sigma.
     \end{align}
     Since $|\Theta_z| \le C \aver{\omega}^{-3}$, it follows that if $z=\alpha \pm i M$ and 
     $i,j,k \in \{\omega, y, \xi, \zeta\}$, then 
         \begin{align}\label{eq:M0}
    \int_{R} |\tr \partial_i \sym_z \sym_{z}^{-1} \partial_j \sym_z \sym_z^{-1} \partial_k \sym_z \sym_z^{-1}| d R_3
    \longrightarrow 0
\end{align}
as $M \rightarrow \infty$. 
Since $\partial_\omega \sym= i$, we know that
\begin{align}\label{eq:thetaijk}
    \Theta_z = -i \tr \eps_{ijk}\partial_\omega \sym\sym_z^{-1} \partial_i \sym_z \sym_z^{-1} \partial_j \sym_z \sym_z^{-1}\nu_k.
\end{align}
It then follows from cyclicity of the trace that
\begin{align}\label{eq:sigmaMinfty}
         \sigma_I = \frac{1}{48\pi^3} \lim_{M \rightarrow \infty} \int_{\partial \Rfour_M} u \cdot \nu d\Sigma_3,
\end{align}
where $\Rfour_M := [-M,M] \times R$.
Indeed, the integral over $[-M,M] \times \partial R \subset \partial \Rfour_M$ is precisely the integral on the right-hand side of \eqref{eq:sigmaM}, while the integral over the rest of $\partial \Rfour_M$ vanishes in the $M \rightarrow \infty$ limit by \eqref{eq:M0}. The factor of $3$ adjustment in \eqref{eq:sigmaMinfty} compared to \eqref{eq:sigmaM} is justified by the rank-four tensor $\epsilon_{ijkl}$, which causes each term on the right-hand side of \eqref{eq:thetaijk} to appear three times in \eqref{eq:sigmaMinfty}.

Thanks to \hone, all singularities of $\sym_z^{-1}$ lie in $\Rfour_M \cap \Rfour$ when $M$ is sufficiently large.
Hence $\sym_z^{-1}$ is well-defined in 
$\Rfour_M \Delta \Rfour := (\Rfour_M \cup \Rfour) \setminus (\Rfour \cap\Rfour_M)$.
It follows that for $(\omega, y, \xi, \zeta) \in \Rfour_M \Delta \Rfour$, we have
\begin{align*}
    \nabla \cdot u &= 
    \epsilon_{ijkl} \tr
    \partial_l
    (
    \partial_i \sym_z \sym_z^{-1} \partial_j \sym_z \sym_z^{-1} \partial_k \sym_z \sym_z^{-1})\\
    &=
    -\epsilon_{ijkl} \tr
    \Big(
    \partial_i \sym_z \sym_z^{-1}
    \partial_l \sym_z \sym_z^{-1}
    \partial_j \sym_z \sym_z^{-1} \partial_k \sym_z \sym_z^{-1}
    +
    \partial_i \sym_z \sym_z^{-1}
    \partial_j \sym_z \sym_z^{-1}
    \partial_l \sym_z \sym_z^{-1} \partial_k \sym_z \sym_z^{-1}\\
    &\qquad +
    \frac{1}{2}
    \partial_i \sym_z \sym_z^{-1}
    \partial_j \sym_z \sym_z^{-1}
    \partial_k \sym_z \sym_z^{-1} \partial_l \sym_z \sym_z^{-1}
    +
    \frac{1}{2}
    \partial_l \sym_z \sym_z^{-1}
    \partial_i \sym_z \sym_z^{-1}
    \partial_j \sym_z \sym_z^{-1} \partial_k \sym_z \sym_z^{-1}\Big)
    = 0,
\end{align*}
where we have used cyclicity of the trace to justify the second equality, and the antisymmetry of $\epsilon_{ijkl}$ for the last equality (the third term cancels half the second term, the fourth term cancels half the first term, half the second term cancels half the first term).
Using the Divergence Theorem, 
we 
thus replace $\partial \Rfour_M$ in \eqref{eq:sigmaMinfty} by $\partial \Rfour$ and the proof is complete.
\end{proof}

\subsection{Simplifying the computation of the integrals} \label{sec:integral}
The integrals \eqref{eq:sigmaI1}-\eqref{eq:BEC2}-\eqref{FH}, while significantly simpler than the trace \eqref{eq:sigmaI} both conceptually and practically, still need to be evaluated. They simplify in a number of situations we now consider. Throughout the section, we use the shorthand $\sigma_I := \sigma_I (H, P, \varphi)$, where it is implied that $P(x)=P \in \fs(0,1)$ and $\varphi \in \fs(0,1;E_1,E_2)$. Recall from Theorem \ref{thm:stabilityall} (\ref{it:P}) that $\sigma_I$ is independent of $\varphi$ and $P$.

Our first result assumes  beyond \hone\ that $(\lambda_\ell (y,\xi,\zeta),\psi_\ell (y,\xi,\zeta))$ for $1\leq \ell\leq n$, the eigenvalues and corresponding eigenvectors of $\sym(x_0, y,\xi,\zeta)$ for a fixed $x_0\in\Rm$, are appropriately differentiable functions in $(y,\xi,\zeta)$. Then we have:

\begin{proposition} \label{wn}
Suppose $H=\Op (\sym)$ satisfies \hone. Fix $x_0 \in \mathbb{R}$ and let $E_1 < \alpha < E_2$ and $R$ as in Theorem \ref{thm:main}. 
Assume that $\lambda_\ell$ and $\psi_\ell$, the eigen-elements of $\sym(x_0, y,\xi,\zeta)$ are differentiable on $\partial R$ up to a set of $\Sigma$-measure zero.
Define $n_+ := \{\ell : \lambda_\ell > \alpha \text{ on } \partial R\}$ and $n_- := \{\ell : \lambda_\ell < \alpha \text{ on } \partial R\}$.
Then we have:
\begin{align}\label{wneq}
    2\pi \sigma_I = 
    \frac{i}{2\pi} \int_{\partial R} \eps_{ijk} \sum_{\ell_+ \in n_+, \ell_- \in n_-} \partial_i \psi_{\ell_+}^* \psi_{\ell_-} \psi_{\ell_-}^* \partial_j \psi_{\ell_+} \nu_k d\Sigma.
\end{align}
\end{proposition}
The main advantage of this result is to integrate out the variables $\omega$ so that the integration in \eqref{wneq} is two-dimensional. We will make use of this formula in the next section.

\medskip

Our next objective is to simplify Theorem \ref{thm:main} for a large class of $2$-dimensional models written in terms of Pauli matrices (or more generally in terms of representations of Clifford algebras \cite{3}). Many analogous result, e.g., for tight-binding models, can be found in the literature \cite{FC}.

The goal is to recast the formulas \eqref{eq:sigmaI1}-\eqref{eq:BEC2}-\eqref{FH} in terms of the topological degree of an appropriate vector field. This requires some definitions. We say that a point $a\in \mathbb{R}^d$ is a \emph{regular value} of a map $f: \mathbb{R}^d \rightarrow \mathbb{R}^d$ if
the preimage $f^{-1} (a)$ is
a finite collection of points $\{x_1,x_2, \dots, x_p\} \subset \mathbb{R}^d$ such that the Jacobian $M_j^{mn} = \partial_m f_n \vert _{x_j}$ is nonsingular for all $j$. We let $\mathcal{R} (f)$ denote the set of all regular values of $f$. By Sard's Theorem \cite{Sard}, we know that $\mathbb{R}^d \setminus \mathcal{R} (f)$ has measure zero.
\begin{proposition} \label{deg}
Let $H=\Op (\sym)$ satisfy \hone, and fix $x_0 \in \mathbb{R}$. Suppose there exist 
(smooth and real-valued) functions $f_1, f_2, f_3$ such that 
\begin{align*}
    \sym(x_0,y, \xi, \zeta) = f_1 (y,\xi, \zeta) \sigma_1 + f_2(y,\xi, \zeta) \sigma_2 + f_3(y,\xi, \zeta) \sigma_3
\end{align*}
for all $(y,\xi,\zeta) \in \mathbb{R}^3$,
where
\begin{align}\label{eq:pauli}
    \sigma_1 = \begin{pmatrix}0 &1\\1& 0 \end{pmatrix}, \qquad \sigma_2 = \begin{pmatrix}0 &-i\\i& 0 \end{pmatrix}, \qquad \sigma_3 = \begin{pmatrix}1 &0\\0& -1 \end{pmatrix}
\end{align}
are the Pauli matrices.
Thanks to \hone, there exists a
three-dimensional ball $S$ 
such that the vector field $f := (f_1, f_2, f_3)$ satisfies $|f| := \sqrt{f_1^2 + f_2^2 + f_3^2} \ge \eps_0$
in $\mathbb{R}^3 \setminus S$, for some $\eps_0 >0$.
Then the set 
$\mathcal{R}_{\eps_0} := \mathcal{R} (f) \cap \{|a| < \eps_0\}$ is nonempty.
Let $a \in \mathcal{R}_{\eps_0}$, and define $f^{-1} (a) =: \{(y_j,\xi_j, \zeta_j)\}_{j=1}^p$. Then
\begin{align} \label{sgndet}
    2 \pi \sigma_I = -\sum_{j=1}^p \sgn \det M_j,
\end{align}
where $M_j \in \mathbb{R}^{3 \times 3}$ is the Jacobian matrix defined by
$M_j^{mn} = \partial_m f_n \vert _{(y_j,\xi_j, \zeta_j)}$.
\end{proposition}
The proofs of Propositions \ref{wn} and \ref{deg} are postponed to Appendix \ref{appendix:deg}.

%
\subsection{Examples of application}
\label{sec:applications}

\noindent {\bf $2\times2$ models.} 
We now consider a number of Hamiltonians that appear in the analysis of topological insulators and superconductors. The $2 \times 2$ Dirac system \cite{2,1,3,bal2023topological,PS,Witten} is defined as
\begin{align} \label{2x2Dirac}
    H_D = D_x \sigma_1 + D_y \sigma_2 + m(y) \sigma_3,
\end{align}
with $\sigma_j$ the Pauli matrices already introduced in \eqref{eq:pauli} and $D_x=-i\partial_x$ while $D_y=-i\partial_y$ are the Hermitian differential operators.

Models for $p$-wave and $d$-wave superconductors \cite{BH,Volovik} are given by
\begin{align} \label{p}
    H_p &= \left( \frac{1}{2m} (D_x^2 + D_y^2) - \mu \right) \sigma_1 + \frac{1}{2}(c(y) D_y + D_y c(y)) \sigma_2 + c_0 D_x \sigma_3 \\
    \label{d}
    H_d &= \left( \frac{1}{2m} (D_x^2 + D_y^2) - \mu \right) \sigma_1
    + c_0 (D_y^2 - D_x^2) \sigma_2 + 
    \frac{1}{2}D_x (c(y)D_y + D_y c(y)) \sigma_3.
\end{align}

Above, $c_0$, $m$, and $\mu$ are fixed positive constants. We verify (see also Appendix \ref{appendix:deg}) that $H_D\in \Op S^1$ while $H_p\in \Op S^2$ and $H_d\in \Op S^2$.
We assume that $m(y)$ and $c(y)$ are smooth domain walls verifying
\begin{align*}
    m(y) = \begin{cases}
    m_-, &y \le -y_0 \\
    m_+, &y \ge y_0
    \end{cases}
    \qquad \text{and} \qquad
    c(y) = \begin{cases}
    c_-, &y \le -y_0 \\
    c_+, &y \ge y_0
    \end{cases}
\end{align*}
for some $y_0 > 0$, where the constants $m_\pm$ and $c_\pm$ are all nonzero. We then have the following accessible formulas for the edge current observable, whose proof is postponed to Appendix \ref{appendix:deg}. 


\begin{proposition}\label{prop:app}
    Each Hamiltonian $H \in \{H_D, H_p, H_d\}$ satisfies \hone\ with energy interval $[E_1, E_2] = [-E,E]$ and $E>0$ (depending on $H$) sufficiently small. 
    The corresponding edge current observables are 
    \begin{align*}
        2\pi \sigma_I (H_D) = \frac{1}{2} (\sgn (m_-) - \sgn (m_+)), \;
        2\pi \sigma_I (H_p) = \sgn (c_-) - \sgn (c_+), \;
        2\pi \sigma_I (H_d) = 2 (\sgn (c_-) - \sgn (c_+)).
    \end{align*}
\end{proposition}

%
%
%
%
%
\medskip \noindent {\bf Regularized model of equatorial waves.} 
%
The two-dimensional water wave model is given by a Hamiltonian $H_0={\rm Op}(\sym_0)$ with 
\begin{equation}\label{eq:a0geo}
  \sym_0(x,y,\xi,\zeta) = \begin{pmatrix} 0 & \xi & \zeta \\ \xi & 0 & -if(y) \\ \zeta& if(y) & 0\end{pmatrix}
\end{equation}
where $f(y)$ is a Coriolis force that is positive when $y>0$ (northern hemisphere of (necessarily) flat Earth) and negative when $y<0$; see \cite{3,delplace,souslov} for background on this water wave problem and in particular the observation that the bulk-interface correspondence fails for certain profiles $f(y)$ \cite{3}.

We thus consider here a regularized version given by Hamiltonian $H_\mu = \Op (\sym)$ given by $\sym= \lambda_+ \Pi_+ + \lambda_- \Pi_- + \lambda_0 \Pi_0$, where, following calculations in \cite{3,delplace}
\begin{align*}
\Pi_j = \psi_j \psi_j^*,\quad
\psi_0 =
\frac{1}{\kappa}
\left[ {\begin{array}{c}
   if\\
   \zeta\\
   -\xi\\
  \end{array} } \right],
  \quad
    \psi_\pm = 
    \frac{1}{\rho}
    \left[ {\begin{array}{c}
   if\xi \pm \kappa \zeta\\
   \xi\zeta \pm if\kappa\\
   \zeta^2+f^2\\
  \end{array} } \right],
  \quad
  \lambda_0 = \mu   \kappa^2(1+\kappa^2)^{-\frac12},
  \quad \lambda_\pm = \pm \kappa,
\end{align*}
with $f \in \fs (f_-, f_+)$, 
$\kappa = \sqrt{f^2+\xi^2+\zeta^2}$, and $\rho = \kappa \sqrt{2(f^2+\zeta^2)}$
for some nonzero constants $\mu$ and $f_\pm$. We verify that $H_\mu=H_0$ when the regularization parameter $\mu=0$.

The role of the regularization is to replace the infinitely degenerate (flat band) eigenvalue $0$ by a topologically trivial band with eigenvalues tending to $\pm\infty$ as $|(\xi,\zeta)|\to\infty$.
The choice of the regularization for $\lambda_0$ ensure that $\sym\in S^1$ and is elliptic when $\mu\not=0$.
We then have
\begin{proposition}\label{prop:3x3}
    Fix $\mu \ne 0$ and set $\mu_1 := \min\{|\mu|,1\}$ and $f_0 := \min \{|f_+|, |f_-|\} > 0$. Then
    the regularized water wave Hamiltonian $H_\mu$ satisfies \hone\ with $E_2= f_0 \min\{\mu_1, |\mu| f_0\}$ and $E_1= -E_2$. Moreover, 
    \begin{align*}
        2\pi \sigma_I (H_\mu) = \sgn (f_+) - \sgn (f_-).
    \end{align*}
\end{proposition}
We refer to Appendix \ref{appendix:deg} for the details of the derivations and proofs.




%
\section{Periodic approximations of infinite-space problems} \label{sec:periodic}

This section addresses the computation of asymmetric transport and edge currents by numerical simulations. In particular, we wish to devise a numerical method that is able to capture $\sigma_I$. The operator $H$ is initially defined on the unbounded domain $\Rm^2$ and thus first needs to be replaced by a sequence of approximations on bounded domains. A standard method to approximate spectral decompositions of operators in $\mathbb{R}^2$ is to consider infrared cutoffs defined as the restriction of the operator on a box $(-\pi L,\pi L)^2$ with periodic boundary conditions and analyze the limit $L\to\infty$. 

As we mentioned in the introduction, this approach is fraught with a number of difficulties. Because of the structure of the domain wall with $\sym(y)$ transitioning from $\sym_-$ to $\sym_+$ as $y$ crosses $0$, any periodization generates another transition from $\sym_+$ to $\sym_-$ as $y$ crosses $\pi L$. The asymmetric transport along the $x$ axis when $y$ is close to $0$ is compensated by an asymmetric transport along the $x$ axis when $y$ is close to $-\pi L\equiv \pi L$ with opposite chirality, resulting in a (globally) topologically trivial material. Furthermore, the function $P(x)$ modeling transport along a line $x\approx x_1$ as $P$ (smoothly) jumps from $0$ to $1$ across that (thickened) interface, necessarily jumps back from $1$ to $0$ somewhere else. 

Periodic systems thus no longer enjoy a non-trivial topology. This is a no-go result similar to a fermion doubling or Nelson-Ninomiya theorem \cite{Witten} ensuring that any domain wall in a mass term on a torus, no matter how large, may continuously be deformed to a constant mass term. The trace in \eqref{eq:sigmaI} therefore needs to be modified so the integral focuses on the original domain wall, which as $L\to\infty$ becomes well separated from the spurious second domain wall. This will be achieved by introducing a {\em spatial filter} localizing the computation of the trace in the vicinity of the origin $(0,0)$; see left panel in Fig. \ref{fig:PQ}.

\begin{figure}[ht!]
    \centering
    \includegraphics[width=6cm]{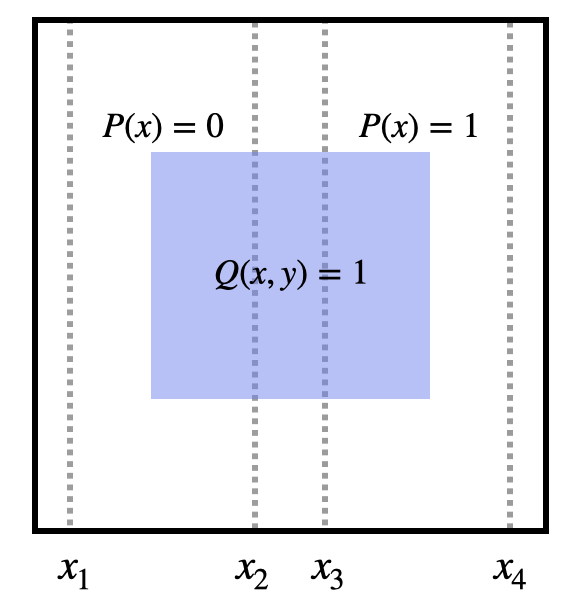} \qquad
    \includegraphics[width=8.5cm]{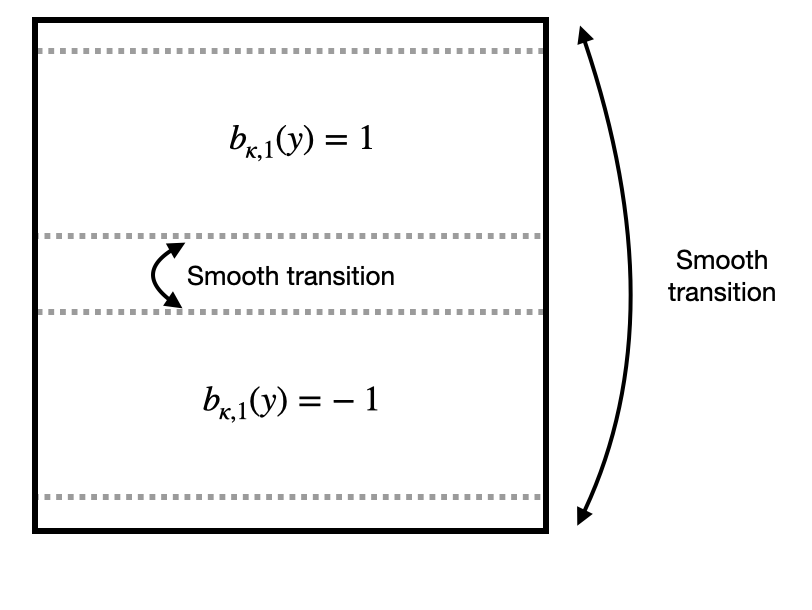}
    \caption{Left panel: illustration of the spatial filters $P$ and $Q$. The function $P$ is smooth and constant within the vertical slabs $x_1 \le x \le x_2$ and $x_3 \le x \le x_4$, while $Q$ is equal to $1$ in the shaded region and rapidly (but smoothly) vanishes outside this region. Right panel: illustration of the coefficients of the periodic operator $H_\kp$. For example, a coefficient can transition smoothly between $-1$ and $1$.}
    \label{fig:PQ}
\end{figure}



\subsection{Periodic Hamiltonians and main results}


While edge currents may be defined and computed for general pseudo-differential Hamiltonians, as we do in section \ref{sec:FHformula}, see also \cite{bal2023topological}, we restrict our analysis of periodized Hamiltonians to differential systems and avoid complications resulting from non-local effects. 
We therefore consider (unperturbed) infinite-space Hamiltonians that satisfy \hone\ and are differential operators of the form
\begin{align}\label{H2}
    H = M_0 D^m_y + M_m D^m_x + \sum_{j=0}^m \cp_j (y) D_x^j D_y^{m-j} + \sum_{i+j \le m-1} \cp_{ij} (y) D^i_x D^j_y.
\end{align}
Here, $D_x=-i\partial_x$ and $D_y=-i\partial_y$ still while $M_0$ and $M_m$ are constant nonsingular Hermitian matrices with
\begin{align}\label{anticommutation}    
    \{M_0, M_m\}\begin{cases}
    =0, & \text{$m$ odd}\\
    \ge 0, & \text{$m$ even}
    \end{cases},
    \quad
    \{M_i, \cp_j (y)\}
    \begin{cases}
    =0, & \text{$i+j$ odd}\\
     \ge 0, & \text{$i+j$ even}
    \end{cases}, i \in \{0,m\},\ j \in \{0,1, \dots, m\},\ y \in \mathbb{R}.
\end{align}
Here, $\{O_1,O_2\}:=O_1O_2+O_2O_1$. 
The $\cp_j(y)$ and $\cp_{ij} (y)$ are smooth matrix-valued functions chosen such that $H$ is symmetric;
in particular, $\cp_j(y)$ are all Hermitian-valued.
The anti-commutation relations \eqref{anticommutation} ensure that $H^2 = M_0^2 D^{2m}_y + M_m^2 D^{2m}_x + A + B$, where $A$ is non-negative and $B$ is a differential operator of order $2m-1$, so that $H^2$ remains a self-adjoint elliptic operator.  The above structure is sufficiently general to apply our theory to all models  from section \ref{sec:applications} including those of superconductors. 

We will finally consider perturbed operators $H_V=H+V$ where $V=V(x,y)$ is a multiplication operator with support that remains sufficiently close, in an appropriate sense, to the center of the torus.

To approximate $H$ by an operator on $(-\pi L, \pi L)^2$, we redefine the coefficients near $y = \pi L$ to obtain Hamiltonians with smooth coefficients. 
The resulting periodic Hamiltonian is unitarily equivalent to an operator $H_\kp$ on the 2-torus $\mathbb{T}^2=[-\pi, \pi)^2$,
with $\kp = L^{-1}$ the relevant infra-red cut-off parameter. 
The details of the routine construction of $H_\kp$ are postponed to section \ref{subsec:construction}. 

For $\varphi \in \fs (0,1;E_1,E_2)$,
we define the \emph{periodic  filtered edge current observable} by
\begin{align} \label{eq:sigmaIQ}
\tilde{\sigma}_I (H_\kp) = \Tr i Q [H_\kp,P] \varphi ' (H_\kp),
\end{align}
where
$P = P(x)$ and $Q = Q (x,y) = Q_X (x) Q_Y(y)$ are smooth point-wise multiplication operators such that
\begin{align*}
    Q_\Theta(\theta) = 
    \begin{cases}
    1, & |\theta| \le \pi/2\\
    0, & |\theta| \ge \pi/2 + \delta_Q
    \end{cases} \quad
    \text{and}
    \quad
    P(x) = 
    \begin{cases}
    0, & x_1 \le x \le x_2   \qquad -\pi < x_1 < -3\pi/4 < -\pi/2 < x_2 \\
    1, & x_3 \le x \le x_4 \qquad x_2 < x_3 < \pi/2 < 3\pi/4 < x_4 < \pi,
    \end{cases}
\end{align*}
for some fixed $0<\delta_Q < \pi/4$.
As we observed earlier, $Q$ is a spatial filter centered at the domain wall of interest where $P'>0$ and vanishing in the vicinity of the unwanted domain wall where $P'<0$.
We refer to Fig. \ref{fig:PQ} (left panel) for an illustration of $P$ and $Q$.

\medskip

Our first main result of this section is the following, where we recall the definition \eqref{eq:sigmaI} of $\sigma_I (H)$.
\begin{theorem} \label{periodicApprox}
For any $p > 0$,
there exists a constant $C_p> 0$ such that $|\sigma_I (H) - \tilde{\sigma}_I (H_\kp)| \le C_p \kp^{p}$ as $\kp \rightarrow 0$.
\end{theorem}

The above theorem establishes 
an almost-exponential rate of convergence, 
which is based on a careful approximation of the periodic eigenfunctions by truncations of the infinite eigenfunctions corresponding to energies in the bulk spectral gap in the vicinity of domain walls and on a proof that such periodic eigenfunctions have to be negligible away from the domain walls. Estimates of the periodic eigenelements by means of the Courant-Fischer min-max theorem allow us to obtain the above trace estimates. 

\medskip

We next establish the stability of the periodic filtered interface current observable against perturbations of the form
\begin{align}\label{eq:Vperiodic}
    V_\kp := \sum_{i+j\le m-1} \kp^{i+j} v_{\kp,ij} (x,y) D^i_x D^j_y
\end{align}
where the 
$v_{\kp, i j}$ are smooth matrix-valued functions that make $V_\kp$ a symmetric differential operator on $L^2 (\mathbb{T}^2) \otimes \mathbb{C}^n$.
We asssume that  $\sum_{i,j} \norm{v_{\kp, i j}}_{L^\infty} \le C$ uniformly in $\kp$ and that there exists a closed set $S_1 \subset \mathbb{T}^2$ such that
$\cup_{i,j} \supp (v_{\kp,ij}) \subset S_1$ for all $\kp \in (0, \kp_0]$
and $\supp ((1-Q)P') \cap S_1 = \emptyset$.

The condition on the sup-norm of $v_{\kp, ij}$
guarantees that the coefficients of $H_\kp + V_\kp$ are bounded uniformly in $\kp$.
Combined with the assumption that $V_\kp$ is symmetric, this ensures that $H_\kp + \mu V_\kp$ is uniformly elliptic (in the sense of Proposition \ref{periodicEllipticity}) and hence self-adjoint.
Note that $L^\infty-$bounds on derivatives of the $v_{\kp,ij}$ are not necessary since $V_\kp$ has no leading order terms.
The condition on the support of $v_{\kp, ij}$ is natural as it ensures that the perturbation is close to the domain wall of interest and well separated from the spurious domain wall introduced by the periodization.
It
allows us to express $\tilde{\sigma}_I (H_\kp + V_\kp) - \tilde{\sigma}_I (H_\kp)$ as the integral over a product of operators, two of which are differential operators with disjoint support (see the proof of Theorem \ref{stabilityPeriodic} below).
Note that the assumption $\supp ((1-Q)P') \cap S_1 = \emptyset$ could easily be replaced by $\supp (QP') \cap S_1 = \emptyset$.

Our main stability result  for the periodic filtered edge current observable is the following
\begin{theorem} \label{stabilityPeriodic}
Take $H_\kp$ and $V_\kp$ as above.
Then for all $N \in \mathbb{N}$,
\begin{align*}
    |\tilde{\sigma}_I (H_{\kp}+V_{\kp}) - \tilde{\sigma}_I (H_\kp)| \le C_N \kp^N,\quad  0<\kp\leq \kp_0.
\end{align*}
\end{theorem}
This result, combined with Theorem \ref{periodicApprox}, shows that the periodic approximation of $\sigma_I$ enjoys a spectral (almost exponential) convergence property with $\sigma_I(H)-\tilde \sigma_I(H_\kp+V_\kp)$ being of order $\kp^p$ for any $p\geq0$.

The long and technical proofs of Theorems \ref{periodicApprox} and \ref{stabilityPeriodic} are postponed to section \ref{sec:pfs_periodic}. Because the periodic edge current $\tilde\sigma_I$ does not enjoy any exact stability property, unlike $\sigma_I$ analyzed in section \ref{sec:stab}, the results of section \ref{sec:stability} and their proofs need to be combined with a direct spectral analysis of the periodic Hamiltonian $H_\kp$. The next section concerns the precise definition of $H_\kp$ while some applications are given in section \ref{sec:Lappli}.

\medskip

\subsection{Construction of periodic operator and spectral approximation} \label{subsec:construction}
We now construct 
the differential operator $H_\kp$ on the torus. We start from our original operator $H = \Op (\sym)$, assumed to be a differential operator on $\mathcal{H} =L^2 (\mathbb{R}^2) \otimes \mathbb{C}^n$ given by \eqref{H2} and \eqref{anticommutation} such that \hone\ holds.

Any operator with smooth coefficients on the torus will have a transition modeling a domain wall between $\sym_-$ and $\sym_+$ in the vicinity of $y=0$ and another domain wall from $\sym_+$ to $\sym_-$ in the vicinity of $y=\pi$. We construct a differential operator $H' = \Op (\sym')$ with such a domain wall, i.e., such that $\sym' = \sym_+$ whenever $y \le \pi - y_0$ and $\sym' = \sym_-$ whenever $y \ge \pi + y_0$. Assume that $\sym'$ is smooth and independent of $x$, and
that for every $y' \in \mathbb{R}$, there exists $y \in \mathbb{R}$ such that $\sym' (y', \xi, \zeta) = \sym(y,\xi,\zeta)$.
We may for instance choose $\sym' (y,\xi,\zeta) = \sym(\pi- y, \xi, \zeta)$.

The above operators are defined on $\Rm$ and need to be mapped to a torus and glued together.
Define the rescaled versions 
\begin{align}\label{eq:rescaledop}
    \tilde{H}_\kp = \Op (\tilde{\sym}_\kp),\quad \tilde{\sym}_\kp (y,\xi,\zeta) := \sym(\frac y{\kp}, \kp \xi, \kp\zeta);\qquad
    \tilde{H}'_\kp = \Op(\tilde{\sym}'_\kp),\quad \tilde{\sym}'_\kp (y,\xi,\zeta) := \sym' (\frac{y-\pi}{\kp} + \pi, \kp\xi, \kp\zeta).
\end{align}
Then there exists $\kp_0 \in (0,1]$ such that for all $\kp \in (0, \kp_0]$,
we can define $H_\kp$ a differential operator on \tcbn{$\mathcal{H}(\mathbb{T}^2)= L^2 (\mathbb{T}^2) \otimes \mathbb{C}^n$} with \emph{smooth coefficients} that are equal to those of $\tilde{H}_\kp$ when $-\pi/2 \le y \le \pi/2$ and $\tilde{H}'_\kp$ when $\pi/2 \le y \le 3 \pi/2$ (with the equivalence $-\pi/2 \equiv 3\pi/2$ on $\mathbb{T}$).
Thus we have
\begin{align} \label{Hlambda2D}
    H_\kp = \kp^m \Big(M_0 D^m_y + M_m D^m_x + \sum_{j=0}^m \cp_{\kp,j} (y) D^j_x D^{m-j}_y \Big)
    + \sum_{i+j \le m-1} \kp^{i+j} \cp_{\kp,ij} (y) D^i_x D^j_y,
\end{align}
where the $\cp_{\kp, j},\cp_{\kp, ij} \in \mathcal{C}^\infty_c (\mathbb{T})$ are constant whenever $\kp y_0 \le |y| \le \pi - \kp y_0$. We refer to Fig. \ref{fig:PQ} (right panel) for an illustration.
By definition, we know that for all $j \in \{0, 1, \dots, m\}$ and $y \in \mathbb{T}$, there exists $y' \in \mathbb{R}$ such that $\cp_{\kp, j} (y) = \cp_j (y')$.
Hence \eqref{anticommutation} holds in the periodic setting, with $\cp_j$ replaced by $\cp_{\kp, j}$.
Moreover, 
for any multi-index $\alpha \in \mathbb{N}^2$, we have $\kp^{|\alpha|} (\sum_{j} \norm{\partial^\alpha \cp_{\kp, j}}_{L^\infty} + \sum_{i,j} \norm{\partial^\alpha \cp_{\kp, ij}}_{L^\infty}) \le C$ uniformly in $\kp$.


\medskip

In the rest of this section, we summarize the main steps in the spectral approximations that are required to prove Theorem \ref{periodicApprox}. These results will all be made rigorous in section \ref{sec:pfs_periodic}. Since the coefficients of $H$ are independent of $x$, we can for $\xi \in \mathbb{R}$ define the operator $\hat{H} (\xi) :\mathcal{H}^m (\mathbb{R}) \to \mathcal{H} (\mathbb{R})$ by 
\begin{align}\label{eq:hatH}
    \hat{H} (\xi) := M_0 D_y^m + M_m \xi^m +
    \sum_{j=0}^m \cp_j (y) \xi^j D_y^{m-j} + \sum_{i+j \le m-1} \cp_{ij} (y) \xi^i D_y^j,
\end{align}
so that
$
    H = \mathcal{F}^{-1}_{\xi \rightarrow x} \int_{\mathbb{R}}^{\oplus} \hat{H} (\xi) d\xi \mathcal{F}_{x \rightarrow \xi},
$
with $\mathcal{F}$ the one dimensional Fourier transform in the $x$-variable.
The edge current observable can then (formally) be written as
\begin{align}\label{eq:tr_ft}
    2\pi \sigma_I = \Tr \int_{\mathbb{R}} \partial_\xi \hat{H} (\xi) \varphi' (\hat{H} (\xi)) d\xi,
\end{align}
with the above trace over $\mathcal{H} (\mathbb{R})$ and
\begin{align*}
    \partial_\xi \hat{H} (\xi) := M_0 D_y^m + m M_m \xi^{m-1} +
    \sum_{j=1}^{m} j\cp_j (y) \xi^{j-1} D_y^{m-j} + \sum_{i+j \le m-1} i \cp_{ij} (y) \xi^{i-1} D_y^j
\end{align*}
the Fourier transform of the commutator $i[H,P]$ after integration in $x$; see, e.g. \cite{3}.
Similarly, the periodic edge current observable will be shown in \eqref{eq:sigmaI_sum} to equal
\begin{align}\label{eq:tr_ft_per}
    2\pi \tilde{\sigma}_I = \kp \Tr \sum_{\xi \in \kp \mathbb{Z}} Q_Y \partial_\xi \hat{H}_\kp (\xi) \varphi' (\hat{H}_\kp (\xi)),
\end{align}
where $\hat{H}_\kp (\xi) : \mathcal{H}^m (\mathbb{T}) \rightarrow \mathcal{H} (\mathbb{R})$ and $\partial_\xi \hat{H}_\kp (\xi) : \mathcal{H}^m (\mathbb{T}) \rightarrow \mathcal{H} (\mathbb{R})$ are defined by
\begin{align*}
    \hat{H}_\kp (\xi) &:= \kp^m M_0 D_y^m + M_m \xi^m +
    \sum_{j=0}^m \kp^{m-j} \cp_{\kp,j} (y) \xi^j D_y^{m-j} + \sum_{i+j \le m-1} \kp^{j} \cp_{\kp,ij} (y) \xi^i D_y^j\\
    \partial_\xi \hat{H}_\kp (\xi) &:= \kp^m M_0 D_y^m + m M_m \xi^{m-1} +
    \sum_{j=1}^{m} j\kp^{m-j} \cp_{\kp,j} (y) \xi^{j-1} D_y^{m-j} + \sum_{i+j \le m-1} i \kp^j \cp_{\kp,ij} (y) \xi^{i-1} D_y^j.
\end{align*}
The formula \eqref{eq:tr_ft_per} resembles a Riemann sum approximation of \eqref{eq:tr_ft}, suggesting that a spectral approximation of $\hat{H} (\xi)$ by $\hat{H}_\kp (\xi)$ would suffice to produce the convergence result in Theorem \ref{periodicApprox}. In the proof of Theorem \ref{periodicApprox} (see \eqref{eq:sigmaI_sum} there), we use the identity $\Tr A = \sum_{j} \aver{\phi_j, A \phi_j}$ (with $\{\phi_j\}$ any orthonormal basis) to expand trace in \eqref{eq:tr_ft_per} in the eigenbasis of $\hat{H}_\kp (\xi)$. The spatial filter $Q_Y$ eliminates the contribution from the spurious domain wall, as eigenfunctions of $\hat{H}_\kp$ localized to the $-\pi \equiv \pi$ interface get exponentially small on the support of $Q_Y$ in the $\kp \to 0$ limit.
We show with Proposition \ref{propEvalApprox} and Lemma \ref{switchBasis} below that the eigenvalues and corresponding eigenspaces of $\hat{H}_\kp (\xi)$ and $\hat{H} (\xi)$ get exponentially close as $\kp \to 0$. The proofs use the Courant-Fischer min-max theorem (stated by Theorem \ref{maxmin}) and the exponential decay of the (generalized) eigenfunctions of $\hat{H} (\xi)$ and $\hat{H}_\kp (\xi)$ in the regions where the coefficients of these differential operators are constant (Proposition \ref{expDecay}).

Our precise result on eigenvalue approximation also involves the operator $H'$ defined in \eqref{eq:rescaledop}. Using $\hat{H}'(\xi)$ as the Fourier transform of $H'$, we prove the following: 
\begin{proposition}\label{prop:eval_approx}
    For $\xi \in \mathbb{R}$, let $\mu_1 (\xi) \le \mu_2 (\xi) \le \dots \le \mu_s (\xi)$ denote the combined eigenvalues of $\hat{H}^2 (\xi)$ and $\hat{H}'^2(\xi)$ 
    in $[0,E^2)$, and let $\mu_{\kp,1}(\xi) \le \mu_{\kp,2}(\xi) \le \dots$ denote the eigenvalues of $\hat{H}_\kp^2 (\xi)$, all counted with multiplicity. Then
    there exist positive constants $C$ and $r$ such that $$\mu_{\ell} - C e^{-r/\kp} \le \mu_{\kp,\ell}\le \mu_{\ell} + C e^{-r/\kp}$$
    uniformly in $\kp \in (0, \kp_0]$ 
    for all $\ell \in \{1,2, \dots, s\}$.
    Moreover, 
    $\mu_{\kp, s + 1} \notin [0,E^2)$
    for all $\kp$ sufficiently small.
\end{proposition}
This proposition is restated as Proposition \ref{propEvalApprox} below and proved there.

\medskip

With Lemma \ref{switchBasis} below, we will show that the eigenspaces of $\mu_{\kp,\ell}$ converge exponentially to the eigenspaces of $\mu_\ell$ as $\kp \to 0$. 
For ease of exposition, we now present the following consequence of Lemma \ref{switchBasis} (and the rapid decay of the eigenfunctions of $\hat{H}^2 (\xi)$ and $\hat{H}'^2 (\xi)$ proved in Proposition \ref{expDecay}),
where we inherit the notation of eigenvalues from Proposition \ref{prop:eval_approx} and use $(\cdot, \cdot)$ and $\aver{\cdot, \cdot}$ to denote the inner products on $L^2 (\mathbb{R}) \otimes \mathbb{C}^n$ and $L^2 (\mathbb{T}) \otimes \mathbb{C}^n$, respectively.
\begin{proposition}\label{prop:efn_approx}
    Fix $\xi \in \mathbb{R}$. For $j \in \mathbb{N}_+$, let $\psi_{j,\xi}$ denote the eigenfunction of $\hat{H}^2 (\xi)$ or $\hat{H}'^2 (\xi)$ corresponding to eigenvalue $\mu_j (\xi)$. We use the orthonormalization convention $(\psi_{i,\xi}, \psi_{j,\xi}) = \delta_{ij}$ for all $\psi_{i,\xi}, \psi_{j,\xi}$ eigenfunctions of the same operator. Let $\theta_{\kp,j,\xi}$ denote the eigenfunction of $\hat{H}_\kp^2 (\xi)$ corresponding to eigenvalue $\mu_{\kp,j} (\xi)$, again assuming that $\aver{\theta_{\kp,i,\xi}, \theta_{\kp,j,\xi}} = \delta_{ij}$. Let $j \le \ell$ be positive integers such that $\ell \le s$ and define
    $\delta := \min\{ \mu_{\kp,\ell+1} (\xi) - \mu_{\kp,\ell} (\xi), \mu_{\kp,j}(\xi) - \mu_{\kp,j-1}(\xi) \}$, where the second argument is ignored if $j=1$. Then there exist positive constants $C$ and $r$ such that
    \begin{align}\label{eq:efn_approx}
        \Big |\sum_{i=j}^\ell \Big (\aver{\theta_{\kp,i,\xi}, A_\kp (\xi) \theta_{\kp,i,\xi}} - (\psi_{i,\xi}, A (\xi) \psi_{i,\xi}) \Big) \Big| \le
        C \delta^{-1} e^{-r/\kp},
    \end{align}
    where $A_\kp (\xi) := Q_Y \partial_\xi \hat{H}_\kp (\xi)$ and $A (\xi):= \partial_\xi \hat{H} (\xi)$.
\end{proposition}

\subsection{Application to Dirac and superconductor models}
\label{sec:Lappli}
To apply the main results from sections \ref{sec:periodic} to several
Hamiltonians from section \ref{sec:applications}, we need to check the anti-commutation relations of the leading order terms \eqref{anticommutation}. Here are some examples:

\medskip
\noindent{\bf $2 \times 2$ Dirac system \eqref{2x2Dirac}.}
Since $H = D_x \sigma_1 + D_y \sigma_2 + m(y) \sigma_3$,
we can take $M_0 = \sigma_2$, $M_1 = \sigma_1$, and $\cp_{00} (y) = m(y) \sigma_3$.
Since the $\sigma_j$ are nonsingular and $\{\sigma_1, \sigma_2\} = 0$, 
\eqref{anticommutation} holds.
It follows that $|\tilde{\sigma}_I (H_\kp) - \sigma_I (H)| \le C_p \kp^p$ as $\kp \rightarrow 0$, for any $p$.
Moreover, $|\tilde{\sigma}_I (H_\kp) - \tilde{\sigma}_I (H_\kp + V_\kp)| \le C_p \kp^p$
for all smooth Hermitian-valued point-wise multiplication operators $V_\kp = V_\kp (x,y)$ such that $\norm{V_\kp}_{L^\infty} \le C$ uniformly in $\kp$,
and that satisfy $S_1 \cap \supp (1-Q) = \emptyset$ 
for some closed set $S_1$ containing $\cup_\kp \supp (V_\kp)$.

\medskip
\noindent{\bf $p$-wave superconductor model \eqref{p}.} 
We can take $M_0 = M_2 = \frac{1}{2m} \sigma_1$, which takes care of all leading order terms.
Since $\sigma_1$ is non-singular, \eqref{anticommutation} holds.
It follows that $|\tilde{\sigma}_I (H_\kp) -\sigma_I (H)| \le C_p \kp^p$ for any $p$. 
We also have $|\tilde{\sigma}_I (H_\kp) - \tilde{\sigma}_I (H_\kp + V_\kp)| \le C_p \kp^p$ 
for all Hermitian-valued first-order differential operators 
$V_\kp := \sum_{i+j\le 1} \kp^{i+j} v_{\kp,ij} (x,y) D^i_x D^j_y$ with smooth coefficients
$v_{\kp, ij}$ such that $\sum_{i,j} \norm{v_{\kp, ij}}_{L^\infty} \le C$
and $\cup_{i,j} \supp (v_{\kp, ij}) \subset S_1$ uniformly in $\kp$, for some closed set $S_1$
satisfying $S_1 \cap \supp (1-Q) = \emptyset$.

\medskip

\noindent{\bf $d$-wave superconductor model \eqref{d}.} 
Take $M_0 = M_2 = \frac{1}{2m} \sigma_1$, $\cp_0 = c_0 \sigma_2$, $\cp_1(y) = c(y) \sigma_3$, $\cp_2 = -c_0 \sigma_2$.
Since $\sigma_1$ is nonsingular and $\{\sigma_1, \sigma_j\}=0$ for $j \in \{2,3\}$, \eqref{anticommutation} holds.
The following conclusions are the same as those for the $p$-wave superconductor model above.

\section{Numerical simulations}\label{sec:numerics}

\begin{figure}
    \begin{subfigure}{.23\textwidth} 
  \centering
  \includegraphics [width=\textwidth] {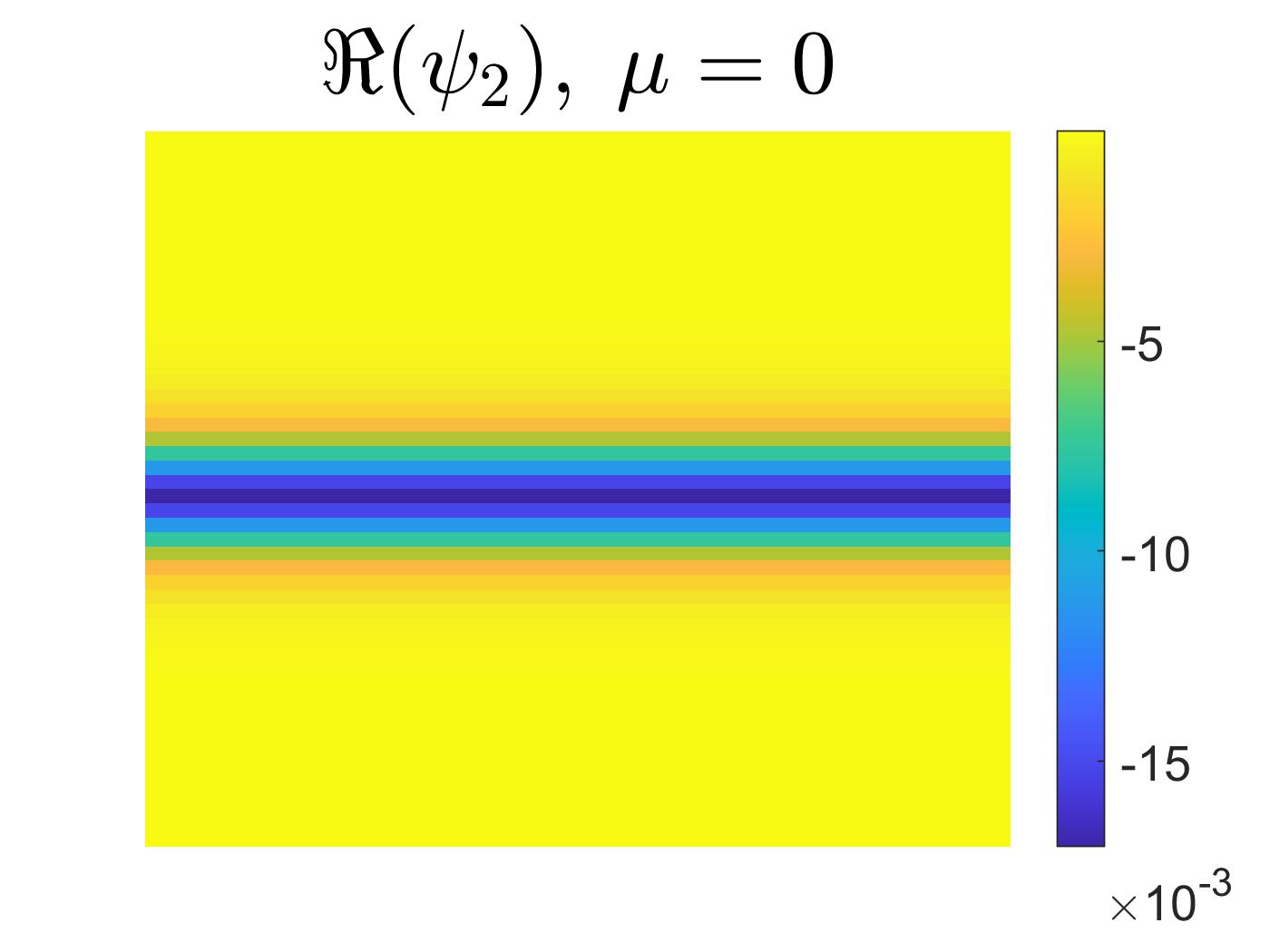}  
\end{subfigure}
\begin{subfigure}{.23\textwidth}
  \centering
  \includegraphics [width=\textwidth] {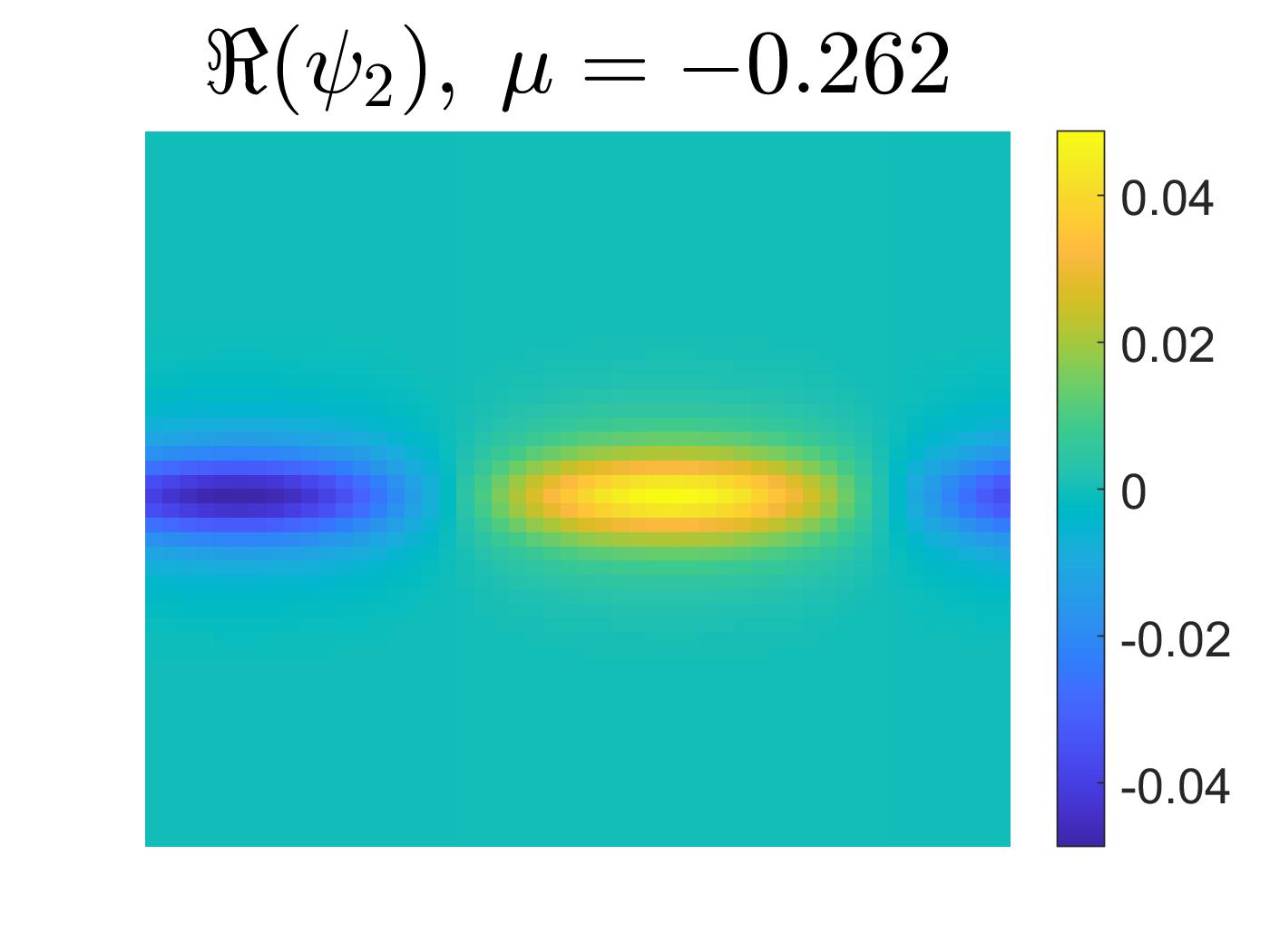}
\end{subfigure}
\begin{subfigure}{.23\textwidth}
  \centering
  \includegraphics [width=\textwidth] {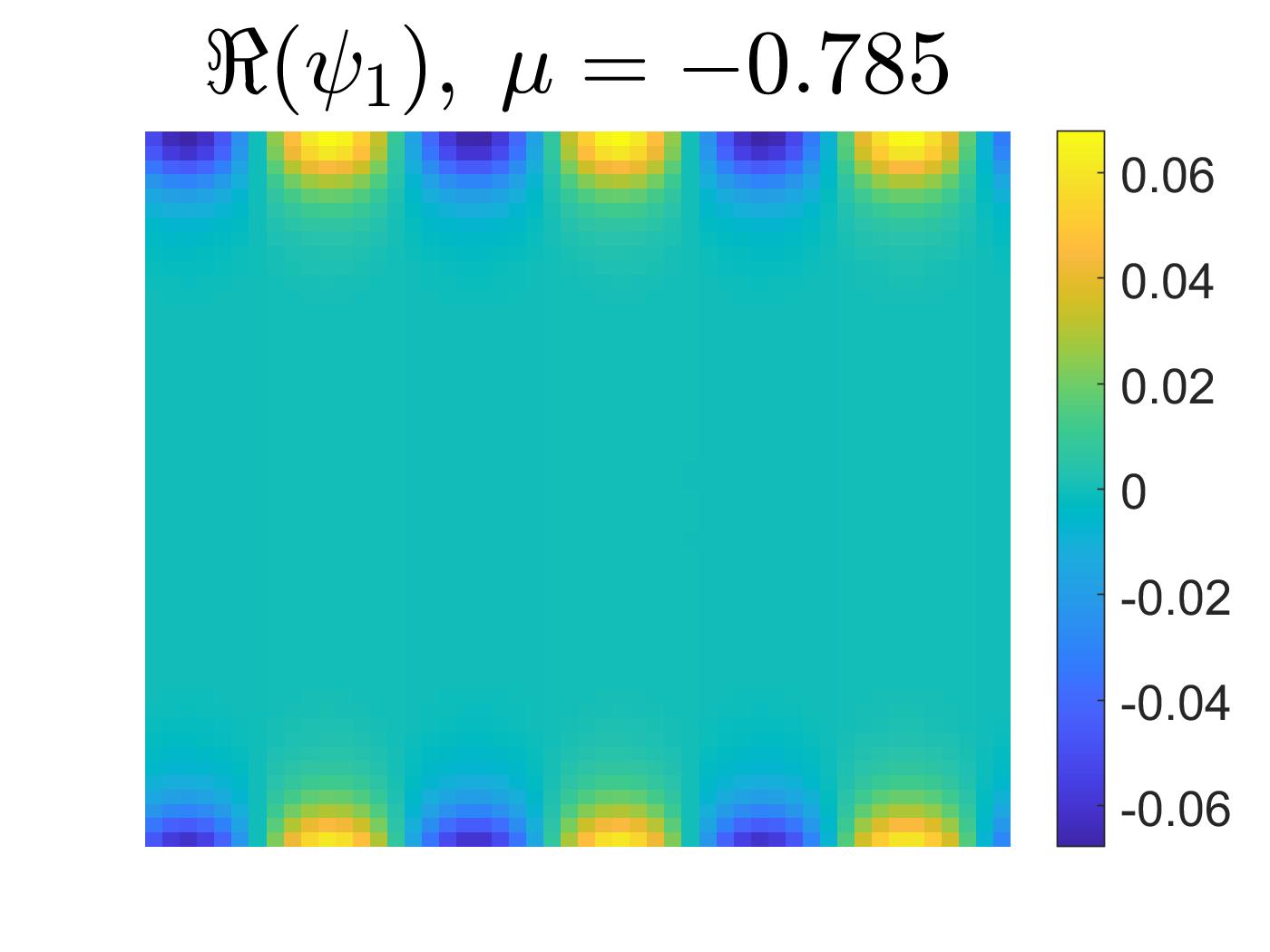}
\end{subfigure}
\begin{subfigure}{.23\textwidth}
  \centering
  \includegraphics [width=\textwidth] {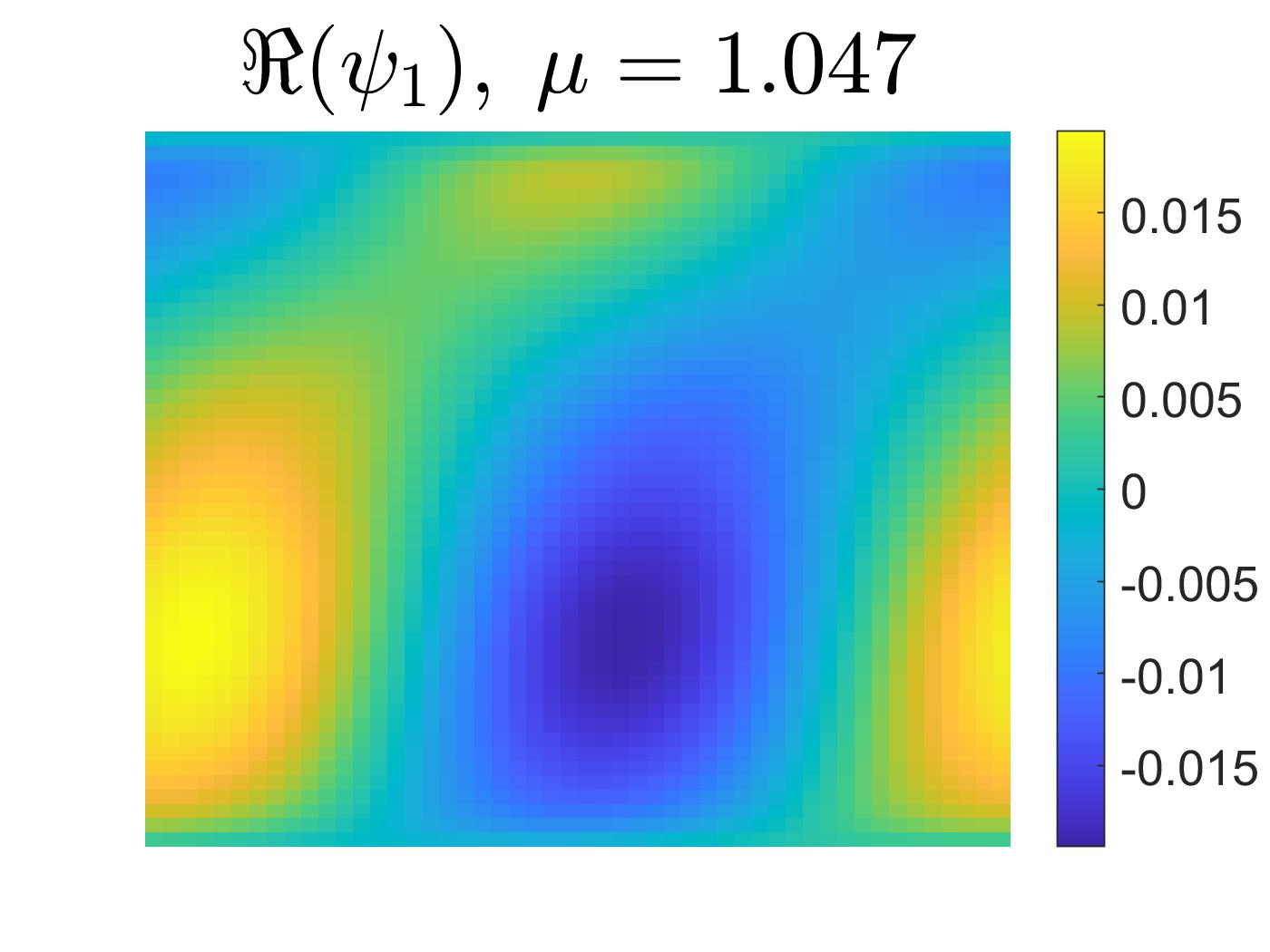}
\end{subfigure}
\newline
    \begin{subfigure}{.23\textwidth} 
  \centering
  \includegraphics [width=\textwidth] {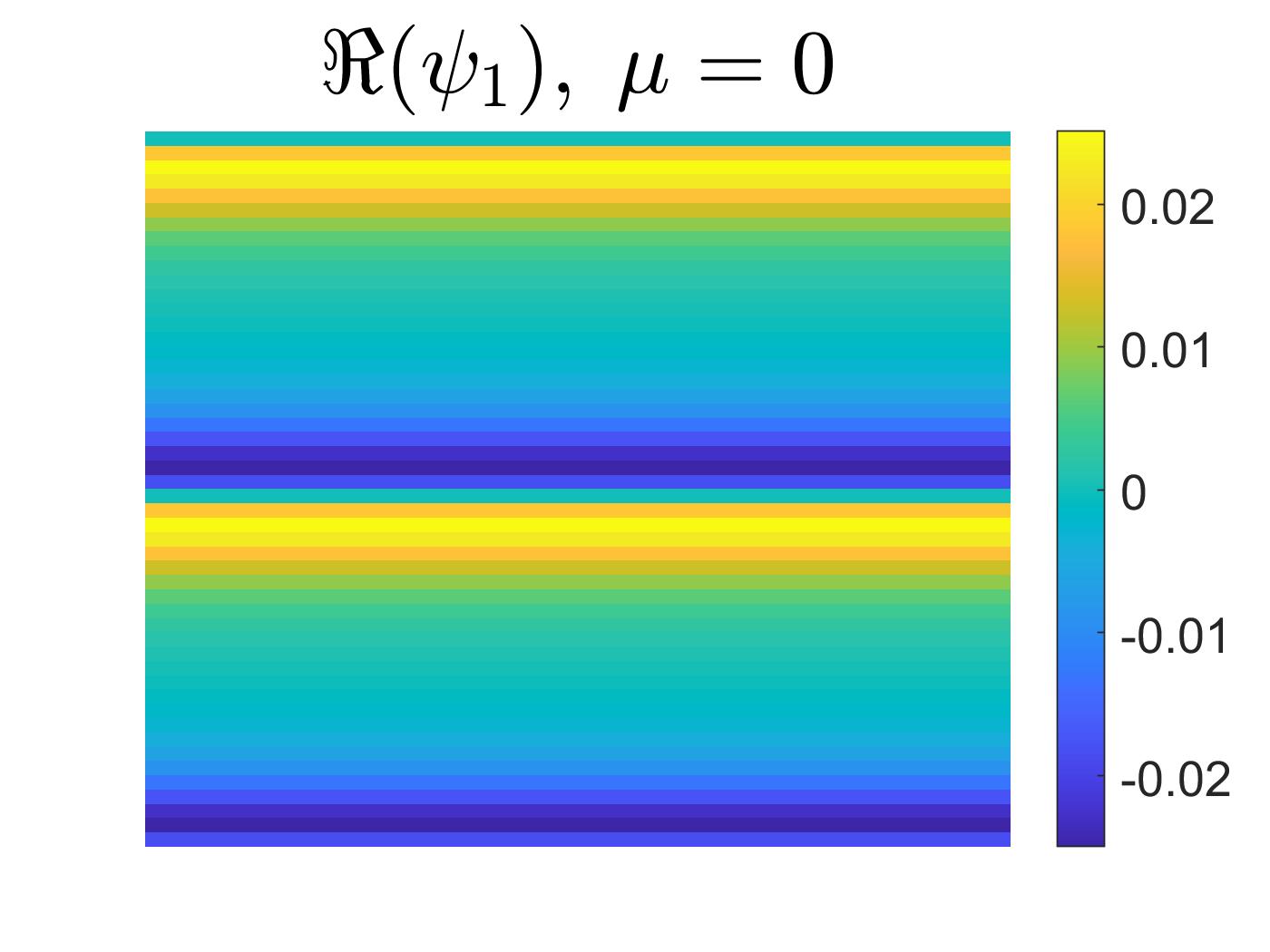}  
\end{subfigure}
\begin{subfigure}{.23\textwidth}
  \centering
  \includegraphics [width=\textwidth] {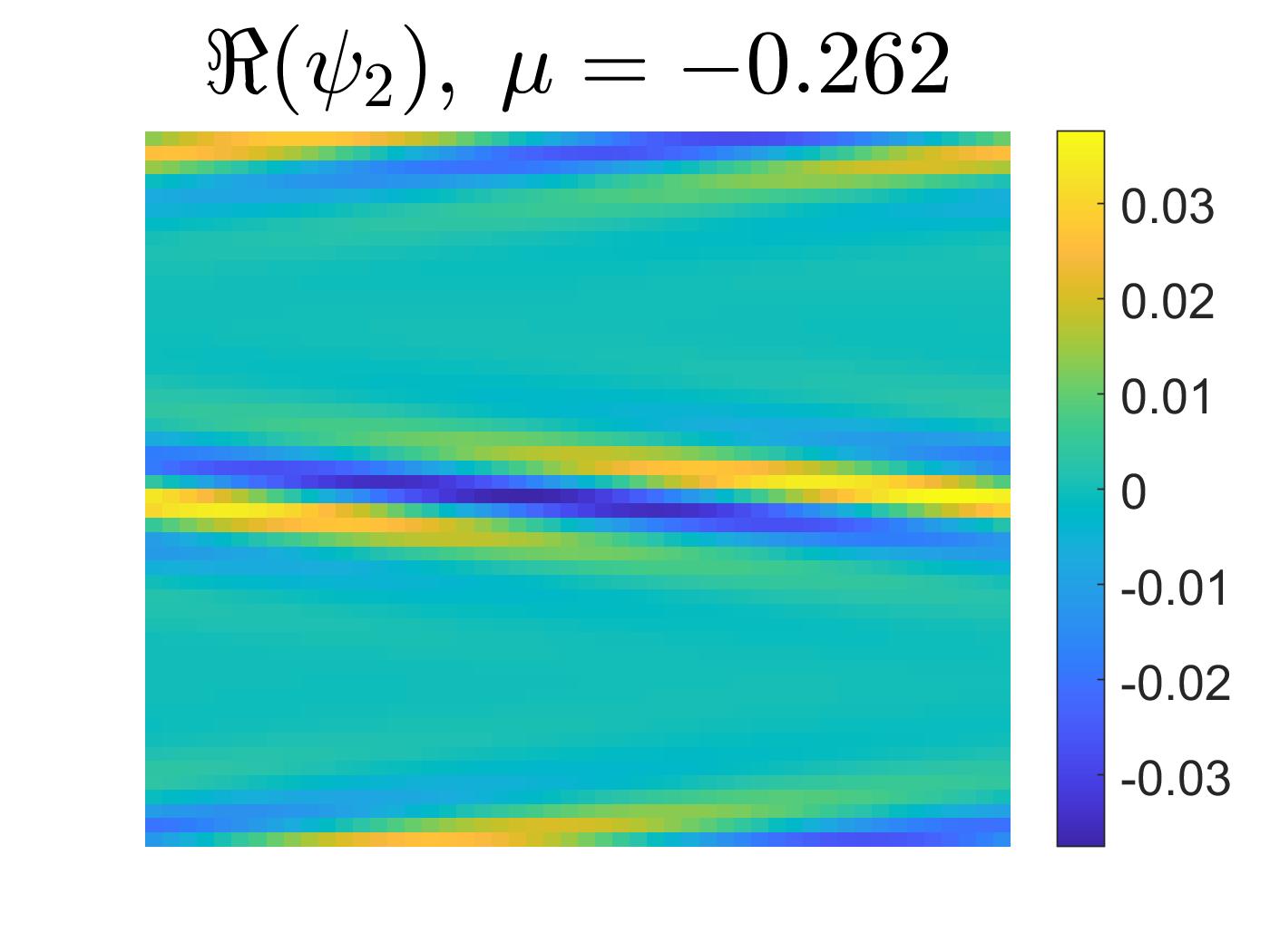}
\end{subfigure}
\begin{subfigure}{.23\textwidth}
  \centering
  \includegraphics [width=\textwidth] {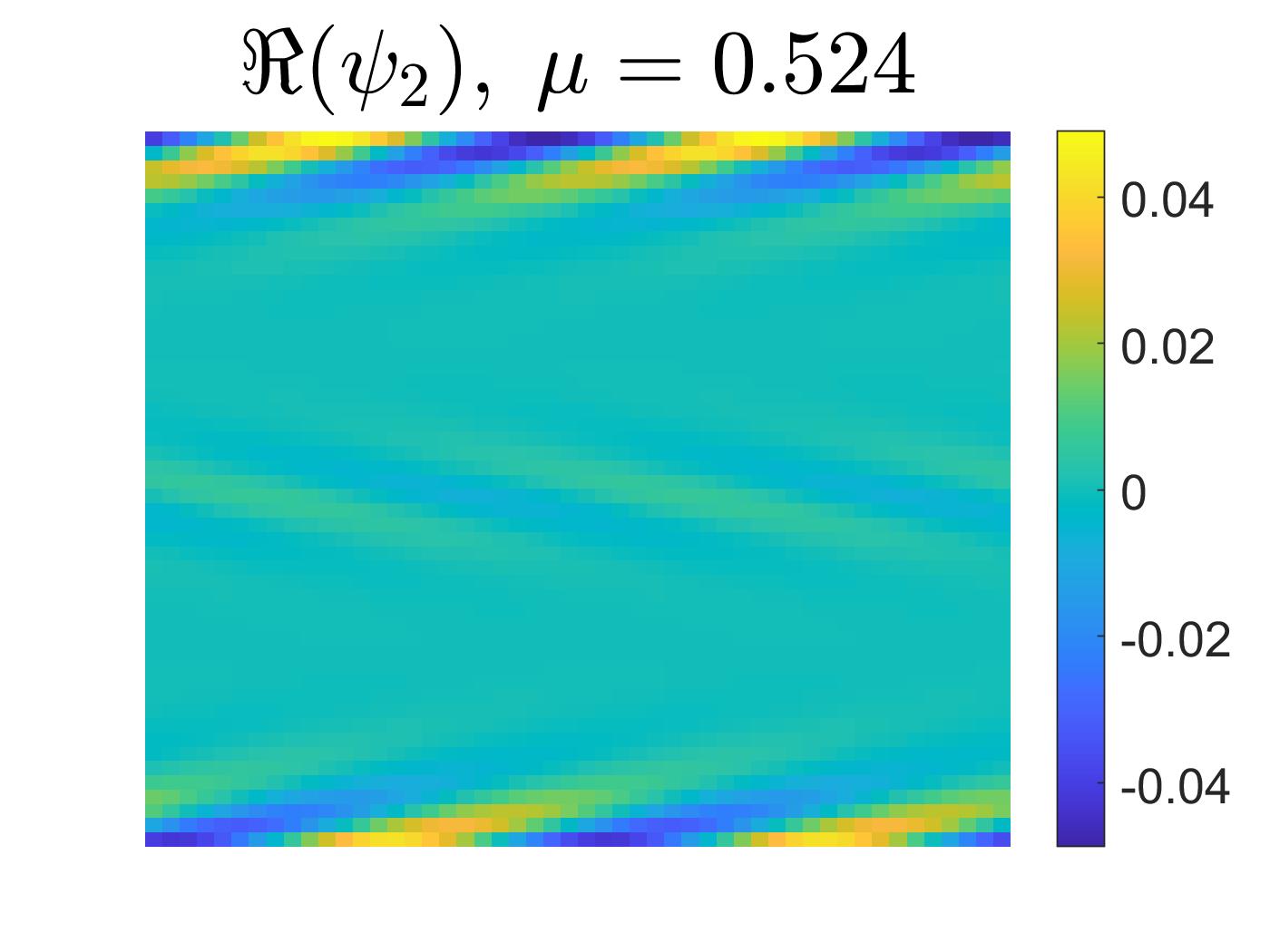}
\end{subfigure}
\begin{subfigure}{.23\textwidth}
  \centering
  \includegraphics [width=\textwidth] {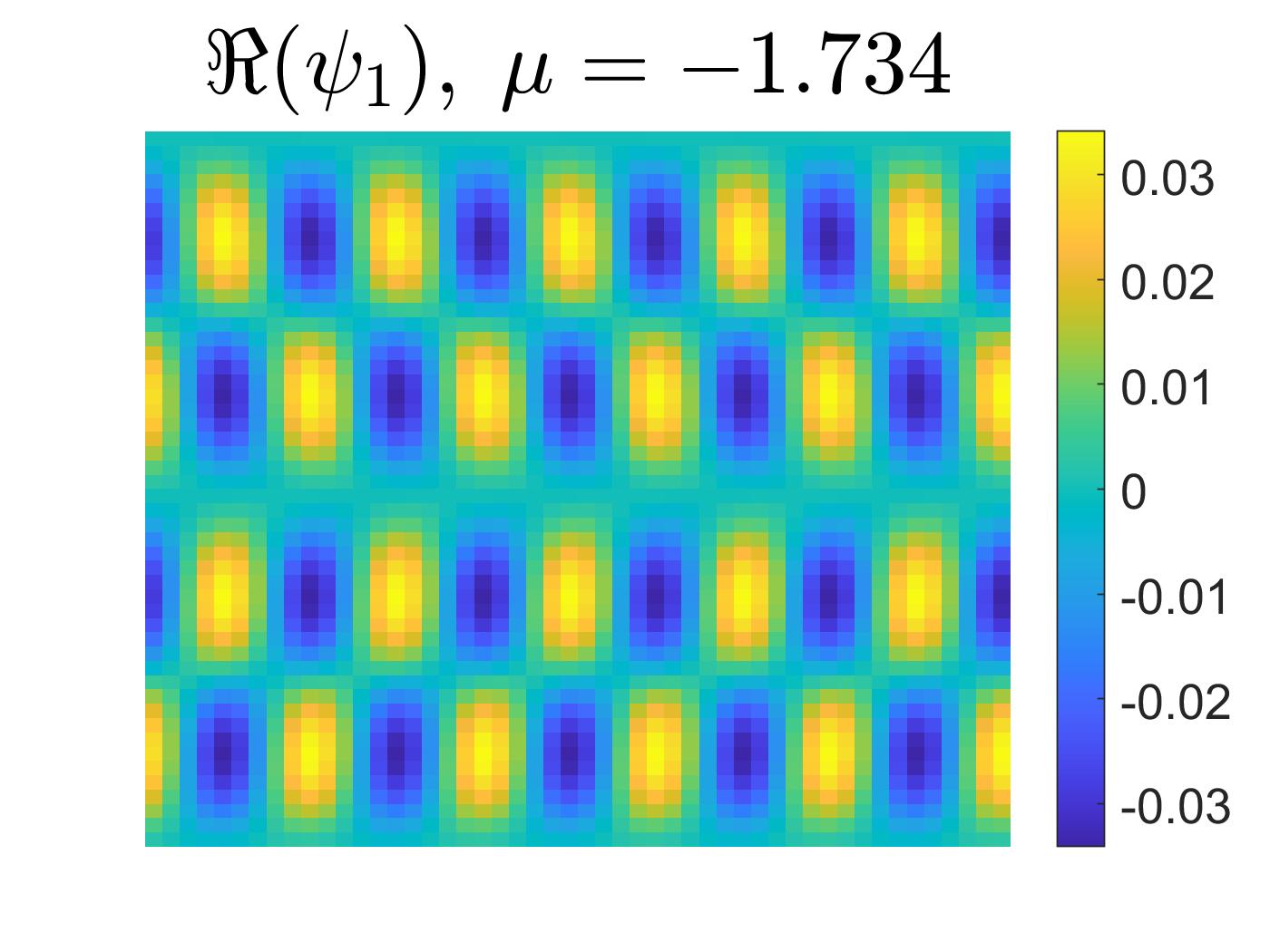}
\end{subfigure}
\caption{Eigenfunctions of the periodic $2 \times 2$ Dirac Hamiltonian (top) and $p$-wave superconductor (bottom) models. 
First three columns: edge states localized at the $y=0$ and $y = \pm L_y/2$ interfaces;
far right: bulk states that do not contribute to the edge current. The plotted vector components and corresponding eigenvalues are labeled.}
\label{efns2x2Dirac}
\end{figure}

To illustrate the stability (or not) of the interface current observable $\sigma_I$ in \eqref{eq:sigmaI}, we consider a number of numerical simulations. The periodized operator $H_\kp$ with $\kp=L^{-1}$ introduced in the preceding section is now amenable to standard discretizations, for instance spectral methods \cite{Trefethen}, which we have implemented for the $2 \times 2$ Dirac model, the $p$-wave superconductor model, and the equatorial wave model analyzed in section \ref{sec:applications}.

In all simulations, we discretize a torus of size $L_x \times L_y = 24 \times 24$ uniformly into $N_x \times N_y = 64 \times 64$ grid points with periodic boundary conditions.
Point-wise multiplication operators are represented as diagonal matrices, and derivatives $\tilde{D}_\alpha \approx \frac{1}{i} \partial_\alpha$ ($\alpha = x,y$) are computed in the Fourier basis 
$
    \tilde{D}_\alpha = F_\alpha^{-1} \Lambda_\alpha F_\alpha,
$
where $F_\alpha$ is the discrete Fourier transform with respect to direction $\alpha$,
and
\begin{align*}
\Lambda_\alpha = \frac{2 \pi}{N_\alpha L_\alpha} \diag\Big(-\frac{N_\alpha}{2}, -\frac{N_\alpha}{2} + 1, \dots, \frac{N_\alpha}{2} - 1\Big).
\end{align*}
Note that $F_\alpha$ is unitary, so that $F_\alpha^{-1} = F_\alpha^*$.
We refer the reader to \cite{Trefethen} for more details on spectral methods and their accuracy.

\begin{figure}[t] 
\begin{subfigure}{.23\textwidth}
  \centering
  \includegraphics [width=\textwidth] {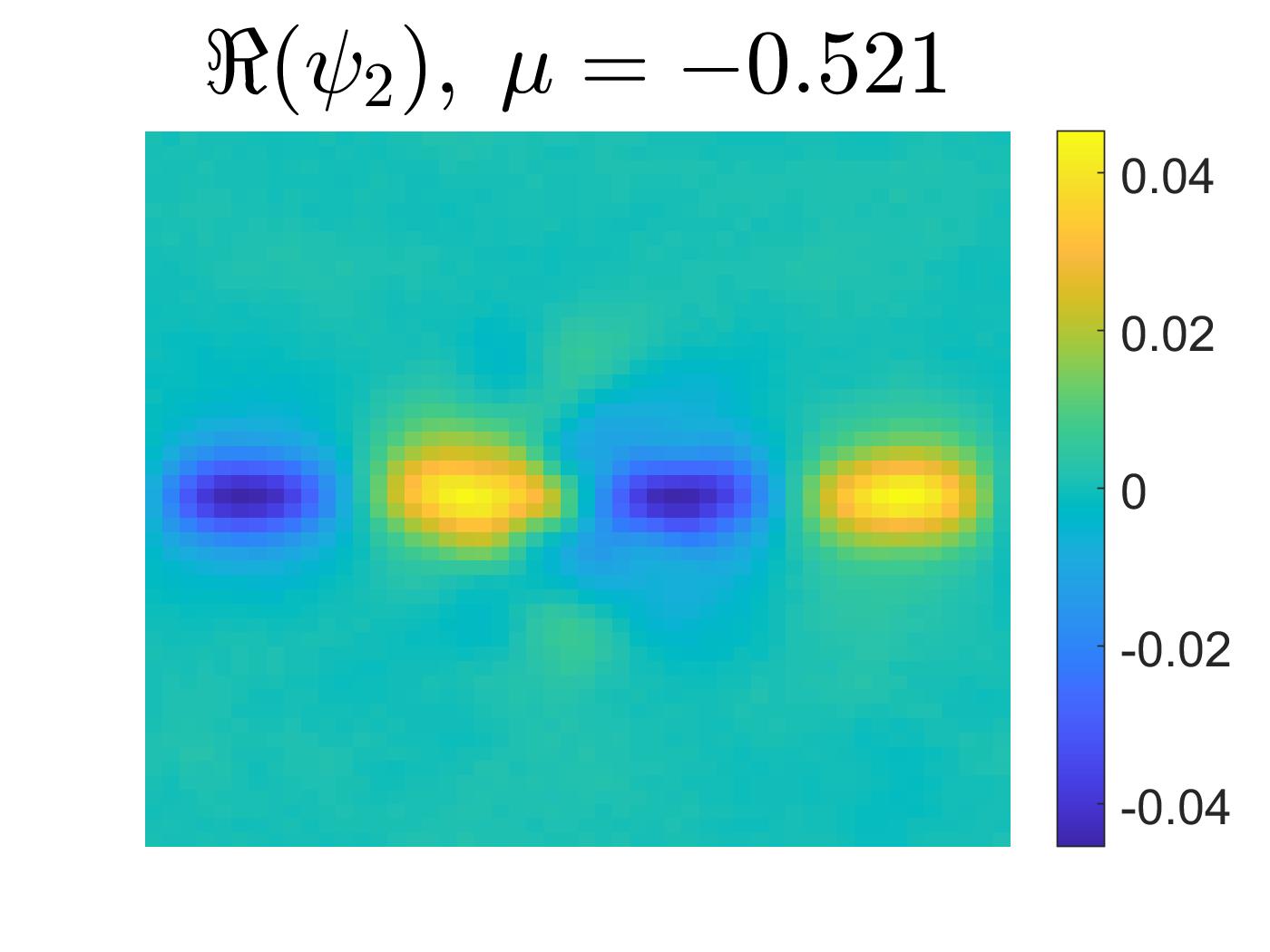}
\end{subfigure}
\begin{subfigure}{.23\textwidth} 
  \centering
  \includegraphics [width=\textwidth] {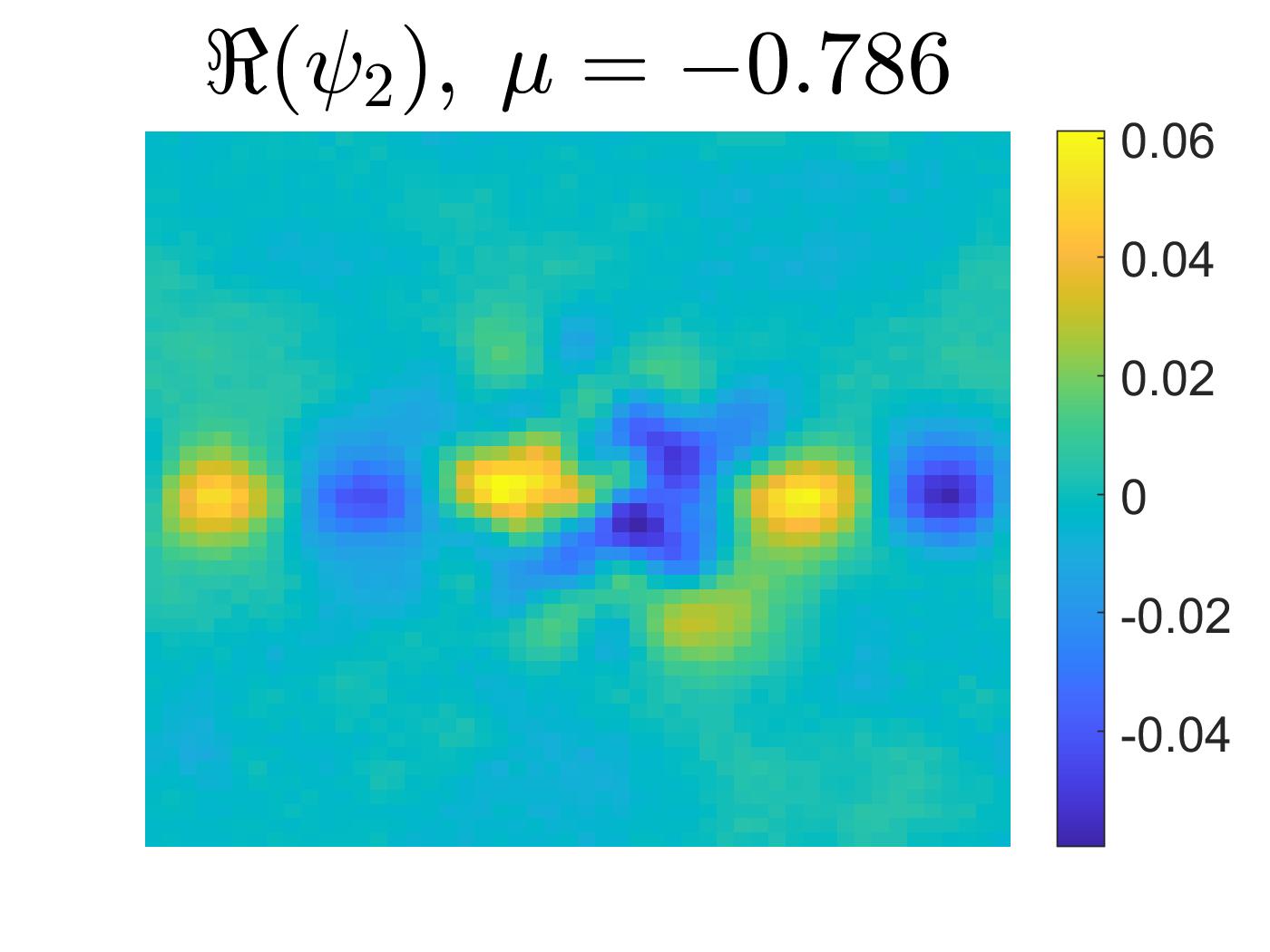}  
\end{subfigure}
\begin{subfigure}{.23\textwidth}
  \centering
  \includegraphics [width=\textwidth] {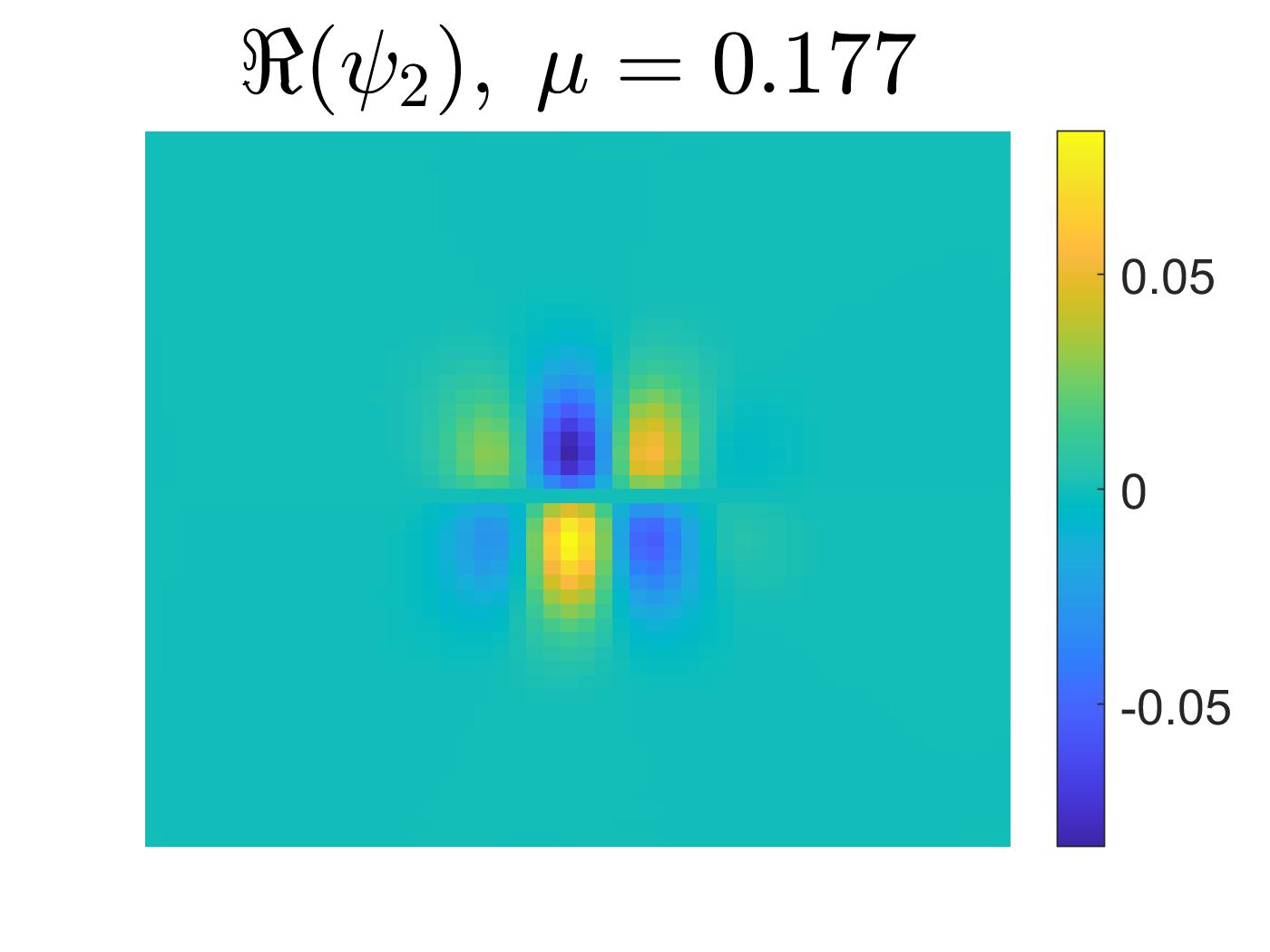}
\end{subfigure}
\begin{subfigure}{.23\textwidth}
  \centering
  \includegraphics [width=\textwidth] {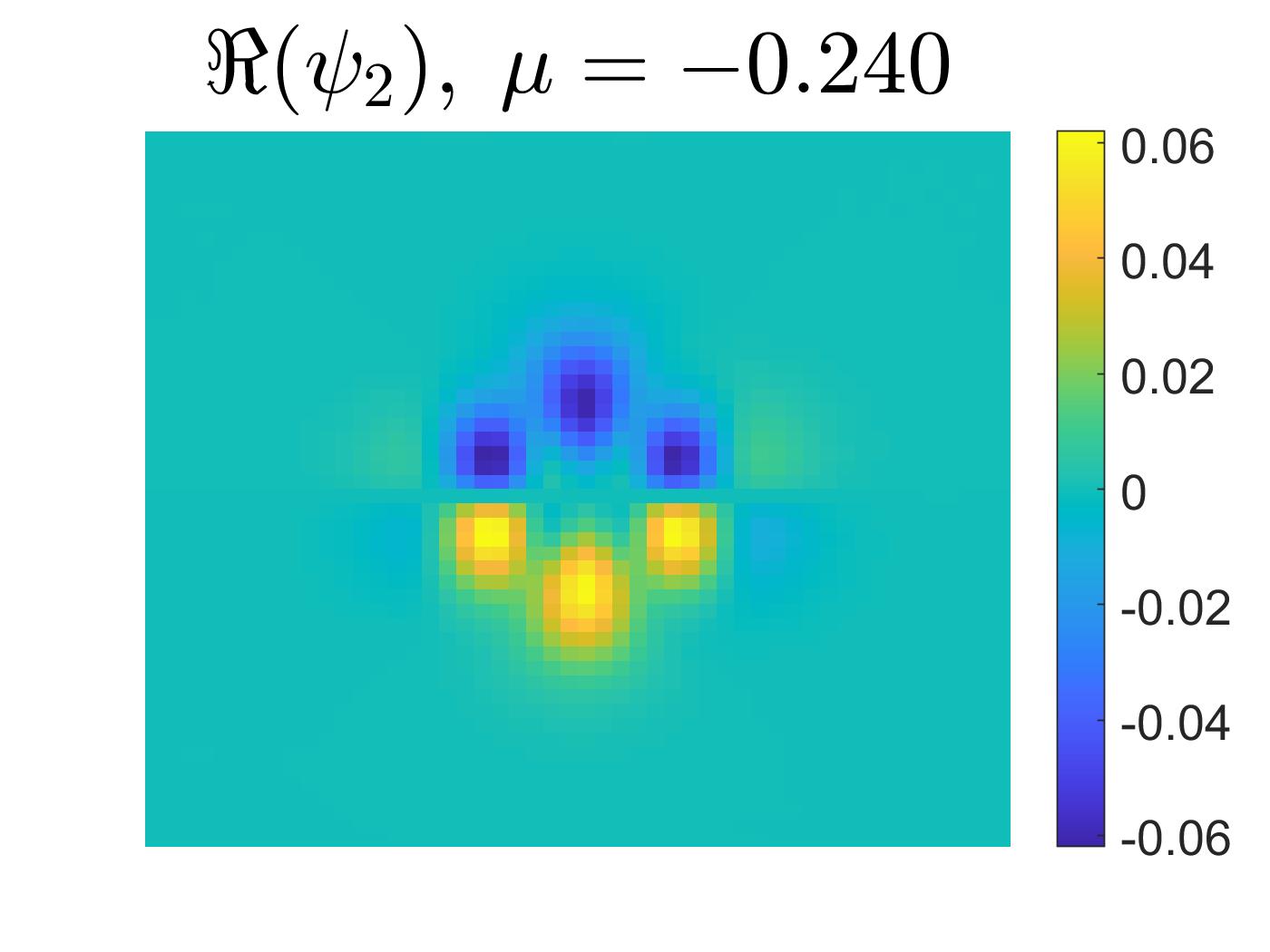}
\end{subfigure}
\newline
\begin{subfigure}{.23\textwidth}
  \centering
  \includegraphics [width=\textwidth] {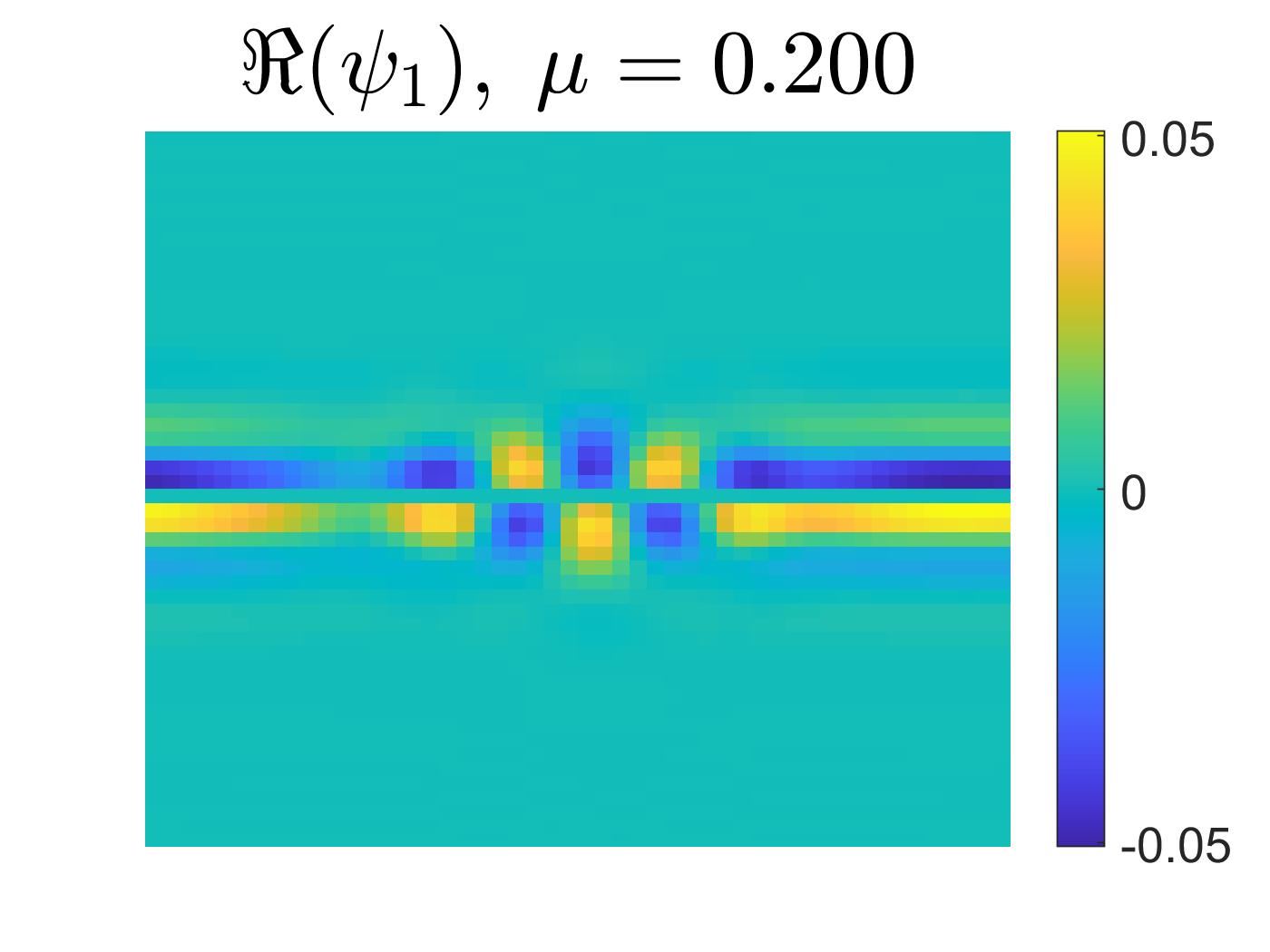}
\end{subfigure}
\begin{subfigure}{.23\textwidth}
  \centering
  \includegraphics [width=\textwidth] {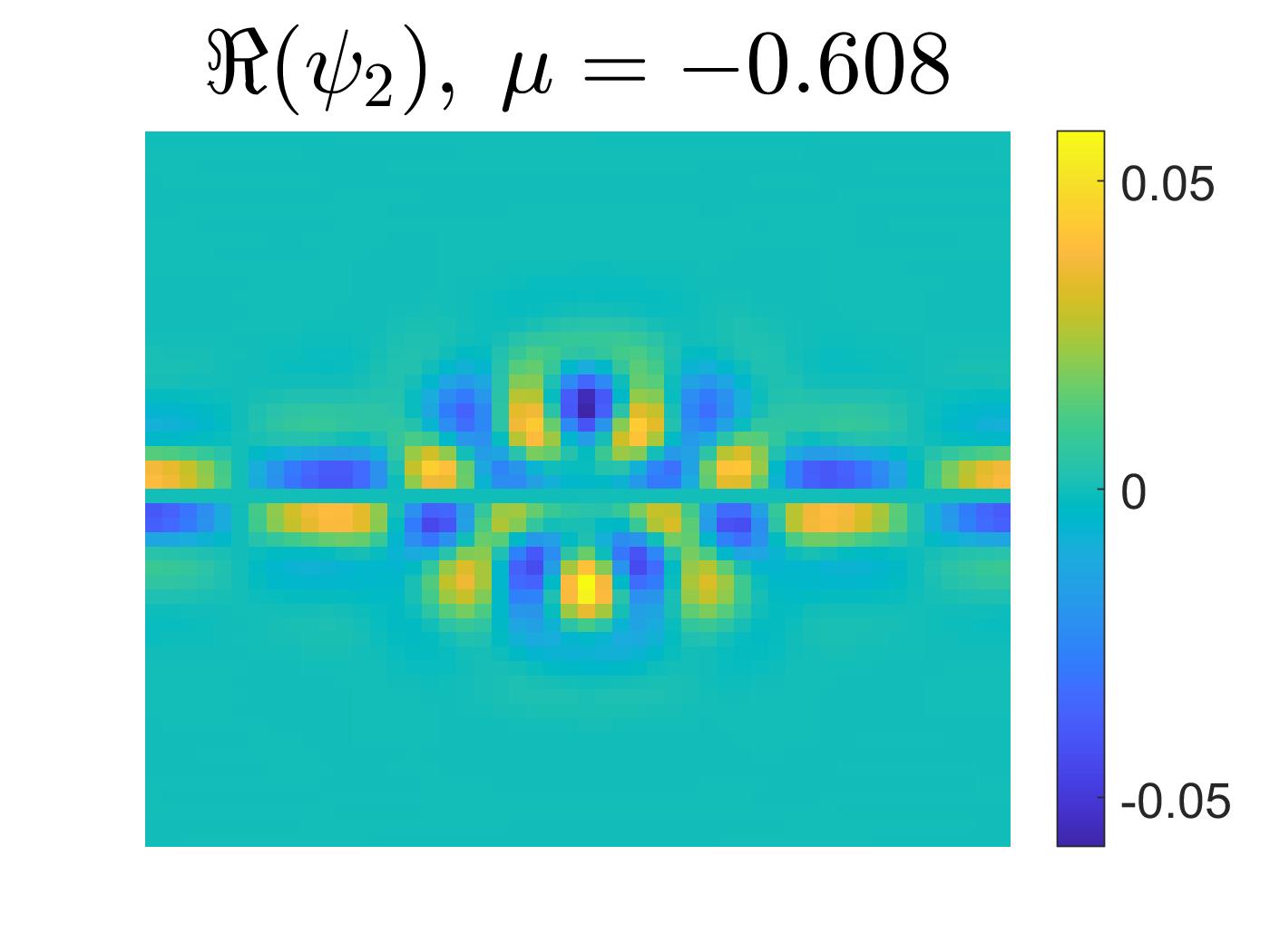}
\end{subfigure}
\begin{subfigure}{.23\textwidth}
  \centering
  \includegraphics [width=\textwidth] {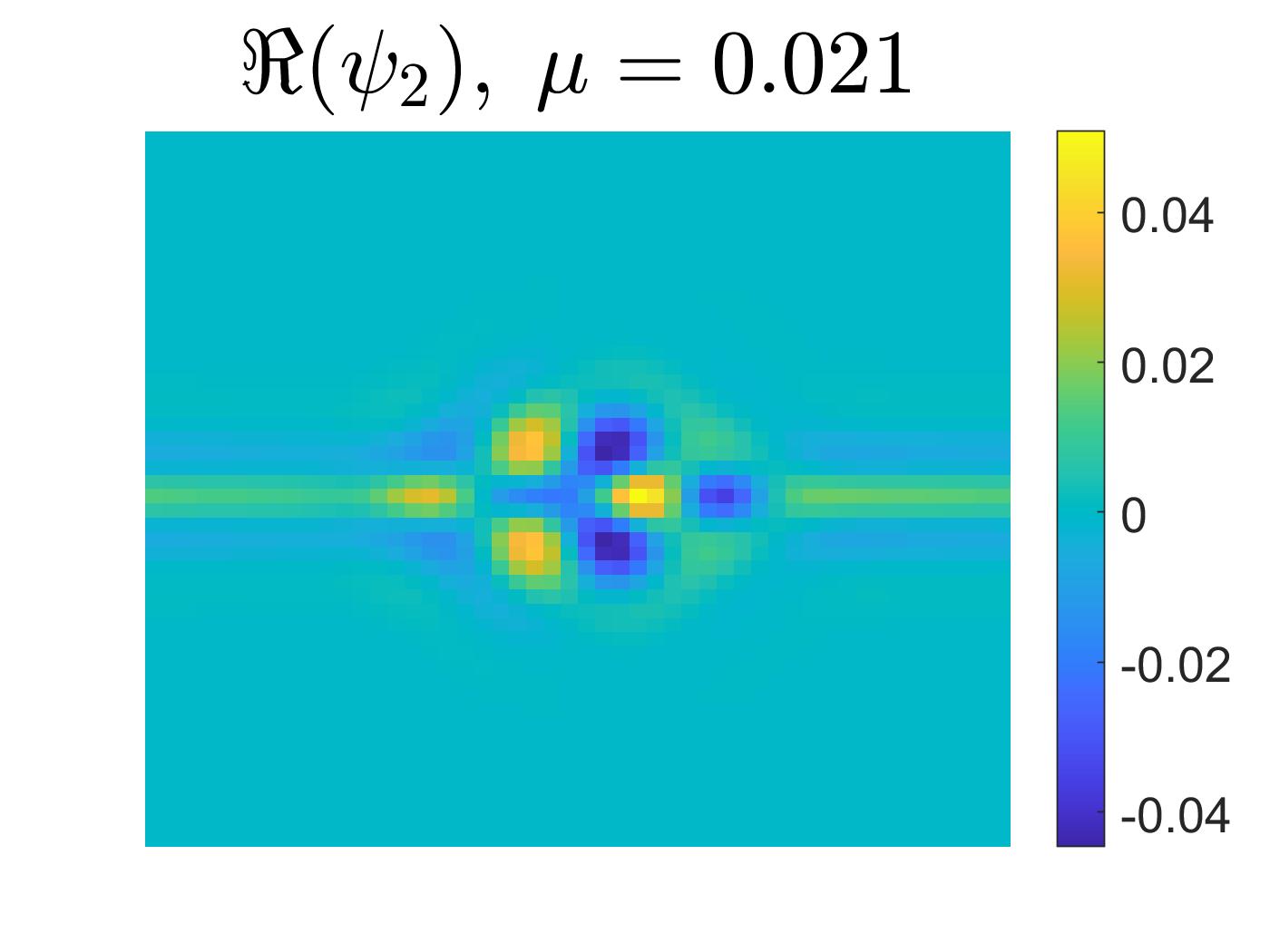}
\end{subfigure}
\begin{subfigure}{.23\textwidth} 
  \centering
  \includegraphics [width=\textwidth] {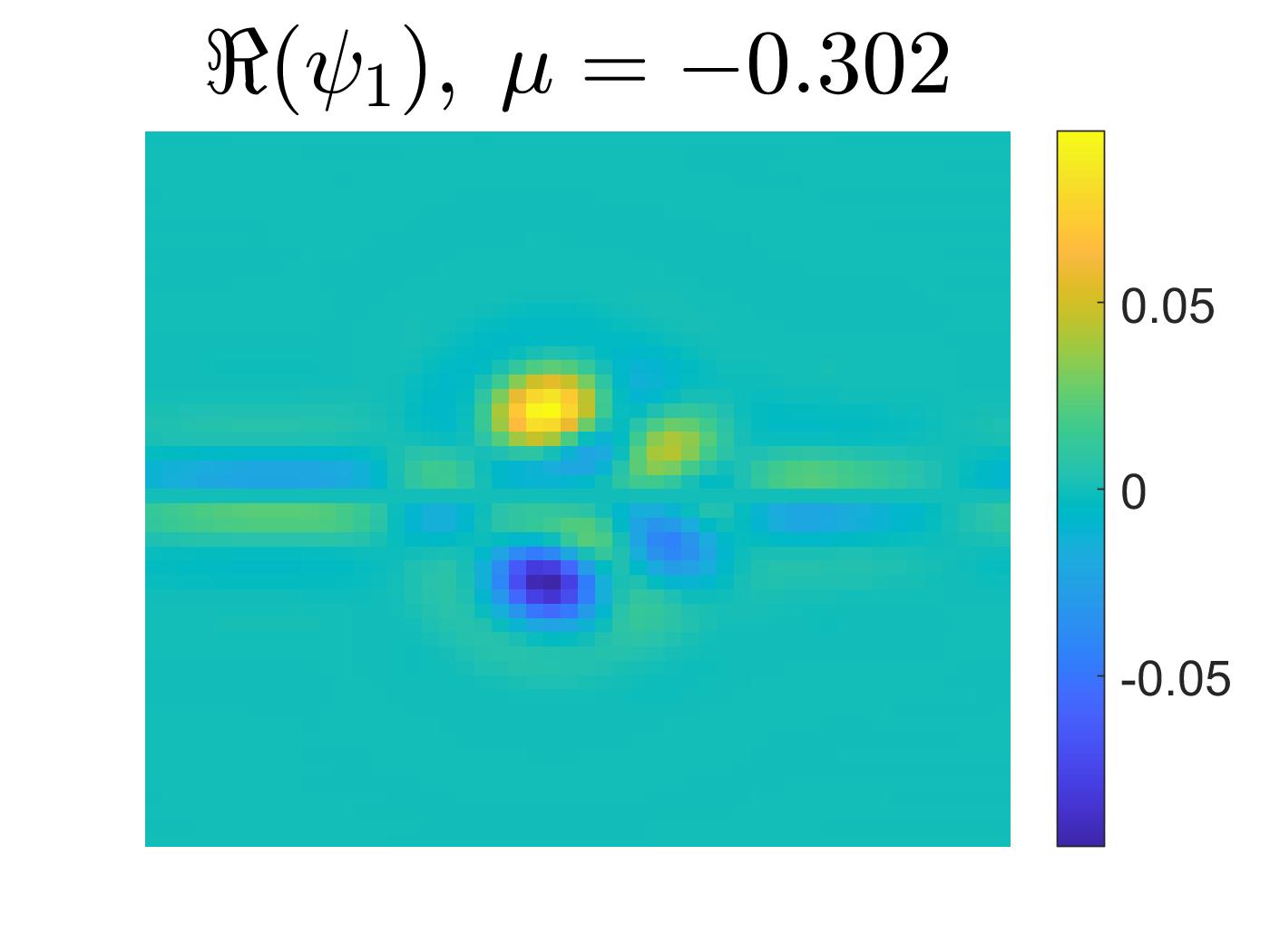}  
\end{subfigure}
\caption{Perturbed eigenfunctions of $2 \times 2$ Dirac (top) $p$-wave superconductor (bottom) models. For all plots, $r = 2$ and $a = 7.5$. 
We see a mix of propagating and evanescent modes. Despite qualitative differences between perturbed and unperturbed eigenfunctions, 
the edge current remains approximately the same.}
\label{efnsPerturbed}
\end{figure}

As in section \ref{sec:periodic}, we define the filtered edge current by $\tilde{\sigma}_I (H) := \Tr i Q [H,P] \varphi '(H)$
where $P = P(x)$ and $Q = Q(x,y) = Q_X (x) Q_Y (y)$ are point-wise multiplication operators satisfying
\begin{align*}
    P (x) = 
    \begin{cases}
    1, &\delta_x \le x \le 3 L_x/8\\
    0, &-3 L_x/8 \le x \le -\delta_x
    \end{cases}
\qquad \text{and} \qquad
    Q_\Theta (\theta) = \begin{cases}
    1, &|\theta| \le L_\theta/4\\
    0, &|\theta| \ge L_\theta/4 + \delta_\theta
    \end{cases},
\end{align*}
for some 
$0 < \delta_\theta \ll L_\theta / 4$.
The Hamiltonians are periodized so that
the $y$ dependent terms ($m(y)$ for the $2 \times 2$ Dirac system and $c(y)$ for the $p$-wave superconductor) are equal to $1$ whenever $\delta ' \le y \le L_y/2 - \delta '$ and $-1$ for $-L_y/2 + \delta ' \le y \le- \delta '$, and are smoothly connected in between. Here, $\delta '$ is a positive constant that is small relative to $L_y$, but large enough for the interval $(-\delta ', \delta ')$ to contain at least $3$ grid points.
The eigenvalues of $H$ in the support of $\varphi '$ (and the corresponding eigenvectors) are computed using the ``eigs'' command in MATLAB.
Since the number of these eigenvalues is typically much smaller than the dimension of $H$, 
we avoid the computational expense of a full spectral decomposition.

\subsection{Dirac model and superconductor $p-$wave model}
For the $2 \times 2$ models,
we consider perturbations of the form 
\begin{align*}
    V(x,y) = (\hat{n} \cdot \vec{\sigma}) v(x,y) , \qquad  v(x,y) =  r \exp(-a^2/(a^2-x^2-y^2))\mathbbm{1}_{x^2 + y^2 < a^2}
\end{align*}
where 
$\hat{n} \in \mathbb{R}^4$
is a unit normal vector and $\vec{\sigma} = (\sigma_0, \sigma_1, \sigma_2, \sigma_3)$. Here, the $\sigma_j$ are the Pauli matrices (with $\sigma_0$ the $2 \times 2$ identity matrix).
We observe that $\tilde{\sigma}_I$ is stable with respect to such perturbations.
Namely, if we fix the support of $v$ to lie well within $\supp (Q)$ while increasing the amplitude of that perturbation, the edge current does not change by much (even if the amplitude is very large).
Similarly, for a fixed amplitude, the edge current remains close to its original value until $v$ becomes highly delocalized. 
See Figure \ref{stabilityBoth} (left), where we have plotted the dependence of the edge current on the perturbation strength and localization for the $2 \times 2$ Dirac and $p$-wave superconductor models. For many other choices of parameters, the edge current remains within $2\%$ of its unperturbed value for values of $r$ and $a$ much larger than $10$.
\begin{figure}[t]
\centering
 \begin{tikzpicture}    
 \matrix (fig) [matrix of nodes]{
 \includegraphics[width=2.00in]{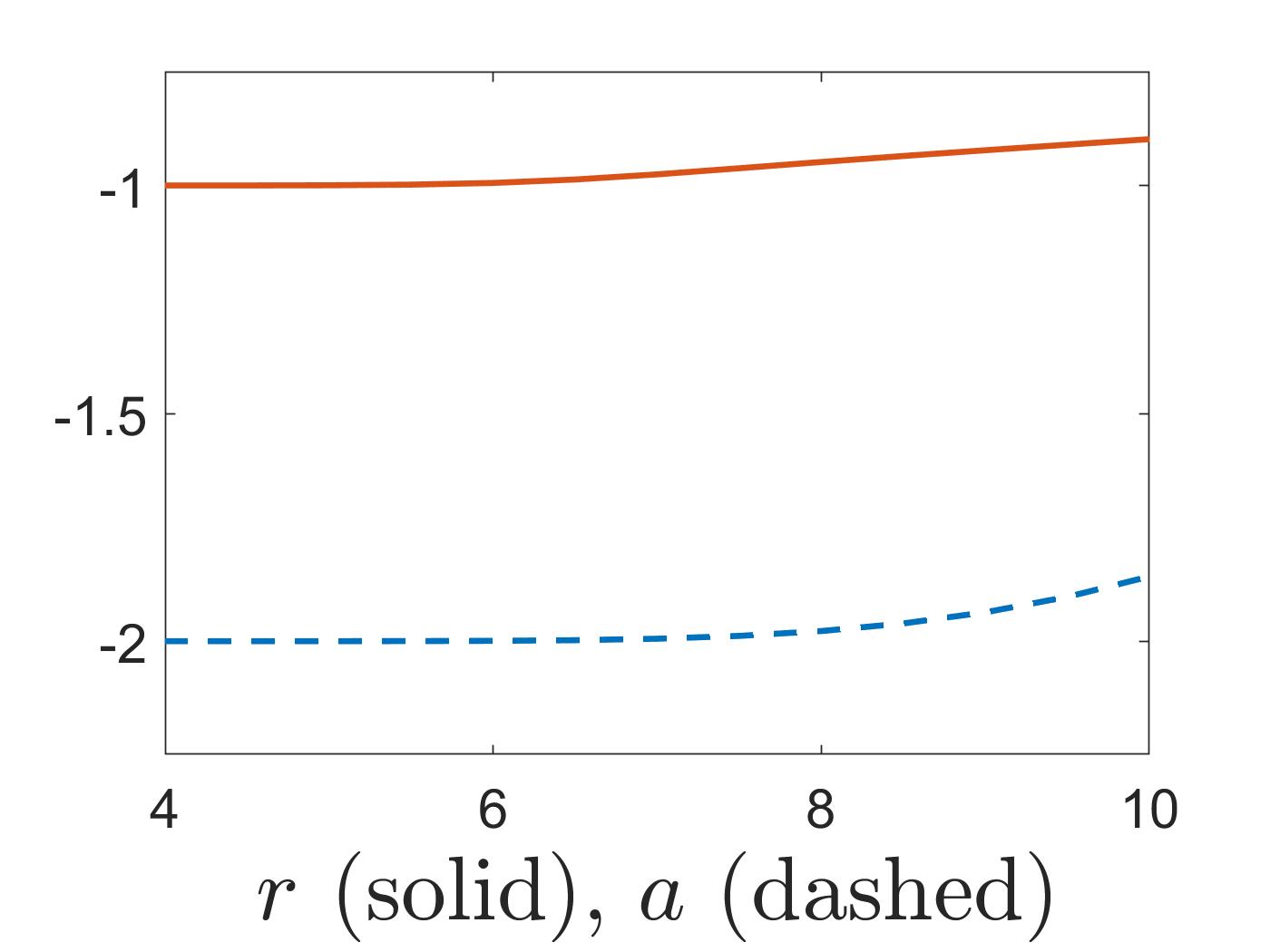}
 &
 \includegraphics[width=2.00in]{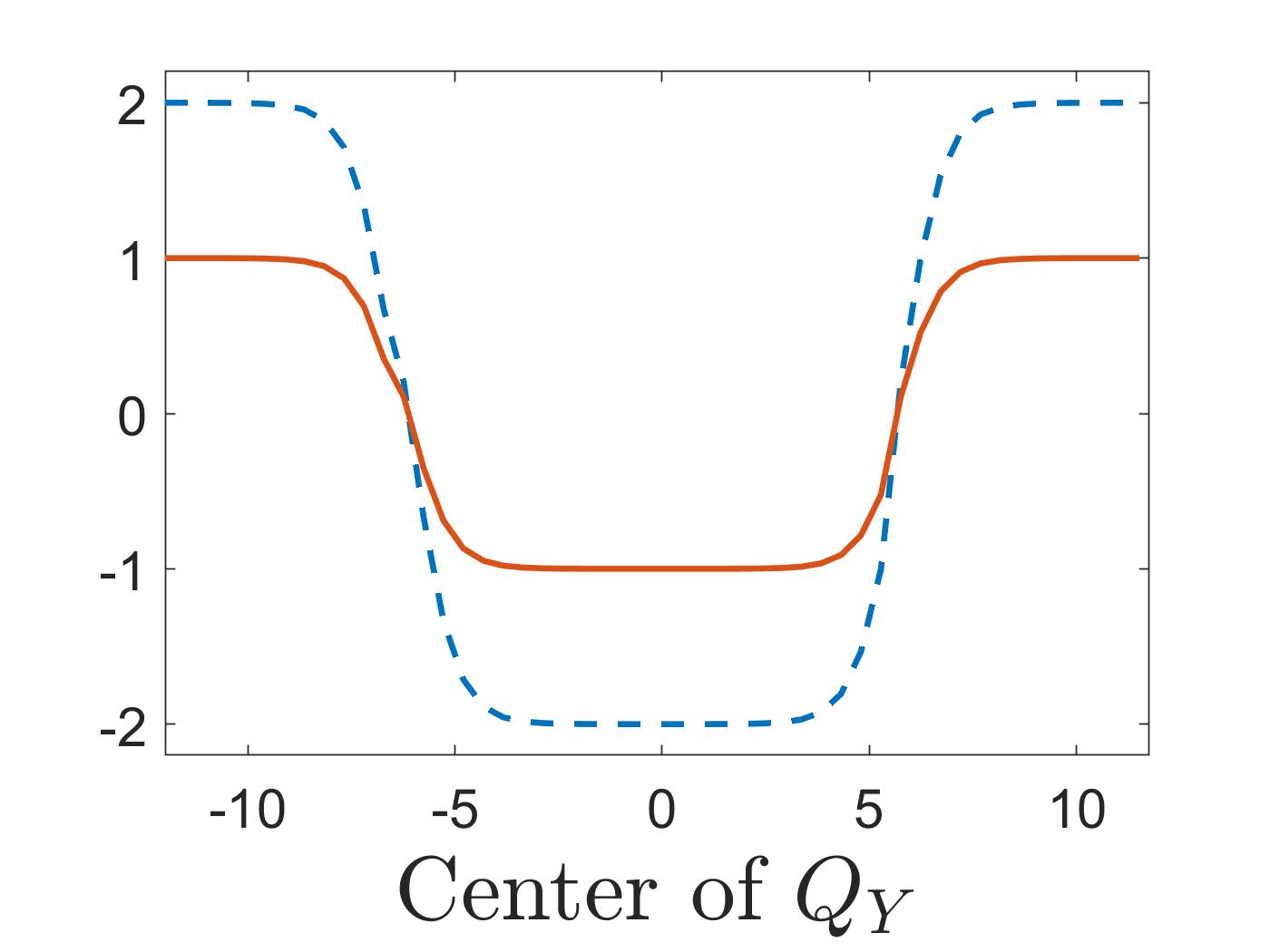}
 \includegraphics[width=2.00in]{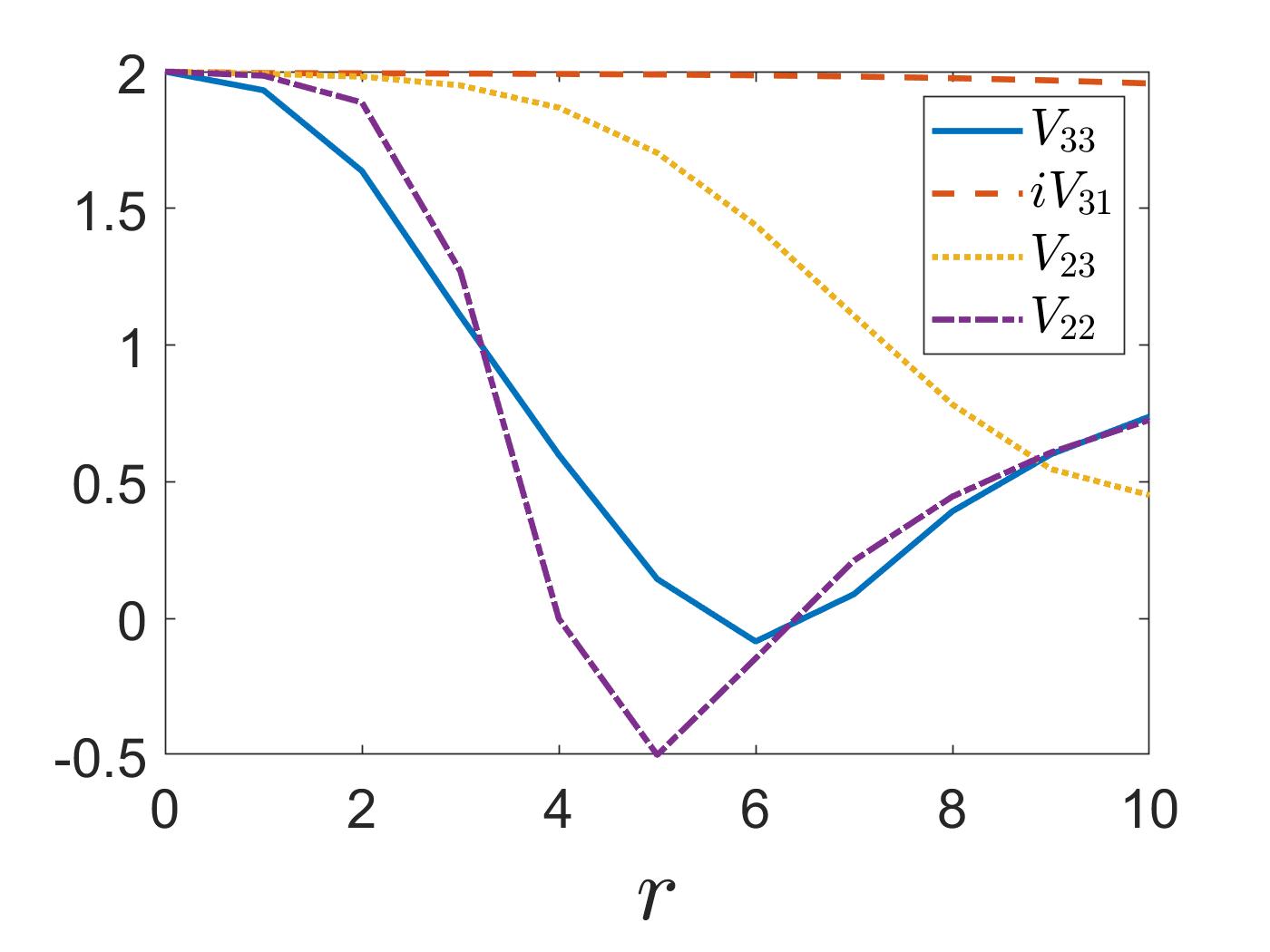}\\
  };
 \path (fig-1-1.south west)  -- (fig-1-1.north west) node[midway,above,sloped]{$2\pi \tilde{\sigma}_I$};
 \end{tikzpicture}
\caption{
The left panel demonstrates the numerical stability of $\tilde{\sigma}_I$ for the $2 \times 2$ Dirac (solid line) and $p$-wave superconductor (dashed line) models.
For the solid line, we fixed $a = 7.5$ (with $r$ varying) and the dashed line corresponds to $r = 10$ (with $a$ varying).
The center panel shows the edge current for the $2 \times 2$ Dirac (solid line) and $p$-wave superconductor (dashed line) models as a function of the center of $Q_Y$. As expected, we get $-1,-2$ when the filter selects the increasing domain wall and $1,2$ when the filter selects the decreasing domain wall, with a sharp transition in between.
The right panel plots the edge current for the $3 \times 3$ model as a function of perturbation strength, for perturbations in four distinct matrix elements.
The edge current is stable with respect to perturbations of the other five matrix elements; see text.}
\label{stabilityBoth}
\end{figure}
The empirical stability of $\tilde{\sigma}_I$ is entirely consistent with the theoretical results from earlier sections.

We present examples of computed eigenfunctions for the $2 \times 2$ Dirac and $p$-wave superconductor models in  Figures \ref{efns2x2Dirac} and \ref{efnsPerturbed}.
We observe that perturbations $V$ of the above form strongly alter the edge states. They are no longer a superposition of terms of the form $e^{i\xi x} \psi (y)$ and we even recognise localized modes (see top right of Figure \ref{efnsPerturbed}) that do not contribute to the edge current.
Thus the robustness of $\tilde{\sigma}_I$ observed in Figure \ref{stabilityBoth} occurs in spite of the instability of the eigenfunctions with respect to perturbations.


One may also consider filters $Q_{\Delta_X, \Delta_Y} (x,y) = Q_{\Delta_X} (x) Q_{\Delta_Y} (y)$ with different centers, where
\begin{align*}
    Q_{\Delta_\Theta} (\theta) = \begin{cases}
    1, &\theta \in [-\frac{L_\theta}{4} - \Delta_\Theta, \frac{L_\theta}{4} - \Delta_\Theta]\\
    0, &\theta \notin (-\frac{L_\theta}{4} - \Delta_\Theta -\delta_\theta, \frac{L_\theta}{4} - \Delta_\Theta + \delta_\theta).
    \end{cases}
\end{align*}
When $\Delta_Y = L_y/2$, $Q$ would instead test the edge current associated to the opposite, spurious, domain wall. As the center $\Delta_Y$ of the filter increases, we expect the numerical value of $\sigma_I$ to shift from the edge current of one domain wall to the that of the next domain wall. This is confirmed by the plot in Figure \ref{stabilityBoth} (center) obtained for the 
$2 \times 2$ Dirac and $p$-wave superconductor models.
Translating the center of $Q$ in the $x$-direction would have the same effect, as the sign of $P'$ on the support of $Q$ changes when $\Delta_X = L_x/2$.

\subsection{Equatorial waves}

We finally carry out some numerical simulations on the model of equatorial waves presented in section \ref{sec:applications}. The unperturbed Hamiltonian ($\mu=0$) with symbol given in \eqref{eq:a0geo} is
\begin{align} \label{3x3}
    H_0 = (D_x, D_y, -f(y)) \cdot \Gamma,
\end{align}
where $\Gamma = (\gamma_1, \gamma_4, \gamma_7)$ with the $\gamma_i$ denoting ($3 \times 3$) Gell-Mann matrices.
The Coriolis force is given by $f(y)$, which changes signs across the equator.
As shown in \cite{3}, the edge current of such a system depends on the choice of $f$ when $\mu=0$ whereas the theory presented in section \ref{sec:applications} shows that the edge current equals $2$ as soon as $\mu\not=0$. For example, $2 \pi \sigma_I = 1$ if $f = f_0 \sgn(y)$, and $2 \pi \sigma_I = 2$ if $f(y) = \beta y$, with $f_0, \beta > 0$. This {\em violates the bulk-edge correspondence}, which states that $2\pi\sigma_I$ is independent of the profile $f(y)$.

In fact, the edge current is stable under perturbations of the form $V = \diag (V_{11}, 0, 0)$, but not $V = \diag (0, V_{22}, V_{33})$ when $\mu=0$ \cite{3}.
Our objective here is to verify these results numerically for $f(y)$ a smooth periodic function satisfying the same conditions as $m(y)$ and $c(y)$ above. 
For example, we see that if
$V_0 = 5 g_{L_y/4} (y)$
(where $g_\sym(y)$ is the pdf of a Gaussian with mean zero and standard deviation $\sym$),
then $2 \pi \tilde{\sigma}_I = 1.9950$ if $V=0$, $2 \pi \tilde{\sigma}_I = 1.9941$ if $V = \diag (V_0, 0,0)$, $2 \pi \tilde{\sigma}_I = -0.5394$ if $V = \diag (0,V_0,0)$, and $2 \pi \tilde{\sigma}_I = 0.1430$ if $V = \diag (0,0,V_0)$.
More generally, we have verified numerically that $\tilde{\sigma}_I$ is stable, as predicted by theory, under perturbations of the form $V(y)$ times any of the matrices
\begin{align}\label{stableV}
    \begin{pmatrix}1 & 0 & 0 \\
   0 & 0 & 0 \\
    0 & 0 & 0 \end{pmatrix},
    \begin{pmatrix}0 & 1 & 0 \\
   1 & 0 & 0 \\
    0 & 0 & 0 \end{pmatrix},
    \begin{pmatrix}0 & 0 & 1 \\
   0 & 0 & 0 \\
    1 & 0 & 0 \end{pmatrix},
    \begin{pmatrix}0 & i & 0 \\
   -i& 0 & 0 \\
    0 & 0 & 0 \end{pmatrix},
    \begin{pmatrix}0 & 0 & i \\
   0 & 0 & 0 \\
    -i & 0 & 0 \end{pmatrix},
    \begin{pmatrix}0 & 0 & 0 \\
   0 & 0 & i \\
    0 & -i & 0 \end{pmatrix},
\end{align}
and unstable under all other Hermitian perturbations, see Figure \ref{stabilityBoth} (right).
\begin{figure}[t]
\begin{tikzpicture}    
 \matrix (fig) [matrix of nodes]{
 \includegraphics[width=1.98in]{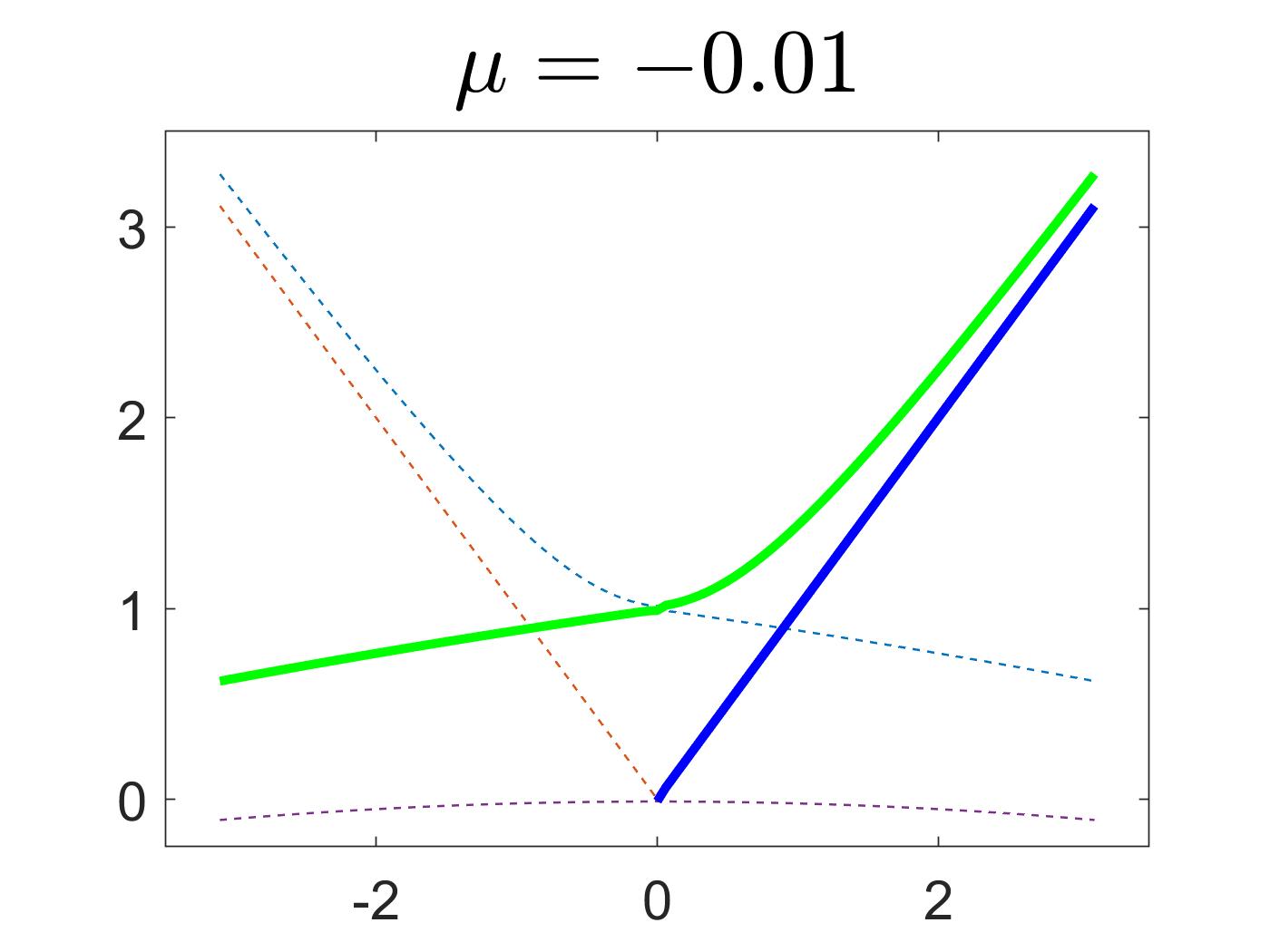}
 &
 \includegraphics[width=1.98in]{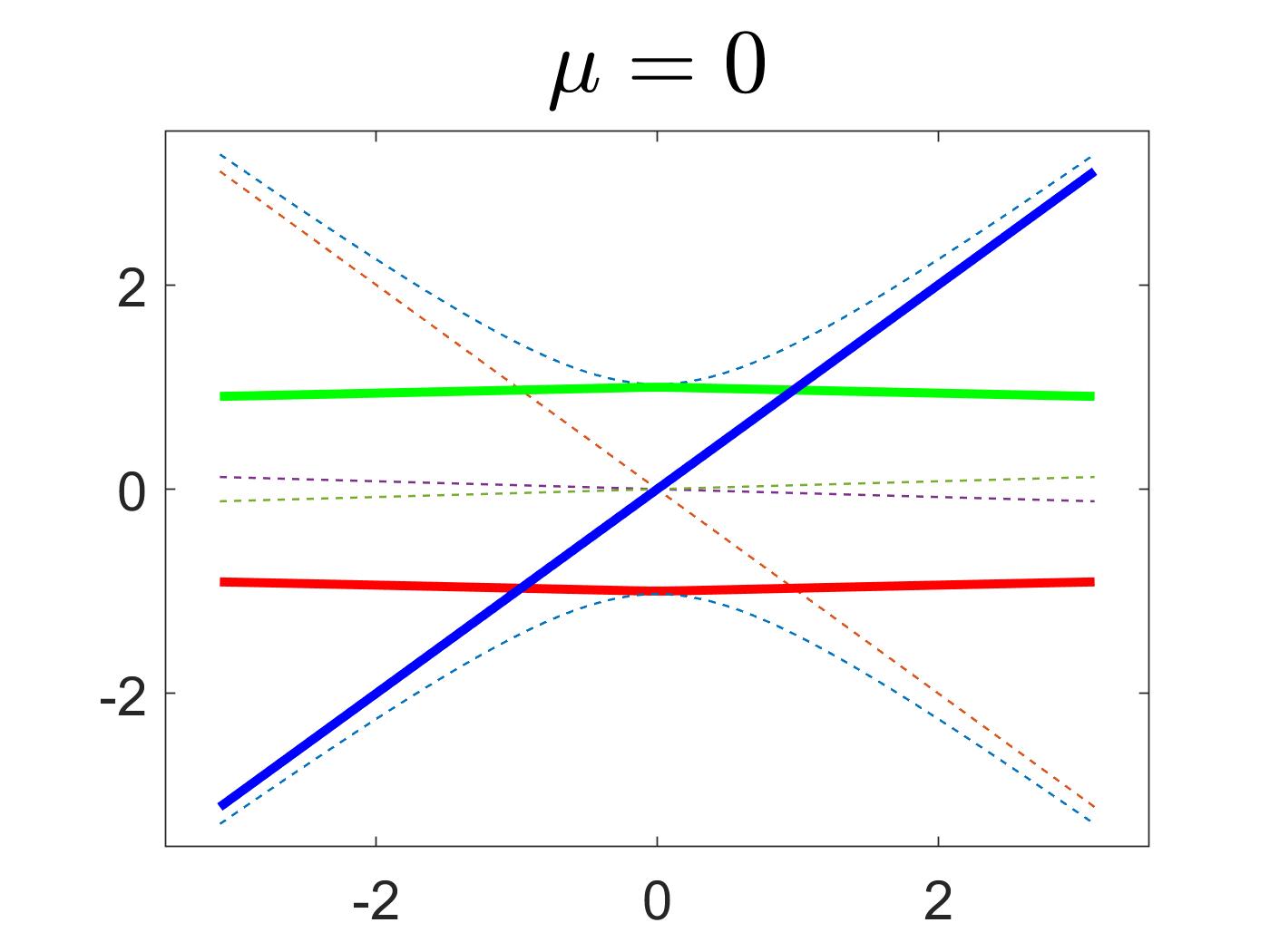}
 &
 \includegraphics[width=1.98in]{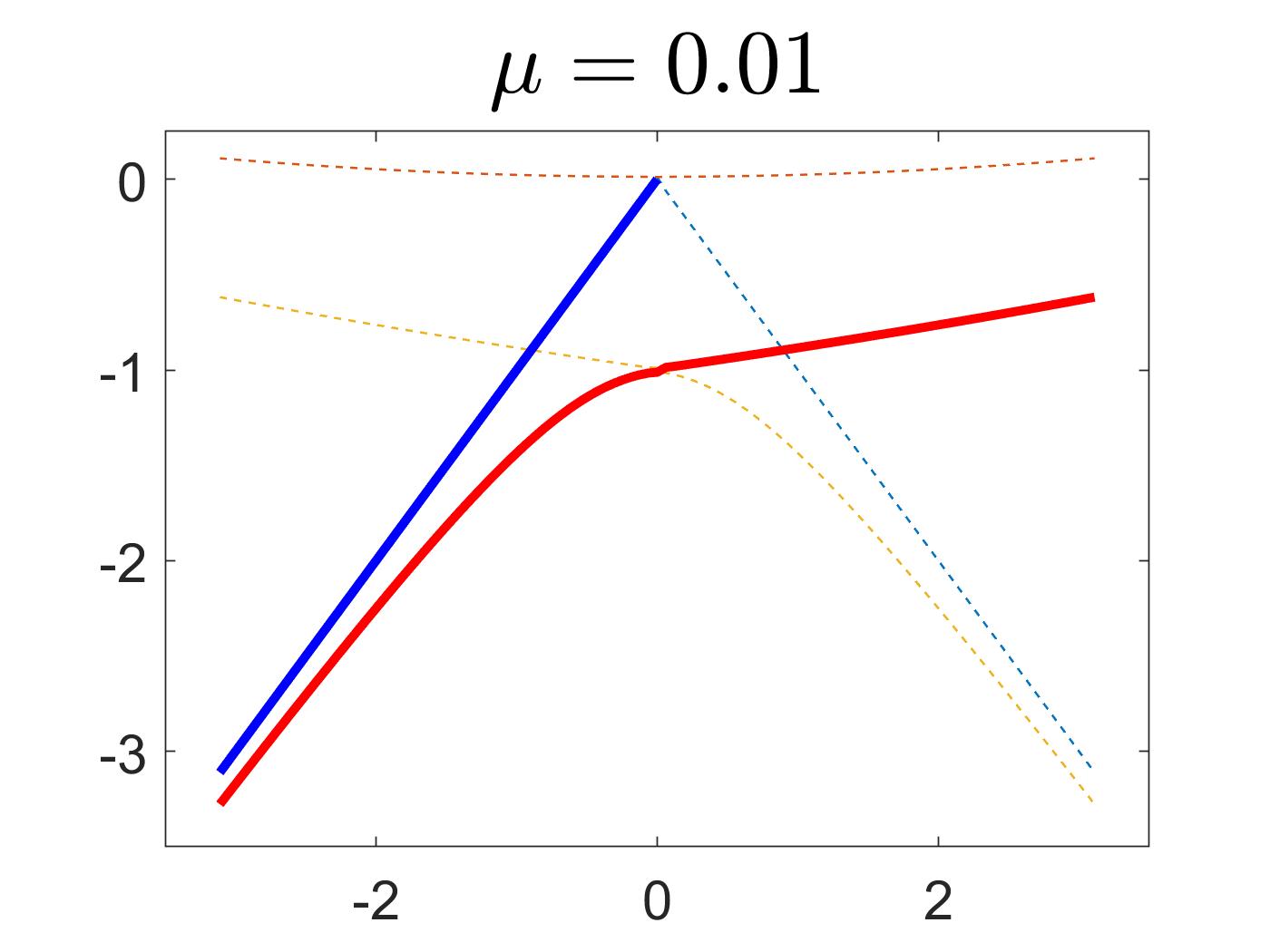}
 \\
 \includegraphics[width=1.98in]{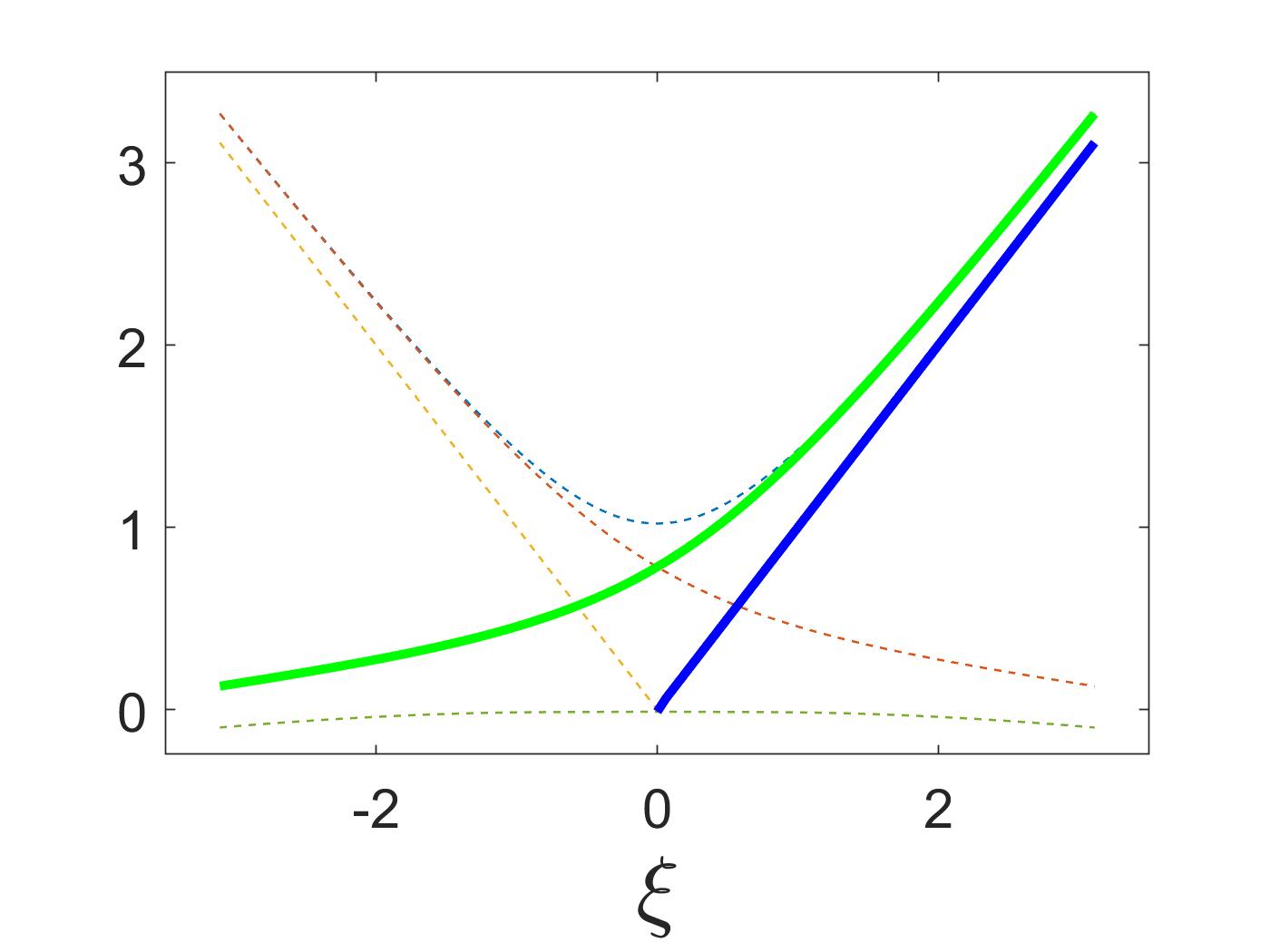}
 &
 \includegraphics[width=1.98in]{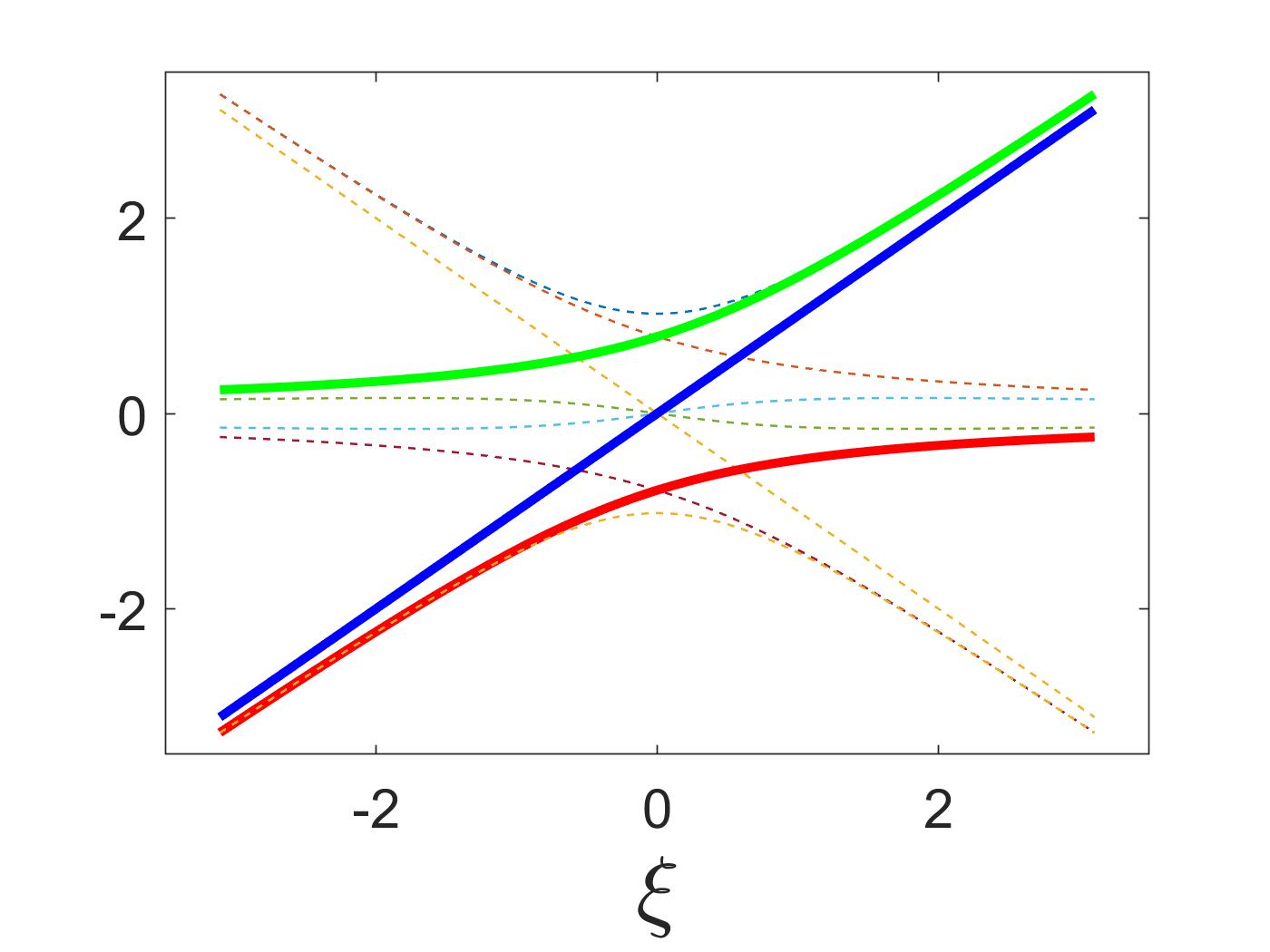}
 &
 \includegraphics[width=1.98in]{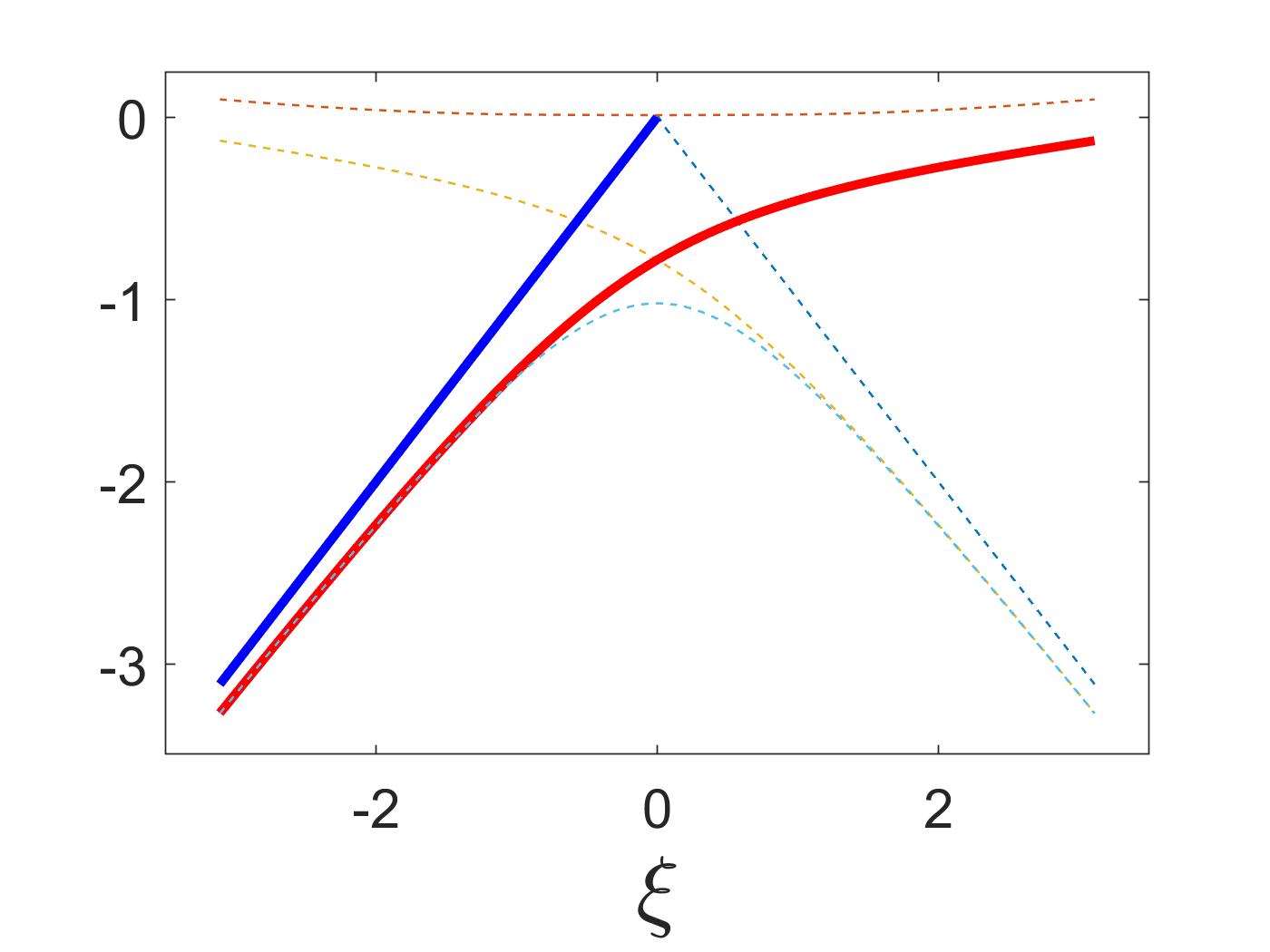}
 \\
  };
 \path (fig-2-1.west)  -- (fig-1-1.west) node[midway,above,sloped]{$E(\xi)$};
 \end{tikzpicture}
\caption{Branches of continuous spectrum for the $3 \times 3$ equatorial wave Hamiltonian, for different values of regularization parameter $\mu$. The top line corresponds to $f(y) = \sgn (y)$ while the bottom line corresponds to $f(y) = \tanh (\beta y)$.
The solid curves represent the nontrivial increasing branches
(and the two flat bands at $\pm 1$ when $\mu = 0$ and $f(y) = \sgn (y)$).
For $\mu = 0$, we omit many eigenvalues approximately equal to $0$, as they correspond to essential spectrum for the continuous problem. When $\mu > 0$ ($\mu < 0$), these eigenvalues populate the region $\{E(\xi) > 0\}$ ($\{E(\xi) < 0\}$), making it difficult to identify branches of spectrum there.
}
\label{tbt}
\end{figure}

Given that $H$ is translationally invariant in $x$, we can approximate its continuous branches of spectrum by computing the eigenvalues of 
$
    \hat{H} [\xi] := (\xi, D_y, -f(y)) \cdot \Gamma
$
as a function of 
\begin{align*}
\xi \in \Big \{\frac{2 \pi j}{N_x}: j \in \{-\frac{N_x}{2}, -\frac{N_x}{2} + 1, \dots, \frac{N_x}{2} - 1\} \Big\},
\end{align*}
see Figure \ref{tbt} ($\mu = 0$).
Consistent with theory \cite{3}, the number of branches passing through $E = \pm 1/2$ depends on the profile of $f$. Namely,
we obtain two nondecreasing branches of continuous spectrum passing through $E = 1/2$ when using $f(y) = \tanh(\beta y)$ and only one with $f(y) = \sgn (y)$.
(The qualitative behavior of the branches is independent of the above perturbations under which $\tilde{\sigma_I}$ is stable.)
Here, $\beta > 0$ is sufficiently small so that the transition of $f(y)$ from values near $-1$ to values near $1$ occurs over at least several grid points. Of course, $f(y)$ is periodically wrapped, so that it in fact equals $-\tanh (\beta (y \mp L_y/2))$ in the vicinity of $y = \pm L_y/2$.

Given \cite[Theorem 2.1 and Appendix B]{3} and the fact that $\sigma_I (H) := \Tr i [H,P] \varphi'(H) = 0$ for all matrices $H$ and $P$, we would expect the number of \emph{signed} crossings of $E = 1/2$ to be $0$. That is,
every nondecreasing branch of spectrum passing through $E = 1/2$ should be accompanied by a nonincreasing branch that also passes through $E = 1/2$, as we observe. But these nonincreasing branches correspond to eigenfunctions that localize at $y = \pm L_y/2$ and thus 
become insignificant when we apply $Q$ to compute $\tilde{\sigma}_I$.

The plots in Figure \ref{tbt} with $\mu \ne 0$ correspond to the regularized model from section \ref{sec:applications}.
As expected, we find two nondecreasing branches of spectrum passing through the energy interval of interest when $f(y)$ is smooth or $\mu \ne 0$.
Again, the qualitative behavior of the branches is robust to perturbations of the form \eqref{stableV}.
We see that the branches corresponding to $f(y) = \sgn (y)$, particularly $E(\xi) = \pm 1$ for $\mu = 0$, are more sensitive to variations in $\mu$ than those corresponding to smooth $f$.
The flat bands are eliminated when $\mu \ne 0$, resulting in two non-trivial increasing branches of spectrum.

\section{Proofs of the edge current stability and the BEC} \label{sec:pfbec}
\subsection{Proofs of Lemma \ref{lemma:tc} and Theorem \ref{thm:stabilityall}}\label{appendix:pfs}
This section is dedicated to proving the results from section \ref{sec:stability}, which state that the the interface current observable $\sigma_I$ is well defined and stable with respect to perturbations of $H$, $P$ and $\varphi$. Recall the definition of $\sigma_I$ given by \eqref{eq:sigmaI}, and the hypotheses \hone\ and \hnot\ from section \ref{sec:stability}.

We begin by showing that any operator $H$ satisfying \hone\ must also satisfy \hnot. We use freely the notation recalled in Appendix \ref{sectionPreliminaries}.
\begin{proposition} \label{trclass}
Suppose $H = \Op (\sym)$ satisfies \hone\ and $\Phi \in \mathcal{C}^\infty_c (E_1, E_2)$. Define $H_h := \Op_h (\sym)$ for $h \in (0,1]$. Then $\Phi (H_h) \in \Op_h (S (\aver{y,\xi,\zeta}^{-\infty}))$.
\end{proposition}
\begin{proof}
We observe that 
$H_h = \Op _1 (\sym(x,y, h\xi, h\zeta))$
satisfies \hone\ for all $0 < h \le 1$,
with
$\sym_\pm (h\xi,h\zeta)$ replacing $\sym_\pm (\xi,\zeta)$.
This means
$H_h$ is self-adjoint with domain 
$\mathcal{H}^m$
that is dense in $\mathcal{H}$,
and
$(i + H_h)^{-1} \in \Op _h (S(\aver{\xi,\zeta}^{-m}))$ as recalled in Proposition \ref{prop:resh} below.
Writing $\Phi (H_h) = (i+H_h)^{-N} \Phi_N (H_h)$ for $\Phi_N$ still compactly supported, we obtain that $\Phi (H_h) \in \Op_h (S (\aver{\xi,\zeta}^{-\infty}))$.

We now prove decay in $y$.
Since $\sym_\pm$ are independent of $(x,y)$ and have a spectral gap in $(E_1, E_2)$, it follows that $H_{\pm,h} := \Op_h (\sym_\pm)$ also have a spectral gap in $(E_1, E_2)$. Hence $\Phi (H_{\pm,h}) = 0$, meaning that
\begin{align*}
    \Phi(H_h) &= 
    \phi (y) (\Phi(H_h) - \Phi(H_{+,h})) + (1-\phi(y))(\Phi(H_h) - \Phi(H_{-,h})),
\end{align*}
where we assume $\phi \in \fs (0,1)$.
By the Helffer-Sj\"ostrand formula \eqref{HSformula}, we see that
\begin{align*}
    \Phi(H_h) - \Phi(H_{+,h}) =
    -\frac{1}{\pi} \int_\mathbb{C} \bar{\partial} \tilde{\Phi} (z)
    (z-H_h)^{-1} (H_h - H_{+,h}) (z-H_{+,h})^{-1} d^2z,
\end{align*}
with 
$H_h-H_{+,h} = \Op_h (\sym-\sym_+)$ and $\sym-\sym_+$ vanishing for $y>0$ sufficiently large.
By Propositions \ref{prop:sharpfinite} and \ref{prop:resh} (and the rapid decay of $\bar{\partial} \tilde{\Phi}$ near the real axis), 
it follows that $\Phi(H_h) - \Phi(H_{+,h}) \in \Op_h (S (\aver{y_+}^{-\infty}))$.
Since $\phi (y) =\phi \in \Op_h (S (\aver{y_-}^{-\infty}))$, the composition calculus implies that
$\phi (y) (\Phi(H_h) - \Phi(H_{+,h})) \in \Op_h (\aver{y}^{-\infty})$.
A parallel argument proves that $(1-\phi(y))(\Phi(H_h) - \Phi(H_{-,h}))\in \Op_h (\aver{y}^{-\infty})$, hence $\Phi(H_h) \in \Op_h (\aver{y}^{-\infty})$.
The result then follows from interpolation.
\end{proof}

Appealing to Proposition \ref{trclass}, we can now work exclusively with operators $H$ that satisfy our more general hypothesis \hnot.

By Proposition \ref{prop:resh} below, any $H$ satisfying \hnot\ is self-adjoint with domain of definition $\mathcal{D}(H) = \mathcal{H}^m$. Moreover, 
$(z-H)^{-1} \in \Op (S(\aver{\xi,\zeta}^{-m}))$ with symbolic bounds that blow up algebraically as $\Im z \rightarrow 0$. 
For any $\psi \in \mathcal{C}^\infty_c (\mathbb{R})$ and $s \in \mathbb{N}$, we have
$\psi (H) = (i-H)^{-s} \phi_s (H)$ for some $\phi_s \in \mathcal{C}^\infty_c (\mathbb{R})$.
Using the Helffer-Sj\"ostrand formula and composition calculus, this implies that 
$\psi (H) \in \Op (S (\aver{\xi,\zeta}^{-\infty}))$ for every $\psi \in \mathcal{C}^\infty_c (\mathbb{R})$.
If, in particular, $\psi \in \mathcal{C}^\infty_c (\teone, \tetwo)$, then \hnot\ and the composition calculus imply that $\psi (H) = \psi (H) \Phi (H) \in \Op (S(\aver{y,\xi,\zeta}^{-\infty}))$.

We are now ready to prove the first result from section \ref{sec:stability}.
\begin{proof}[Proof of Lemma \ref{lemma:tc}]
     Recall that $\varphi' \in \mathcal{C}^\infty_c (\teone, \tetwo)$, and hence $\varphi' (H) \in \Op (S(\aver{y,\xi,\zeta}^{-\infty}))$. Since
    $[H,P]=(1-P)HP-PH(1-P) \in \Op (S (\aver{x}^{-\infty}\aver{y,\xi,\zeta}^m))$, the composition calculus implies that $ [H,P] \varphi ' (H)$ is trace-class.
\end{proof}

The rest of the section is devoted to the proof of Theorem \ref{thm:stabilityall}. 
\paragraph{Proof of Theorem \ref{thm:stabilityall} (\ref{it:P}).}
Before proving invariance of $\sigma_I$ with respect to $\varphi$ and $P$, we verify the following important trace-class property.
\begin{lemma}\label{lemma:tc_poly}
    Suppose $H$ satisfies \hnot, and let $P(x) = P \in \fs (0,1)$.
    If $\phi \in \mathcal{C}^\infty_c (E_1, E_2)$, then $q(H) [\phi (H), P]$ is trace-class for any polynomial $q$.
\end{lemma}
\begin{proof}
    We know by assumption that $\phi (H) \in \Op (S (\aver{y,\xi,\zeta}^{-\infty}))$, and thus $[\phi(H), P] = (1-P) \phi(H) P - P \phi (H) (1-P) \in \Op (S (\aver{x,y,\xi,\zeta}^{-\infty}))$ by the composition calculus. Since $q(H) \in \Op (S (\aver{\xi,\zeta}^{km}))$ for some $k \in \mathbb{N}$, it follows that $q(H) [\psi(H),P] \in \Op (S (\aver{x,y,\xi,\zeta}^{-\infty}))$ is trace-class.
\end{proof}

We next show
that the trace in \eqref{eq:sigmaI} is not modified after regularization of the commutator.
\begin{lemma} \label{Psi}
Let $P(x) = P \in \fs (0,1)$ and $\varphi \in \fs (0,1;\teone,\tetwo)$.
Let $S$ be any open interval containing $\supp \varphi '$, and let
$\Psi \in \mathcal{C}^\infty_c (\mathbb{R})$
such that $\Psi (\lambda) = \lambda$ in $S$.
Then given any self-adjoint $H$ for which $[H,P]\varphi' (H)$ is trace-class, 
it follows that $[\Psi (H), P] \varphi' (H)$ is trace-class, with
\begin{align}\label{eq:Psi}
    \sigma_I (H,P,\varphi) = \Tr i[\Psi (H), P] \varphi' (H).
\end{align}
\end{lemma}
By Lemma \ref{lemma:tc}, we know that Lemma \ref{Psi} applies to any $H$ satisfying \hnot.
\begin{proof}
Using the Helffer-Sj\"ostrand formula, we write
\begin{align*}
    [\Psi(H),P] =\Big[-\frac{1}{\pi} \int_{\mathbb{C}} \bar{\partial} \tilde{\Psi} (z) (z-H)^{-1} d^2z, P\Big]
    =-\frac{1}{\pi} \int_{\mathbb{C}} \bar{\partial} \tilde{\Psi} (z) (z-H)^{-1}[H,P] (z-H)^{-1} d^2z,
\end{align*}
which implies that
\begin{align*}
    [\Psi(H),P]\varphi'(H) = -\frac{1}{\pi} \int_{\mathbb{C}} \bar{\partial} \tilde{\Psi} (z) (z-H)^{-1}[H,P] \varphi' (H) (z-H)^{-1} d^2z.
\end{align*}
Since $\norm{(z-H)^{-1}} \le |\Im z|^{-1}$ with $\bar{\partial} \tilde{\Psi} \in \mathcal{C}^\infty_c (\mathbb{C})$ and $|\bar{\partial} \tilde{\Psi}|(z) \le C|\Im z|^2$, the assumption that $[H,P] \varphi' (H)$ is trace-class implies that $[\Psi(H),P]\varphi'(H)$ must be trace-class as well.

Let $\Phi_0, \Phi_{00} \in \mathcal{C}^\infty_c (\mathbb{R})$ such that
$\varphi' = \varphi' \Phi_0 = \varphi' \Phi_{00}$ and $\Phi_0 = \Phi_0 \Phi_{00}$, with $\Psi \Phi_{00} = \lambda \Phi_{00}$.
Below, we use the shorthand $f := f(H)$ for all compactly supported functions $f$.
It follows from cyclicity of the trace \cite{Kalton} that
\begin{align*}
    \Tr [H,P]\varphi'&=\Tr [H,P]\varphi' \Phi_0 = \Tr \Phi_0 [H,P] \varphi' =\Tr \Phi_0 (\Phi_{00}HP-PH\Phi_{00}) \varphi'\\& = 
    \Tr \Phi_0 (\Phi_{00}\Psi P-P \Psi \Phi_{00}) \varphi'
    =\Tr \Phi_0[\Psi,P]\varphi' 
    = \Tr [\Psi,P]\varphi',
\end{align*}
and the proof is complete.
\end{proof}

We now show that the edge current is independent of $\varphi \in \fs (0,1;\teone,\tetwo)$.
\begin{proposition}\label{prop:invvarphi}
Suppose $H$ satisfies \hnot, and let $P(x)=P \in \fs (0,1)$ and
    $\varphi_1, \varphi_2 \in \fs (0,1;\teone,\tetwo)$. Then
        $\sigma_I (H,P, \varphi_1) = \sigma_I (H,P,\varphi_2)$.
\end{proposition}
\begin{proof}
    It suffices to show that
    $
        \Tr i[H,P] \phi'(H) = 0
    $
    for all $\phi \in \mathcal{C}^\infty_c (\teone, \tetwo)$. 
    Let $\Phi_0, \Phi_{00} \in \mathcal{C}^\infty_c (\teone, \tetwo)$ such that $\supp (\phi) \subset \supp (\Phi_0) \subset \supp (\Phi_{00})$ with $\Phi_{00} (\lambda) = 1$ for all $\lambda \in \supp (\Phi_0)$ and $\Phi_0 (\lambda) = 1$ for all $\lambda \in \supp (\phi)$. It follows that
    \begin{align*}
        \Tr [H,P] \phi '(H) &= \Tr [H,P] \phi' (H) \Phi_{00} (H) = \Tr \Phi_{00} (H) [H,P] \phi' (H)\\
        &= \Tr \Phi_{00} (H) [H,P] \Big( -\frac{1}{\pi}\int_{\mathbb{C}} \bar{\partial} \tilde{\phi'} (z) (z-H)^{-1}d^2 z
        \Big)\\
        &= \Tr \Phi_{00} (H) [H,P] \Big( -\frac{1}{\pi}\int_{\mathbb{C}} \bar{\partial} \tilde{\phi} (z) (z-H)^{-2}d^2 z
        \Big)\\
        &= \Tr \Big( -\frac{1}{\pi}\int_{\mathbb{C}} \bar{\partial} \tilde{\phi} (z) \Phi_{00} (H) [H,P] (z-H)^{-2}d^2 z
        \Big)\\
        &=\Tr \Big( -\frac{1}{\pi}\int_{\mathbb{C}} \bar{\partial} \tilde{\phi} (z) \Phi_{00} (H) (z-H)^{-1} [H,P] (z-H)^{-1}d^2 z
        \Big)\\
        &=\Tr \Big( -\frac{1}{\pi}\int_{\mathbb{C}} \bar{\partial} \tilde{\phi} (z) \Phi_{00} (H) [(z-H)^{-1},P] d^2 z
        \Big)\\
        &=\Tr \Phi_{00} (H) \Big[ -\frac{1}{\pi}\int_{\mathbb{C}} \bar{\partial} \tilde{\phi} (z) (z-H)^{-1}d^2 z, P\Big]
        = \Tr \Phi_{00} (H) [\phi (H), P]\\
        &= \Tr [\phi (H), P] + \Tr (1-\Phi_{00} (H)) P \phi (H)
        = \Tr [\phi (H), P] + \Tr (1-\Phi_{00} (H)) P \phi (H) \Phi_0 (H)\\
        &= \Tr [\phi (H), P] + \Tr \Phi_0 (H)(1-\Phi_{00} (H)) P \phi (H)
        = \Tr [\phi (H), P],
    \end{align*}
    where we have used cyclicity of the trace to justify the second, sixth and twelfth equalities, the Helffer-Sj\"ostrand formula for the third and ninth equalities, integration by parts in $\partial$ for the fourth equality, the fact that $[(z-H)^{-1},\Phi_{00} (H)]=0$ for the sixth equality, the identity $[(z-A)^{-1},B] = (z-A)^{-1} [A,B] (z-A)^{-1}$ for the seventh equality, the fact that $(1-\Phi_{00}) \phi \equiv 0$ for the tenth equality (invoking Lemma \ref{lemma:tc_poly} to ensure that $[\phi (H),P]$ is trace-class), and the fact that $\Phi_0 (1-\Phi_{00}) \equiv 0$ for the thirteenth equality.

    Now, let $P_1, P_2 \in \fs (0,1)$ with $P_j = P_j (x)$ such that $P P_1 = P$ and $(1-P) (1-P_2) = 1-P$. Then
    \begin{align*}
        \Tr [\phi (H), P] &= \Tr ((1-P) \phi (H) P - P \phi (H) (1-P)) = \Tr (1-P) \phi (H) P - \Tr P \phi (H) (1-P)\\
        &=\Tr (1-P) \phi (H) P P_1 - \Tr P \phi (H) (1-P) (1-P_2)\\
        &=\Tr P_1(1-P) \phi (H) P - \Tr (1-P_2)P \phi (H) (1-P)\\
        &=\Tr P P_1(1-P) \phi (H) - \Tr (1-P)(1-P_2)P \phi (H)
        =\Tr P (1-P) \phi (H) - \Tr (1-P)P \phi (H)=0,
    \end{align*}
    and the proof is complete.
\end{proof}

We now prove the stability of the trace with respect to changes in the switch function $P$.
\begin{proposition} \label{ip}
Suppose $H$ satisfies (H0), 
$\varphi \in \fs (0,1;\teone,\tetwo)$, and
$P_1, P_2 \in \fs (0,1)$ with $P_j = P_j (x)$.
Then 
    $\sigma_I (H,P_1,\varphi) = \sigma_I (H,P_2,\varphi)$.
\end{proposition}

\begin{proof}
Using Lemma \ref{Psi}, we have
\begin{align*}
    \sigma_I (H,P_2,\varphi)-\sigma_I (H,P_1,\varphi)=
    \Tr i [\Psi (H), P_2 - P_1] \varphi'(H).
\end{align*}
Since $P_2 - P_1 \in \aver{x}^{-\infty}$,
the assumption \hnot\ implies that $(P_2 - P_1) \varphi' (H)$ is trace-class.
Therefore,
\begin{align*}
    \Tr i [\Psi (H), P_2 - P_1] \varphi'(H) &= 
    \Tr i \Psi (H) (P_2 - P_1) \varphi'(H) - \Tr i (P_2 - P_1) \Psi (H) \varphi'(H)\\
    &=
    \Tr i (P_2 - P_1) \varphi' (H) \Psi(H) -
    \Tr i (P_2 - P_1) \Psi (H) \varphi' (H) =0,
\end{align*}
where we have used cyclicity of the trace to justify the second equality, and the fact that $[\varphi' (H), \Psi (H)] = 0$ for the last equality.
\end{proof}
With Propositions \ref{prop:invvarphi} and \ref{ip} above, we have proved Theorem \ref{thm:stabilityall} (\ref{it:P}).

\paragraph{Proof of Theorem \ref{thm:stabilityall} (\ref{it:compact}).}
Now we want to analyze the stability of $\sigma_I (H,P,\varphi)$ with respect to perturbations of $H$. For $W$ a symmetric linear operator (with various additional assumptions in the results below), let 
\begin{equation}\label{eq:Hmu}
    \Hmu = H + \mu W \quad \mbox{ for } \quad\mu \in [0,1].
\end{equation}
We begin with a preliminary result that introduces 
a class of appropriately decaying perturbations under which the interface current observable is stable.
\begin{proposition}\label{thm:decayxy}
     Suppose $H$ satisfies \hnot, $P(x)=P \in \fs(0,1)$ and $\varphi \in \fs(0,1;\teone,\tetwo)$.
     Assume that $W$ is symmetric with $W \in \Op (S^m_{1,0} \cap S (\aver{x,y}^{-\delta}))$ for some $\delta > 0$. If $\Hone$ satisfies \hnot, then $\sigma_I (\Hone, P,\varphi) = \sigma_I (H,P,\varphi)$. 
\end{proposition}
\begin{proof}
    The assumption that $\Hone$ satisfies \hnot\ implies 
    $\sigma_I (\Hone,P,\varphi)$ is well defined (see Lemma \ref{lemma:tc}), with 
    \begin{align*}
        \sigma_I (\Hone,P,\varphi) - \sigma_I (H,P,\varphi) = \Tr i [\Psi (\Hone),P] (\varphi ' (\Hone) - \varphi ' (H)) + \Tr i [\Psi (\Hone) - \Psi (H), P] \varphi ' (H)
    \end{align*}
    by Lemma \ref{Psi}. Using cyclicity of the trace as in the proof of Proposition \ref{prop:invvarphi}, we find that
    \begin{align*}
        \Tr i [\Psi (\Hone) - \Psi (H), P] \varphi ' (H) = -\Tr i [\varphi '(H),P] (\Psi (\Hone)-\Psi(H)).
    \end{align*}
    Thus by Proposition \ref{ip}, it suffices to show that $\Tr [A,P_{x_0}] B \rightarrow 0$ as $x_0 \rightarrow \infty$ for $A \in \{\Psi (\Hone),\varphi'(H)\}$ and $B \in \{\varphi'(\Hone)-\varphi'(H), \Psi(\Hone)-\Psi(H)\}$, where $P_{x_0} (x) := P (x-x_0)$. 
    
    Fix $\eps>0$.
    For any $\phi \in \mathcal{C}^\infty_c (\teone, \tetwo)$, the Helffer-Sj\"ostrand formula implies that
\begin{align*}
    \phi (\Hone) - \phi (H) = \frac{1}{\pi} \int_{\mathbb{C}}\bar{\partial} \tilde{\phi} (z) (z-\Hone)^{-1} W (z-H)^{-1} d^2 z.
\end{align*}
Propositions \ref{prop:sharpfinite} and \ref{prop:resh} (together with the rapid decay of $\bar{\partial} \tilde{\phi}$ near the real axis) then imply that 
\begin{align}\label{eq:thetadecay}
\phi (\Hone) - \phi (H) \in \Op (S (\aver{x,y}^{-\delta} \aver{\xi,\zeta}^{-m})) \subset \Op (S (\aver{x,y,
\xi,\zeta}^{-\delta})),
\end{align}
where we have assumed without loss of generality that $\delta < m$.
Thus $B \in \Op (S (\aver{x,y,
\xi,\zeta}^{-\delta}))$, so
    there exist $b_0 \in \mathcal{C}^\infty_c (\mathbb{R}^4)$ and $b_1\in S(1)$ as small as necessary such that $B = B_0+ B_1$ with $B_j := \Op (b_j)$ and $\norm{B_1} < \eps$. 
    Writing $[A,P_{x_0}] = (1-P_{x_0}) A P_{x_0} - P_{x_0} A (1-P_{x_0})$, Lemma \ref{lemma:decaycompact} and the composition calculus imply that $[A,P_{x_0}] \in \Op (S (\aver{x-x_0,y,\xi,\zeta}^{-\infty}))$ uniformly in $x_0$. We conclude that $\norm{[A,P_{x_0}]}_1 \le C$ uniformly in $x_0$, hence
    $\limsup_{x_0 \rightarrow \infty} \norm{[A,P_{x_0}] B_1}_1< C \eps$.
    Moreover, the decay of $b_0$ implies that
    \begin{align*}
        \norm{[A,P_{x_0}] B_0}_1 \le C \int_{\mathbb{R}^4}\aver{x-x_0,y,\xi,\zeta}^{-5} \aver{x,y,\xi,\zeta}^{-5}dxdyd\xi d\zeta
        \longrightarrow 0
    \end{align*}
as $x_0 \rightarrow \infty$.
    We have thus shown that
    \begin{align*}
        \limsup_{x_0 \rightarrow \infty} \norm{[A,P_{x_0}] B}_1 \le \limsup_{x_0 \rightarrow \infty} \norm{[A,P_{x_0}] B_0}_1 + \limsup_{x_0 \rightarrow \infty} \norm{[A,P_{x_0}] B_1}_1 < C \eps.
    \end{align*}
    Since $\eps$ was arbitrary, the proof is complete.
\end{proof}

Next, we show that for the class of relatively compact perturbations $W$ from Theorem \ref{thm:stabilityall} (\ref{it:compact}), the symbols of compactly supported functionals of $\Hone$ decay rapidly in $\aver{y,\xi,\zeta}$.

\begin{lemma}\label{lemma:decaycompact}
    Suppose $H$ satisfies \hnot, and let
    $W$ be a symmetric pseudo-differential operator such that $W \in \Op (S^m_{1,0} \cap S (\aver{\xi, \zeta}^{m-\delta} \aver{x,y}^{-\delta}))$ for some $\delta > 0$. Then $\Hone \in \Op (\smeh)$ is self-adjoint with domain of definition $\mathcal{H}^m$, and 
    $\phi (\Hone) \in \Op (S (\aver{y,\xi,\zeta}^{-\infty}))$ for any $\phi \in \mathcal{C}^\infty_c (\teone, \tetwo)$.
\end{lemma}

\begin{proof}
That $\Hone \in \Op (\smeh)$ follows immediately from ellipticity of $H$ and the decay 
of $W$. 
By Proposition \ref{prop:resh}, this means $\Hone$ is self-adjoint with domain of definition $\mathcal{H}^m$. Moreover, 
$(z-\Hone)^{-1} \in \Op (S(\aver{\xi,\zeta}^{-m}))$ for all $\Im z \ne 0$, with symbolic bounds blowing up at worst algebraically as $\Im z \rightarrow 0$ (see Proposition \ref{prop:resh}).

Let $\phi,\phi_0 \in \mathcal{C}^\infty_c (\teone, \tetwo)$ such that $\phi \phi_0 = \phi$, and define $\Theta := \phi (\Hone) - \phi (H)$ and $\Theta_0 := \phi_0 (\Hone) - \phi_0 (H)$. Then
$\Theta = \Theta \Theta_0 + \phi (H) \Theta_0 + \Theta \phi_0 (H)$, and hence
\begin{align*}
    \Theta (1-\Theta_0)= \phi (H) \Theta_0 + \Theta \phi_0 (H).
\end{align*}
By \eqref{eq:thetadecay}, we have $\Theta, \Theta_0 \in \Op (S (\aver{x,y,
\xi,\zeta}^{-\delta}))$.
Thus there exist $\theta_{00} \in \mathcal{C}^\infty_c (\mathbb{R}^4)$ and $\theta_{01} \in S(1)$ as small as necessary such that $\Theta_{00} = \Op (\theta_{00})$, $\Theta_{01} = \Op (\theta_{01})$ with $\norm{\Theta_{01}} < 1$ (see Proposition \ref{prop:CV}), and $\Theta_0 = \Theta_{00} + \Theta_{01}$. Since $\phi (H), \phi_0 (H) \in \Op (S (\aver{y,\xi,\zeta}^{-\infty}))$ by assumption on $H$, it follows that
\begin{align*}
    \Theta (1-\Theta_{01}) = \Theta (1-\Theta_{0}) + \Theta \Theta_{00}\in \Op (S (\aver{y,\xi,\zeta}^{-\infty})),
\end{align*}
hence
\begin{align*}
    \Theta = (\Theta (1-\Theta_{0}) + \Theta \Theta_{00})(1-\Theta_{01})^{-1}\in \Op (S (\aver{y,\xi,\zeta}^{-\infty})).
\end{align*}
We conclude that $\phi (\Hone) = \phi (H) + \Theta \in \Op (S (\aver{y,\xi,\zeta}^{-\infty}))$, and the proof is complete.
\end{proof}

The proof of Theorem \ref{thm:stabilityall} (\ref{it:compact}) is then concluded as follows.
Fix $\phi \in \mathcal{C}^\infty_c (\teone, \tetwo)$ such that $\phi \equiv 1$ in some open interval containing $\supp (\varphi')$. Lemma \ref{lemma:decaycompact} then implies that $\Hone$ satisfies \hnot, with $\Phi$ replaced by $\phi$.
The result then follows from Proposition \ref{thm:decayxy}.

\paragraph{Proof of Theorem \ref{thm:stabilityall} (\ref{it:bounded}).}
Next, we derive a stability result that no longer requires the perturbation to be relatively compact. We 
instead assume that $W$ is relatively bounded with respect to $H$, and require that $W$ be ``sufficiently small.''

In the proof below, we will be analyzing operators $A = \Op (a)$ that depend on the parameter $\mu$. For $\ofn : \mathbb{R}^{4} \rightarrow [0,\infty)$ an order function, we will write $A \in \Op (\mu S(\ofn))$ to mean that 
$a \in S (\ofn)$ for all $\mu$, with $|\partial^\alpha a| \le C_\alpha \mu \ofn$ uniformly in $\mu$.
\begin{proof}
    Since $H \in \Op (\smeh)$ and $W$ is symmetric, it follows that $\Hmu \in \Op (\smeh)$ is self-adjoint (with domain of definition $\mathcal{H}^m$) whenever $\mu > 0$ is sufficiently small.
    
    Let $\phi, \phi_0 \in \mathcal{C}^\infty_c (E_1, E_2)$ such that $\phi \phi_0 = \phi$. Following the proof of Lemma \ref{lemma:decaycompact}, we define
$\Theta := \phi (\Hmu) - \phi (H)$ and $\Theta_0 := \phi_0 (\Hmu) - \phi_0 (H)$, and verify that
\begin{align*}
    \Theta (1-\Theta_0)= \phi (H) \Theta_0 + \Theta \phi_0 (H).
\end{align*}
The Helffer-Sj\"ostrand formula implies that
\begin{align*}
    \Theta_0 = -\frac{\mu}{\pi} \int_{\mathbb{C}}\bar{\partial} \tilde{\phi_0} (z) (z-\Hmu)^{-1} W (z-H)^{-1} d^2 z.
\end{align*}
By Propositions \ref{prop:sharpfinite} and \ref{prop:reseps}, this means $\Theta_0 \in \Op (\mu S (\aver{\xi,\zeta}^{-m}))$.
It follows that $(1-\Theta_0)^{-1} \in \Op (S(1))$ whenever $\mu>0$ is sufficiently small, hence
\begin{align*}
    \Theta= (\phi (H) \Theta_0 + \Theta \phi_0 (H)) (1-\Theta_0)^{-1} \in \Op (S (\aver{y,\xi,\zeta}^{-\infty}))
\end{align*}
by \hnot.
We conclude that $\phi (\Hmu) = \phi (H) + \Theta \in \Op (S (\aver{y,\xi,\zeta}^{-\infty}))$. Since $\phi$ was arbitrary, this means $\sigma_I (\Hmu, P, \varphi)$ is well defined whenever $\mu > 0$ is small enough.

\medskip

Let $\chi \in \fs (1,0;1,2)$ be monotonically non-increasing, and define $\chi_\eps (x,y,\xi,\zeta) := \chi (\eps |(x,y,\xi,\zeta)|)$ and $\Hmueps := \Op (\sym+ \mu \chi_\eps w)$ for $\eps \in (0,1]$. 
Since $\chi_\eps w \in S(\aver{x,y,\xi,\zeta}^{-\infty})$, 
Theorem \ref{thm:stabilityall}(\ref{it:compact}) implies that $$\sigma_I (\Hmueps, P,\varphi) = \sigma_I(H,P,\varphi), \qquad \eps \in (0,1].$$ 
It remains to show that $\sigma_I (\Hmueps, P,\varphi) - \sigma_I(\Hmu,P,\varphi) \rightarrow 0$ as $\eps \downarrow 0$.
In order to do this, we will first need symbolic bounds that are uniform in $\mu$ and $\eps$.

For any multi-index $\alpha$ and any $p>0$, we have $|\partial^\alpha \chi_\eps|(x,y,\xi,\zeta) \le C_{\alpha,p} \eps^{|\alpha|} \aver{\eps x, \eps y, \eps \xi, \eps \zeta}^{-p}$ uniformly in $\eps$.
In particular, this means
\begin{align*}
    |\partial^\alpha \chi_\eps|(x,y,\xi,\zeta) \le C_{\alpha,|\alpha|} \eps^{|\alpha|} \aver{\eps x, \eps y, \eps \xi, \eps \zeta}^{-|\alpha|} \le C_{\alpha,|\alpha|} \aver{x,y,\xi,\zeta}^{-|\alpha|}, \qquad \eps \in (0,1].
\end{align*}
By definition, $\Op (w) := W$ satisfies $|\partial^{\beta_1}_{x,y} \partial^{\beta_2}_{\xi,\zeta} w| (x,y,\xi,\zeta) \le C_\beta \aver{\xi,\zeta}^{m-|\beta_2|}$ for any multi-index $\beta = (\beta_1, \beta_2)$.
It follows that
\begin{align}\label{eq:unifsm10}
    |\partial^\alpha \chi_\eps \partial^\beta w| (x,y,\xi,\zeta) \le C_{\alpha,\beta}\aver{x,y,\xi,\zeta}^{-|\alpha|}\aver{\xi,\zeta}^{m-|\beta_2|} \le C_{\alpha,\beta} \aver{\xi,\zeta}^{m-|\alpha_2|-|\beta_2|}, \qquad \eps \in (0,1],
\end{align}
hence the set $\mathcal{W} := \{\chi_\eps w : \eps \in (0,1]\}$ satisfies the assumptions of Proposition \ref{prop:reseps}.
It follows that
$(z-\Hmueps)^{-1} \in \Op (S (\aver{\xi,\zeta}^{-m}))$ with all symbolic bounds uniform in $\eps \in (0,1]$ and blowing up algebraically as $\Im z \rightarrow 0$.
Recall that $\Hmu \in \Op (\smeh )$ whenever $\mu>0$ is sufficiently small, thus the same can be said of the ($\eps$-independent) operator $(z-\Hmu)^{-1}$.
With $\Thetaeps := \phi (\Hmueps) - \phi (\Hmu)$ and $W_\eps := \Op ((1-\chi_\eps) w)$, the Helffer-Sj\"ostrand formula and Proposition \ref{prop:sharpfinite} then imply that
\begin{align}\label{eq:thetaeps}
    \Thetaeps = \frac{\mu}{\pi}\int_{\mathbb{C}} \bar{\partial} \tilde{\phi} (z) (z-\Hmueps)^{-1} W_\eps (z-\Hmu)^{-1} d^2 z \in \Op (\mu S(\aver{\xi,\zeta}^{-m})) 
\end{align}
uniformly in $\mu>0$ sufficiently small and $\eps \in (0,1]$.
Since $\phi$ is arbitrary, we also have that $\Thetaepsnot := \phi_0 (\Hmueps) - \phi_0 (\Hmu) \in \Op (\mu S(\aver{\xi,\zeta}^{-m}))$ uniformly in $\mu$ and $\eps$.
Thus if $\mu > 0$ is small enough, then $(1-\Thetaepsnot)^{-1} \in \Op (S (1))$ uniformly in $\eps \in (0,1]$.
Since $\phi (\Hmu) \in \Op (S (\aver{y,\xi,\zeta}^{-\infty}))$, the 
familiar identity
\begin{align*}
    \Thetaeps= (\phi (\Hmu) \Thetaepsnot + \Thetaeps \phi_0 (\Hmu)) (1-\Thetaepsnot)^{-1} 
\end{align*}
implies that $\phi (\Hmueps) = \Thetaeps + \phi (\Hmu) \in \Op (S (\aver{y,\xi,\zeta}^{-\infty}))$ uniformly in $\eps \in (0,1]$.

\medskip

As in the proof of Theorem \ref{thm:stabilityall} (\ref{it:compact}), cyclicity of the trace implies that
\begin{align*}
    \sigma_I (\Hmueps,P,\varphi) - &\sigma_I (\Hmu,P,\varphi) = \\
    &\Tr i [\Psi (\Hmueps),P] (\varphi ' (\Hmueps) - \varphi ' (\Hmu)) - \Tr i [\varphi ' (\Hmu),P] (\Psi (\Hmueps) - \Psi (\Hmu)).
\end{align*}
Thus it suffices to show that $\Tr [A_\eps,P] B_\eps \rightarrow 0$ as $\eps \downarrow 0$, for $A_\eps \in \{\varphi'(\Hmu), \Psi(\Hmueps)\}$ and $B_\eps \in \{\varphi ' (\Hmueps) - \varphi ' (\Hmu), \Psi (\Hmueps) - \Psi (\Hmu)\}$.
In the paragraph above, we showed that $A_\eps \in \Op (S (\aver{y,\xi,\zeta}^{-\infty}))$ uniformly in $\eps \in (0,1]$. The composition calculus then implies that $[A_\eps, P]\in \Op (S (\aver{x,y,\xi,\zeta}^{-\infty}))$ 
uniformly in $\eps \in (0,1]$.
Using \eqref{eq:thetaeps} and the paragraph above it, we see that $B_\eps \in \Op (S (1-\chi_\eps (x,y,\xi,\zeta) + \eps))$ uniformly in $\eps \in (0,1]$.
We thus have $[A_\eps,P] B_\eps \in \Op (S (\aver{x,y,\xi,\zeta}^{-5}(1-\chi_\eps (x,y,\xi,\zeta) + \eps)))$ uniformly in $\eps \in (0,1]$, meaning that (see appendix \ref{sectionPreliminaries})
\begin{align*}
    \norm{[A_\eps,P] B_\eps}_1 &\le C \int_{\mathbb{R}^4} \aver{x,y,\xi,\zeta}^{-5}(1-\chi_\eps (x,y,\xi,\zeta) + \eps) dx dy d\xi d\zeta\\
    &\le
    C \Big(\int_{\{|(x,y,\xi,\zeta)|\ge \eps^{-1}\}} \aver{x,y,\xi,\zeta}^{-5} dx dy d\xi d\zeta + \eps \int_{\mathbb{R}^4} \aver{x,y,\xi,\zeta}^{-5}dx dy d\xi d\zeta \Big) \le C \eps.
\end{align*}
This completes the proof.
\end{proof}
\paragraph{Proof of Theorem \ref{thm:stabilityall} (4).}
We now turn our attention to operators $H$ satisfying \hone.
Recall Proposition \ref{trclass}, which ensures that all the previous results obtained in this section (for $H$ satisfying \hnot) still apply.

\begin{proof}
    Fix $h \in (0,1]$. Since $H_h = \Op (\sym_h)$ with $\sym_h (x,y,\xi,\zeta) = \sym(x,y,h\xi,h\zeta)$, we know that $H_h$ satisfies \hone. Observe also that for any $h' \in (0,1]$, we have $H_h - H_{h'} \in \Op (S^m)$.
    For $h' \in (0,1]$ and $\mu \in [0,1]$, define $\sym_{h,h'}^{(\mu)}:=\sym_h + \mu (\sym_{h'} - \sym_h)$.
    By continuity and ellipticity of $\sym$, we know that whenever $|h'-h|$ is sufficiently small, $\sym_{h,h'}^{(\mu)} \in \smeh$ for all $\mu \in [0,1]$.
    Moreover, \hone\ implies that $\sym_{h,h'}^{(\mu)}(x,y,\xi,\zeta) = \sym_{h,h',\pm}^{(\mu)}(\xi,\zeta)$ whenever $\pm y$ is sufficiently large, where 
    $$\sym_{h,h',\pm}^{(\mu)} (\xi,\zeta) := \sym_\pm (h\xi,h\zeta) + \mu (\sym_\pm (h'\xi h'\zeta)-\sym_\pm (h\xi, h\zeta)).$$
    Since $\sym_\pm$ have a spectral gap in $(E_1,E_2)$, it follows (also from $\sym_{h,h'}^{(\mu)} \in \smeh$ and continuity) that whenever $|h'-h|$ is sufficiently small,
    $\sym_{h,h',\pm}^{(\mu)}$ has a spectral gap in $(E_1 + \Delta, E_2 - \Delta)$ for all $\mu \in [0,1]$, where $\Delta := (E_2-E_1)/4$.
    Thus we have shown that whenever $|h'-h|$ is sufficiently small and $\mu \in [0,1]$, the operator $H^{(\mu)}_{h,h'} := \Op (\sym_{h,h'}^{(\mu)})$ satisfies \hone\ with $(E_1, E_2)$ replaced by $(E_1 + \Delta, E_2 - \Delta)$.
    By Theorem \ref{thm:stabilityall} (\ref{it:bounded}), we conclude that whenever $|h'-h|$ is sufficiently small, 
    $\sigma_I (H^{(\mu)}_{h,h'},P,\varphi_1)$ is independent of $\mu \in [0,1]$, provided $\varphi_1 \in \fs (0,1;E_1+\Delta, E_2-\Delta)$.
    In particular, this means $\sigma_I (H_{h'},P,\varphi_1) = \sigma_I (H_h,P,\varphi_1)$ if $|h'-h|$ is sufficiently small.
    But since $H_{h'}$ satisfies \hone, we know by Proposition \ref{prop:invvarphi} that 
    $\sigma_I (H_{h'},P,\varphi) = \sigma_I (H_h,P,\varphi)$.
    We have thus shown that $\sigma_I (H_{h'},P,\varphi)$ is constant over $h'$ in some small neighborhood of $h$. Since $h \in (0,1]$ was arbitrary, the result is complete.
\end{proof}
\paragraph{Proof of Theorem \ref{thm:stabilityall} (5).}
Finally, we show that the interface edge current is immune to oscillations in the $x$-variable.
\begin{proof}
Let $\chi \in \mathcal{C}^\infty_c (-2,2)$ such that $\chi (x) = 1$ whenever $x \in [-1,1]$. Define $\chi_\eps (x) := \chi (\eps x)$.
For $\eps \in (0,1]$, define $\sigmaeps (x,y,\xi,\zeta) := \sym(x - (x-x_0)\chi_\eps (x-x_0), y, \xi, \zeta)$, so that $\sigmaeps = \sym_{x_0}$ whenever 
$|x-x_0|\le \eps^{-1}$ and $\sigmaeps = \sym$ whenever $|x-x_0| \ge 2 \eps^{-1}$.
Since $\Hepsnot := \Op (\sigmaeps)$ satisfies \hone\ and $\sigmaeps - \sym$ vanishes whenever $\aver{x,y}$ is sufficiently large, 
Proposition \ref{thm:decayxy} implies that $\sigma_I (\Hepsnot, P, \varphi) = \sigma_I (H,P,\varphi)$ for all $\eps \in (0,1]$.

It remains to show that $\sigma_I (\Hepsnot, P, \varphi) - \sigma_I(H_{x_0}, P, \varphi) \rightarrow 0$ as $\eps \downarrow 0$.
To do this, we will emulate the argument used to prove Theorem \ref{thm:stabilityall} (\ref{it:bounded}).
In particular, it suffices to show that
$\Tr [A_\eps,P] B_\eps \rightarrow 0$ as $\eps \downarrow 0$, for $A_\eps \in \{\phi (H_{x_0}), \phi(\Hepsnot)\}$, $B_\eps =\phi (\Hepsnot) - \phi (H_{x_0})$ and $\phi \in \mathcal{C}^\infty_c (E_1, E_2)$.
As before, Proposition \ref{prop:reseps} and the composition calculus imply that $[A_\eps, P] \in \Op (S (\aver{x,y,\xi,\zeta}^{-\infty}))$ uniformly in $\eps \in (0,1]$. By the Helffer-Sj\"ostrand formula,
\begin{align*}
    B_\eps = \frac{1}{\pi}\int_{\mathbb{C}} \bar{\partial} \tilde{\phi} (z) (z-\Hepsnot)^{-1} (H_{x_0} - \Hepsnot) (z-H_{x_0})^{-1} d^2 z. 
\end{align*}
Since $\sym_{x_0} - \sigmaeps$ vanishes whenever $|x-x_0| \le \eps^{-1}$, it follows that $\sym_{x_0} - \sigmaeps \in S ((1-\chi_\eps (x-x_0)+\eps) \aver{\xi,\zeta}^m)$ uniformly in $\eps \in (0,1]$.
By Propositions \ref{prop:sharpfinite} and \ref{prop:reseps} (and the rapid decay of $\bar{\partial} \tilde{\phi}$ near the real axis), we conclude that $B_\eps \in \Op (S (1-\chi_\eps (x-x_0)+\eps))$ uniformly in $\eps \in (0,1]$. Therefore,
\begin{align*}
    \norm{[A_\eps,P] B_\eps}_1 &\le C \int_{\mathbb{R}^4} \aver{x}^{-2} \aver{y,\xi,\zeta}^{-4} (1-\chi_\eps (x-x_0) + \eps) dx dy d\xi d\zeta \le 
    C \int_{\mathbb{R}} \aver{x}^{-2}(1-\chi_\eps (x-x_0) + \eps) dx\\
    &\le
    C \Big(\int_{\{|x-x_0|\ge \eps^{-1}\}} \aver{x}^{-2} dx + \eps \int_{\mathbb{R}} \aver{x}^{-2}dx \Big) \le C \eps,
\end{align*}
and the proof is complete.
\end{proof}

\subsection{Proof of the BEC}
In this section, we prove the analytical formula \eqref{eq:sigmaI1} for the edge current observable. The proof involves a semiclassical expansion and uses the fact that $\sigma_I (H_h)$ is independent of the semiclassical parameter $h$ (Theorem \ref{thm:stabilityall} (4)).
In particular, the resolvent operator $(z-H_h)^{-1}$ will be approximated in the semiclassical limit, thus we will need to understand how bounds on its symbol depend on $h$. We present these useful estimates below in section \ref{subsubsec:res}. Although similar to existing results in, e.g. \cite{Bony,DS,Zworski}, to our knowledge the following propositions cannot be found in the existing literature. 
In section \ref{subsec:main}, We then use these estimates (in particular, Proposition \ref{prop:fc}) to prove Theorem \ref{thm:main}.

\subsubsection{Elliptic operators and resolvent estimates}\label{subsubsec:res}

\begin{proposition}\label{prop:resh}
    Let $\sym\in \smeh$. Then for all $h \in (0,1]$, the operator $H_h := \Op_h (\sym)$ is self-adjoint with domain of definition $\mathcal{H}^m$.
    This means we can define
    $\Op_h (r_{z,h}) := (z-H_h)^{-1}$ 
    whenever $\Im z \ne 0$. 
    Let $Z \subset \mathbb{C}$ be bounded such that $\Im z \ne 0$ for all $z \in Z$. Then 
    there exists $s \in \mathbb{N}$ such that for any 
    $\alpha \in \mathbb{N}^{2d}$,
    \begin{align*}
        |\partial^\alpha r_{z,h}(x,\xi)| \le C_{\alpha} |\Im z|^{-s-|\alpha|} \aver{\xi}^{-m}
    \end{align*}
    uniformly in $z \in Z$ and $h \in (0,1]$.
\end{proposition}
\begin{proof}
    By \cite[Corollary 2 and the paragraph following Theorem 3]{Bony}, we know that $H_h$ is self-adjoint with domain of definition $\mathcal{H}^m$.
    Applying also \cite[the paragraphs between equation (8.11) and Proposition 8.5]{DS}, it follows that
    $(i-H_h)^{-1}\in \Op_h (S(\aver{\xi,\zeta}^{-m}))$ is a bijection of $L^2 (\mathbb{R}^d) \otimes \mathbb{C}^n$ onto $\mathcal{H}^m$.
For $\Im z \ne 0$, we have that $A_{z,h} := 1-(i-z)(i-H_h)^{-1}$ is a bijection of $L^2 (\mathbb{R}^d) \otimes \mathbb{C}^n$ onto itself, with $\norm{A_{z,h}^{-1}} \le C |\Im z|^{-1}$ uniformly in $z \in Z$.
Applying \cite[Proposition 8.4]{DS} to $\Op_h (b_{z,h}) := A^{-1}_{z,h}$, we obtain that
\begin{align*}
    |\partial^\alpha b_{z,h}(x,\xi)| \le C_{\alpha} |\Im z|^{-2d-2-|\alpha|}
\end{align*}
uniformly in $z \in Z$ and $h \in (0,1]$.
The result then follows from Proposition \ref{prop:sharpfinite} (with $N=0$) and the fact that
$(z-H_h)^{-1} = (i-H_h)^{-1} A^{-1}_{z,h}$.
\end{proof}

\begin{proposition}\label{prop:reseps}
    Let $\sym\in \smeh$ and $\mathcal{W} \subset \mathcal{C}^\infty (\mathbb{R}^{2d}; \mathbb{M}_n)$.
    Suppose that for any $(\alpha,\beta) \in \mathbb{N}^{2d}$,
    \begin{align}\label{eq:unifw}
    |\partial^\alpha_x \partial^\beta_\xi w (x,\xi)| \le C_\alpha \aver{\xi}^{m-|\beta|}
    \end{align}
    uniformly in $w \in \mathcal{W}$ (meaning that $\mathcal{W} \subset S^m$).
    For $\mu \in [0,1]$ and $w \in \mathcal{W}$, define $H_{\mu,w} := \Op (\sym+ \mu w)$.
    Then there exists $\mu_0 \in (0,1]$ such that 
    the following conditions hold:
    \begin{enumerate}
        \item If $w \in \mathcal{W}$ and $\mu \in [0,\mu_0]$, then $H_{\mu,w}$ is self-adjoint with domain of definition $\mathcal{H}^m$.
        \item For $\Im z \ne 0$, define $\Op (r_{z,\mu,w}) := (z-H_{\mu,w})^{-1}$.
        Let $Z \subset \mathbb{C}$ be bounded such that $\Im z \ne 0$ for all $z \in Z$. Then there exists $s \in \mathbb{N}$ such that
        for any $\alpha \in \mathbb{N}^{2d}$,
        \begin{align*}
        |\partial^\alpha r_{z,\mu,w}(x,\xi)| \le C_\alpha |\Im z|^{-s-|\alpha|} \aver{\xi}^{-m}
        \end{align*}
        uniformly in $z \in Z$, $w \in \mathcal{W}$ and $\mu \in [0, \mu_0]$.
    \end{enumerate}
\end{proposition}

\begin{proof}
    Define $\sigmamuw := \sym+ \mu w$. For {\it 1.}: The uniform bounds \eqref{eq:unifw} 
        imply that whenever $\mu > 0$ is sufficiently small, $\sigmamuw \in \smeh$ for all $w \in \mathcal{W}$.
        Hence in this case, $H_{\mu,w}$ is self-adjoint with domain of definition $\mathcal{H}^m$.

        For {\it 2.}: We write $(i-H_{\mu,w})^{-1} = (1 + (i-H_{\mu,w})^{-1} \mu W) (i-H)^{-1}$, which implies
        \begin{align*}
            (i-H_{\mu,w})^{-1} (1 - \mu W (i-H)^{-1}) = (i-H)^{-1},
        \end{align*}
        where $W := \Op_h (w)$ and $H := \Op (\sym)$.
        By \eqref{eq:unifw} and Proposition \ref{prop:CV}, we conclude that whenever $\mu > 0$ is sufficiently small, $A_{\mu,w} := 1 - \mu W (i-H)^{-1}$ is a bijection of $L^2 (\mathbb{R}^d) \otimes \mathbb{C}^n$ onto itself, with $A_{\mu,w}^{-1} \in \Op (S (1))$ uniformly in $\mu$ and $w$.
        Since $(i-H)^{-1} \in \Op (S (\aver{\xi}^{-m}))$ by Proposition \ref{prop:resh}, it follows that $(i-H_{\mu,w})^{-1} \in \Op (S (\aver{\xi}^{-m}))$ uniformly in $\mu$ and $w$. The estimates for
        $r_{z,\mu,w}$ then follow from the same argument that was used in the proof of Proposition \ref{prop:resh} (after it was shown that $(i-H_h)^{-1}\in \Op_h (S(\aver{\xi,\zeta}^{-m}))$ there).
\end{proof}

\begin{proposition}\label{prop:fc}
    Let $\sym\in \smeh$ and define $H_h := \Op_h (\sym)$ for $h \in (0,1]$.
    Let $\phi \in \mathcal{C}^\infty_c (E_1, E_2)$ 
    and define $\Op_h (\nu_h) := \phi (H_h)$.
    Let  $\ofn: \mathbb{R}^{2d} \rightarrow [0,\infty)$ be any order function such that $\nu_h \in S (\ofn^{-\infty})$ and all eigenvalues of $\sym$ lie outside the interval $(E_1, E_2)$ whenever $\ofn$ is sufficiently large.
    For $z \in \mathbb{C}$, define $\sym_z := z-\sym$. For $N \in \mathbb{N}_+$, define $q_{z,h,N}$ recursively by
    $q_{z,h,1} := \sym_z^{-1}$ and
    \begin{align*}
        q_{z,h,N} = \sym_z^{-1} \Big(1-\sum_{j=1}^{N-1}
        \Big(\frac{(ih(D_\xi\cdot D_y - D_x \cdot D_\eta)/2)^j}{j!} \sym_z (x,\xi) q_{z,h,N-j} (y,\eta) \Big) |_{y=x,\eta=\xi}\Big), \quad N\ge 2.
    \end{align*}
    Then for all $N \in \mathbb{N}$,
    \begin{align}\label{eq:qzhn}
        \nu_h + \frac{1}{\pi} \int_\mathbb{C} \bar{\partial} \tilde{\phi} (z) q_{z,h,N} {\rm d}^2 z \in S^{-N+1/2} (\ofn^{-\infty}).
    \end{align}
\end{proposition}
\begin{proof}
Note that
$\supp (\tilde{\phi}) \subset (E_1, E_2) \times (-M,M) =: Z$ for some $0<M< \infty$, 
hence the integral over $\mathbb{C}$ in \eqref{eq:qzhn} can be restricted to an integral over $Z$.
    By assumption, for any $N \in \mathbb{N}$, $z \mapsto q_{z,h,N}$ is analytic in $z \in Z$ so long as $\ofn (x,\xi)$ is sufficiently large (independent of $h$). Integrating by parts, 
    we conclude that $\int_Z \bar{\partial} \tilde{\phi} (z) q_{z,h,N} {\rm d}^2 z$ vanishes whenever $\ofn (x,\xi)$ is sufficiently large. This implies
    \begin{align*}
        \nu_h + \frac{1}{\pi} \int_Z \bar{\partial} \tilde{\phi} (z) q_{z,h,N} {\rm d}^2 z \in S(\ofn^{-\infty}).
    \end{align*}
    The result will follow from an $h$-dependent bound of the above left-hand side and interpolation.
    Namely, it suffices to show that
    \begin{align}\label{eq:suff}
        \nu_h + \frac{1}{\pi} \int_Z \bar{\partial} \tilde{\phi} (z) q_{z,h,N} {\rm d}^2 z \in S^{-N+1/4}(1).
    \end{align}
    By the Helffer-Sj\"ostrand formula, the above left-hand side equals
    \begin{align*}
        -\frac{1}{\pi} \int_\mathbb{C} \bar{\partial} \tilde{\phi} (z) r_{z,h} {\rm d}^2 z+\frac{1}{\pi} \int_Z \bar{\partial} \tilde{\phi} (z) q_{z,h,N} {\rm d}^2 z
        =\frac{1}{\pi} \int_Z \bar{\partial} \tilde{\phi} (z) (q_{z,h,N}-r_{z,h}) {\rm d}^2 z,
    \end{align*}
    where $\Op_h (r_{z,h}) := (z-H_h)^{-1}$. 
    Therefore, given the rapid decay of $\bar{\partial} \tilde{\phi}$ near the real axis, \eqref{eq:suff} holds if there exists $s \in \mathbb{N}$
    such that for all $\alpha \in \mathbb{N}^d$ and $N \in \mathbb{N}$,
    \begin{align}\label{eq:qr}
        |\partial^\alpha (q_{z,h,N} - r_{z,h})| \le C_{\alpha, N} |\Im z|^{-(2N+2s+|\alpha|+1)} h^N
    \end{align}
    uniformly in $z \in Z$ and $h \in (0,1]$.
    We will now prove \eqref{eq:qr} by induction.

    Using that $\sym_z \sharp_h r_{z,h} = 1$ for all $h \in (0,1]$, it follows from Proposition \ref{prop:sharpfinite} that $$|\partial^\alpha (1-\sym_z r_{z,h})| \le C_{\alpha,1} \sum_{j=0}^{|\alpha|} \sbd_{2+s+j}(\sym_z, \aver{\xi}^m) \sbd_{2+s+|\alpha|-j}(r_{z,h}, \aver{\xi}^{-m}) h,$$
    hence
    \begin{align*}
        |\partial^\alpha (\sym_z^{-1} - r_{z,h})| &= |\partial^\alpha (\sym_z^{-1} (1-\sym_z r_{z,h}))| \\
        &\le C_{\alpha,1} \sum_{j=0}^{|\alpha|} \sum_{k=0}^{|\alpha|-j} \sbd_{j} (\sym_z^{-1}, \aver{\xi}^{-m}) \sbd_{2+s+k}(\sym_z, \aver{\xi}^m) \sbd_{2+s+|\alpha|-j-k}(r_{z,h}, \aver{\xi}^{-m}) h
    \end{align*}
    uniformly in $z \in Z$ and
    $h \in (0,1]$.
    Observe that for $k \in \mathbb{N}$, 
    $$\sbd_{k} (\sym_z^{-1}, \aver{\xi}^{-m}) \le C |\Im z|^{-1-k}, \qquad \sbd_{k}(\sym_z, \aver{\xi}^m) \le C$$ 
    uniformly in $z \in Z$.
    Since 
    $\sbd_{k}(r_{z,h}, \aver{\xi}^{-m}) \le C |\Im z|^{-s-k}$ 
    by Proposition \ref{prop:resh},
    we have verified \eqref{eq:qr} when $N=1$.
     Now, fix $N \in \mathbb{N}$ and suppose 
     that for all $k \in \{1, \dots, N\}$,
     \begin{align*}
         |\partial^\alpha (q_{z,h,k} - r_{z,h})| \le C_{\alpha, k} |\Im z|^{-(2k+2s+|\alpha|+1)} h^k
     \end{align*}
     uniformly in $z \in Z$ and $h \in (0,1]$.
    Then
    \begin{align*}
        q_{z,h,N+1} - r_{z,h} 
        &= \sym_z^{-1} \Big(1-\sym_z r_{z,h} -\sum_{j=1}^{N}
        \Big(\frac{(ih(D_\xi\cdot D_y - D_x \cdot D_\eta)/2)^j}{j!} \sym_z (x,\xi) q_{z,h,N+1-j} (y,\eta) \Big) |_{y=x,\eta=\xi}\Big)\\
        &=
        \sym_z^{-1} \Big(1-\sum_{j=0}^{N}
        \Big(\frac{(ih(D_\xi\cdot D_y - D_x \cdot D_\eta)/2)^j}{j!} \sym_z (x,\xi) r_{z,h} (y,\eta) \Big) |_{y=x,\eta=\xi}\Big)\\
        &\quad +\sym_z^{-1} \sum_{j=1}^{N}
        \Big(\frac{(ih(D_\xi\cdot D_y - D_x \cdot D_\eta)/2)^j}{j!} \sym_z (x,\xi) (r_{z,h} (y,\eta) -q_{z,h,N+1-j} (y,\eta)) \Big) |_{y=x,\eta=\xi}
    \end{align*}
    which we write as $t_1 + t_2$.
    Proposition \ref{prop:sharpfinite} implies that
    \begin{align*}
    |\partial^\alpha t_1| &\le C_{\alpha,N+1}\sum_{j=0}^{|\alpha|} \sum_{k=0}^{|\alpha|-j}
    \sbd_{j} (\sym_z^{-1}, \aver{\xi}^{-m}) \sbd_{2(N+1)+s+k} (\sym_z, \aver{\xi}^m) \sbd_{2(N+1)+s+|\alpha|-j-k} (r_{z,h}, \aver{\xi}^{-m})  h^{N+1}\\
    &\le C_{\alpha,N+1} |\Im z|^{-(2(N+1)+2s+|\alpha|+1)} h^{N+1}.
    \end{align*}
    Define $t_2 =: \sym_z^{-1} \sum_{j=1}^N t_{2,j}$. By our inductive hypothesis,
    \begin{align*}
        |\partial^\alpha t_{2,j}| &\le C_{\alpha,j} h^j \sum_{i=0}^{|\alpha|+j} \sbd_{i} (\sym_z,\aver{\xi}^m) \sbd_{|\alpha|+j-i} (r_{z,h} - q_{z,h,N+1-j},1)\aver{\xi}^m\\
        &\le C_{\alpha,j} |\Im z|^{-(2(N+1-j)+2s+|\alpha|+j+1)} h^{N+1}\aver{\xi}^m
        =C_{\alpha,j} |\Im z|^{-(2(N+1)+2s+|\alpha|+1-j)} h^{N+1}\aver{\xi}^m\\
        &\le C_{\alpha,j} |\Im z|^{-(2(N+1)+2s+|\alpha|)} h^{N+1}\aver{\xi}^m.
    \end{align*}
    It follows that
        $|\partial^\alpha (\sym_z^{-1} t_{2,j})| \le C_{\alpha,j} |\Im z|^{-(2(N+1)+2s+|\alpha|+1)} h^{N+1}$.
    for all $j \in \{1, \dots, N\}$.
    We have thus 
    verified \eqref{eq:qr} for all $N \in \mathbb{N}$, and the proof is complete.
\end{proof}

\subsubsection{Proof of Theorem \ref{thm:main}}\label{subsec:main}

\begin{proof}
By Proposition \ref{prop:invvarphi}, $\sigma_I (H,P,\varphi)$ is independent of $\varphi \in \fs (0,1;E_1,E_2)$, thus we can take $\varphi ' \in \mathcal{C}^\infty_c (\anot, \alpha)$ for some $\anot > E_1$. 
Moreover, by Theorem \ref{thm:stabilityall} (5), we can without loss of generality assume that  $\tilde{\sym} (x,y,\xi,\zeta) =\sym(y, \xi, \zeta)$ for all $(x,y,\xi,\zeta) \in \mathbb{R}^4$.
Theorem \ref{thm:stabilityall} (4) states that with $H_h := \Op_h (\tilde{\sym})$, $\sigma_I (H_h,P,\varphi)$ is independent of $h \in (0,1]$. Thus we will expand $\sigma_I (H_h, P, \varphi)$ in powers of $h$ and ignore terms that are not $O(1)$.
We will use the shorthand $\sigma_I := \sigma_I (H_h,P,\varphi)$. 

Let $\Op_h \nu_h := \varphi ' (H_h)$.
By Proposition \ref{prop:fc}, we have
\begin{align}\label{eq:nuh}
        \nu_h + \frac{1}{\pi} \int_\mathbb{C} \bar{\partial} \tilde{\varphi '} (\zp) \tilde{q}_{\zp,h} d^2 \zp \in S^{-3/2} (\aver{y,\xi,\zeta}^{-\infty}),
\end{align}
where
\begin{align*}
    \tilde{q}_{\zp,h} = \sym_\zp^{-1} +\frac{ih}{2} \{\sym_\zp^{-1}, \sym_\zp\}_{\zeta,y} \sym_\zp^{-1}, \qquad \{a,b\}_{\zeta,y} := \partial_\zeta a \partial_y b - \partial_y a \partial_\zeta b, 
\end{align*}
and the right-hand side of \eqref{eq:nuh} is defined by \eqref{eq:symbolm}.
With $\Op _h (\kappa_{h}) := [H_h,P]$, we have that
\begin{align*}
    \kappa_{h} + ih \kone - \frac{h^2}{4} \ktwo \in S^{-3} (\aver{x}^{-\infty} \aver{\xi, \zeta}^m), \qquad \kone := \partial_\xi \sym P'(x),
    \qquad \ktwo := \partial_{\xi \xi} \sym P''(x).
\end{align*}
Since $\nu_h \in S (\aver{y,\xi,\zeta}^{-\infty})$ and $\kappa_h \in S^{-1} (\aver{x}^{-\infty} \aver{\xi,\zeta}^m)$,
the composition calculus implies that
\begin{align*}
    \kappa_h \sharp_h \nu_h - \kappa_h \nu_h + \frac{ih}{2} \{\kappa_h, \nu_h\} \in S^{-3} (\aver{x,y,\xi,\zeta}^{-\infty}),
    \qquad \{a,b\} := \partial_\xi a \partial_x b + \partial_\zeta a \partial_y b - \partial_x a \partial_\xi b - \partial_y a \partial_\zeta b,
\end{align*}
with $S^{-3}$ (rather than $S^{-2}$) above because $\kappa_h$ is $O(h)$ in $S(\aver{x}^{-\infty} \aver{\xi,\zeta}^m)$.
Therefore,
\begin{align*}
    \sigma_I &= \frac{i}{(2\pi h)^2} \tr \int_{\mathbb{R}^4}\kappa_h \sharp_h \nu_h d R_4 
    = \frac{i}{(2\pi h)^2} \tr \int_{\mathbb{R}^4}\Big(\kappa_h \nu_h - \frac{ih}{2} \{\kappa_h, \nu_h\}\Big)dR_4+o(1)
    \end{align*}
    as $h \rightarrow 0$, with $dR_4:= d x d y d \xi d \zeta$.
    Since
    \begin{align*}
        \kappa_h \nu_h = (ih \kone -\frac{h^2}{4} \ktwo)\frac{1}{\pi}\int_{\mathbb{C}}\bar{\partial} \tilde{\varphi '} (\zp) \Big(\sym_\zp^{-1} +\frac{ih}{2} \{\sym_\zp^{-1}, \sym_\zp\}_{\zeta,y} \sym_\zp^{-1}\Big) d^2\zp + h^{5/2}a_h, \qquad a_h \in S (\aver{x,y,\xi,\zeta}^{-\infty})
    \end{align*}
    and
    \begin{align*}
        \{\kappa_h, \nu_h\} = \Big\{ ih \kone, \frac{1}{\pi}\int_{\mathbb{C}}\bar{\partial} \tilde{\varphi '} (\zp) \sym_\zp^{-1} d^2 \zp\Big\} + h^2 b_h, \qquad b_h \in S (\aver{x,y,\xi,\zeta}^{-\infty}),
    \end{align*}
    it follows that
    \begin{align*}
    \sigma_I
    =
    \frac{i}{(2\pi h)^2}\frac{1}{\pi} \tr \int_{\mathbb{R}^4}\int_\mathbb{C} \bar{\partial} \tilde{\varphi '} (\zp) \Big( ih \kone \sym_\zp^{-1}
    - \frac{h^2}{2}\kone \{\sym_\zp^{-1}, \sym_\zp\}_{\zeta,y}\sym_\zp^{-1}
    - \frac{h^2}{4} \ktwo \sym_\zp^{-1} +\frac{h^2}{2} \{ \kone, \sym_\zp^{-1} \}\Big) d^2 \zp d R_4 + o(1)
\end{align*}
as $h \rightarrow 0$. 
Since $\sigma_I$ is independent of $h$, it follows that the $O(h^{-1})$ term above vanishes, and thus
\begin{align}\label{eq:sigmaDiv}
    \sigma_I = \frac{i}{(2\pi)^3}\tr \int_{\mathbb{R}^4} \int_{\mathbb{C}}
    \bar{\partial} \tilde{\varphi '} (\zp) \Big( -\kone \sym_\zp^{-1} \{\sym_\zp, \sym_\zp^{-1}\}_{\zeta,y} -\frac{1}{2}k_2 \sym_\zp^{-1}
    +\{ \kone, \sym_\zp^{-1} \}\Big) d^2 \zp d R_4.
\end{align}
Observe that whenever $(y,\xi,\zeta) \notin R$, 
$\zp \mapsto \sym_\zp^{-1}$ is holomorphic and thus the above integral over $\zp$ vanishes (this is verified via an integration by parts in $\bar{\partial}$).
Since $k_1$ and $k_2$ vanish whenever 
$|x|$ is sufficiently large, the integration region $\mathbb{R}^4$ in \eqref{eq:sigmaDiv} can be replaced by the 
volume $I \times R$, with $I \subset \mathbb{R}$ a bounded interval.

Recalling that $\partial_x \sym = 0$, 
it follows that 
$\{ \kone, \sym_\zp^{-1} \} = P'(x)\{ \partial_\xi \sym, \sym_\zp^{-1} \}_{\zeta,y}-P''(x) \partial_\xi \sym\partial_\xi \sym_\zp^{-1}$ in \eqref{eq:sigmaDiv}.
Using that $\int P' = 1$ and $\int P''=0$, we apply Fubini's Theorem and integrate \eqref{eq:sigmaDiv} in $x$ to obtain
\begin{align}\label{eq:x0}
    \sigma_I = \frac{i}{(2\pi)^3}\tr \int_{R} \int_{\mathbb{C}}
    \bar{\partial} \tilde{\varphi '} (\zp) \Big( \partial_\xi \sym_\zp \sym_\zp^{-1} \{\sym_\zp, \sym_\zp^{-1}\}_{\zeta,y}
    -\{ \partial_\xi \sym_\zp, \sym_\zp^{-1} \}_{\zeta,y} \Big) d^2 \zp d R_3,
\end{align}
with 
$dR_3 := dy d\xi d\zeta$.
At this point, we 
use the identity 
$\{a,b\}_{\zeta, y} = \partial_\zeta (a \partial_y b) - \partial_y (a \partial_\zeta b)$ and integration by parts in $(\zeta,y)$ to write
\begin{align}\label{eq:divform}
    \int_{R} \int_{\mathbb{C}}
    \bar{\partial} \tilde{\varphi '} (\zp) \{ \partial_\xi \sym_\zp, \sym_\zp^{-1} \}_{\zeta,y}d^2 \zp d R_3 = 
    \int_{\partial R} \int_{\mathbb{C}}
    \bar{\partial} \tilde{\varphi '} (\zp) \partial_\xi \sym_\zp (\partial_y \sym_\zp^{-1} \nu_\zeta - \partial_\zeta \sym_\zp^{-1} \nu_y)d^2 \zp d \Sigma.
\end{align}
Since $\zp \mapsto \sym_\zp^{-1}$ is holomorphic whenever $(y,\xi,\zeta) \in \partial R$, an integration by parts in $\bar{\partial}$ reveals that the right-hand side of \eqref{eq:divform} vanishes.
Thus we are left with
\begin{align*}
    \sigma_I = \frac{i}{(2\pi)^3}\tr \int_{R} \int_{\mathbb{C}}
    \bar{\partial} \tilde{\varphi '} (\zp)\partial_\xi \sym_\zp \sym_\zp^{-1} \{\sym_\zp, \sym_\zp^{-1}\}_{\zeta,y}
    d^2 \zp d R_3.
\end{align*}
Integrating by parts in $\bar{\partial}$ with $\zp =: \lambda+i\omega$, we see that
\begin{align}\label{eq:final0}
    \sigma_{I} = \frac{1}{2 (2\pi)^3} \tr \int_{R} \int_{\anot}^{\alpha} \varphi ' (\lambda) \partial_\xi \sym_\zp \sym_\zp^{-1} \{\sym_\zp, \sym_\zp^{-1}\}_{\zeta,y} \Big \vert ^{\omega = 0^+}_{\omega=0^-} d\lambda dR_3.
\end{align}
Since $\partial_\xi \sym_\zp \sym_\zp^{-1} \{\sym_\zp, \sym_\zp^{-1}\}_{\zeta, y}\rightarrow 0$ as $|\omega| \rightarrow \infty$, we have
\begin{align*}
    \partial_\xi \sym_\zp \sym_\zp^{-1} \{\sym_\zp, \sym_\zp^{-1}\}_{\zeta, y}\Big \vert^{\omega=0^+}_{\omega =0^-} = -\int_{- \infty}^{+\infty} \partial_\omega (\partial_\xi \sym_\zp \sym_\zp^{-1} \{\sym_\zp, \sym_\zp^{-1}\}_{\zeta, y}) d\omega,
\end{align*}
with the above integral understood as a principal value about the origin.
Cyclicity of the trace and the fact that $\partial_\omega \sym_\zp = i$ imply that
\begin{align}\label{eq:omegaeps}
    \tr \partial_\omega (\partial_\xi \sym_\zp \sym_\zp^{-1} \{\sym_\zp, \sym_\zp^{-1}\}_{\zeta, y}) = -i \tr \eps_{ijk} \partial_k (\sym_\zp^{-1} \partial_i \sym_\zp \sym_\zp^{-1} \partial_j \sym_\zp \sym_\zp^{-1}),
\end{align}
where $\eps_{ijk}$ is the anti-symmetric tensor with $\eps_{123} = 1$, and the variables are identified by $(1,2,3) = (y,\xi,\zeta)$.
Pulling $\partial_k$ out of the integral over $\omega$ and integrating by parts, we get
\begin{align*}
    \tr \int_{R}\partial_\xi \sym_\zp \sym_\zp^{-1} \{\sym_\zp, \sym_\zp^{-1}\}_{\zeta, y} \Big \vert^{\omega=0^+}_{\omega = 0^-} dR_3 =i \int_{\partial R} \int_{-\infty}^{+\infty} \Theta_\zp d\omega d\Sigma, 
\end{align*}
where we 
recall the definition of $\Theta_\zp$ in \eqref{eq:sigmaI1}, and note that the above right-hand side is now defined as a Legesgue integral since $\Theta_\zp$ is defined on $\partial R$ for all $\omega \in \mathbb{R}$.
Thus we have shown that
\begin{align*}
    \sigma_{I} = \frac{i}{16\pi^3} \int_{[\anot,\alpha]} \varphi ' (\lambda)\int_{\partial R} \int_{-\infty}^{+\infty} \Theta_\zp d\omega d\Sigma d\lambda.
\end{align*}
Integrating by parts in $\lambda$, we obtain
\begin{align*}
    \sigma_{I} = \frac{i}{16\pi^3} \int_{\partial R} \int_{-\infty}^{+\infty} \Theta_z d\omega d\Sigma,
\end{align*}
with now $z = \alpha + i\omega$ in the above integrand. 
The fact that only the boundary term survives follows from
analyticity of $\Theta_\zp$ in $\zp$ over the region of integration (so that $\partial_\lambda \Theta_\zp = -i \partial_\omega \Theta_\zp$). 
This completes the proof.
\end{proof}

\section{Proofs of main results on a periodic domain} \label{sec:pfs_periodic}
Let us introduce the following Hilbert spaces
\begin{align*}
    \mathcal{H}^{j}(X) := 
    \{\Psi \in L^2 (X) \otimes \mathbb{C}^n \quad | \quad \partial_\alpha \Psi \in L^2 (X) \otimes \mathbb{C}^n \quad \forall \ |\alpha| \le j\},
\end{align*}
where 
$j \in \mathbb{N}$ and
$X \in\{\mathbb{R}^d, \mathbb{T}^d\}$ with $d \in \{1,2\}$.
For $X = \mathbb{T}^d$ and a parameter $\kp \in (0,1]$, we define on $\mathcal{H}^{j}(X)$ the inner products and norms
\begin{align*}
    \aver{f,g}_{\kp,j} :=\sum_{|\alpha| \le j} \aver{\kp^{|\alpha|} D^\alpha f, \kp^{|\alpha|} D^\alpha g}, \qquad \norm{f}_{\kp,j} := \aver{f,f}_{\kp,j}^{1/2},
\end{align*}
with $\aver{\cdot, \cdot}$ the standard inner product in $L^2 (\mathbb{T}^d) \otimes \mathbb{C}^n$ .
For $X = \mathbb{R}^d$, 
the corresponding inner product and norm are
\begin{align*}
    (f,g)_{j} := \sum_{|\alpha| \le j} (D^\alpha f, D^\alpha g), \qquad 
    \vertiii{f}_j := (f,f)_j^{1/2},
\end{align*}
with $(\cdot, \cdot)$ the standard inner product in $L^2 (\mathbb{R}^d) \otimes \mathbb{C}^n$.
We will use the shorthand
\begin{align*}
    \mathcal{H} (X) := \mathcal{H}^0 (X), \qquad \mathcal{H} := \mathcal{H}(\mathbb{R}^2), \qquad \aver{\cdot, \cdot}_j := \aver{\cdot, \cdot}_{1,j}, \qquad \aver{\cdot, \cdot} := \aver{\cdot, \cdot}_0, \qquad (\cdot, \cdot) := (\cdot, \cdot)_0.
\end{align*}

We will repeatedly need to rescale functions with $\kp$, embed functions on $\mathbb{R}$ in the torus, and extend functions on $\mathbb{T}$ as functions on $\mathbb{R}$.
We define these operations as follows.
For $\kp \in (0, 1]$ and $y_1 \in \mathbb{R}$, define the unitary map $\ul_{\kp,y_1} : \mathcal{H}(\mathbb{R}) \rightarrow \mathcal{H}(\mathbb{R})$ by
\begin{align*}
    (\ul_{\kp, y_1} (u))(y) := \kp^{1/2} u(\kp(y-y_1) + y_1), \qquad y \in \mathbb{R}.
\end{align*}
Observe that $\ul_{\kp,y_1}^{-1} = \ul_{\kp^{-1}, y_1}$.

Let $L^\infty_{c, 2\pi} (\mathbb{R})$ denote the space of bounded functions $f \in L^\infty_c (\mathbb{R})$ such that
$\supp (f)\in (\tau,\tau+2\pi)$ for $\tau\in\Rm$. 
We define by $f_\sharp(y)=\sum_{q\in\Zm}f(2\pi q+y)$ their periodization (which is smooth when $f$ is smooth) and then  $\mathcal{P}: L^\infty_{c,2\pi}(\mathbb{R}) \rightarrow L^\infty(\mathbb{T})$ 
by
\begin{align*} 
    \mathcal{P}u (y) = 
    u_\sharp(y) = \sum_{q\in\Zm} u (2\pi q+y).
\end{align*}
Thus, $\mathcal{P}$ is an embedding of functions on $\mathbb{R}$ with sufficiently small support to functions on $\mathbb{T}$.


%
We 
define $\tilde \mP$ as the operator mapping a function in $u\in L^\infty(\mathbb{T})$ to $L^\infty_c(\Rm)$ by $$\tilde \mP u(y)=\chi_{[-\pi,\pi)}(y) u((y+\pi)\ {\rm mod} \ 2\pi -\pi)$$ with $\chi_I$ the indicatrix function of $I\subset\Rm$.
%
%
For $u$ vector-valued, $\mathcal{P} u$  and $\tilde{\mathcal{P}} u$ are defined as above component-wise.

Recall the definition of $\hat{H} (\xi)$ 
in \eqref{eq:hatH}.
We will show with Proposition \ref{acProp} below that the spectrum of $\hat{H} (\xi)$ in the interval $[0,E^2)$ consists entirely of eigenvalues. The same holds for the operator $\hat{H}'(\xi)$ (also defined in section \ref{subsec:construction}).
We will henceforth denote the \emph{combined} eigenvalues of $\hat{H} (\xi)$ and $\hat{H}'(\xi)$ in the interval $[0,E^2)$ by $\mu_1 (\xi) \le \mu_2 (\xi) \le \dots$; 
that is, each $\mu_j (\xi)$ is either an eigenvalue of $\hat{H} (\xi)$ or $\hat{H}'(\xi)$, with a given eigenvalue counted twice if it is an eigenvalue of both $\hat{H} (\xi)$ and $\hat{H}'(\xi)$. The corresponding eigenfunctions will be denoted by $\psi_{j,\xi}$, and chosen such that all eigenfunctions corresponding to the same operator are orthonormal.
Similarly, $\mu_{\kp,1} (\xi), \mu_{\kp,2} (\xi), \dots$ will denote the eigenvalues of $\hat{H}_\kp (\xi)$, with $\theta_{\kp,j,\xi}$ the corresponding orthonormalized eigenfunctions.

\subsection{Proof of Theorem \ref{periodicApprox}}\label{ref:subsec_pf_periodicApprox}
Recall that Theorem \ref{periodicApprox} assumes that the Hamiltonian $H$ satisfies hypothesis \hone\ and is a differential operator with structure given by \eqref{H2}.
Since the Hamiltonian $H - \frac{E_1+ E_2}{2}$ satisfies the same assumptions as $H$, we can take $-E_1 = E_2 =: E>0$ without loss of generality.
To simplify the 
proof of Theorem \ref{periodicApprox}, we will assume that $\varphi '$ is even, so we can
let $\Upsilon \in \mathcal{C}^\infty_c (-1, E^2-\delta_0)$ for some $\delta_0 > 0$, such that
$\varphi ' (x) = \Upsilon (x^2)$ for all $x \in \mathbb{R}$.
Recall that by 
Theorem \ref{thm:stabilityall} (\ref{it:P}), the infinite-space edge current $\sigma_I (H)$ is independent of $\varphi \in \fs (0,1; -E,E)$.

Before proving Theorem \ref{periodicApprox}, we will need to establish several results regarding elliptic partial differential operators on the torus. We show with Proposition \ref{sa} that the periodic operator $H_\kp$ is self-adjoint. Relevant spectral properties of the periodic and infinite-space operators $H_\kp$ and $H$ are then stated in Propositions \ref{evalsInfty} and \ref{acProp}. Proposition \ref{propEvalApprox} and Lemma \ref{switchBasis} then establish the convergence of the periodic eigenelements to their infinite-space analogues. Using these results, we prove Theorem \ref{periodicApprox} at the end of this section. 

\medskip

We begin by showing that $H_\kp$ is self-adjoint \cite{RS}. We will use the self-adjointness of $H_\kp$ to approximate its spectrum by the spectrum of $H$, which will be fundamental to our proof of Theorem \ref{periodicApprox}. The self-adjointness of $H_\kp$ will also ensure the implicit assertion that $\varphi' (H_\kp)$ is well defined.
We prove that $H_\kp$ is self-adjoint by creating a continuous path between $H_\kp$ and a corresponding differential operator with constant coefficients. To this end, define
\begin{align*}
    H_{\kp,\mu} := \kp^m \Big(M_0 D^m_y + M_m D^m_x\Big ) + \mu \Big(\kp^m \sum_{j=0}^m \cp_{\kp,j} (y) D^j_x D^{m-j}_y
    + \sum_{i+j \le m-1} \kp^{i+j} \cp_{\kp,ij} (y) D^i_x D^j_y\Big)
\end{align*}
for $\mu \in [0,1]$.
Note that $H_{\kp, \mu}$ is symmetric since $H_\kp$ is.

\begin{proposition} \label{periodicEllipticity}
There exist positive constants $C_1$ and $C_2$ such that
for all $f \in \mathcal{H}^{m}(\mathbb{T}^2)$,
\begin{align*}
    \norm{H_{\kp,\mu} f}^2 \ge C_1 \norm{f}^2_{\kp,m} - C_2 \norm{f}^2
\end{align*}
uniformly in $\kp \in (0,\kp_0]$ and $\mu \in [0,1]$. 
\end{proposition}

In the following proof, we use the Gagliardo-Nirenberg inequality on the torus which states that for any non-negative integers $i$ and $j$ satisfying $i+j \le m$,
\begin{align} \label{GN}
    \norm{\partial^i_x \partial^j_y f} \le \norm{\partial^m_x f}^{\frac{i}{m}} \norm{\partial^m_y f}^{\frac{j}{m}} \norm{f}^{1 - \frac{i+j}{m}}
\end{align}
for all $f \in \mathcal{H}^m(\mathbb{T}^2)$.

\begin{proof}
From the anti-commutation properties \eqref{anticommutation}, 
we see that $H_{\kp,\mu}^2 = \kp^{2m} \Big(M_0^2 D^{2m}_y + M_m^2 D^{2m}_x\Big ) + A_{\kp,\mu} + B_{\kp,\mu}$, where 
\begin{align*}
A_{\kp,\mu} = \kp^{2m}\Bigg( \{M_0, M_m\} D^m_y D^m_x + \mu \sum_{j =0}^{\lfloor \frac{m}{2} \rfloor} \Big( D^{2j}_x D^{m-j}_y \{M_0, \cp_{\kp,2j}\} D^{m-j}_y &+ D^{2(m-j)}_x D^{j}_y \{M_m, \cp_{\kp,m-2j}\} D^{j}_y
\Big) \\
&+ \Big(\mu \sum_{j=0}^m \cp_{\kp,j} (y) D^j_x D^{m-j}_y\Big)^2
\Bigg)
\end{align*}
is non-negative and 
$$B_{\kp,\mu} = \sum_{i+j \le 2m-1} \kp^{i+j} \tilde{\cp}_{\kp,\mu,ij} (y) D^i_x D^j_y$$ 
is of order $2m-1$, with $\kp^{|\alpha|}\sum_{i,j} \norm{\partial^\alpha \tilde{\cp}_{\kp,\mu,ij}}_{L^\infty} \le C_\alpha$ uniformly in $\kp$ and $\mu$ for all $\alpha \in \mathbb{N}^2$.
Using that $M_0$ and $M_m$ are non-singular and Hermitian, we have $$\norm{H_{\kp,\mu} f}^2 \ge \kp^{2m} c (\norm{D^{m}_y f}^2 + \norm{D^{m}_x f}^2) + (f, B_{\kp,\mu} f)$$
for some $c > 0$.
The result then follows from \eqref{GN} and the fact that
$|(f, B_{\kp,\mu} f)| \le C \norm{f}_{\kp,m} \norm{f}_{\kp,m-1}$.
\end{proof}

Proposition \ref{periodicEllipticity} implies the existence of a positive constant $C$ such that 
$\norm{f}_{\kp,m} \le C \norm{g}$
for all $f \in \mathcal{H}^m(\mathbb{T}^2)$ and $g \in \mathcal{H}(\mathbb{T}^2)$ satisfying
$(i - H_{\kp,\mu}) f = g$
uniformly in $\kp$ and $\mu$.
Indeed, redefining the constants $C_1$ and $C_2$, we have
\begin{align} \label{eq:alphaBdd}
    \norm{f}_{\kp,m}^2 \le
    C_1 \norm{H_{\kp,\mu} f}^2 + C_2 \norm{f}^2 =
    C_1 \norm{g - if}^2 + C_2 \norm{f}^2 \le C^2 \norm{g}^2,
\end{align}
where the last inequality 
follows from the fact that
$
    \norm{f}^2 = \Im (f, (i-H_{\kp,\mu}) f) = \Im (f,g) \le \norm{f}\norm{g}.
$
Clearly the same bound \eqref{eq:alphaBdd} holds if $i-H_{\kp,\mu}$ is replaced by $i+H_{\kp,\mu}$.

\begin{proposition} \label{sa}
For all $\kp \in (0,\kp_0]$, $H_\kp$ is self-adjoint on $\mathcal{H}(\mathbb{T}^2)$ with domain of definition 
$\mathcal{H}^m(\mathbb{T}^2)$.
\end{proposition}

\begin{proof}
By \eqref{eq:alphaBdd} and regularity of the coefficients of $H_{\kp,\mu}$, we can choose $N \in \mathbb{N}$ sufficiently large such that $\norm{(H_{\kp,(k+1)/N} - H_{\kp, k/N})f} \le \norm{g}/2$ for all $g \in \mathcal{H}(\mathbb{T}^2)$, $f \in \mathcal{H}^m(\mathbb{T}^2)$ satisfying $(i-H_{\kp,k/N}) f = g$, and
$k \in \{0,1, \dots, N-1\}$.
Since $H_{\kp, 0}$ is symmetric with constant coefficients, 
we know (or easily verify) that $i-H_{\kp, 0}$ is a bijection $\mathcal{H}^m(\mathbb{T}^2) \rightarrow \mathcal{H}(\mathbb{T}^2)$.
Now, suppose $i-H_{\kp, k/N}$ is a bijection $\mathcal{H}^m(\mathbb{T}^2) \rightarrow \mathcal{H}(\mathbb{T}^2)$ for some $k \in \{0,1, \dots, N-1\}$, meaning that $\norm{(H_{\kp,(k+1)/N} - H_{\kp, k/N})(i-H_{\kp,k/N})^{-1}} \le 1/2$. Then
\begin{align*}
    i-H_{\kp, (k+1)/N} = \Big(1 - (H_{\kp, (k+1)/N} - H_{\kp, k/N})(i-H_{\kp,k/N})^{-1}\Big) (i-H_{\kp, k/N})
\end{align*}
is a bijection $\mathcal{H}^m(\mathbb{T}^2) \rightarrow \mathcal{H}(\mathbb{T}^2)$, as the first factor on the above right-hand side can be inverted using the Neuman series.
We conclude by induction that $i-H_{\kp,1} = i-H_{\kp}$ is
a bijection $\mathcal{H}^m(\mathbb{T}^2) \rightarrow \mathcal{H}(\mathbb{T}^2)$.
The same reasoning implies that $i+H_\kp$ is a bijection $\mathcal{H}^m(\mathbb{T}^2) \rightarrow \mathcal{H}(\mathbb{T}^2)$.
Since $H_\kp$ is symmetric, it follows from
\cite[Theorem VIII.3]{RS} that $H_{\kp}$ is self-adjoint with domain of definition $\mathcal{H}^m(\mathbb{T}^2)$. This completes the proof.
\end{proof}

For $\xi \in \mathbb{R}$, define the differential operator $\hat{H}(\xi)$ on $\mathcal{H}(\mathbb{R})$ by
\begin{align*}
    \hat{H} (\xi) := M_0 D_y^m + M_m \xi^m +
    \sum_{j=0}^m \cp_j (y) \xi^j D_y^{m-j} + \sum_{i+j \le m-1} \cp_{ij} (y) \xi^i D_y^j.
\end{align*}
Since the coefficients of $H$ are independent of $x$, we have the decomposition
\begin{align*}
    H = \mathcal{F}^{-1}_{\xi \rightarrow x} \int_{\mathbb{R}}^{\oplus} \hat{H} (\xi) d\xi \mathcal{F}_{x \rightarrow \xi},
\end{align*}
with $\mathcal{F}$ the one dimensional Fourier transform in the $x$-variable.

For $\kp \in (0, \kp_0]$ and \tcbn{$\xi \in \mathbb{R}$}, define the differential operator $\hat{H}_\kp (\xi)$ on $\mathcal{H}(\mathbb{T})$ by
\begin{align} \label{Hlambda}
    \hat{H}_\kp (\xi) := \kp^m M_0 D^m_y + M_m \xi^m +
    \sum_{j=0}^m \kp^{m-j} \cp_{\kp,j} (y) \xi^j D^{m-j}_y + \sum_{i+j \le m-1} \kp^{j} \cp_{\kp,ij} (y) \xi^i D^j_y.
\end{align}
Noting that the coefficients of $H_\kp$ are independent of $x$, we see that for any $\psi \in \mathcal{H}^m(\mathbb{T})$ and $\xi \in \mathbb{Z}$,
\begin{align*}
    (\hat{H}_\kp (\kp \xi) \psi) (y) = e^{-i\xi x} (H_\kp \psi_\xi) (x,y)
\end{align*}
for all $(x,y) \in \mathbb{T}^2$, where $\psi_\xi \in\mathcal{H}^m(\mathbb{T}^2)$ is given by $\psi_\xi (x,y) = e^{i\xi x} \psi (y)$.


One can easily apply the logic from the proofs of Propositions 
\ref{periodicEllipticity} and \ref{sa} 
to obtain analogous results for $\hat{H}_\kp (\xi)$.
Namely, 
there exist positive constants $C_1$ and $C_2$ such that for all $f \in \mathcal{H}^m(\mathbb{T})$,
\begin{align}\label{periodicEllipticity1D}
    \norm{\hat{H}_\kp (\xi) f}^2 \ge (C_1 \xi^{2m} - C_2) \norm{f}^2 + C_1 \norm{\kp^m D^m_y f}^2
\end{align}
uniformly in $\kp \in (0,\kp_0]$ and $\xi \in \mathbb{R}$.
As above, this implies that $\hat{H}_\kp (\xi)$ is self-adjoint with domain of definition $\mathcal{H}^m(\mathbb{T})$. 
In addition, we obtain
\begin{proposition} \label{evalsInfty}
The operator $\hat{H}_\kp ^2 (\xi)$ is self-adjoint with domain of definition $\mathcal{H}^{2m}(\mathbb{T})$ and its spectrum consists only of eigenvalues that go to $+\infty$.
\end{proposition}

\begin{proof}
Clearly $\hat{H}_\kp^2 (\xi)$ is symmetric.
Moreover, we see that $i+\hat{H}_\kp^2 (\xi) = (\frac{1 - i}{\sqrt{2}} + \hat{H}_\kp (\xi)) (\frac{-1 + i}{\sqrt{2}} + \hat{H}_\kp (\xi))$, with each factor on the right-hand side a bijection $\mathcal{H}^{k+2m}(\mathbb{T}) \rightarrow \mathcal{H}^k(\mathbb{T})$ for any $k \in \mathbb{N}$.
Thus $i+\hat{H}_\kp^2 (\xi)$ is a bijection $\mathcal{H}^{2m}(\mathbb{T}) \rightarrow \mathcal{H}(\mathbb{T})$. The same holds for $i-\hat{H}_\kp^2 (\xi)$, meaning that
$\hat{H}_\kp^2 (\xi)$ is self-adjoint with domain $\mathcal{H}^{2m}(\mathbb{T})$.

To prove the spectral property, we write $(1 + \hat{H}_\kp^2 (\xi))^{-1} = (1 + \hat{H}_\kp^2 (\xi))^{-1} (1 + \hat{H}_{\kp,0}^2 (\xi)) (1 + \hat{H}_{\kp,0}^2 (\xi))^{-1}$ for $\hat{H}_{\kp,0} (\xi) := M_0 \kp^m D^m_y + M_m \xi^m$.
Since $\hat{H}_{\kp,0} (\xi)$ has constant coefficients, it is clear that
$(1 + \hat{H}_{\kp,0}^2 (\xi))^{-1}$ is Hilbert-Schmidt.
Since $(1 + \hat{H}_\kp^2 (\xi))^{-1} (1 + \hat{H}_{\kp,0}^2 (\xi))$ is bounded, we conclude that
$(1 + \hat{H}_\kp^2 (\xi))^{-1}$ is Hilbert-Schmidt, hence compact.
Therefore the spectrum of $(1 + \hat{H}_\kp^2 (\xi))^{-1}$ consists entirely of eigenvalues that converge to $0$, meaning that the spectrum of $\hat{H}_\kp^2 (\xi)$ consists only of eigenvalues tending to $+\infty$.
\end{proof}

We recall the following Courant-Fischer min-max theorem characterizing the (discrete) spectrum of self-adjoint operators \cite{RS4, Teschl}:
\begin{theorem} \label{maxmin}
Let $A:\mathcal{H}_0 \rightarrow \mathcal{H}_0$ be self-adjoint with $\mathcal{H}_0$ a separable Hilbert space, and let $E_1 \le E_2 \le E_3 \dots$ be the eigenvalues of $A$ below the essential spectrum (counted with multiplicity), respectively, the infimum of the essential spectrum, once there are no more eigenvalues left. Let $v_j$ denote the eigenfunction corresponding to $E_j$ (when the latter is an eigenvalue), such that $\aver{v_i, v_j} = \delta_{ij}$ for all $i$ and $j$. Then for any $\ell \in \mathbb{N}_+$ and $u \in \s (v_1, v_2, \dots, v_{\ell-1})^{\perp}$, $\aver{u, Au} \ge E_{\ell}\norm{u}^2$.
Moreover, we have
\begin{align*}
    E_\ell &= 
    \max_{\phi_1, \dots, \phi_{\ell-1}} \min \{ \aver{\phi, A \phi}: \phi \in \s (\phi_1, \dots, \phi_{\ell-1})^\perp, \;\norm{\phi} = 1\}\\
    &= \min_{\phi_1, \dots, \phi_\ell} \max \{ \aver{\phi, A \phi} : \phi \in \s (\phi_1, \dots, \phi_\ell), \; \norm{\phi} = 1\}.
\end{align*}
\end{theorem}

We now gather useful properties of the infinite-space operator $\hat{H} (\xi)$ and its spectrum, which will use Lemma \ref{GardingUnbdd} below.

\begin{proposition} \label{acProp}
Take $H = \Op (\sym)$ as above, and fix $\delta > 0$.
Then for each $\xi \in \mathbb{R}$, $\hat{H} (\xi)$ and $\hat{H}^2 (\xi)$ are self-adjoint with respective domains of definition $\mathcal{H}^m (\mathbb{R})$ and $\mathcal{H}^{2m} (\mathbb{R})$. The spectrum of $\hat{H}^2 (\xi)$ in the interval $[0, E^2 - \delta)$ consists of a finite number of eigenvalues, each with finite multiplicity. These eigenvalues and corresponding eigenfunctions can be chosen so that they are analytic in $\xi$.
Moreover,
the rank of $\hat{H}^2 (\xi)$
in $[0, E^2 - \delta)$ is uniformly bounded in $\xi \in \mathbb{R}$
and vanishes whenever $|\xi|$ is sufficiently large.
\end{proposition}

\begin{proof}
Fix $\xi \in \mathbb{R}$. 
Since $\hat{H} (\xi) = \Op (\sym_\xi)$ with $\sym_\xi \in S^m$ elliptic, it follows from \cite{Bony, Hormander} that $\hat{H} (\xi)$ is self-adjoint with domain of definition $\mathcal{H}^m(\mathbb{R})$.
We know that $\hat{H}^2 (\xi) = \Op (\tau_\xi)$, where $\tau_\xi (y, \zeta) = (\sym\sharp \sym) (y,\xi,\zeta)$.
Applying the same argument to the elliptic symbol $\tau_\xi \in S^{2m}$, we see that $\hat{H}^2 (\xi)$ is self-adjoint
with domain of definition $\mathcal{H}^{2m}(\mathbb{R})$.
Define $\sym_{\pm,\xi} (\zeta) := \sym_\pm (\xi,\zeta)$ and $\tau_{\pm, \xi} := \sym_{\pm,\xi}^2$, and observe that
$\tau_\xi =\tau_{+,\xi}$ whenever $y \ge y_0$, and $\tau_\xi = \tau_{-,\xi}$ whenever $y \le -y_0$.
Since the $\tau_{\pm,\xi}$ are independent of $y$,
$$T_{\pm,\xi} := \Op (\tau_{\pm,\xi}) = (\Op (\sym_{\pm,\xi}))^2 \ge 0.$$
Let $\chi \in \fs (0,\pi/2,-1,1)$, and define $\chi_\eps(y) := \chi (\eps y)$ and
$$\Op (\tau_{\eps, \xi}) =T_{\eps,\xi} := \sin (\chi_\eps (y)) T_{+,\xi} \sin(\chi_\eps (y)) + \cos(\chi_\eps(y)) T_{-,\xi} \cos (\chi_\eps (y))$$
for $\eps \in (0,1]$.
We see that $\tau_{\eps,\xi} =\sin^2(\chi_\eps (y)) \tau_{+,\xi} + \cos^2(\chi_\eps(y)) \tau_{-,\xi} + \eps r_{\eps,\xi}$, with $r_{\eps,\xi} \in S(\aver{\zeta}^{2m})$ uniformly in $\eps$, and $\tau_{\pm,\xi} \ge \max\{E^2, c \aver{\zeta}^{2m}\}$.
It follows that
$\tau_{\eps,\xi} \ge \max \{E^2, c \aver{\zeta}^{2m}\} - c' \eps \aver{\zeta}^{2m}$ for some positive constant $c'$, hence
$\tau_{\eps, \xi} \ge E^2 - \delta/4$ for $\eps$ sufficiently small.
Using the fact that $T_{\eps, \xi} \ge 0$, we apply
Lemma \ref{GardingUnbdd} below (with a unitary transformation) to
conclude that $T_{\eps, \xi}$
has no spectrum in the interval $(-\infty, E^2 - \delta/2)$ for $\eps$ sufficiently small.

We now observe that
\begin{align*}
    (i-T_{\eps,\xi})^{-1} - (i-\hat{H}^2 (\xi))^{-1} =
    (i-T_{\eps,\xi})^{-1} (\hat{H}^2 (\xi) - T_{\eps,\xi}) (i-\hat{H}^2 (\xi))^{-1} \in \Op (S (\aver{y}^{-\infty} \aver{\zeta}^{-2m})),
\end{align*}
as $\tau_\xi - \tau_{\eps, \xi} \in S(\aver{y}^{-\infty} \aver{\zeta}^{2m})$ and both resolvents have symbol in $S (\aver{\zeta}^{-2m})$.
Since $\aver{\zeta}^{-2m} \in L^2 (\mathbb{R})$, it follows that $(i-T_{\eps,\xi})^{-1} - (i-\hat{H}^2 (\xi))^{-1}$ is Hilbert-Schmidt, hence compact.
By \cite[Theorem 6.18]{Teschl}, we conclude that the essential spectra of $T_{\eps, \xi}$ and $\hat{H}^2 (\xi)$ are the same, meaning that
the spectrum of $\hat{H}^2 (\xi)$ in $[0, E^2 - \delta/2)$ consists of eigenvalues, each with finite multiplicity.
The number of these eigenvalues must be finite, as we can verify from the proof of Proposition \ref{trclass} that $\Phi_0 (\hat{H}^2 (\xi))$ is trace-class for all $\Phi_0 \in \mathcal{C}^\infty_c (-1, E^2)$. 

Fix $\xi \in \mathbb{R}$.
Let $\psi \in \mathcal{C}^\infty_c (\mathbb{R}) \otimes \mathbb{C}^n$ such that $\vertiii{\psi} = 1$.
Define 
$\tilde{\psi}:= \ul^{-1}_{\kp_1,0} \psi$, where
$\kp_1 \in (0, \kp_0]$
is sufficiently small so that $\supp (\tilde{\psi}) \subset (-\pi/2, \pi/2)$.
Then by \eqref{periodicEllipticity1D}, $$\vertiii{\hat{H} (\xi) \psi} = \norm{\hat{H}_{\kp_1} (\xi) \mathcal{P} \tilde{\psi}} \ge C_1 \xi^m - C_2.$$
Since $\mathcal{C}^\infty_c (\mathbb{R})$ is dense in $\mathcal{H}^m (\mathbb{R})$ with respect to the $\mathcal{H}^m$-norm,
we have proved that $\hat{H}^2 (\xi)$ 
has no spectrum in $[0, E^2)$ whenever $|\xi|$ is sufficiently large.

Let 
$\nu_1 (\xi)\le \nu_2 (\xi)\le \dots\le \nu_{\tilde{N}(\xi)} (\xi)$ denote the eigenvalues of $\hat{H}^2 (\xi)$ in the interval $[0,E^2 - \delta/4)$. We showed above that $\tilde{N} (\xi)$ is finite for all $\xi \in \mathbb{R}$. 
Let $N(\xi)$ be the number of eigenvalues of $\hat{H}^2 (\xi)$ in $[0, E^2 -\delta)$.
To show that $N(\xi)$ is bounded uniformly in $\xi$,
we will first prove that the eigenvalues are uniformly equicontinuous in $\xi$.
By ``uniformly equicontinuous'', we mean that for all $\eps > 0$, there exists $\eta > 0$ such that 
for any $|\xi_1 - \xi_2| < \eta$ and any $j$ satisfying $j \le \tilde{N}(\xi_1)$, 
it follows that $|\nu_j (\xi_2) - \nu_j (\xi_1)| < \eps$.

Fix $C_0 > 0$. Let $\bar{\xi} > 0$ such that $\hat{H}^2 (\xi)$
has no spectrum in $[0, E^2 - \delta/4)$ whenever $|\xi| > \bar{\xi}$.
Since $\hat{H}^2(\xi+\rho) - \hat{H}^2(\xi)$ is a differential operator of order $2m-1$ with coefficients bounded by $C |\rho|$, the Gagliardo-Nirenberg inequality \eqref{GN} on $\mathbb{R}^2$ implies that
\begin{align} \label{eq:ue}
    |(u, (\hat{H}^2(\xi+\rho)- \hat{H}^2(\xi)) u)| \le C |\rho| \vertiii{u}^{\frac{2m-1}{m}}_m \vertiii{u}^{\frac{1}{m}} \le C C_0^{\frac{2m-1}{m}} |\rho| \vertiii{u}^2; \qquad \quad \rho \in [-1,1], \quad \xi \in [-\bar{\xi}, \bar{\xi}],
\end{align}
for all $u \in \mathcal{H}^m (\mathbb{R})$ satisfying 
$\vertiii{u}_{m} \le C_0 \vertiii{u}$.
By ellipticity of $\hat{H} (\xi)$, we can choose $C_0$ sufficiently large so that
$\vertiii{\hat{H}(\xi) u} \ge E\vertiii{u}$ whenever $\vertiii{u}_{m} > C_0 \vertiii{u}$ 
and $\xi \in \mathbb{R}$.
Hence given any 
$\ell$-dimensional subspace $S_\ell \subset \mathcal{H}^m (\mathbb{R})$, 
the bound \eqref{eq:ue} holds uniformly in $u \in S_\ell$ 
and $\xi \in [-\bar{\xi}, \bar{\xi}]$ satisfying $\vertiii{\hat{H}(\xi) u} < E \vertiii{u}$. 
It follows that
\begin{align*}
    \Big |\max_{u \in S_\ell, \vertiii{u} =1} (u, \hat{H}^2(\xi+\rho) u) - \max_{u \in S_\ell,\vertiii{u}=1} (u,\hat{H}^2(\xi) u)\Big | \le C C_0^{\frac{2m-1}{2m}} |\rho|; \qquad \quad \rho \in [-1,1] 
\end{align*}
uniformly in $\xi \in [-\bar{\xi}, \bar{\xi}]$ satisfying $\max_{u \in S_\ell,\vertiii{u}=1} \vertiii{\hat{H}(\xi) u} < E$.
Therefore by Theorem \ref{maxmin},
\begin{align*}
|\nu_\ell (\xi+\rho) - \nu_\ell (\xi)| =
    \Big |\min_{S_\ell} \max_{u \in S_\ell,\vertiii{u}=1} (u, \hat{H}^2(\xi+\rho) u)- \min_{S_\ell} \max_{u \in S_\ell,\vertiii{u}=1} (u,\hat{H}^2(\xi) u)\Big | \le C C_0^{\frac{2m-1}{2m}} |\rho| 
\end{align*}
uniformly in $\rho \in [-1,1]$ and $\ell \le \tilde{N} (\xi)$.
Thus we have proven uniform equicontinuity of the eigenvalues.


Suppose by contradiction that $N(\xi)$ is not bounded uniformly in $\xi$.
Since $N(\xi)$ vanishes outside $[-\bar{\xi}, \bar{\xi}]$, there exists a sequence $(\xi_k) \subset [-\bar{\xi}, \bar{\xi}]$ and a number $\xi_* \in [-\bar{\xi}, \bar{\xi}]$ such that $\xi_k \rightarrow \xi_*$ and $N(\xi_k) \rightarrow \infty$ as $k \rightarrow \infty$.
Uniform equicontinuity of the 
eigenvalues implies that for all $k$ sufficiently large, 
all eigenvalues $\nu_j$ satisfying $\nu_j (\xi_k) < E^2 - \delta$ must also satisfy $\nu_j (\xi_*) < E^2 - \delta/4$.
This implies
\begin{align*}
    \tilde{N} (\xi_*) \ge N(\xi_k) \rightarrow \infty,
\end{align*}
which is a contradiction since we already proved that $\tilde{N} (\xi_*)$ must be finite.

It remains to show that the eigenvalues and eigenfunctions can be chosen analytic in $\xi$ when they lie in $[0, E^2 - \delta)$. 
We have shown that there exists $\tilde{N}_0 < \infty$ such that $\tilde{N} (\xi) \le \tilde{N}_0$ for all $\xi \in \mathbb{R}$.
Thus for all $\xi$ such that $\nu_1 (\xi) < E^2 - \delta/2$, there exists $k(\xi) \in \{1,2, \dots, \tilde{N}(\xi)\}$ such that $\nu_{k(\xi)+1} (\xi) \ge E^2 - \delta/2$ and $\nu_{k(\xi)+1} (\xi) - \nu_{k(\xi)} (\xi) \ge \frac{\delta}{4 \tilde{N}_0}$.
Here, $\nu_{\tilde{N} (\xi)+1} (\xi)$ is defined as the smallest eigenvalue of $\hat{H}^2 (\xi)$ in the interval $[E^2-\delta/4, E^2-\delta/8)$; or if no such eigenvalue exists, then $\nu_{\tilde{N} (\xi)+1} (\xi):= E^2-\delta/8$.

Uniform equicontinuity implies that there exists $\eta > 0$ such that 
$$|\nu_{k(\xi) + 1} (\hat{\xi}) - \nu_{k(\xi)+1} (\xi)| + |\nu_{k(\xi)} (\hat{\xi}) - \nu_{k(\xi)} (\xi)| \le \frac{\delta}{8 \tilde{N}_0} \qquad \text{whenever} \qquad |\hat{\xi} - \xi| < \eta.$$
In particular, this means
$$\nu_{k(\xi) + 1} (\hat{\xi}) - \nu_{k(\xi)} (\hat{\xi}) \ge \frac{\delta}{4\tilde{N}_0} + \nu_{k(\xi) + 1} (\hat{\xi}) - \nu_{k(\xi) + 1} (\xi) + \nu_{k(\xi)} (\xi) - \nu_{k(\xi)} (\hat{\xi}) \ge \frac{\delta}{8 \tilde{N}_0} \quad \text{whenever} \quad |\hat{\xi} - \xi| < \eta.$$
That is, $\nu_{k(\xi)+1} (\hat{\xi})$ never enters the interval $[0,E^2-\delta)$ and is bounded away from $\nu_{k(\xi)} (\hat{\xi})$ so long as $|\hat{\xi} - \xi| < \eta$.
Moreover, $\nu_{k(\xi)} (\hat{\xi})\in [0,E^2-\delta/8)$ for all $|\hat{\xi} - \xi| < \eta$.

Now write $[-\bar{\xi}, \bar{\xi}] \subset \cup_j B_j$, where $\{B_j\}$ is a finite collection of overlapping open intervals ($B_j \cap B_{j+1} \ne \emptyset$) of length $2 \eta$.
Let $\xi_j$ be the midpoint of $B_j$.
Then for every $j$, we can apply \cite[Theorem VII.1.7 and Section VII.3.1]{kato2013perturbation} to the finite system of eigenvalues $\nu_1 (\xi), \nu_2 (\xi), \dots, \nu_{k(\xi_j)} (\xi)$ and corresponding eigenfunctions to conclude that they can be chosen analytic in $\xi \in B_j$.
Since the $B_j$ form an open cover of $[-\bar{\xi}, \bar{\xi}]$, this proves that all eigenvalues (and corresponding eigenfunctions) in the energy interval $[0, E^2- \delta)$ can be chosen analytic in $\xi$.
This completes the result.
\end{proof}

\begin{lemma} \label{GardingUnbdd}
Suppose $A_h = \Op_h (\sg)$ is non-negative for all $h \in (0,1]$, 
where $\sg \in \smeh$. 
Then for any constant $c\ge 0$ satisfying $\sg_{\min} \ge c$, there exist constants $\beta > 0$ and $0 < h_0 \le 1$ such that $\aver{\psi, A_h \psi} \ge c - \beta h$ for all $\norm{\psi} = 1$ and $h \in (0,h_0]$.
\end{lemma}

\begin{proof}
We know that $A_h$ is self-adjoint by our assumptions on $\sg$.
Since $A_h \ge 0$, it follows that $(1+A_h)^{-1} = \Op_h (r_h)$, where $S(1) \ni r_h = (1+\sg)^{-1} + O(h)$.
By the semi-classical sharp G\aa rding inequality \cite[Theorem 7.12]{DS}, there exist constants $\beta_0 > 0$ and $0 < h_0 \le 1$ such that
\begin{align} \label{Garding}
0 < \aver{\phi, (1+A_h)^{-1}\phi} \le \Big((1+c)^{-1} +\beta_0 h\Big) \norm{\phi}^2
\end{align}
for all nonzero $\phi \in \mathcal{H}$ and $h \in (0,h_0]$.
Let $\psi \in \mathcal{D} (A_h)$ such that $\norm{\psi} =1$, and define $\phi = (1+A_h)^{1/2} \psi$.
Then
\begin{align*}
    \aver{\psi, (1+A_h) \psi} = \norm{\phi}^2 = \frac{\norm{\phi}^2}{\norm{\psi}^2} = \frac{\norm{\phi}^2}{\aver{\phi, (1+A_h)^{-1} \phi}} \ge 1+c - \beta h, \qquad h \in (0, h_0]
\end{align*}
for some fixed constant $\beta$. This completes the proof.
\end{proof}

\paragraph{Spectral approximations.}
Recall that $H' = \Op (\sym')$, where for every $y' \in \mathbb{R}$, there exists $y \in \mathbb{R}$ such that $\sym' (y', \xi, \zeta) = \sym(y,\xi,\zeta)$. Since $\sym$ satisfies \hone, it follows that $\sym'$ is elliptic: $\sym'_{\min} \ge c \aver{\xi,\zeta}^m - 1$.
Hence the fact that $\sym' = \sym_+$ whenever $y \le \pi - y_0$ and $\sym' = \sym_-$ whenever $y \ge \pi + y_0$ implies that $\sym'$ also satisfies \hone, with the roles of $\sym_+$ and $\sym_-$ reversed.
We conclude that Proposition \ref{acProp} 
still holds if $H$ is replaced by $H'$, where we would instead define
$\tilde{\psi} := \ul^{-1}_{\kp_1, \pi} \psi$
in the third paragraph of the proof.

The approximation result of Theorem \ref{periodicApprox} requires a detailed analysis of the spectrum of $H_\kp$ constructed in \eqref{Hlambda2D}, which it inherits from that of the operators $H$ and $H'$.  For $\xi \in \mathbb{R}$,
let $\mu_1 (\xi) \le \mu_2 (\xi) \le \mu_3 (\xi) \dots$ 
denote the {\em combined} eigenvalues of $\hat{H}^2 (\xi)$ and $\hat{H}'^2 (\xi)$
below $\Sigma_{\ess} (\hat{H}^2 (\xi)) \cup \Sigma_{\ess} (\hat{H}'^2 (\xi))$, respectively, the infimum of $\Sigma_{\ess} (\hat{H}^2 (\xi)) \cup \Sigma_{\ess} (\hat{H}'^2 (\xi))$, once there are no more eigenvalues left.
Let $R_\xi$ (resp. $R'_\xi$) be the set consisting of the indices $j$ for which $\mu_j (\xi)$ is an eigenvalue of $\hat{H}^2 (\xi)$ (resp. $\hat{H}'^2 (\xi)$).
Thus $R_\xi$ and $R'_\xi$ form a partition of the indices below
$\Sigma_{\ess} (\hat{H}^2 (\xi)) \cup \Sigma_{\ess} (\hat{H}'^2 (\xi))$.
When $\mu_j (\xi) \in R_\xi \cup R'_\xi$, we denote the corresponding normalized eigenfunction by $\psi_{j,\xi}(y)$.
We choose the eigenfunctions so that $(\psi_{i,\xi}, \psi_{j,\xi}) = \delta_{ij}$ whenever $i,j \in R_\xi$ or $i,j \in R'_\xi$.

Let $\tilde{s} (\xi)$ be the total number of eigenvalues of $\hat{H}^2$ and $\hat{H}'^2$ lying in $[0, E^2 - \frac{\delta_0}{2})$.
By Proposition \ref{acProp} 
(and the above paragraphs), we know that 
$\tilde{s}(\xi)$ is indeed finite and that
the eigenvalues $\mu_1 (\xi) \dots \mu_{\tilde{s} (\xi)} (\xi)$
make up the entire spectrum of $\hat{H}^2 (\xi)$ and $\hat{H}'^2 (\xi)$ in $[0, E^2 - \frac{\delta_0}{2})$. \tcbn{Moreover, there exists $s_0 \in \mathbb{N}$ such that $\tilde{s}(\xi) \le s_0$ for all $\xi \in \mathbb{R}$, and there exists a compact interval $I \subset \mathbb{R}$ such that $\tilde{s}(\xi) = 0$ for all $\xi \notin I$.}

Let $\mu_{\kp,1}(\xi) \le \mu_{\kp,2}(\xi) \le \mu_{\kp,3}(\xi) \dots$ denote the eigenvalues of the operator $\hat{H}^2_{\kp} (\xi)$ defined on the torus, which we know are well defined and go to infinity
by Proposition \ref{evalsInfty}.
Let $\theta_{\kp,1,\xi}, \theta_{\kp,2,\xi},\theta_{\kp,3,\xi}, \dots$ be the corresponding (orthonormalized) eigenfunctions.

We now prove 
a result of (collective) exponential decay of the eigenfunctions of an operator with constant coefficients outside a bounded interval.
This verifies that the low-energy eigenfunctions are localized to the vicinity of the domain walls, in both the periodic and infinite-space settings.
\begin{proposition} \label{expDecay}
For any $\alpha \in \mathbb{N}$,
there are positive constants $C$ and $r$ such that $|\partial^\alpha_y \psi_{j,\xi}| (y) \le C e^{-r |y|}$ uniformly in $j \in \{1,2, \dots, \tilde{s}(\xi)\}$ and $\xi \in I$. 
Similarly, given any closed interval $T \subset \mathbb{T}$ not containing $0$ or $\pi$, $\norm{\kp^\alpha \partial^\alpha_y \theta_{\kp,j,\xi}}_{T} \le C e^{-r/\kp}$ uniformly in 
$j$ and $\xi$ satisfying $\mu_{\kp,j} (\xi) < E^2 - \delta_0/2$.
\end{proposition}
Here and below, we use the shorthand $\norm{\cdot}_{T} := \norm{\cdot}_{L^2 (T) \otimes \mathbb{C}^n}$ if $T \subset \mathbb{T}$.
\begin{proof}
We first prove the claim regarding the infinite-space eigenfunctions and assume $y>0$ for concreteness. 
Restricted to the region $y > y_0$, we have
$
    \hat{H} (\xi) = \Op (\sym_{+,\xi} (\zeta)),
$
where $\sym_{+,\xi} (\zeta) = \sym_+ (\xi,\zeta)$.
Our spectral gap assumption in \hone\ is that $\sym^2_+(\xi,\zeta)$ does not have eigenvalues inside $[0,E^2)$ while here by construction, $\mu_{j}\in [0,E_0^2]$ 
with $E_0 := \sqrt{E^2 - \frac{\delta_0}{2}} < E$.
By construction of $\psi_{j,\xi}$, the statement thus follows if we can prove the bound for any solution $\psi$ of 
\begin{align}\label{eq:Decay}
  (\sym^2_{+,\xi}(D_y) -\mu) \psi(y) =0,\qquad y > y_0 
\end{align}
for an eigenvalue $\mu\in [0,E_0^2]$. 
%
We replace the above system by the first-order system 
\[
  \frac{d}{dy} u = Au
\]
for $A = A_\xi$ a $N\times N$ matrix with $N=2mn$ and $u=(\psi,\psi',\ldots,\psi^{(2m-1)})^t$. We then write $A=PJP^{-1}$ with $J$ in Jordan form and $u=Pv$. 
Solving $J v = \frac{d}{dy}v$, we see that each eigenvalue $\nu$ of $J$ must correspond to an eigenfunction of the form
$v=e^{\nu y}v_0$ so that $u=e^{\nu y} Pv_0$. 
This means that for every such $\nu$, there exists a solution $\psi$ 
of \eqref{eq:Decay} that is of the form $\psi (y) = e^{\nu y} \bar{\psi}_0$, where $\bar{\psi}_0 \in \mathbb{C}^n$ is independent of $y$.
Since $\psi \ne 0$ by definition, this shows that
\begin{align}\label{eq:detCondition}
\det (\sym^2_{+,\xi} (-i\nu) - \mu) = 0; 
\end{align}
that is, $\sym^2_{+,\xi} (-i\nu)$ must have an eigenvalue that is equal to $\mu$.
We will now show that $\nu$ must be bounded away from the imaginary axis, uniformly in all variables.
First, observe that since $\sym_{+,\xi} (\zeta)$ is bounded above by $C\aver{\xi, \zeta}^m$ and elliptic in $\zeta$, there exist positive constants $\nu_1$ and $\delta_1$ such that whenever $|\Im \nu| > \nu_1$ and $|\Re \nu| < \delta_1$, all eigenvalues of $\sym_{+,\xi} (-i\nu)$ lie outside $(-E, E)$.
Since $\sym_{+,\xi} (\zeta)$ is continuous in $\zeta$ and $\sym_{+,\xi}(-i\nu)$ does not have spectrum in $(-E, E)$ when $\Re \nu = 0$, there exists $\delta_2 > 0$ such that whenever $|\Im \nu| \le \nu_1$ and $|\Re \nu| < \delta_2$, all eigenvalues of $\sym_{+,\xi} (-i\nu)$ lie outside $[-E_0, E_0]$.
Thus, taking $\tilde{\delta} := \min\{\delta_1, \delta_2\}$, we have shown that all $\nu$ satisfying \eqref{eq:detCondition} must also satisfy $|\Re \nu| \ge \tilde{\delta}$.

It is clear that $\tilde{\delta} = \tilde{\delta} (\xi,\mu)$ is continuous in both variables.
Since
$\sym_{+,\xi} (\zeta)$ is bounded above by $C\aver{\xi, \zeta}^m$ and elliptic in $\xi$, there exist positive constants
$\xi_0$ and $\delta_0$ such that for all $|\xi| \ge \xi_0$, any $\nu$ solving \eqref{eq:detCondition} must also satisfy $|\Re \nu| \ge \delta_0$.
Thus, setting $K := [-\xi_0, \xi_0] \times [-E_0, E_0]$ and $\delta := \min\{\delta_0, \min_{(\xi,\mu) \in K} \tilde{\delta} (\xi, \mu)\} > 0$,
we have shown that all solutions $\nu$ of \eqref{eq:detCondition} satisfy $|\Re \nu| \ge \delta > 0$ uniformly in $\xi \in \mathbb{R}$ and $\mu\in [0, E_0^2]$.

Note that every solution $\psi$ of \eqref{eq:Decay} is a finite linear combination of terms of the form $y^k e^{\nu y} \bar{\psi}_k$, with 
$\nu$ satisfying \eqref{eq:detCondition}.
Since $\psi$ is square integrable in $y$ (and thus cannot be exponentially increasing), 
this shows that $\psi$ and all of its derivatives are exponentially decreasing as $y \rightarrow \infty$, uniformly in $\xi$ and $\mu$.
Thus we have proven the desired decay property for eigenfunctions of $\hat{H} (\xi)$ as $y \rightarrow \infty$. The argument for exponential decay as $y \rightarrow -\infty$ is identical, as is the proof for eigenfunctions of $\hat{H}' (\xi)$.

The same ideas are used for the periodic eigenfunctions, only now we must handle the $\kp$-dependence.
By definition of $T$, there exists a closed interval $T_1 \subset \mathbb{T}$, also not containing $0$ or $\pi$, such that $T \subset T_1^\circ$. 
Suppose that $T_1 \subset (0, \pi)$ for concreteness.
The defining equation for the eigenfunctions $\theta_{\kp,j,\xi}$ in $T_1$ is
$$
\sym_+^2(\xi,\kp D_y) \theta_{\kp, j,\xi}(y) = \mu_{\kp, j}(\xi)\theta_{\kp, j,\xi}(y).
$$
This is equivalent to solving an equation of the form
\begin{align}\label{eq:DecayPeriodic}
  (\sym_+^2(D_y) -\mu) f(y) =0,\qquad f(y)= \theta_\kp(\kp y),
\end{align}
where $y \in T_{\kp,1} := \kp^{-1} T_1$ 
and $0\le\mu\leq E_0^2<E^2$.
The differential equation in \eqref{eq:DecayPeriodic}
is identical to \eqref{eq:Decay}, with $f$ replacing $\psi$. 
Hence
$\theta_\kp (y) = f(y/\kp)$ is exponentially increasing or decreasing in $T_{1}$ with a rate that is proportional to $\kp^{-1}$.
Since $\theta_\kp (y)$ has bounded $L^2-$norm on $T_{1}$, 
its $L^2-$norm on $T$ is bounded by $C e^{-r/\kp}$ for $r>0$.
Derivatives of $\theta_\kp$ also satisfy such an estimate. This completes the proof.
\end{proof}

What remains before proving Theorem \ref{periodicApprox} is to show that the
combined spectrum (including eigenspaces) of the infinite-space Hamiltonians is a good approximation of the spectrum of the periodic Hamiltonian.
We will truncate the infinite-space eigenfunctions $\psi_{j,\xi}$ so that they have compact support and can be embedded in the torus.
Due to the rapid decay (in $y$) of these eigenfunctions and their derivatives, the truncated versions \emph{almost} solve the eigenproblem, admitting a small residual term that goes to zero with $\kp$.
The fact that the eigenvalues 
for the periodic problem are themselves well approximated
is the statement of Proposition \ref{propEvalApprox}, and
follows from Lemma \ref{lemmaOrthogonal} and Theorem \ref{maxmin}.
Once the error in eigenvalues is controlled, we use
Lemma 
\ref{switchBasis}
to show that the periodic eigenfunctions are well approximated by their infinite-space analogues, which will allow us to
directly bound the error $|\tilde{\sigma}_I (H_\kp) - \sigma_I (H)|$.


\tcbn{
For the following results, $\xi \in \mathbb{R}$ is assumed to be arbitrary and is dropped from the notation for brevity. It will be clear that the below estimates are all uniform in $\xi$, given the uniform bounds from Proposition \ref{expDecay}.}

\medskip

\paragraph{Eigenvalue approximation.}

Let $s$ be the total number of eigenvalues of $\hat{H}^2$ and $\hat{H}'^2$ lying in $[0, E^2 - \delta_0)$.
For $\ell \in \{0,1, \dots, s\}$, define
$\tilde{R}_{\ell}:= \{1,2, \dots, \ell\} \cap R$ and $\tilde{R}'_{\ell}:= \{1,2, \dots, \ell\} \cap R'$. \footnote{The sets  $R = R_\xi$ and $R'= R'_\xi$ are defined between Propositions \ref{acProp} and \ref{expDecay}.}
We will need to approximate the eigenfunctions of the infinite-space Hamiltonian by functions that live on the torus.
Let $\hat{\chi}_0 \in \mathcal{C}^\infty_c (-\pi/4, \pi/4)$ such that $0 \le \hat{\chi}_0 \le 1$ on $\mathbb{R}$ and
$\hat{\chi}_0 = 1$ in $[-\pi/8, \pi/8]$.
Similarly, let
$\hat{\chi}_\pi \in \mathcal{C}^\infty_c (\pi-\pi/4, \pi+\pi/4)$
such that $0 \le \hat{\chi}_\pi \le 1$ on $\mathbb{R}$ and
$\hat{\chi}_\pi = 1$ in $[\pi-\pi/8, \pi+\pi/8]$.
Define
$\psi_{\kp, j}^0 := \hat{\chi}_{t_j} \ul^{-1}_{\kp,t_j} \psi_{j}$, where $t_j = 0$ if $j \in R$ and $t_j = \pi$ if $j \in R'$.
Set $\psi_{\kp,1} := N_{\kp,1}\psi_{\kp,1}^0$,
and for all $j > 1$ define
\begin{align*}
    \psi_{\kp,j} = N_{\kp,j} (\psi_{\kp,j}^0 - \sum_{i=1}^{j-1} \aver{\psi_{\kp,i}, \psi_{\kp,j}^0} \psi_{\kp,i}),
\end{align*}
with $N_{\kp, j} > 0$ such that $\norm{\psi_{\kp,j}} = 1$.
Thus the $\psi_{\kp, j}$ form an orthonormal set of functions on $\mathbb{R}^2$.
Let $\tilde{\psi}_{\kp,j} := \mathcal{P} \psi_{\kp,j}$, which is well defined since the support of each $\psi_{\kp,j}$ is contained in an interval of length $\pi/2$.

In words, what we have done is embed the infinite-space eigenfunctions on the torus by rescaling in $\kp$ as necessary, truncating them to have compact support and orthonormalizing at the end.
By the rapid decay 
of the $\psi_{j,\xi}$ (Proposition \ref{expDecay}), the error we get from truncating and orthonormalizing goes to zero 
exponentially in $\kp$.
Namely, we have that for any 
$\alpha \in \mathbb{N}$, there exist positive constants $C$ and $r$ such that
\begin{align}\label{decayInfiniteApprox}
    \vertiii{\kp^{|\alpha|}\partial_y^\alpha (\psi_{\kp,j} - \ul^{-1}_{\kp,t_j} \psi_{j})} \le C e^{-r/\kp}
\end{align}
uniformly in $\kp \in (0, \kp_0]$ 
for all $j \in \{1,2, \dots, s\}$.

\begin{lemma}\label{lemmaOrthogonal}
There exist positive constants $C$ and $r$ such that for all $\kp \in (0, \kp_0], \ell \in \{0,1, \dots, s\}$, and 
$u_0 \in ( \s \{\tilde{\psi}_{\kp,j}\}_{j =1}^\ell)^\perp$ 
such that either $\supp (u_0) \subset (-\frac{3\pi}{4}, \frac{3\pi}{4})$ or $\supp (u_0) \subset (\frac{\pi}{4}, \frac{7\pi}{4})$, 
\begin{align} \label{eqOrthogonal}
\aver{u_0, \hat{H}_\kp^2 u_0} \ge \Big(\min\{\mu_{\ell+1}, E^2 + 1\} - C e^{-r/\kp}\Big) \norm{u_0}^2.
\end{align}
\end{lemma}
\begin{proof}
Suppose $\supp (u_0) \subset (-\frac{3\pi}{4}, \frac{3\pi}{4})$ 
for concreteness and $\norm{u_0} = 1$ without loss of generality.
Then for $u := \ul_{\kp,0}\tilde{\mathcal{P}} u_0$, we have
$\aver{u_0, \hat{H}_\kp^2 u_0} = (u, \hat{H}^2 u)$.
Define $\tilde{u} := u - \sum_{j\in \tilde{R}_{\ell}}(\psi_{j}, u) \psi_{j}$, so that $(\psi_{j}, \tilde{u}) = 0$ for all $j \in R_{\ell}$, and hence $$(\tilde{u}, \hat{H}^2 \tilde{u}) \ge \vertiii{\tilde{u}}^2\min\{ \mu_j : j \in (R \cup \{s_0 + 1\}) \cap \{\ell + 1, \ell + 2, \dots, s_0 + 1\}\} \ge \mu_{\ell+1} \vertiii{\tilde{u}}^2$$ 
by Theorem \ref{maxmin}.
Since $\vertiii{\tilde{u}}^2 = 1 - \sum_{j\in \tilde{R}_{\ell}} |(\psi_{j}, u)|^2$,
it follows that
\begin{align*}
    (u, \hat{H}^2 u) = (\tilde{u}, \hat{H}^2 \tilde{u}) + \sum_{j\in \tilde{R}_{\ell}}\mu_j |(\psi_{j}, u)|^2
    \ge \mu_{\ell+1} - \sum_{j\in \tilde{R}_{\ell}}(\mu_{\ell+1} -\mu_j) |(\psi_{j}, u)|^2.
\end{align*}
Using that $(\ul_{\kp,0} \psi_{\kp,j}, u) = 
0$, we see by \eqref{decayInfiniteApprox} that
$$|(\psi_{j}, u)| = |(\ul_{\kp,0} \psi_{\kp,j} - \psi_{j}, u)| \le \vertiii{\ul_{\kp,0} \psi_{\kp,j} - \psi_{j}} \le C e^{-r/\kp}$$ for all $j \in \tilde{R}_{\ell}$.
We conclude that
\begin{align*}
    (u, \hat{H}^2 u) \ge
    \mu_{\ell + 1} (1 - C e^{-r/\kp}) 
    \ge \min\{\mu_{\ell+1}, E^2 + 1\} - C e^{-r/\kp},
\end{align*}
and the result is complete.
\end{proof}

We will now show that the eigenvalues for the periodic problem are well approximated by those for the infinite-space problem. 
The lower bounds for $\mu_{\kp,j} (\xi)$ will be obtained using Theorem \ref{maxmin} and Lemma \ref{lemmaOrthogonal}.

\begin{proposition} \label{propEvalApprox}
There exist positive constants $C$ and $r$ such that $$\mu_{\ell} - C e^{-r/\kp} \le \mu_{\kp,\ell}\le \mu_{\ell} + C e^{-r/\kp}$$
uniformly in $\kp \in (0, \kp_0]$ 
for all $\ell \in \{1,2, \dots, s\}$.
Moreover, 
$\mu_{\kp, s + 1} \notin \supp (\Upsilon)$
for all $\kp$ sufficiently small. 
\end{proposition}

\begin{proof}
We recall that $\varphi ' (x) = \Upsilon (x^2)$. We first prove the upper bound.
Fix $\kp \in (0,\kp_0]$ 
and $\ell \in \{1, \dots, s\}$.
Let $u = \sum_{j=1}^{\ell} c_j \tilde{\psi}_{\kp, j}$ with the $c_j \in \mathbb{C}$ such that $\sum_{j=1}^{\ell} |c_j|^2 = 1$.
Then
\begin{align*}
    \aver{u, \hat{H}_\kp^2 u} =
    \sum_{j= 1}^{\ell} c_j\aver{u,\hat{H}_\kp^2\tilde{\psi}_{\kp,j}}
    =
    \sum_{j=1}^{\ell} |c_j|^2 \mu_{j} + 
    \sum_{j= 1}^{\ell} c_j\aver{u ,(\hat{H}_\kp^2-\mu_{j}) \tilde{\psi}_{\kp,j}},
\end{align*}
with the first term on the right-hand side bounded above by $\mu_{\ell}$.
To control the second term, observe that
\begin{align*}
    \norm{(\hat{H}_\kp^2-\mu_{j})\tilde{\psi}_{\kp,j}} =
    \vertiii{(A_j-\mu_{j})\ul_{\kp,t_j} \psi_{\kp,j}},
\end{align*}
with $(A_j, t_j) = (\hat{H}^2, 0)$ if $j \in R$ and $(A_j, t_j) = (\hat{H}'^2, \pi)$ if $j \in R'$. 
Using that $(A_j - \mu_j) \psi_{j} = 0$, we conclude that
\begin{align*}
    \aver{u, \hat{H}^2_\kp u} \le \mu_\ell + \sum_{j=1}^\ell \vertiii{(A_j-\mu_{j})(\ul_{\kp,t_j} \psi_{\kp,j} - \psi_{j})} \le
    \mu_\ell + C e^{-r/\kp},
\end{align*}
where the last inequality follows from \eqref{decayInfiniteApprox}.
Since $u$ was an arbitrary function in a subspace of dimension $\ell$, Theorem \ref{maxmin} implies that
$\mu_{\kp,\ell}\le \mu_{\ell}+C e^{-r/\kp}$.

We will now 
prove the significantly more challenging lower bound.
Let $T_0 = (-\frac{3 \pi}{4}, -\frac{\pi}{4}) \cup (\frac{\pi}{4}, \frac{3\pi}{4})$ and $T_1 = (-\frac{7 \pi}{8}, -\frac{\pi}{8}) \cup (\frac{\pi}{8}, \frac{7\pi}{8})$ be subsets of the torus $\mathbb{T}$, so that
$T_0\ssubset T_1\ssubset \mathbb{T}$.
Fix $\chi_0 >0$, and let $\tilde{\chi} : \mathbb{T} \rightarrow [0, \chi_0]$ be a smooth function supported in $T_1$ such that $\tilde{\chi} = \chi_0$ in $\bar{T}_0$.
Let 
\[
  B_\kp = \tilde{\chi} R_m \tilde{\chi}
\]
with $R_m$ an elliptic operator 
$1 + \kp^{2m} D^{2m}_y$.
This is an operator that is large on $T_0$ as well as non-negative.

Let 
$\chi: \mathbb{T} \rightarrow [0,1]$ be a smooth function supported in $(-\frac{3\pi}{4}, \frac{3\pi}{4})$
such that 
$\chi = 1$ in $[-\frac{\pi}{4}, \frac{\pi}{4}]$.
Note that $\supp (\chi ') \subset T_0$.
Let $u$ be an arbitrary function in $\mathcal{H}^m (\mathbb{T})$. Then,
\begin{align} \label{xTerms}
    \aver{(1-\chi) u, \hat{H}^2_\kp \chi u} = \aver{(1-\chi) \hat{H}_\kp u, \chi \hat{H}_\kp u} + \aver{(1-\chi) \hat{H}_\kp u, [\hat{H}_\kp, \chi] u} -
    \aver{[\hat{H}_\kp, \chi] u, \hat{H}_\kp \chi u}.
\end{align}
The first term on the above right-hand side is non-negative since both $\chi$ and $1-\chi$ are.
We now control the remaining terms. 
Observe that
all operators involved are differential operators, and the coefficients of $[\hat{H}_\kp, \chi]$ vanish wherever $\chi'$ does.
Thus
\begin{align*}
    \aver{(1-\chi) \hat{H}_\kp u, [\hat{H}_\kp, \chi] u} = 
    \frac{1}{\chi_0^2}
    \aver{(1-\chi) \hat{H}_\kp \tilde{\chi} u, [\hat{H}_\kp, \chi] \tilde{\chi} u}.
\end{align*}
We can similarly insert factors of $\tilde{\chi}/\chi_0$ in the third term of \eqref{xTerms}.
Since $\hat{H}_\kp$ is a differential operator of order $m$ and the commutators $[\hat{H}_\kp, \chi]$ introduce an extra factor of $\kp$,
there exists a constant $C_0 > 0$ such that
\begin{align} \label{xTerms2}
    \Re \aver{(1-\chi) u, \hat{H}^2_\kp \chi u} \ge
    - \frac{C_0 \kp}{\chi_0^2} (\norm{\kp^{m} D^{m}_y \tilde{\chi} u}^2 + \norm{\tilde{\chi} u}^2)
\end{align}
uniformly in $u \in \mathcal{H}^m (\mathbb{T})$.
Since
$
    \aver{u, B_\kp u} = \norm{\tilde{\chi} u}^2 + \norm{\kp^m D^m_y \tilde{\chi} u}^2,
$
we can choose $\chi_0$ sufficiently large so that
\begin{align} \label{xTerms3}
    \aver{u, B_\kp u} + 2 \Re \aver{(1-\chi) u, \hat{H}^2_\kp \chi u} \ge \frac{C}{\chi_0^2} (\norm{\kp^{m} D^{m}_y \tilde{\chi} u}^2 + \norm{\tilde{\chi} u}^2) \ge C \norm{u}_{T_0}^2
\end{align}
for $C > 0$ as large as necessary. Recall that $\norm{\cdot}_{T_0} := \norm{\cdot}_{L^2 (T_0) \otimes \mathbb{C}^n}$.

By Proposition \ref{expDecay}, we know that
the low-energy spectrum of $\hat{H}^2_\kp$ satisfies
\begin{equation}\label{eq:expbound}
  \max_{\psi\in \tilde S_s;\ \|\psi\|=1} \aver{B_\kp \psi,\psi} =\eps \leq e^{-r/\kp}
\end{equation}
for some $c>0$, with $\tilde S_s={\rm span}(\theta_{\kp, 1},\ldots \theta_{\kp, s})$. 
Recall from the definition of $s$ that
$\mu_{\kp, \ell}<E^2$ remains in the (infinite domain) spectral gap for all $\ell \le s$.

Let us consider the operator $\hat{H}^2_\kp+B_\kp$, also self-adjoint (this is proved as for $\hat{H}^2_\kp$), with eigenvalues $\nu_{\kp, j}$. 
For $1\leq \ell\leq s$, we have by the min-max principle,
\[
 \nu_{\kp, \ell} = \min_{S_\ell} \max_{\psi\in S_\ell,\ \|\psi\|=1} \aver{\hat{H}^2_\kp\psi,\psi} + \aver{B_\kp \psi,\psi} \leq \max_{\psi\in \tilde S_\ell,\ \|\psi\|=1} \aver{\hat{H}^2_\kp\psi,\psi} + \aver{B_\kp \psi,\psi} \leq \mu_{\kp, \ell} + \eps.
\]
We may therefore obtain a lower bound for $\nu_{\kp, \ell}$ now thanks to the regularization $B_\kp$. We have
\[
  \nu_{\kp,\ell}\geq \min_{u\in S_{\ell-1}^\perp,\ \|\psi\|=1} \aver{(\hat{H}^2_\kp+B_\kp)u,u}
\]
with now $S_{\ell-1}$ the span of $\tilde{\psi}_{\kp, j}$ for $1\leq j\leq \ell-1$ using both families of infinite-domain eigenfunctions properly embedded in the torus. 

It follows from Lemma \ref{lemmaOrthogonal} and \eqref{xTerms3} that
for $u \in S_{\ell - 1}^{\perp}$,
\begin{align*}
    \aver{u, (\hat{H}^2_\kp + B_\kp) u} &=
    \aver{\hat{H}^2_\kp \chi u,\chi u} + \aver{\hat{H}^2_\kp (1-\chi) u, (1-\chi) u} + 2 \Re \aver{\hat{H}^2_\kp\chi u,(1-\chi)u} + \aver{B_\kp u,u}\\
    &\ge
    (\mu_\ell - \eta) \|\chi u\|^2 + (\mu_\ell-\eta) \|(1-\chi)u\|^2 + C \|u\|^2_{T_0},
\end{align*}
with 
$\eta \le C e^{-r/\kp}$.
So, with $C$ large, we have
$\aver{u, (\hat{H}^2_\kp + B_\kp) u} \ge (\mu_\ell - \eta) \norm{u}^2$ and thus get
\[
  \mu_{\kp, \ell} \geq \nu_{\kp, \ell}-\eps \geq \mu_\ell-\eps -\eta \geq \mu_\ell-C e^{-r/\kp}
\]
for all $\kp$ sufficiently small. 

It remains to prove that $\mu_{\kp, s+1}$ lies above the support of $\Upsilon$ when $\kp$ is small enough.
By definition of $\Upsilon$, there exists $\delta '> 0$ such that $\Upsilon \in \mathcal{C}^\infty_c (-1, E^2 - \delta_0 - \delta')$.
Let $u \in \mathcal{H}^m (\mathbb{T})$ 
such that 
$u \in S_{s}^\perp$ and
$\norm{u} = 1$.
Take $\chi$ as above, with the additional requirement that $\supp (\chi ') \subset T_0 '$ for some $T_0 ' \subset T_0$ satisfying $\norm{u}_{T_0'}^2 \le \frac{\delta '}{E^2+1}$.
By \eqref{eq:alphaBdd}, which can easily be shown to hold in this one-dimensional setting (see \eqref{periodicEllipticity1D}),
it follows that $$\norm{[\hat{H}_\kp, \chi] (i-\hat{H}_\kp)^{-1}}\le C \kp, \qquad \kp \in (0, \kp_0].$$ 
Writing $[\hat{H}_\kp, \chi] = [\hat{H}_\kp, \chi] (i-\hat{H}_\kp)^{-1} (i-\hat{H}_\kp)$ and $\hat{H}_\kp \chi = [\hat{H}_\kp, \chi] + \chi \hat{H}_\kp$, we use \eqref{xTerms} to conclude that
\begin{align*}
    \aver{(1-\chi)u, \hat{H}_\kp^2 \chi u} \ge -C \kp \Big(\norm{\hat{H}_\kp u}^2 + \norm{u}^2\Big) = -C \kp \Big(\norm{\hat{H}_\kp u}^2 + 1\Big).
\end{align*}
Hence by Lemma \ref{lemmaOrthogonal}, we have
\begin{align*}
    \aver{u, \hat{H}_\kp^2 u} &= \aver{\chi u, \hat{H}_\kp^2 \chi u} + \aver{(1-\chi)u, \hat{H}_\kp^2 (1-\chi)u} + 2 \Re \aver{(1-\chi)u, \hat{H}_\kp^2 \chi u}\\
    &\ge (\check{\mu}_{s+1} - C e^{-r/\kp}) \norm{\chi u}^2 + (\check{\mu}_{s+1} - C e^{-r/\kp}) \norm{(1-\chi) u}^2 -C \kp \Big(\norm{\hat{H}_\kp u}^2 + 1\Big),
\end{align*}
where $\check{\mu}_{s+1} := \min\{\mu_{s +1}, E^2 + 1\}$.
It follows that $\aver{u, \hat{H}_\kp^2 u} \ge \check{\mu}_{s+1} (\norm{\chi u}^2 + \norm{(1-\chi) u}^2)-C\kp$.
Since
\begin{align*}
    1 =\norm{u}^2 = \norm{\chi u}^2 + \norm{(1-\chi) u}^2 + 2\aver{(1-\chi)u, \chi u} \le \norm{\chi u}^2 + \norm{(1-\chi) u}^2 + \frac{1}{2} \norm{u}_{T_0'}^2,
\end{align*}
we have shown that
\begin{align*}
    \aver{u, \hat{H}_\kp^2 u} \ge \check{\mu}_{s+1} - \frac{1}{2} \check{\mu}_{s+1}\norm{u}_{T_0'}^2 -C\kp \ge 
    \check{\mu}_{s+1} - \frac{1}{2} \delta ' -C\kp.
\end{align*}
Since $u$ was arbitrary, Theorem \ref{maxmin} implies that $\mu_{\kp,s+1} \ge \check{\mu}_{s+1} - \delta '$ for all $\kp$ sufficiently small.
We know that $\mu_{s+1} \ge E^2 - \delta_0$ by definition, hence
$\mu_{\kp,s+1} \ge E^2 - \delta_0 - \delta '$. 
This completes the proof.
\end{proof}

\paragraph{Eigenspace approximation.} 
For the following lemma,
fix $\kp \in (0,1]$ and 
let $\{\phi_1, \phi_2, \dots, \phi_{s}\} \subset \mathcal{H}^m (\mathbb{T})$ 
such that $\aver{\phi_i, \phi_j} = \delta_{ij}$.
For any integers $j,k$, and $\ell$ satisfying $1 \le j \le k \le \ell \le s$, define
\begin{equation}\label{eq:alphas}
\alpha_{j,k,\ell} := \Big( 1 - \sum_{i=j}^\ell |\aver{\theta_{\kp,i}, \phi_{k}}|^2\Big) ^{1/2}, \qquad
   \alpha_{k,\ell} := \alpha_{1,k,\ell},\qquad 
    r_j:= (\hat{H}^2_\kp - \mu_{\kp,j}) \phi_{j} .
\end{equation}
Note that
if we denote by $\Pi_{j,\ell}$ the orthogonal projector onto the span of $\{\theta_{\kp,j},\ldots,\theta_{\kp,\ell}\}$, then we have
\begin{align} \label{jlPi}
    \alpha_{j,k,\ell}=\|(I-\Pi_{j,\ell})\phi_k\|.
\end{align}

\begin{lemma} \label{lemmaAlpha2}
Let $j,k,\ell$ be integers satisfying $1 \le j \le k \le \ell \le s$.
Then
\begin{align} \label{eqAlpha2}
\alpha_{j,k,\ell}^2 \le \norm{r_k}^2(\mu_{\kp,\ell+1} - \mu_{\kp,k})^{-2} + \sum_{i=1}^{j-1} \norm{r_{i}}^2 (\mu_{\kp,j} - \mu_{\kp,i})^{-2}.
\end{align}
\end{lemma}
\eqref{jlPi} means 
$\phi_k$ lives approximately in 
$\s (\theta_{\kp,j}, \theta_{\kp,j+1}, \dots, \theta_{\kp,\ell})$
when the right-hand side of \eqref{eqAlpha2} is small.
\begin{proof}
For brevity, set $\theta_i := \theta_{\kp,i}$.
We see that
$
    \alpha_{j,k,\ell}^2 = \alpha_{k,\ell}^2 + \sum_{i=1}^{j-1}|\aver{\theta_{i}, \phi_k}|^2,
$
with 
\begin{align*}
    \sum_{i=1}^{j-1}|\aver{\theta_{i}, \phi_k}|^2
    &\le
    \sum_{i=1}^{j-1}\Big(1 - \sum_{h=1}^{j-1}|\aver{\theta_{i}, \phi_{h}}|^2\Big)
    =\sum_{h=1}^{j-1}\Big(1 - \sum_{i=1}^{j-1}|\aver{\theta_{i}, \phi_{h}}|^2\Big)
    = \sum_{h=1}^{j-1} \alpha_{h, j-1}^2.
\end{align*}
It remains to bound $\alpha_{k,\ell}^2$ and the $\alpha^2_{h,j-1}$.
Define $\tilde{\phi}_k := \phi_k - \sum_{i=1}^\ell \aver{\theta_i, \phi_k} \theta_i$, so that $\aver{\theta_i, \tilde{\phi}_k}=0$ for all $i \in \{1,\dots, \ell\}$, and thus
\begin{align} \label{ge}
\aver{\tilde{\phi}_k, \hat{H}^2_\kp \tilde{\phi}_k} \ge
\mu_{\kp,\ell+1} \norm{\tilde{\phi}_j}^2 = \mu_{\kp, \ell+1} \alpha_{k,\ell}^2
\end{align}
by Theorem \ref{maxmin}.
We also see that
$
    \hat{H}^2_\kp \tilde{\phi}_k = \mu_{\kp, k} \phi_k - \sum_{i=1}^\ell \mu_{\kp,i} \aver{\theta_i, \phi_k} \theta_i + r_k,
$
and hence
\begin{align} \label{le}
    \aver{\tilde{\phi}_k,\hat{H}^2_\kp \tilde{\phi}_k} &=
    \mu_{\kp, k} - \mu_{\kp, k} \sum_{i=1}^\ell |\aver{\theta_i, \phi_k}|^2 +
    \aver{\tilde{\phi}_k,r_k} 
    =
    \mu_{\kp, k} \alpha_{k,\ell}^2 + \aver{\tilde{\phi_k},r_k}
    \le \mu_{\kp, k} \alpha_{k,\ell}^2 + \norm{r_k} \alpha_{k,\ell}.
\end{align}
Combining \eqref{ge} and \eqref{le}, we obtain that $\alpha_{k,\ell} \le \norm{r_k} (\mu_{\kp,\ell+1} - \mu_{\kp,k})^{-1}$.
The same argument proves that $\alpha_{h, j-1} \le \norm{r_h} (\mu_{\kp,j} - \mu_{\kp,h})^{-1}$ for all $h \in \{1,2, \dots, j-1\}$, and the result is complete.
\end{proof}



We will apply Lemma \ref{lemmaAlpha2} to the functions $\phi_j = \tilde{\psi}_{\kp,j}$.
We write
$r_j = (\hat{H}^2_\kp - \mu_j) \tilde{\psi}_{\kp,j} + (\mu_j - \mu_{\kp,j}) \tilde{\psi}_{\kp,j}$, meaning that
$\norm{r_j} \le \norm{(\hat{H}^2_\kp - \mu_j) \tilde{\psi}_{\kp,j}} + |\mu_j - \mu_{\kp,j}| \le C e^{-r/\kp}$, where the last inequality follows from Proposition \ref{propEvalApprox} and its proof.
Thus if $\min \{\mu_{\kp,\ell+1} - \mu_{\kp,k}, \mu_{\kp,j}- \mu_{\kp,j-1} \} \ge \delta$ for some $\delta > 0$, then $\alpha^2_{j,k,\ell} \le C \delta^{-2} e^{-r/\kp}$.
In addition, we have
\begin{align} \label{betaBd}
\begin{split}
    \beta^2_{j,k,\ell} :&= \norm{\theta_{\kp,k} - \sum_{i=j}^\ell \aver{\tilde{\psi}_{\kp,i}, \theta_{\kp,k}} \tilde{\psi}_{\kp,i}}^2 =
    1 - \sum_{i=j}^\ell |\aver{\tilde{\psi}_{\kp,i}, \theta_{\kp,k}}|^2 \le
    \sum_{i' = j}^\ell (1 - \sum_{i=j}^\ell |\aver{\tilde{\psi}_{\kp,i}, \theta_{\kp,i'}}|^2)\\
    &=
    \sum_{i = j}^\ell (1 - \sum_{i'=j}^\ell |\aver{\tilde{\psi}_{\kp,i}, \theta_{\kp,i'}}|^2) = \sum_{i = j}^\ell \alpha^2_{j,i,\ell} \le C (\ell + 1 - j) \delta^{-2} e^{-r/\kp}.
    \end{split}
\end{align}

\begin{lemma}\label{switchBasis}
Let $A_\kp := \sum_{i=0}^{m-1} \cg_{\kp,i} (y) \kp^i D^i_y$ be an operator on $\mathcal{H}(\mathbb{T})$, for some smooth coefficients $\cg_{\kp,i}$ satisfying $\sum_{i=0}^{m-1} \norm{\kp^k D^k_y \cg_{\kp,i}}_{L^\infty (\mathbb{T})} \le C'$ uniformly in $\kp$ for all $k \in \{0,1, \dots, m\}$.
Let $j \le \ell$ be positive integers such that $\ell \le s$ and define
$\delta := \min\{ \mu_{\kp,\ell+1} - \mu_{\kp,\ell}, \mu_{\kp,j} - \mu_{\kp,j-1} \}$, 
where the second argument is ignored if $j=1$.
Then 
there exist positive constants $C$ and $r$ such that
\begin{align*}
    \Big |\sum_{i=j}^\ell \Big (\aver{\theta_{\kp,i}, A_\kp \theta_{\kp,i}} - \aver{\tilde{\psi}_{\kp,i}, A_\kp \tilde{\psi}_{\kp,i}} \Big) \Big| \le
    C \delta^{-1} e^{-r/\kp}.
\end{align*}
\end{lemma}

\begin{proof}
First, we observe that for all
$i \in \{1, \dots, s\}$ and $k \in \{0,1, \dots, m\}$,
\begin{align} \label{derivativesBdd}
    \norm{\kp^{k} D^{k}_y \nu_{\kp, i}}^2 \le \norm{\kp^m D^m_y \nu_{\kp,i}}^{\frac{2k}{m}} \le  C_0 \Big(\norm{\hat{H}_\kp \nu_{\kp,i}}^2 + C_1 \Big)^{\frac{k}{m}} \le
    C_0 (E^2 + C_1)^{\frac{k}{m}}, \qquad \nu = \theta, \tilde{\psi}
\end{align}
uniformly in $\kp$,
with the second inequality following from ellipticity of $\hat{H}_\kp$.
We write
\begin{align*}
\sum_{i=j}^\ell \Big(
    \aver{\theta_{\kp,i}, A_\kp \theta_{\kp,i}} &- \aver{\tilde{\psi}_{\kp,i}, A_\kp \tilde{\psi}_{\kp,i}}\Big) =\\
    &\sum_{i=j}^\ell \Big(\aver{\theta_{\kp,i} - \sum_{i' = j}^\ell \aver{\tilde{\psi}_{\kp,i'},\theta_{\kp,i}} \tilde{\psi}_{\kp,i'}, A_\kp \theta_{\kp,i}} - \aver{\tilde{\psi}_{\kp,i}, A_\kp (\tilde{\psi}_{\kp,i} - \sum_{i' = j}^\ell \aver{\theta_{\kp,i'}, \tilde{\psi}_{\kp,i}} \theta_{\kp,i'})}\Big),
\end{align*}
which 
implies that
\begin{align*}
    \Big |\sum_{i=j}^\ell \Big(\aver{\theta_{\kp,i}, A_\kp \theta_{\kp,i}} - \aver{\tilde{\psi}_{\kp,i}, A_\kp \tilde{\psi}_{\kp,i}}\Big) \Big| \le
  \sum_{i=j}^\ell \Big(
  \beta_{j,i,\ell} \norm{A_\kp \theta_{\kp,i}} + \alpha_{j,i,\ell} \norm{A_\kp^* \tilde{\psi}_{\kp,i}} \Big),
\end{align*}
where $A_\kp^* := \sum_{i=0}^{m-1} \kp^i D^i_y \cg^*_{\kp,i} (y)$ is the formal adjoint of $A_\kp$.
By \eqref{derivativesBdd}, we know that the norms of $A_\kp \nu_{\kp,i}$ and $A^*_\kp \nu_{\kp,i}$ are bounded uniformly in $\kp$, for $\nu = \theta, \tilde{\psi}$.
Using 
\eqref{betaBd} and the corresponding bound for $\alpha_{j,i,\ell}^2$, we conclude that
$
    \Big |\sum_{i=j}^\ell \Big(\aver{\theta_{\kp,i}, A_\kp \theta_{\kp,i}} - \aver{\tilde{\psi}_{\kp,i}, A_\kp \tilde{\psi}_{\kp,i}} \Big)\Big|
    \le C \delta^{-1} e^{-r/\kp}.
$ 
This completes the result.
\end{proof}

\paragraph{Edge current approximation.} 
We are now ready to prove Theorem \ref{periodicApprox}. We will write $\tilde{\sigma}_I$ as a sum of inner products over $\xi \in \kp \mathbb{Z}$, thus it is helpful to reintroduce the $\xi$-dependent notation.

\begin{proof}[Proof of Theorem \ref{periodicApprox}]
Since the coefficients of $H_\kp$ are independent of $x$, we have
\begin{align}\label{eq:sigmaI_sum}
    2\pi\tilde{\sigma}_I 
=
2\pi\Tr i Q [H_\kp, P] \Upsilon (H^2_\kp)
=
\sum_{\xi \in \kp \mathbb{Z}} \sum_{j\in \mathbb{N}} \aver{\theta_{\kp,j,\xi}, i\kp Q_Y G_\kp (\xi) \Upsilon(\hat{H}_\kp^2 (\xi))\theta_{\kp,j,\xi}},
\end{align}
where $G_\kp (\xi) := m M_m \xi^{m-1} + \sum_{j=1}^m j \kp^{m-j} \cp_{\kp, j} (y)\xi^{j-1}D_y^{m-j} + \sum_{i+j \le m-1} i \kp^j \cp_{\kp,ij} (y) \xi^{i-1} D^j_y$.
To obtain $G_\kp$, we used that the contributions from $Q_X [H_\kp ,P]$ of all second and higher order derivatives of $P$ must vanish, as $\int_\mathbb{T} Q_X P^{(q)} (x) dx = P^{(q-1)} (\pi/2) - P^{(q-1)} (-\pi/2) = 0$ for all $q > 1$.
Using that the $\theta_{\kp,j,\xi}$ are eigenfunctions of $\hat{H}^2_\kp (\xi)$, we obtain that
$$2\pi\tilde{\sigma}_I = \sum_{\xi \in \kp \mathbb{Z} \cap I} \sum_{j=1}^{s_0} \aver{\theta_{\kp,j,\xi}, i\kp Q_Y G_\kp (\xi) \theta_{\kp,j,\xi}} \Upsilon(\mu_{\kp,j}(\xi)).$$


The next step is to partition the eigenvalues into clusters so that
eigenvalues that are nearby belong to the same cluster and
the separation between clusters is controlled (bounded from below).
We will choose the clusters with diameter sufficiently small so that we can approximate all eigenvalues in a given cluster by the smallest eigenvalue in the cluster, with negligible error.
Once $\Upsilon$ is constant over each cluster, we will 
approximate $\tilde{\sigma}_I$ by a corresponding sum involving the $\tilde{\psi}_{\kp,j}$ and apply
Lemma \ref{switchBasis} to control the error.

By Proposition \ref{propEvalApprox}, the number of eigenvalues of $\hat{H}^2_\kp (\xi)$ in $\supp (\Upsilon)$ is bounded by $s(\xi)$
uniformly in
$\kp > 0$ sufficiently small and $\xi \in \kp \mathbb{Z}$.
Recall that $s(\xi)$ is the total number of eigenvalues of $\hat{H}^2 (\xi)$ and $\hat{H}'^2 (\xi)$ lying in $[0, E^2 - \delta_0)$, which is itself bounded by $s_0$ uniformly in $\xi$ and vanishes whenever $\xi \notin I$, with $I \subset \mathbb{R}$ a compact interval; see Proposition \ref{acProp}.
Let $\xi \in \kp \mathbb{Z} \cap I$ 
and 
\tcbn{$\delta := \delta (\kp)$}.
Define $k_0 (\xi) := 1$, and for all $i \ge 1$ 
define $k_i (\xi) := \inf \{ \ell > k_{i-1}(\xi) : \mu_{\kp, \ell}(\xi) - \mu_{\kp, \ell - 1}(\xi) \ge \delta \quad \text{or} \quad \ell > s(\xi) \}$.
We see that the $k_i (\xi)$ form an increasing sequence of integers, with $k_{J(\xi)} (\xi) = s(\xi) + 1$ for some $J(\xi) \in \mathbb{N}$.
Define
$S_{\xi,j} := \{k_j, \dots, k_{j+1} - 1\}$ for all $j \in \{0,1,\dots, J(\xi)-1\}$,
so that
$S_{\xi,0}, S_{\xi,1}, \dots , S_{\xi, J(\xi)-1}$ forms a partition of $\{1,2,\dots, s (\xi)\}$.
Thus we have
\begin{align*}
    2\pi\tilde{\sigma}_I =\sum_{\xi \in \kp \mathbb{Z} \cap I} \sum_{i=0}^{J(\xi) - 1} \sum_{j \in S_{\xi,i}} \aver{\theta_{\kp,j,\xi}, i\kp Q_Y G_\kp (\xi) \theta_{\kp,j,\xi}} \Upsilon(\mu_{\kp,j}(\xi)).
\end{align*}
By \eqref{derivativesBdd}, we see that
$|\aver{\theta_{\kp,j,\xi}, G_\kp (\xi) \theta_{\kp,j,\xi}}| \le C$ uniformly in $\kp$ and $\xi$. 
Defining
\begin{align*}
    2\pi\tilde{\sigma}_{I,1} :=\sum_{\xi \in \kp \mathbb{Z} \cap I} \sum_{i=0}^{J(\xi) - 1} \sum_{j \in S_{\xi,i}} \aver{\theta_{\kp,j,\xi}, i\kp Q_Y G_\kp (\xi) \theta_{\kp,j,\xi}} \Upsilon(\mu_{\kp,k_i}(\xi)),
\end{align*}
which replaces each eigenvalue by the smallest eigenvalue in its cluster,
we thus obtain that
\begin{align} \label{LipBd}
    |\tilde{\sigma}_{I,1} - \tilde{\sigma}_{I}| \le s_0 (|I|+1) C
    \sup_{|x-y| < \delta} |\Upsilon (x) - \Upsilon (y)| \le 
    C \delta 
\end{align}
by regularity of $\Upsilon$, with $|I|$ the length of the interval $I$. 
We will now control the error from replacing the periodic eigenfunctions $\theta_{\kp,j,\xi}$ in $\tilde{\sigma}_{I,1}$ by the truncations 
$\tilde{\psi}_{\kp,j,\xi}$.
Namely, it follows from 
Lemma \ref{switchBasis} that
\begin{align*}
    2\pi\tilde{\sigma}_{I,2} :=\sum_{\xi \in \kp \mathbb{Z} \cap I} \sum_{i=0}^{J(\xi) - 1} \sum_{j \in S_{\xi,i}} \aver{\tilde{\psi}_{\kp,j,\xi}, i\kp Q_Y G_\kp (\xi) \tilde{\psi}_{\kp,j,\xi}} \Upsilon(\mu_{\kp,k_i}(\xi))
\end{align*}
satisfies
$
    |\tilde{\sigma}_{I,2} - \tilde{\sigma}_{I,1}| \le C \delta^{-1} e^{-r/\kp}.
$
Here, we have used the extra factor of $\kp$ that appears in the inner product to cancel the $\kp^{-1}$ scaling of the number of terms in the sum over $\xi$.
We can simplify $\tilde{\sigma}_{I,2}$ as
\begin{align*}
    2\pi\tilde{\sigma}_{I,2}= \sum_{\xi \in \kp \mathbb{Z} \cap I} \sum_{i=0}^{J(\xi) - 1} \sum_{j \in S'_{\xi,i}} \aver{\tilde{\psi}_{\kp,j,\xi}, i\kp G_\kp (\xi) \tilde{\psi}_{\kp,j,\xi}} \Upsilon (\mu_{\kp,k_i}(\xi)),
\end{align*}
with $S'_{\xi, i}$ containing only the indices $j$ in $S_{\xi, i}$ such that $\supp (\tilde{\psi}_{j,\xi,\kp}) \cap \supp Q_Y \ne \emptyset$.
That is, $S'_{\xi, i} \subseteq S_{\xi, i}$ and $\bigcup_{i=0}^{J(\xi)-1} S'_{\xi, i} = R_\xi$.
Here, we used the fact that 
$Q_Y = 1$ on $\supp (\tilde{\psi}_{\kp,j,\xi})$ for all $j \in S'_{\xi,i}$.

Note that for all $\xi \in \kp \mathbb{Z} \cap I$, $i \in \{1, \dots, J(\xi)-1\}$ and $j \in S'_{\xi, i}$,
we have
$$|\mu_j (\xi) - \mu_{\kp, k_i} (\xi)| 
\le|\mu_j (\xi) - \mu_{\kp, j} (\xi)| + |\mu_{\kp,j} (\xi) - \mu_{\kp, k_i} (\xi)| \le C (e^{-r/\kp} + \delta)$$
by Proposition \ref{propEvalApprox}.
Thus, by the same logic used to justify 
\eqref{LipBd}, 
we know that
\begin{align*}
    2\pi\tilde{\sigma}_{I,3} := \sum_{\xi \in \kp \mathbb{Z} \cap I} \sum_{i=0}^{J(\xi) - 1} \sum_{j \in S'_{\xi,i}} \aver{\tilde{\psi}_{\kp,j,\xi}, i\kp G_\kp (\xi) \tilde{\psi}_{\kp,j,\xi}} \Upsilon (\mu_{j}(\xi))
\end{align*}
satisfies $|\tilde{\sigma}_{I,3} - \tilde{\sigma}_{I,2}| \le C (e^{-r/\kp} + \delta)$.

We now express $\tilde{\sigma}_{I,3}$ in terms of functions in 
$\mathcal{H}(\mathbb{R})$ as
\begin{align*}
    2\pi\tilde{\sigma}_{I,3} = \kp \sum_{\xi \in \kp \mathbb{Z} \cap I} \sum_{i=0}^{J(\xi) - 1} \sum_{j \in S'_{\xi,i}} (\ul_{\kp,0} \psi_{\kp,j,\xi}, \hat{G}(\xi) \ul_{\kp,0} \psi_{\kp,j,\xi}) \Upsilon(\mu_{j}(\xi)),
\end{align*}
where $\hat{G} (\xi) = \Op (\tau_\xi)$ with $\tau_\xi (y, \zeta) := \partial_\xi \sym(y,\xi,\zeta)$ Hermitian-valued.
Defining
\begin{align*}
    2\pi\tilde{\sigma}_{I,4} :=\kp \sum_{\xi \in \kp \mathbb{Z} \cap I} \sum_{i=0}^{J(\xi) - 1} \sum_{j \in S'_{\xi,i}} (\psi_{j,\xi}, \hat{G} (\xi) \psi_{j,\xi})\Upsilon(\mu_{j}(\xi)),
\end{align*}
it follows that $2\pi(\tilde{\sigma}_{I,4} - \tilde{\sigma}_{I,3})$ is equal to
\begin{align*}
    &\kp \sum_{\xi \in \kp \mathbb{Z} \cap I} \sum_{i=0}^{J(\xi) - 1} \sum_{j \in S'_{\xi,i}}\Big( 
    (\psi_{j,\xi} - \ul_{\kp,0} \psi_{\kp,j,\xi}, \hat{G} (\xi) \psi_{j,\xi})
    +
    (\ul_{\kp,0} \psi_{\kp,j,\xi}, \hat{G} (\xi) (\psi_{j,\xi} - \ul_{\kp,0} \psi_{\kp,j,\xi}))
    \Big)
    \Upsilon(\mu_{j}(\xi)),
\end{align*}
so that $2\pi|\tilde{\sigma}_{I,4} - \tilde{\sigma}_{I,3}|$ is bounded above by
\begin{align*}
    &\kp \sum_{\xi \in \kp \mathbb{Z} \cap I} \sum_{i=0}^{J(\xi) - 1} \sum_{j \in S'_{\xi,i}}\Big( 
    \vertiii{\psi_{j,\xi} - \ul_{\kp,0} \psi_{\kp,j,\xi}} \vertiii{\hat{G} (\xi) \psi_{j,\xi}} 
    +
    \vertiii{\hat{G} (\xi) \ul_{\kp,0} \psi_{\kp,j,\xi}} \vertiii{\psi_{j,\xi} - \ul_{\kp,0} \psi_{\kp,j,\xi}} 
    \Big)
    \norm{\Upsilon}_\infty. 
\end{align*}
Note that $\hat{G} (\xi) \in \Op (S^m)$ and $(i-\hat{H} (\xi))^{-1} \in \Op (S^{-m})$ by \hone, 
hence $\hat{G} (\xi) (i-\hat{H} (\xi))^{-1}$ is bounded for all $\xi \in I$.
The norm is continuous as a function of $\xi$ in $I$ a compact interval, meaning that $\norm{\hat{G} (\xi) (i-\hat{H} (\xi))^{-1}}$ is bounded uniformly in $\xi \in I$.
Writing 
\begin{align*}
    \hat{G} (\xi) \psi_{j,\xi} = \hat{G} (\xi) (1 + \hat{H}^2(\xi))^{-1}(1 + \hat{H}^2 (\xi))\psi_{j,\xi}
\end{align*}
with $\vertiii{(1 + \hat{H}^2 (\xi))\psi_{j,\xi}} \le 1 + E^2$ for all $j \le s(\xi)$,
it follows from \eqref{decayInfiniteApprox} that
$|\tilde{\sigma}_{I,4} - \tilde{\sigma}_{I,3}| \le C e^{-r/\kp}$.

So far, we have shown that $|\tilde{\sigma}_{I,4} - \tilde{\sigma}_I| \le C (\delta^{-1} e^{-r/\kp} + e^{-r/\kp} + \delta)$, for some positive constants $C$ and $r$. Therefore, setting $\delta := e^{-r/2\kp}$, we obtain that $|\tilde{\sigma}_{I,4} - \tilde{\sigma}_I| \le C e^{-r/2\kp}$.

Finally,
the translation invariance of $H$ in the $x$-direction implies that
\begin{align}\label{eq:SFsigmaI}
    2\pi\sigma_I = 
    \Tr \int_{\mathbb{R}} \hat{G} (\xi) \Upsilon(\hat{H}^2 (\xi)) d\xi =
    \int_\mathbb{R} \sum_{i=0}^{J(\xi) - 1} \sum_{j \in S'_{\xi,i}} (\psi_{j,\xi}, \hat{G} (\xi) \psi_{j,\xi}) \Upsilon(\mu_{j}(\xi)) d\xi,
\end{align}
with the above integrand
a smooth, compactly supported function of $\xi$ by Proposition \ref{acProp}.
It is well known (see, e.g. \cite{KT}) that such integrals
can be approximated by sampling over a uniform grid of size $\kp$, with convergence faster than any power of $\kp$.
We conclude that
$|\tilde{\sigma}_{I,4} - \sigma_I| \le C \kp^p$, 
and the proof is complete.
\end{proof}

\tcbn{Observe that $|\tilde{\sigma}_{I,4} - \sigma_I|$ is the only error term that does not converge exponentially in $\kp$.
Thus if it were not for the approximation of the integral over $\xi$ by a discrete sum, we would get exponential convergence of the periodic edge current to its corresponding infinite-space value.}

\subsection{Proof of Theorem \ref{stabilityPeriodic}}
The aim of this section is to prove stability of $\tilde{\sigma}_I (H_\kp)$ under perturbations in the limit $\kp \rightarrow 0$, with $H_\kp$ given by \eqref{Hlambda2D} and $\tilde{\sigma}_I$ by \eqref{eq:sigmaIQ}.
For $\eps \in \{0,1\}$, define $H_{\kp,\eps} := H_\kp + \eps V_\kp$.\footnote{Not to be confused with the operators $H_{\kp,\mu}$ from section \ref{sec:periodic}.}
The arguments from section \ref{sec:periodic} can easily be adapted to prove that the $H_{\kp,\eps}$ are self-adjoint for all 
$\kp$, with
\begin{align} \label{muRes}
\norm{H_{\kp,\eps} f}^2 \ge C_1 \norm{f}^2_{\kp,m} - C_2 \norm{f}^2, \qquad \qquad
    \norm{\kp^{|\alpha|} D^\alpha (i-H_{\kp,\eps})^{-1}} \le C
\end{align}
for all $f \in \mathcal{H}^m (\mathbb{T}^2)$ and $|\alpha| \le m$
uniformly in $\kp \in (0,\kp_0]$. 

We first show that $\tilde{\sigma}_I$ is unchanged for $Q$ replaced by $Q-1$, using $1=Q+1-Q$ and the following result:
\begin{lemma} \label{topologicallytrivial}
Let $H$ be a self-adjoint linear operator and $\Phi \in \mathcal{C}^\infty_c (\mathbb{R})$ such that $[H,P] \Phi (H)$, $H \Phi (H)$, and $\Phi (H)$ are trace-class.
Then $\Tr i [H,P] \Phi (H) = 0$.
\end{lemma}

\begin{proof}
Since $P$ is bounded, $P H \Phi (H)$ is trace class, and thus so is
$H P \Phi (H) = [H,P] \Phi (H) + P H \Phi (H)$.
Using Lemma \ref{Psi} (which still applies in the periodic setting) and cyclicity of the trace, we have
\begin{align*}
    \Tr i [H,P] \Phi (H) =
    \Tr i [\Psi(H),P] \Phi (H) &=
    \Tr i \Psi (H) P \Phi (H) - \Tr i P \Psi (H) \Phi (H)\\
    &=
    \Tr i P \Phi (H) \Psi (H) - \Tr i P \Psi (H) \Phi (H) = 0,
\end{align*}
as desired.
\end{proof}

Clearly, and thankfully, the assumptions of Lemma \ref{topologicallytrivial} were not satisfied in the infinite-space setting. However, due to compactness of the torus, these assumptions hold in the periodic setting for the Hamiltonians we consider
as we show below.
The filter $Q$ that appears in $\tilde{\sigma}_I$ thus allows for non-vanishing edge currents.

We now bound the trace-norm of appropriate functionals of the $H_{\kp,\eps}$.
\begin{proposition}\label{prop:1normBd}
For all $\Phi \in \mathcal{C}^\infty_c (\mathbb{R})$ and $\eps \in \{0,1\}$, there is a positive constant $C$ such that $\norm{\Phi (H_{\kp,\eps})}_1 \le C \kp^{-2}$ uniformly in $\kp \in (0,\kp_0]$.
\end{proposition}
\begin{proof}
Let $M>0$ such that $\Phi \in \mathcal{C}^\infty_c (-M,M)$.
For $(\xi, \zeta) \in \mathbb{Z}^2$ and $j \in \{1,2, \dots, n\}$, define $\phi_{\xi,\zeta,j}(x,y) = \frac{1}{2\pi}e^{i(\xi x+\zeta y)} v_j$, where $v_j$ is the $j$th column of the $n\times n$ identity matrix. Thus the $\phi_{\xi,\zeta,j}$ form an orthonormal basis for $\mathcal{H} (\mathbb{T}^2)$.
Let $k \in \mathbb{N}$ and 
define 
$S_k := (\mathbb{Z} \cap [-k+1, k-1])^2 \times (\mathbb{N} \cap [1,n])$ and $T_k := (\mathbb{Z}^2 \times (\mathbb{N} \cap [1,n])) \setminus S_k$.
Let
$u = \sum_{(\xi,\zeta,j) \in T_{k}} c_{\xi,\zeta,j} \phi_{\xi,\zeta,j}$,
with $(c_{\xi,\zeta,j}) \subset \mathbb{C}$ chosen so that $u \in \mathcal{H}^m (\mathbb{T}^2)$.
It follows that
\begin{align*}
    \norm{u}_{\kp,m}^2 = \sum_{|\alpha| \le m} \: \sum_{(\xi,\zeta,j) \in T_{k}} |c_{\xi,\zeta,j}|^2 (\kp \xi)^{2\alpha_1} (\kp \zeta)^{2\alpha_2} \ge 
    (\kp k)^{2m}\sum_{(\xi,\zeta,j) \in T_{k}} |c_{\xi,\zeta,j}|^2= (\kp k)^{2m} \norm{u}^2.
\end{align*}
To justify the inequality, we 
sum over $\alpha \in \{(m,0), (0,m)\}$ and
use the fact that  for all $(\xi, \zeta, j) \in T_k$, $\xi \ge k$ or $\zeta \ge k$.
Thus if $(\kp k)^{2m} \ge C_1^{-1} (M^2 + C_2)$, then \eqref{muRes} implies that $\norm{H_{\kp,\eps} u}^2 \ge M^2 \norm{u}^2$.
Since $u$ is an arbitrary function in $\s (\phi_{\xi,\zeta,j}: (\xi,\zeta,j) \in S_k)^\perp$,
Theorem \ref{maxmin} implies that the spectrum of $H_{\kp,\eps}^2$ in $(-\infty,M^2)$ consists entirely of eigenvalues, and the number of these eigenvalues is bounded by $C \kp^{-2}$ for some $C>0$. Hence the number of eigenvalues of $H_{\kp,\eps}$ in $(-M,M)$ is also bounded by $C \kp^{-2}$, and the result follows.
\end{proof}

We are now ready to prove the main result of this section. Below we bound the difference of edge currents $|\tilde{\sigma}_I (H_{\kp,1}) - \tilde{\sigma}_I (H_{\kp,0})|$ by the product of a trace-norm and operator-norm.
Proposition \ref{prop:1normBd} 
provides a ($\kp-$dependent) bound on the involved trace-norm.
Combined with \eqref{muRes}, this bound verifies the assumptions of Lemma \ref{topologicallytrivial}.
We conclude by showing that the operator norm goes to zero faster than any power of $\kp$.

\begin{proof}[Proof of Theorem \ref{stabilityPeriodic}]
Let $W = W(x,y)$ be a smooth point-wise multiplication operator proportional to the identity matrix.
By \eqref{muRes}, we have
\begin{align} \label{adBound10}
    \norm{[H_{\kp,\eps'}, W] (z-H_{\kp,\eps})^{-1}}
     &=\norm{[V_\kp, W](i-H_{\kp,\eps})^{-1}(i-H_{\kp,\eps}) (z-H_{\kp,\eps})^{-1}} \nonumber\\
     &\le \norm{[V_\kp, W](i-H_{\kp,\eps})^{-1}} \big(1 + |i-z| \norm{(z-H_{\kp,\eps})^{-1}}\big)\le C |\Im z|^{-1} \kp
\end{align}
for all $(\eps, \eps') \in \{0,1\}^2$
uniformly in $\kp \in (0,\kp_0]$ 
and $z \in Z \setminus \{\Im z = 0\}$, where $Z := [E_1,E_2] \times [-2,2]\subset \mathbb{C}$.
Similarly,
\begin{align} \label{VBound10}
    \norm{V_\kp (z-H_{\kp, \eps})^{-1}} \le C |\Im z|^{-1}
\end{align}
uniformly in $\kp \in (0,\kp_0]$ and $z \in Z\setminus \{\Im z = 0\}$.
Let $\Phi \in \mathcal{C}^\infty_c (E_1, E_2)$ such that $\Phi \varphi' = \varphi '$.
Applying Proposition \ref{prop:1normBd} to the compactly supported function $x \mapsto (i-x) \Phi (x)$, we obtain by \eqref{adBound10} that
\begin{align} \label{1normBd0}
\norm{[H_{\kp,\eps}, W] \Phi (H_{\kp,\eps})}_1 \le \norm{[H_{\kp,\eps}, W](i-H_{\kp,\eps})^{-1}} \norm{(i-H_{\kp,\eps})\Phi (H_{\kp,\eps})}_1 \le C\kp^{-1}
\end{align}
uniformly in $\kp$.

Fix $N \in \mathbb{N}$ and define $\tilde{Q} := Q-1$.
With Proposition \ref{prop:1normBd} and \eqref{1normBd0}, we have
verified that $H_{\kp, \eps}$ satisfies the assumptions of Lemma \ref{topologicallytrivial} for $\eps \in \{0,1\}$, hence
$
    \tilde{\sigma}_I (H_{\kp, \eps}) =
    \Tr i \tilde{Q} [H_{\kp, \eps},P] \varphi ' (H_{\kp, \eps}).
$
By assumption,
$\tilde{Q} [V_\kp, P] = 0$, so that $\tilde{\sigma}_I (H_{\kp, \eps}) = \Tr i \tilde{Q} [H_{\kp, 0}, P] \varphi ' (H_{\kp, \eps})$.
Moreover, there exists a ($\kp$-independent) smooth point-wise multiplication operator $R$ proportional to the identity matrix, such that $R = 1$ on $\supp (\tilde{Q} P')$ 
and $\supp (R) \cap S_1 = \emptyset$. 
It follows that
\begin{align*}
    \tilde{\sigma}_I (H_{\kp, 0}) -\tilde{\sigma}_I (H_{\kp, 1}) &=  \Tr i R\tilde{Q} [H_{\kp, 0},P] (\varphi ' (H_{\kp, 0}) -\varphi ' (H_{\kp, 1}))\\
    &=
    \Tr i \tilde{Q} [H_{\kp, 0},P] \Phi (H_{\kp,0})(\varphi ' (H_{\kp, 0}) -\varphi ' (H_{\kp, 1})) R\\
    &\qquad +
    \Tr i \tilde{Q} [H_{\kp, 0},P] R (\Phi (H_{\kp,0}) -\Phi (H_{\kp, 1}))\varphi ' (H_{\kp,1})
     =: \Delta_0 + \Delta_1,
\end{align*}
where we have used 
cyclicity of the trace to move $R$ to the right-most position on the second line.
By the Helffer-Sj\"ostrand formula, we see that
\begin{align*}
    (\varphi ' (H_{\kp, 0}) -\varphi ' (H_{\kp, 1})) R =
    \frac{1}{\pi} \int_Z \bar{\partial} \tilde{\varphi '}(z)
    (z - H_{\kp, 1})^{-1} V_\kp (z-H_{\kp,0})^{-1} R d^2 z.
\end{align*}
Since $\supp (R)$ and $S_1$ are closed and disjoint sets, there exists a ($\kp$-independent)
collection $\{W_0, W_1, \dots, W_{N}\}$
of smooth, point-wise multiplication operators proportional to the identity matrix such that $W_j V_\kp = 0$ and $W_{j+1} W_j = W_j$ for all $j$ and $\kp$,
with $W_0 := R$.
Using the fact that 
$[(z-H_{\kp, 0})^{-1}, W_j] =
(z-H_{\kp, 0})^{-1} [H_{\kp, 0}, W_j] (z-H_{\kp, 0})^{-1}$, we obtain that
\begin{align*}
    V_\kp (z-H_{\kp,0})^{-1} R &=
    V_\kp [(z-H_{\kp,0})^{-1},W_0]
    =
    V_\kp (z-H_{\kp, 0})^{-1} [H_{\kp, 0}, W_0] (z-H_{\kp, 0})^{-1}\\
    &=
    V_\kp (z-H_{\kp, 0})^{-1} W_1 [H_{\kp, 0}, W_0] (z-H_{\kp, 0})^{-1}\\
    &=
    V_\kp (z-H_{\kp, 0})^{-1} [H_{\kp,0},W_1] (z-H_{\kp,0})^{-1} [H_{\kp, 0}, W_0] (z-H_{\kp, 0})^{-1}\\
    &=
    V_\kp (z-H_{\kp, 0})^{-1} W_2 [H_{\kp,0},W_1] (z-H_{\kp,0})^{-1} [H_{\kp, 0}, W_0] (z-H_{\kp, 0})^{-1}
    \quad = \ \  \dots \\
    &=
    V_\kp (z-H_{\kp, 0})^{-1} [H_{\kp,0},W_{N}] (z-H_{\kp, 0})^{-1}[H_{\kp,0},W_{N-1}] (z-H_{\kp, 0})^{-1}\\
    &\qquad \qquad \qquad \dots (z-H_{\kp,0})^{-1} [H_{\kp, 0}, W_0] (z-H_{\kp, 0})^{-1}.
\end{align*}
By \eqref{adBound10}, $\norm{[H_{\kp,0}, W_j] (z-H_{\kp,0})^{-1}} \le C_j |\Im z|^{-1} \kp$ for every $j$.
Thus by \eqref{VBound10} and
the rapid decay of $\bar{\partial} \tilde{\varphi '}$ near the real axis,
it follows that
$\norm{(\varphi ' (H_{\kp, 0}) -\varphi ' (H_{\kp, 1})) R} \le C \kp^{N+1}$.
Using \eqref{1normBd0}, we conclude that
\begin{align*}
    |\Delta_0| \le
    \norm{\tilde{Q} [H_{\kp, 0},P] \Phi (H_{\kp,0})}_1
    \norm{(\varphi ' (H_{\kp, 0}) -\varphi ' (H_{\kp, 1})) R}
    \le C \kp^N.
\end{align*}
By cyclicity of the trace,
$\Delta_1 = \Tr i\varphi ' (H_{\kp,1}) \tilde{Q} [H_{\kp, 0},P] R (\Phi (H_{\kp,0}) -\Phi (H_{\kp, 1})),$
so that 
the above argument can be repeated to obtain the same bound for $|\Delta_1|$.
We conclude that
$|\tilde{\sigma}_I (H_{\kp, 0}) -\tilde{\sigma}_I (H_{\kp, 1})| \le C \kp^{N}$, and the proof is complete.
\end{proof}

\section*{Acknowledgment}
This research was partially supported by the National Science Foundation, Grants DMS-1908736 and EFMA-1641100.

\appendix
\numberwithin{equation}{section}

\section{Pseudo-differential and semiclassical calculus} \label{sectionPreliminaries} 
\paragraph{Notation and functional setting.} 
Given a bounded linear operator $A: \mathcal{H} \rightarrow \mathcal{H}$ for $\mathcal{H}$ a Hilbert space, we denote by $A^*$ its adjoint and $\|A\|$ its operator norm. If, in addition, $A^* A$ is compact, then  by the spectral theorem, $A^* A$ admits a countable collection of eigenvalues $\{\lambda_j\} \subset [0, \infty)$ converging to $0$. The operator $A$ is Hilbert-Schmidt if $\norm{A}_2:=\sum_j \lambda_j < \infty$ and trace-class if $\norm{A}_1 := \sum_j \sqrt{\lambda_j} < \infty$. If $A$ is trace-class, we define the trace of $A$ by 
\begin{align*}
    \Tr A := \sum_{j\in \mathbb{N}} (\psi_j, A \psi_j),
\end{align*}
where $\{\psi_j\}_{j \in \mathbb{N}}$ is any Hilbert basis of $\mathcal{H}$ (the trace is independent of the chosen Hilbert basis).

\medskip
\noindent{\bf Weyl Quantization}. See \cite[Chapter 7]{DS}.
Let $\mathcal{S} (\mathbb{R}^d) \otimes \mathbb{C}^n$ be the Schwartz space of vector-valued functions and 
$\mathcal{S}' (\mathbb{R}^d) \otimes \mathbb{C}^n$ its dual.
Let $\mathbb{M}_n$ denote the space of Hermitian $n \times n$ matrices.
Given a parameter $h \in (0,1]$ and a symbol $a(x,\xi;h) \in \mathcal{S}' (\mathbb{R}^d \times \mathbb{R}^d) \otimes \mathbb{M}_n$, we define 
the Weyl quantization of $a$ as the operator 
\begin{align}\label{eq:weylquanth}
    \Op _h (a) \psi (x) :=
    \frac{1}{(2\pi h)^d} \int_{\mathbb{R}^{2d}}
    e^{i(x-y)\cdot \xi/h}
    a(\frac{x+y}{2}, \xi;h) \psi (y) dy d\xi,
    \qquad
    \psi \in \mathcal{S} (\mathbb{R}^d) \otimes \mathbb{C}^n.
\end{align}
When $a$ is 
polynomial in $\xi$, it follows that $\Op_h (a)$ is a differential operator. We denote by $\Op a=\Op_1 a$ for $h=1$. We note that for general $a \in \mathcal{S}'$, the integral in \eqref{eq:weylquanth} is defined in the distributional sense. In this paper, we deal exclusively with \emph{smooth} symbols $a$ (still in $\mathcal{S}'$) so that the Lebesgue integrals are well defined.

\medskip
\noindent{\bf Order functions and symbol classes}.
See \cite{Bony} and \cite[Chapter 7]{DS}.
For $(x, \xi) = X\in \mathbb{R}^{2d}$, we define 
$
    \aver{X} := \sqrt{1 + |X|^2}.
$
A function $\fm: \mathbb{R}^{2d} \rightarrow [0,\infty)$ is called an order function if there exist constants $C_0 > 0$, $N_0 > 0$ such that
$\fm(X) \le C_0 \aver{X-Y}^{N_0} \fm(Y)$ for all $X,Y \in \mathbb{R}^{2d}$.
Note that $\aver{X}^p$ and $\aver{X_\pm}$ are order functions for all $p \in \mathbb{R}$, where $X_+ := \max\{X,0\}$ (with the $\max$ defined element-wise) and $X_- := -(-X)_+$.
Moreover, if $\fm_1$ and $\fm_2$ are order functions, then so are $\fm_1^{-1}$ and $\fm_1 \fm_2$.

We say that $a \in S(\fm)$ if
for every $\alpha \in \mathbb{N}^{2d}$, there exists $C_\alpha > 0$ such that $|\partial^\alpha a (X;h)| \le C_\alpha \fm(X)$
for all $X \in \mathbb{R}^{2d}$ and $h \in (0,1]$.
We write $S(\fm^{-\infty})$ to denote the intersection over $s \in \mathbb{N}$ of $S(\fm^{-s})$.
For $\delta \in [0,1]$ and $k \in \mathbb{R}$, we say that $a(X;h) \in S^k_\delta (\fm)$ if for every $\alpha \in \mathbb{N}^{2d}$, there exists $C_\alpha > 0$ such that
\begin{align}\label{eq:symbolm}
    |\partial^\alpha a (X;h)| \le 
    C_\alpha \fm(X) h^{-\delta |\alpha| - k},
\end{align}
uniformly in $X \in \mathbb{R}^{2d}$ and $h \in (0,1]$.
If either $k$ or $\delta$  are omitted, they are assumed to be zero.
We will always write the order function $\fm$ when using these symbol classes.

By \cite[Chapter 7]{DS}, we know that if $a \in S(\fm_1)$ and $b \in S(\fm_2)$, then $\Op_h (c) := \Op_h (a) \Op_h(b)$ is a pseudo-differential operator, with
\begin{align*}
    c (x,\xi;h) =
    (a \sharp_h b) (x,\xi;h) :=
    \Big( e^{i\frac{h}{2} (\partial_x \cdot \partial_\zeta - \partial_y \cdot \partial_\xi)}
    a(x,\xi;h) b(y,\zeta;h) \Big)
    \Big \vert_{y=x,\zeta=\xi}
\end{align*}
and $c \in S(\fm_1 \fm_2)$.
See also Proposition \ref{prop:sharpfinite} for explicit bounds on $c$.


For $m \in \mathbb{Z}$, define the standard Hilbert spaces
\begin{align}\label{eq:Hm}
   \mathcal{H}^m := \{\Psi \in \mathcal{S}' (\mathbb{R}^2) \otimes \mathbb{C}^n \quad | \quad \partial_\alpha \Psi \in L^2(\mathbb{R}^2) \otimes \mathbb{C}^n \quad \forall \ |\alpha| \le m\}.
\end{align}
Following \cite{Bony,Hormander1979,Lerner}, we define 
the H\"ormander class $\sm$ to be the space of symbols $a(x,\xi)$ that satisfy
\begin{align}\label{eq:symbol10}
    |(\partial^\alpha_\xi \partial^\beta_x a) (x,\xi)|
    \le 
    C_{\alpha, \beta} \aver{\xi}^{m - |\alpha|};
    \qquad
    \alpha,\beta \in \mathbb{N}^{d}.
\end{align}
We define $\smeh$ to be the space of \emph{Hermitian-valued} symbols $a(x,\xi)$ that satisfy \eqref{eq:symbol10} and 
\begin{align}\label{eq:elliptic}
a_{\min} (x,\xi) \ge c \aver{\xi}^m - 1
\end{align}
for some $c>0$, where $a_{\min}$ is the smallest singular value of $a$. 

If $\mathcal{A}$ is a symbol class (e.g. $S(\ofn),\sm, \smeh$), we write $A \in \Op (\mathcal{A})$ to mean that $A = \Op (a)$ for some $a \in \mathcal{A}$.
In the case $\mathcal{A} = S (\ofn)$, the notation $A \in \Op_h (S (\ofn))$ means that $A = \Op_h (a)$ for some $a \in S (\ofn)$.

\medskip\noindent{\bf Helffer-Sj\"ostrand formula.}
See  \cite[Chapter 8]{DS}.
Given $f \in \mathcal{C}^\infty_c (\mathbb{R})$, there exists an almost analytic extension $\tilde{f} \in \mathcal{C}^\infty_c (\mathbb{C})$ that satisfies
\begin{align} \label{aae}
    |\bar{\partial} \tilde{f}| \le C_N |\omega|^N, &\quad 
    N \in \{0,1,2,\dots\}; \qquad
    \tilde{f} (\lambda) = f(\lambda), \quad \lambda \in \mathbb{R}.
\end{align}
Here, $z =: \lambda + i\omega$ and $\bar{\partial} := \frac{1}{2} \partial_\lambda + \frac{i}{2} \partial_\omega$.
We now recall \cite[Theorem 8.1]{DS}.
If $H$ is a self-adjoint operator on a Hilbert space, then
\begin{align} \label{HSformula}
    f (H) =
    -\frac{1}{\pi} \int_\mathbb{C} \bar{\partial} \tilde{f} (z) (z-H)^{-1} d^2 z,
\end{align}
where $d^2 z$ is the Lebesgue measure on $\mathbb{C}$.
(\ref{HSformula}) is known as the Helffer-Sj\"ostrand formula,
and we will use it repeatedly in this paper.

\medskip
\noindent{\bf Trace-class operators.}
See  \cite[Chapter 9]{DS}.
Suppose $\fm \in L^1 (\mathbb{R}^{2d})$, and 
$|\partial^\alpha a(x,\xi;h)| \le C_\alpha \fm(x,\xi)$ for all $\alpha \in \mathbb{N}^{2d}$ and $h \in (0,1]$ (meaning that $a \in S(\fm)$).
Then $\Op_h (a)$ is trace-class with
$
    \norm{\Op_h (a)}_1 \le C \max_{|\alpha| \le 2d+1} C_\alpha \norm{\fm}_{L^1}
$
and
\begin{align}\label{eq:ophtrace}
    \Tr \Op_h (a) =
    \frac{1}{(2\pi h)^d}
    \int_{\mathbb{R}^{2d}} \tr a(x,\xi;h) dx d\xi,
\end{align}
where $C$ depends only on $d$ {and $\tr$ is the standard matrix trace}.
To obtain the above equality, we use \cite[Theorem 9.4]{DS} to write
\begin{align*}
    \Tr \Op_h (a(x,\xi;h)) =
    \Tr \Op_1 (a(x,h\xi;h)) 
    =
    \frac{1}{(2\pi)^d}
    \int_{\mathbb{R}^{2d}} \tr a(x,h\xi;h) dx d\xi
    =
    \frac{1}{(2\pi h)^d}
    \int_{\mathbb{R}^{2d}} \tr a(x,\xi;h) dx d\xi.
\end{align*}


The results from \cite{Bony,DS} that were used in this section are stated for scalar symbols. They extend to the matrix-valued case; see \cite{BG}.

\paragraph{Composition calculus}
For $\ofn : \mathbb{R}^{2d} \rightarrow [0,\infty)$ an order function, $u \in S (\ofn)$ and $N \in \mathbb{N}$, define $$\sbd_N(u,\ofn) := \sum_{|\alpha| \le N} \inf \{C > 0 : |\partial^\alpha u| \le C \ofn \}.$$
For $u_1, u_2 \in \mathcal{C}^\infty (\mathbb{R}^{2d}; \mathbb{C}^{n\times n})$ and $N \in \mathbb{N}$, define 
$$\shp_N (u_1, u_2)=(\shp_N (u_1, u_2))(x,\xi) := \sum_{j=0}^{N-1} \Big(\frac{(ih (D_\xi\cdot D_y - D_x \cdot D_\eta)/2)^j}{j!} u_1 (x,\xi) u_2 (y,\eta)\Big) |_{y=x,\eta=\xi}.$$
\begin{proposition}\label{prop:sharpfinite}
Let $\ofn_1, \ofn_2 :\mathbb{R}^{2d} \rightarrow [0,\infty)$ be order functions. 
Then there exists $s \in \mathbb{N}$ such that
for every $N \in \mathbb{N}$ and $\alpha \in \mathbb{N}^{2d}$,
\begin{align*}
    |\partial^\alpha (u_1 \sharp_h u_2 - \shp_N) | \le C_{\alpha,N} \sum_{j=0}^{|\alpha|}\sbd_{2N+s+j} (u_1, \ofn_1) \sbd_{2N+s+|\alpha|-j} (u_2, \ofn_2)  h^{N} \ofn_1 \ofn_2 
\end{align*}
uniformly in $u_1 \in S (\ofn_1)$, $u_2 \in S (\ofn_2)$ and $h \in (0,1]$.
\end{proposition}
\begin{proof}
    We follow arguments from \cite[Chapter 7]{DS} and \cite[Chapters 3 \& 4]{Zworski}.
    Since $$u_1 \sharp_h u_2 = (e^{i\frac{h}{2} (D_\xi\cdot D_y - D_x \cdot D_\eta)} u_1 (x,\xi) u_2 (y,\eta)) |_{y=x,\eta=\xi} =: (e^{ih A(D)} u_1 (x,\xi) u_2 (y,\eta)) |_{y=x,\eta=\xi}$$ and $D^\alpha$ commutes with $e^{ih A(D)}$ (for any $\alpha \in \mathbb{N}^d$), it suffices to show that there exists $s \in \mathbb{N}$ such that for any $N \in \mathbb{N}$ and order function $\ofn : \mathbb{R}^{4d} \rightarrow [0,\infty)$,
    \begin{align*}
        \Big|e^{ih A(D)} u(X)- \sum_{j \le N-1} \frac{(ih A(D))^j}{j!} u(X)\Big| \le C_N \sbd_{2N+s} (u, \ofn) h^{N} \ofn (X)
    \end{align*}
    uniformly in $u \in S (\ofn)$ and $h \in (0,1]$. 
    Here, we use the shorthand $X := (x,y,\xi,\zeta) \in \mathbb{R}^{4d}$.

    Since $A$ is non-degenerate, 
    we can let $A^{-1} (X) = \frac{1}{2} \aver{Q^{-1} X, X}$ be the dual quadratic form on $\mathbb{R}^{4d}$. Then, as stated in \cite[bottom of page 80]{DS}, $e^{ihA(D)} u = K_h * u$, where
        $K_h (x) = C h^{-2d} e^{-i A^{-1} (X)/h}$.
    We then write
    \begin{align*}
        e^{ihA(D)} u = (\chi K_h) * u + ((1-\chi) K_h) * u =: t_1 + t_2, 
    \end{align*}
    where $\chi \in \mathcal{C}^\infty_c (B(0,2))$ is equal to $1$ in $B(0,1)$.
    Set $p:= 2d+1$.
    By \cite[Theorems 3.14 \& 4.16 and their proofs]{Zworski}, we obtain that
    \begin{align*}
        \Big | t_1 (X)-\sum_{j \le N-1} \frac{(ih A(D))^j}{j!} u(X) \Big | \le C_N h^{N} \sum_{|\alpha| \le 2N+p} \sup_{B(X,2)} |\partial^\alpha u| \le C_N \sbd_{2N+p} (u, \ofn) h^{N} \ofn (X). 
    \end{align*}
    For the second term, we directly apply \cite[equation (7.19)]{DS} to obtain that for every $k \in \mathbb{N}$,
    \begin{align*}
        |t_2 (X)| \le C_k h^k \sum_{|\alpha| \le k+p} \norm{\aver{X-\cdot}^{-k-2d} \partial^\alpha u (\cdot)}_{L^1} \le C_k \sbd_{k+p}(u,\ofn) h^{k}\norm{\aver{X-\cdot}^{-k-2d} \ofn (\cdot)}_{L^1}.
    \end{align*}
    Using that $\ofn (Y) \le C \aver{X-Y}^{N_0} \ofn (X)$ for some $N_0 >0$, it follows that
    \begin{align*}
        |t_2 (X)| \le C_k \sbd_{k+p}(u,\ofn) h^{k}\ofn(X) \norm{\aver{\cdot}^{N_0-k-2d}}_{L^1}.
    \end{align*}
    Thus for all $k > N_0 + p$ (so that the above integral is finite),
        $|t_2 (X)| \le C_k \sbd_{k+p}(u,m) h^{k}\ofn (X)$.
\end{proof}
\begin{proposition}\label{prop:CV}
    There exists $N \in \mathbb{N}$ such that $\norm{\Op_h (u)} \le \sbd_{N} (u,1)$ uniformly in $u \in S(1)$ and $h \in (0,1]$.
\end{proposition}
\begin{proof}
    See \cite[Theorem 7.11 and its proof]{DS}.
\end{proof}

\section{Proofs of Propositions \ref{wn}--\ref{prop:3x3}}\label{appendix:deg}
\begin{proof}[Proof of Proposition \ref{wn}]
Recall that by the definition of $\partial R$, $n_+ \cup n_- = \{1,2,\dots,n\}$.
Define the projectors $\Pi_\ell := \psi_\ell \psi_\ell ^*$.
We know that $\sym= \sum_{\ell=1}^n \lambda_\ell \Pi_\ell$, and hence $\sym_z^{-1} = \sum_{\ell=1}^n (z-\lambda_\ell)^{-1} \Pi_\ell$ whenever $\Im z \ne 0$.
It follows from \eqref{eq:sigmaI1} and cyclicity of the trace that
\begin{align}\label{eq:ell12}
    2\pi\sigma_I =
    \frac{i}{8\pi^2} \int_{\partial R} \int_{\mathbb{R}}
    \sum_{\ell_1,\ell_2=1}^n (z-\lambda_{\ell_1})^{-2} (z-\lambda_{\ell_2})^{-1}d \omega
    \eps_{ijk}
    \tr \Pi_{\ell_1} \partial_i \sym\Pi_{\ell_2} \partial_j \sym\nu_k
    d\Sigma.
\end{align}
If $\lambda_{\ell_1} \ne \lambda_{\ell_2}$, then
$$(z-\lambda_{\ell_1})^{-2} (z-\lambda_{\ell_2})^{-1} =
(\lambda_{\ell_2} - \lambda_{\ell_1})^{-2} ((z-\lambda_{\ell_2})^{-1} - (z-\lambda_{\ell_1})^{-1}) - (\lambda_{\ell_2}-\lambda_{\ell_1})^{-1} (z-\lambda_{\ell_1})^{-2}.$$
Thus the integral over $\omega$ vanishes if $\lambda_{\ell_1}-\alpha$ and $\lambda_{\ell_2}-\alpha$ have the same sign, and otherwise equals $2\pi (\lambda_{\ell_2} - \lambda_{\ell_1})^{-2}$.
It follows that
\begin{align*}
    2\pi\sigma_I &=
    \frac{i}{2\pi} \int_{\partial R}
    \sum_{\ell_+ \in n_+,\ell_-\in n_-}
    (\lambda_{\ell_+} - \lambda_{\ell_-})^{-2}
    \eps_{ijk}
    \tr 
    \Pi_{\ell_+} \partial_i \sym\Pi_{\ell_-} \partial_j \sym
    \nu_k
    d\Sigma
    ,
\end{align*}
where we have 
added the contributions of $(\ell_1, \ell_2) \in n_+ \times n_-$ and $(\ell_1, \ell_2) \in n_- \times n_+$ in \eqref{eq:ell12}.
Finally, we observe that
\begin{align*}
    \tr 
    \Pi_{\ell_+} \partial_i \sym\Pi_{\ell_-} \partial_j \sym
    &=
    \tr
    (
    \lambda_{\ell_+}^2 \Pi_{\ell_+}\partial_i \Pi_{\ell_+} \Pi_{\ell_-} \partial_j \Pi_{\ell_+}
    +
    \lambda_{\ell_+}\lambda_{\ell_-} \Pi_{\ell_+}\partial_i \Pi_{\ell_+} \Pi_{\ell_-} \partial_j \Pi_{\ell_-}\\
    &\qquad
    +
    \lambda_{\ell_+}\lambda_{\ell_-} \Pi_{\ell_+}\partial_i \Pi_{\ell_-} \Pi_{\ell_-} \partial_j \Pi_{\ell_+}
    +
    \lambda_{\ell_-}^2 \Pi_{\ell_+}\partial_i \Pi_{\ell_-} \Pi_{\ell_-} \partial_j \Pi_{\ell_-}
    )\\
    &=
    -(\lambda_{\ell_+} - \lambda_{\ell_-})^2
    \tr \psi_{\ell_+} \partial_i \psi_{\ell_+}^* \psi_{\ell_-} \partial_j \psi_{\ell_-}^*\
    =\ 
    (\lambda_{\ell_+} - \lambda_{\ell_-})^2
    \partial_i \psi_{\ell_+}^* \psi_{\ell_-} \psi_{\ell_-}^*\partial_j \psi_{\ell_+},
\end{align*}
and the result is complete.
\end{proof}

\begin{proof}[Proof of Proposition \ref{deg}]
The fact that $\mathcal{R}_{\eps_0} \ne \emptyset$ is an immediate consequence of Sard's Theorem \cite{Sard}.
Given the structure of $\sym$, we 
take $\alpha = 0$ in Theorem \ref{thm:main}, meaning that $z = i \omega$ in
\eqref{eq:sigmaI1}.
It follows that
\begin{align*}
    \sym_z^{-1} = -  (\omega^2+|f|^2)^{-1} 
    (i \omega + \sym),
\end{align*}
from which we conclude that
\begin{align*}
    &\sym_z^{-1} \partial_i \sym_z \sym_z^{-1} \partial_j \sym_z \sym_z^{-1} = 
    -(\omega^2+|f|^2)^{-3}
    (i \omega + \sym) \partial_i \sym(i \omega + \sym) \partial_j \sym(i \omega + \sym)\\
    =&
    -(\omega^2+|f|^2)^{-3}
    \sym\partial_i \sym\sym\partial_j \sym\sym
    \quad + 
    \omega^2 (\omega^2+|f|^2)^{-3}
    \left( 
    \partial_i \sym\partial_j \sym\sym+ \partial_i \sym\sym\partial_j \sym+ \sym\partial_i \sym\partial_j \sym
    \right)
    \ + \dots,
\end{align*}
where we have left out terms that are odd in $\omega$ in the last line.
Carrying out the integral over $\omega$ in \eqref{eq:sigmaI1}, we find that
\begin{align*}
    2 \pi \sigma_I = &
    \frac{i}{8 \pi ^2} \int_{\partial R} \Big(
    -\frac{3 \pi \epsilon_{ijk}}{8 
    |f|^5
    } \tr \sym\partial_i \sym\sym\partial_j \sym\sym
    +
    \frac{\pi \epsilon_{ijk}}{8 
    |f|^3
    } \tr (\partial_i \sym\partial_j \sym\sym+ \partial_i \sym\sym\partial_j \sym+ \sym\partial_i \sym\partial_j \sym)
    \Big)
    \nu_k d\Sigma.
\end{align*}
Simplifying the above using cyclicity of the trace and the fact that $\sym^2 = |f|^2$, 
we have
\begin{align*}
    2 \pi \sigma_I = 
    \frac{i}{16 \pi} \int_{\partial R}
    \frac{\epsilon_{ijk}}{|f|^3} \tr \sym\partial_i \sym\partial_j \sym\nu_k d\Sigma.
\end{align*}
Using commutation relations of the Pauli matrices, we see that
\begin{align*}
    \tr \sym\partial_i \sym\partial_j \sym&= 
    4 i
    \left(
    f_1(\partial_i f_2 \partial_j f_3 - \partial_i f_3 \partial_j f_2)
    + [(1,2,3) \rightarrow (2,3,1)] + [(1,2,3) \rightarrow (3,1,2)]
    \right),
\end{align*}
which we may rewrite as $4i \epsilon_{mnp} f_m \partial_i f_n \partial_j f_p$.
Letting $w \in \mathbb{R}^3$ such that $w_k = \frac{\epsilon_{ijk}}{|f|^3} \tr \sym\partial_i \sym\partial_j \sym$, we have
\begin{align*}
    \nabla \cdot w = 
    \frac{4 i}{|f|^3}
    \tilde{\epsilon}_{ijk} \epsilon_{mnp} \partial_k f_m \partial_i f_n \partial_j f_p
    \left(
    1 - \frac{3 f_m^2}{f_1^2 + f_2^2 + f_3^2}
    \right)
\end{align*}
away from the zeros of $|f|^2$, where the only nonzero entries of $\tilde{\epsilon}_{ijk}$ are $\tilde{\epsilon}_{123} = \tilde{\epsilon}_{231}= \tilde{\epsilon}_{312}=1$.
By symmetry, we know that
$
    \tilde{\epsilon}_{ijk} \epsilon_{mnp} f_m^2 \partial_k f_m \partial_i f_n \partial_j f_p  = 
    \tilde{\epsilon}_{ijk} \epsilon_{mnp} f_n^2 \partial_k f_m \partial_i f_n \partial_j f_p  = \tilde{\epsilon}_{ijk} \epsilon_{mnp} f_p^2 \partial_k f_m \partial_i f_n \partial_j f_p,  
$
meaning that
\begin{align*}
\nabla \cdot w = 
    \frac{4 i}{|f|^3}
    \tilde{\epsilon}_{ijk} \epsilon_{mnp} \partial_k f_m \partial_i f_n \partial_j f_p
    \left(
    1 - \frac{f_m^2 + f_n^2 + f_p^2}{f_1^2 + f_2^2 + f_3^2}
    \right)
    = 
    0.
\end{align*}
Thus
$
    2 \pi \sigma_I = 
    \frac{i}{16 \pi}\int_{\partial R} w \cdot \hat{\nu} d\sym,
$
where $\hat{\nu}$ is the unit vector normal to the surface $\partial R$ and $w$ has zero divergence everywhere except for a finite number of points in $\mathbb{R}^3$.
Hence we can 
deform $\partial R$ in any way so long as it encloses the same zeros of $|f|^3$. 
It follows that
\begin{align*}
    2 \pi \sigma_I = 
    \frac{i}{16 \pi} \int_{\partial S} w \cdot \hat{\nu} d\Sigma;
\ \
    w \cdot \hat{\nu} = 
    \frac{-4i}{|f|^3 \sqrt{\xi^2 + \zeta^2 + y^2}}
    \det
    \left[ {\begin{array}{cccc}
   0 & f_1 & f_2 & f_3 \\
   \xi & \partial_{\xi} f_1 & \partial_\xi f_2 & \partial_\xi f_3 \\
   \zeta & \partial_{\zeta} f_1 & \partial_\zeta f_2 & \partial_\zeta f_3 \\
   y & \partial_{y} f_1 & \partial_y f_2 & \partial_y f_3 \\
  \end{array} } \right].
\end{align*}
Using the geometric interpretation of the determinant, we have
\begin{align*}
    2 \pi \sigma_I &= 
    -\frac{1}{4 \pi} \int_{\partial S} \frac{1}{|f|^3}
    \det
    \left[ {\begin{array}{ccc}
   f_1 & f_2 & f_3 \\
   \partial_{u} f_1 & \partial_u f_2 & \partial_u f_3 \\
   \partial_{v} f_1 & \partial_v f_2 & \partial_v f_3 \\
  \end{array} } \right]
  du \wedge dv,
\end{align*}
where $u = -\hat{\xi} \sin \phi  + \hat{\zeta} \cos \phi $ and 
$v = -\hat{\xi} \cos \theta \cos \phi  - \hat{\zeta} \cos \theta \sin \phi + \hat{y} \sin \theta $,
with $\theta$ and $\phi$ the polar and azimuthal angles that parametrize $\partial S$.
We conclude from \cite[Corollary 14.2.1]{Dubrovin} and continuity of the degree of a map that $2 \pi \sigma_I = -\deg g$,
where $g: \partial S \rightarrow \mathbb{S}^2$ is the \emph{Gauss map} defined by
\begin{align*}
    (y,\xi, \zeta) \mapsto \frac{f(y,\xi,\zeta) - a}{|f(y,\xi,\zeta) - a|}.
\end{align*}
The fact that
$\deg g = \sum_{j=1}^p \sgn \det M_j$
is
a direct consequence of  \cite[Theorems 14.4.3 and 14.4.4]{Dubrovin}.
\end{proof}

\begin{proof}[Proof of Proposition \ref{prop:app}]
    By the bulk-edge correspondence (Corollary \ref{cor:bic}), we can without loss of generality 
    assume $m'(y) \ne 0$ for all $y \in m^{-1} (0)$, and $c'(y) \ne 0$ for all $y\in c^{-1} (0)$.

    Throughout this proof, let $\sym_{\min}$ denote the smallest singular value of $\sym$.
    We continue to use the shorthand $\sigma_I := \sigma_I (H,P,\varphi)$, where $P(x)=P \in \fs(0,1)$ and $\varphi \in \fs(0,1;E_1,E_2)$ is implied.

    \medskip
\noindent{\bf $2 \times 2$ Dirac system \eqref{2x2Dirac}.} We have $H = \Op (\sym)$, with $\sym= \xi \sigma_1 + \zeta \sigma_2 + m(y) \sigma_3$.
Thus it is clear that $\sym\in S^{1}$ with
$
    \sym^2 = \xi^2 + \zeta^2+ m^2 (y) \ge \xi^2 + \zeta^2.
$
Using that $\aver{\xi,\zeta} - |(\xi,\zeta)| \le 1$, it follows that $|\sym_{\min}| \ge \aver{\xi,\zeta} - 1$. 
Now, take $\sym_{\pm} := \xi \sigma_1 + \zeta \sigma_2 \pm m_\pm \sigma_3$, with $E_2 := \min \{|m_+|, |m_-|\}$ and $E_1 := - E_2$.
Since $\sym_{\pm}^2 = \xi^2 + \zeta^2 + m^2_\pm \ge m^2_\pm$,
\hone\ is satisfied.

We now apply Proposition \ref{deg}
with $f = (\xi, \zeta, m(y))$, so that $f^{-1} (0)= \{(y,0,0): m(y) = 0\}$.
Given a point $y \in m^{-1} (0)$, the determinant of the Jacobian of $f$ evaluated at $(y,0,0)$ is $m'(y)$.
Thus it follows that $2 \pi \sigma_I = \frac{1}{2} (\sgn (m_-) - \sgn (m_+))$, as in e.g., \cite{2,3}.

\medskip \noindent {\bf $p$-wave superconductor model \eqref{p}.}
We have $H = \Op (\sym)$, with 
\begin{align*}
    \sym= \Big (\frac{1}{2m} (\xi^2+ \zeta^2) - \mu\Big ) \sigma_1 + 
    c(y) \zeta \sigma_2 + c_0 \xi \sigma_3.
\end{align*}
It follows that $\sym\in S^2$ with, moreover,
$
    \sym^2 = \Big (\frac{1}{2m} (\xi^2+ \zeta^2) - \mu\Big )^2 + c^2 (y) \zeta^2 + c_0^2 \xi^2 \ge \Big (\frac{1}{2m} (\xi^2+ \zeta^2) - \mu\Big )^2,
$
which implies
\begin{align*}
    |\sym_{\min}| \ge \Big \vert \frac{1}{2m} (\xi^2+ \zeta^2) - \mu \Big \vert
    =
    \Big \vert\frac{1}{2m} (\aver{\xi,\zeta}^2 - 1) - \mu \Big \vert.
\end{align*}
Thus there exists $C_1>0$ and a compact set $K \subset \mathbb{R}^2$ such that $|\sym_{\min}| \ge C_1 \aver{\xi, \zeta}^2$ for all $(\xi, \zeta) \in \mathbb{R}^2 \setminus K$.
Moreover, there exists $C_2 >0$ such that $\aver{\xi,\zeta}^2 \le C_2$ and hence
$1+|\sym_{\min}| \ge 1 \ge \frac{1}{C_2} \aver{\xi,\zeta}^2$
for all $(\xi, \zeta) \in K$.
It follows that for $c := \min \{C_1, \frac{1}{C_2}\}$, we have
$
    1+|\sym_{\min}| \ge c \aver{\xi, \zeta}^2$, for $(\xi, \zeta) \in \mathbb{R}^2.
$
Now, take
$
    \sym_\pm := (\frac{1}{2m} (\xi^2+ \zeta^2) - \mu) \sigma_1 + 
    c_\pm \zeta \sigma_2 + c_0 \xi \sigma_3,
$
so that
$
    \sym_\pm^2 = (\frac{1}{2m} (\xi^2+ \zeta^2) - \mu)^2 + c_\pm^2 \zeta^2 + c_0^2 \xi^2.
$
Therefore, $\sym_\pm^2$ is bounded away from zero so that \hone\ holds.

We now apply Proposition \ref{deg} with $f_1 = \frac{1}{2m}(\xi^2 + \zeta^2) - \mu$, $f_2 = \zeta c(y)$, and $f_3 = c_0 \xi$.
The zeros of $f_1^2 + f_2^2 + f_3^2$ are $(y,\xi, \zeta) = (y_1,0, \pm \sqrt{2 m \mu})$, for all $y_1 \in c^{-1} (0)$.
A straightforward calculation reveals that $\sgn \det \partial_m f_n \vert_{(y_1, 0, \pm \sqrt{2 m \mu})} = \sgn c'(y_1)$, hence $2 \pi \sigma_I = \sgn (c_-) - \sgn (c_+)$.


\medskip \noindent {\bf $d$-wave superconductor model \eqref{d}.}
Here, the symbol of $H$ is
\begin{align*}
    \sym= \Big(\frac{1}{2m} (\xi^2 + \zeta^2) - \mu\Big) \sigma_1+
    c_0 (\zeta^2 - \xi^2) \sigma_2 +
    \xi \zeta c(y) \sigma_3.
\end{align*}
We see that $\sym\in S^2$, and
$
    \sym^2 = 
    (\frac{1}{2m} (\xi^2 + \zeta^2) - \mu)^2 +
    c_0^2 (\zeta^2 - \xi^2)^2 +
    \xi^2 \zeta^2 c^2(y)
    \ge
    (\frac{1}{2m} (\xi^2 + \zeta^2) - \mu)^2,
$
hence $\sym$ is elliptic.
As with the $p$-wave superconductor, 
all eigenvalues of 
$$\sym_\pm := \Big(\frac{1}{2m} (\xi^2 + \zeta^2) - \mu\Big) \sigma_1+
    c_0 (\zeta^2 - \xi^2) \sigma_2 +
    \xi \zeta c_\pm \sigma_3$$
are bounded away from zero, hence $\sym$ satisfies \hone.

We now apply Proposition \ref{deg}, with
$f_1 = \frac{1}{2m}(\xi^2 + \zeta^2) - \mu$, $f_2 = c_0(\zeta^2 - \xi^2)$, and $f_3 = \xi \zeta c(y)$. The zeros of $f_1^2 + f_2^2 + f_3^2$ are $(y,\xi, \zeta) = (y_1, \eps_1 \sqrt{m \mu}, \eps_2 \sqrt{m \mu})$, for all $y_1 \in c^{-1} (0)$ and $\eps_1, \eps_2 \in \{-1,1\}$.
(That is, there are four zeros for every $y_1 \in c^{-1} (0)$.)
One can then easily verify that $\sgn \det \partial_m f_n \vert_{(y_1, \eps_1 \sqrt{m \mu}, \eps_2 \sqrt{m \mu})} = \sgn c'(y_1)$. Thus we conclude that $2 \pi \sigma_I = 2 (\sgn (c_-) - \sgn (c_+))$.
\end{proof}

\begin{proof}[Proof of Proposition \ref{prop:3x3}]
    As above, we set $\sym_\pm$  equal to $\sym$ with $f(y)$ replaced by $f_\pm$. A direct calculation then reveals that \hone\ is satisfied for $E_1$ and $E_2$ as defined in the proposition.

To obtain $\sigma_I$, we apply Proposition \ref{wn}, with $\mu > 0$ for concreteness.
Take $\partial R_3 := \{y^2 + \xi^2 + \zeta^2 = r^2\} \subset \mathbb{R}^3$, with $r > 0$ sufficiently large so that 
$|f(y)| \ge f_0$ whenever $|y| \ge r$.
The $\lambda_j$ and $\psi_j$ are differentiable everywhere on $\partial R_3$, except perhaps where 
$\rho = 0$.
But $\{\rho = 0\} \cap \partial R_3 \subset \{\zeta = 0\} \cap \{y^2 + \xi^2 = r^2\}$ has $\sym$-measure zero, so indeed the regularity requirement is satisfied.
Observe that
\begin{align*}
    \partial_\xi \psi_0 = 
    \frac{1}{\kappa}
    \left[ {\begin{array}{c}
   0\\
   0\\
   -1\\
  \end{array} } \right] + c_1 \psi_0,
  \quad
  \partial_\zeta \psi_0 = 
    \frac{1}{\kappa}
    \left[ {\begin{array}{c}
   0\\
   1\\
   0\\
  \end{array} } \right] + c_2 \psi_0,
  \quad
  \partial_y \psi_0 = 
    \frac{1}{\kappa}
    \left[ {\begin{array}{c}
   if'\\
   0\\
   0\\
  \end{array} } \right] + c_3 \psi_0,
\end{align*}
where the $c_i$ are scalar-valued functions.
The terms in the integrand of \eqref{wneq} corresponding to the pair $\{\lambda_0, \lambda_-\}$ are
$\Im ((\psi_-^* \partial_\xi \psi_0)^* \psi_-^* \partial_\zeta \psi_0) =
    -\frac{1}{\kappa^2 \rho^2} f\kappa (\zeta^2+f^2)$, $\Im ((\psi_-^* \partial_\xi \psi_0)^* \psi_-^* \partial_y \psi_0) =
    \frac{1}{\kappa^2 \rho^2} f'\kappa \zeta(\zeta^2+f^2)$, and $\Im ((\psi_-^* \partial_\zeta \psi_0)^* \psi_-^* \partial_y \psi_0) =
    -\frac{1}{\kappa^2 \rho^2} f'\kappa\xi (\zeta^2+f^2)$.
Thus the contribution to $2\pi\sigma_I$ of $\{\lambda_0, \lambda_-\}$ is
\begin{align*}
\frac{i}{2\pi} &\int_{\partial R_3} \eps_{ijk} \partial_i \psi_{0}^* \psi_{-} \psi_{-}^* \partial_j \psi_{0} \nu_k d\Sigma\\
&=
\frac{i}{2\pi} \int_{\partial R_3} 2i(\Im
((\psi_-^* \partial_\xi \psi_0)^* \psi_-^* \partial_\zeta \psi_0)\nu_y -
\Im((\psi_-^* \partial_\xi \psi_0)^* \psi_-^* \partial_y \psi_0)\nu_\zeta +
\Im((\psi_-^* \partial_\zeta \psi_0)^* \psi_-^* \partial_y \psi_0)\nu_\xi
) d\Sigma\\
&=
    \frac{1}{\pi} \int_{\partial R_3}
    \frac{\zeta^2+f^2}{\kappa^2 \rho^2} 
    (fy+f'\zeta^2+f'\xi^2)
    d\Sigma=
    \frac{1}{2\pi} \int_{\partial R_3}
    \frac{1}{\kappa^4} 
    (fy+f'\zeta^2+f'\xi^2)
    d\Sigma.
\end{align*}
Considering the case $f_- < 0 < f_+$, we can without loss of generality take $f(y) = y$ on $\partial R_3$, so that the integral becomes
$
    \frac{1}{2\pi} \int_{\partial R_3} \frac{1}{\kappa^2} d\Sigma= 2.
$
One verifies that $\psi_-^* \partial_\xi \psi_+ = \psi_-^* \partial_\zeta \psi_+ = \psi_-^* \partial_y \psi_+ = 0$, meaning there is no contribution to the edge current from the pair $\{\lambda_+, \lambda_-\}$.

We conclude that $2 \pi \sigma_I = 2$, which in practice corresponds to two (observed) eastward-moving asymmetric modes along the equator. A similar calculation shows that $2 \pi \sigma_I = 2$ also when $\mu < 0$, so that the edge current is independent of the regularization parameter $\mu\not=0$.
If instead we assume that $f_+ < 0 < f_-$ (with south pole pointing upwards), then $2 \pi \sigma_I = -2$.
\end{proof}

{\small
\bibliographystyle{siam}
\bibliography{refs} 
}

\end{document}